\documentclass[a4paper,leqno]{amsart}

\usepackage{latexsym}
\usepackage[english]{babel}
\usepackage{fancyhdr}
\usepackage[mathscr]{eucal}
\usepackage{amsmath}
\usepackage{mathrsfs}
\usepackage{amsthm}
\usepackage{amsfonts}
\usepackage{amssymb}
\usepackage{amscd}
\usepackage{bbm}
\usepackage{graphicx}
\usepackage{graphics}
\usepackage{latexsym}
\usepackage{color}
\usepackage{calc}

\newcommand{\ud}{\mathrm{d}}

\newcommand{\xx}{\mathbf{x}}
\newcommand{\yy}{\mathbf{y}}
\newcommand{\zz}{\mathbf{z}}
\newcommand{\pp}{\mathbf{p}}
\newcommand{\qq}{\mathbf{q}}

\newcommand{\ii}{\mathrm{i}}
\newcommand{\cH}{\mathcal{H}}

\newcommand{\op}{
  \mathop{
    \vphantom{\bigoplus} 
    \mathchoice
      {\vcenter{\hbox{\resizebox{\widthof{$\displaystyle\bigoplus$}}{!}{$\boxplus$}}}}
      {\vcenter{\hbox{\resizebox{\widthof{$\bigoplus$}}{!}{$\boxplus$}}}}
      {\vcenter{\hbox{\resizebox{\widthof{$\scriptstyle\oplus$}}{!}{$\boxplus$}}}}
      {\vcenter{\hbox{\resizebox{\widthof{$\scriptscriptstyle\oplus$}}{!}{$\boxplus$}}}}
  }\displaylimits 
}

\theoremstyle{plain}
\newtheorem{theorem}{Theorem}[section]
\newtheorem{lemma}[theorem]{Lemma}
\newtheorem{corollary}[theorem]{Corollary}
\newtheorem{proposition}[theorem]{Proposition}

\theoremstyle{definition}

\newtheorem{remark}[theorem]{Remark}
\newtheorem*{remark*}{Remark}

\numberwithin{equation}{section}

\begin{document}

\title[Zero-range interaction for the bosonic trimer at unitarity]
{Models of zero-range interaction \\ for the bosonic trimer at unitarity}

\author[A.~Michelangeli]{Alessandro Michelangeli}
\address[Alessandro Michelangeli]{Institute for Applied Mathematics and Hausdorff Center of Mathematics, University of Bonn \\ Endenicher Allee 60 \\ 
D-53115 Bonn (Germany).}
\email{michelangeli@iam.uni-bonn.de}


\begin{abstract}
We present the mathematical construction of the physically relevant quantum Hamiltonians for a three-body systems consisting of identical bosons mutually coupled by a two-body interaction of zero range. For a large part of the presentation, infinite scattering length will be considered (the unitarity regime). The subject has several precursors in the mathematical literature. We proceed through an operator-theoretic construction of the self-adjoint extensions of the minimal operator obtained by restricting the free Hamiltonian to wave-functions that vanish in the vicinity of the coincidence hyperplanes: all extensions thus model an interaction precisely supported at the spatial configurations where particles come on top of each other. Among them, we select the physically relevant ones, by implementing in the operator construction the presence of the specific short-scale structure suggested by formal physical arguments that are ubiquitous in the physical literature on zero-range methods. This is done by applying at different stages the self-adjoint extension schemes a la Kre{\u\i}n-Vi\v{s}ik-Birman and a la von Neumann. We produce a class of canonical models for which we also analyse the structure of the negative bound states. Bosonicity and zero range combined together make such canonical models display the typical Thomas and Efimov spectra, i.e., sequence of energy eigenvalues accumulating to both minus infinity and zero. We also discuss a type of regularisation that prevents such spectral instability while retaining an effective short-scale pattern. Beside the operator qualification, we also present the associated energy quadratic forms. We structured our analysis so as to clarify certain steps of the operator-theoretic construction that are notoriously subtle for the correct identification of a domain of self-adjointness.
\end{abstract}

\date{\today}




\thanks{The author is grateful to S.~Albeverio, A.~Ottolini, and A.~Trombettoni for fruitful and instructive discussions on this subject, as well as to the Italian National Institute for Higher Mathematics (INdAM) that funded the third meeting \emph{Mathematical Challenges of Zero-Range Physics: rigorous results and open problems}, co-organised by the author, where many enlightening exchanges with the distinguished participants inspired the early phase of this project. This work is partially supported by the Alexander von Humboldt foundation.}

\maketitle


\newpage

\tableofcontents

\section{Introduction. A plurality of approaches and models.}
\label{sec:intro}

We are concerned in this work with a class of models for a three-dimensional quantum system of this kind: three non-relativistic, identical bosons are coupled among themselves by means of an isotropic two-body interaction of zero spatial range and, for the main part of our analysis, with infinite scattering length. The interaction does not couple the spins.

We shall discuss in particular which models are mathematically well-posed, besides being physically meaningful, which leads to an amount of very instructive subtleties.

It is fair to say that the system under consideration has undergone various phases of interests over the decades, both in the physical and in the mathematical literature, until the present days. Originally, and also without imposing the bosonic symmetry, it emerged as the typical picture for interacting nucleons in early nuclear physics, at the scale of which the inter-particle interaction may well be considered of zero range as compared to the atomic scales. Instead, in more recent times it has been a system of interest in cold atom physics, given the modern experimental advances in inducing effective zero-range interactions in a Bose gas or in heteronuclear gaseous mixtures by means of sophisticated Feschbach-resonance techniques.

As our perspective here is mainly mathematical, even if driven by strong physical inspiration, it is worth stressing an important and long lasting differences of the approaches.

\emph{Physical} investigations of the quantum three-body problem with zero-range interaction have always had as primary interest the characterisation of the \emph{bound states} of the system. To this aim, at least in the more modern literature (given its vastness, we refer to the recent reviews \cite{Braaten-Hammer-2006,Naidon-Endo-Review_Efimov_Physics-2017}), the eigenvalue problem is invariably set up in terms of the free Hamiltonian (all in all particles subject to a zero-range interaction are meant to move as free bodies except when they come on top of each other), with the constraint that the three-body eigenfunction must display the `physical' short-range asymptotics 
\begin{equation}\label{eq:preBP}
 \psi(\xx_1,\xx_2,\xx_3)\;\sim\;\frac{1}{|\xx_i-\xx_j|}-\frac{1}{a}\qquad \textrm{as }|\xx_i-\xx_j|\to 0\,,
\end{equation}
where $a$ is the $s$-wave scattering length in each two-body channel. The behaviour \eqref{eq:preBP} was identified by Bethe and Peierls in 1935 \cite{Bethe_Peierls-1935} as the actual leading behaviour of eigenfunctions with `contact' interaction. Next, solutions are obtained, with an ad hoc analysis applicable to the eigenvalue problem only, and not to the generality of states in the domain of the underlying Hamiltonian, by reducing the three-body \emph{eigenfunction} to a convenient triple of two-body channel `Faddeev components', in a combination of which that encodes the possible bosonic or fermionic symmetry, where each Faddeev component is a function of one pair of internal Jacobi coordinates. In the case of three identical bosons,
\begin{equation}
 \psi(\xx_1,\xx_2,\xx_3)\;=\;\chi(\xx_{12},\xx_{12,3})+\chi(\xx_{23},\xx_{23,1})+\chi(\xx_{31},\xx_{31,2})\,,
\end{equation}
where
\begin{equation}
 \xx_{ij}\;=\;\xx_j-\xx_i\,,\qquad \xx_{ij,k}\;=\;\xx_k-\frac{\xx_i+\xx_j}{2}\,.
\end{equation}
Based on the Faddeev equations formalism for the three-body system \cite{Faddeev-1963-eng-1965-3body,Fedorov-Jensen-1993}, the original problem is thus boiled down to a single Faddeev component $\chi$. At this level the problem is conveniently separable upon switching from Jacobi to hyper-radial coordinates and expanding $\chi$ into definite angular momentum terms, and in each sector of definite angular symmetry the problem becomes tractable analytically and numerically.

The above line of reasoning, in fact encompassing a multitude of similar variants, is most presumably due to an original idea of Landau, elaborated in the mid 1950's by Skornyakov and Ter-Martirosyan \cite{TMS-1956} in a famous study of the three-body quantum system with \emph{zero-range} interactions. (Actually, \cite{TMS-1956} predates by a couple of years Faddeev's first work \cite{Faddeev-scattering-1960} on the three-body scattering theory, and makes use of Green's function methods. Then in \cite{Faddeev-scattering-1960} Faddeev showed that the equation identified by Skornyakov and Ter-Martirosyan for solving the three-body eigenvalue problem could be recovered in the formal limit of zero interaction range from the ordinary scheme of Faddeev equations.)

In atomic physics the approach sketched above is the basis of what one has customarily referred to since then as `zero-range methods' \cite{Demkov-Ostrovskii-book}. The same approach resurfaced in the early 1970's by Efimov \cite{Efimov-1971,Efimov-1973} in his famous work on quantum three-body systems with \emph{finite-range} two-body interactions (with important precursors such as Macek \cite{Macek-1968} in the usage of hyper-radial equations for three-body energy levels). Efimov's analysis established a reference for the subsequent literature on cold-atom few-body systems.

The catch here is that such a physical scheme is solid when the inter-particle interactions are realised, say, by potentials $V_{ij}$ that are sufficiently regular and have short range, thereby making the underlying three-body Hamiltonian 
\begin{equation}
 -\frac{1}{2m_1}\Delta_{\xx_1}-\frac{1}{2m_2}\Delta_{\xx_2}-\frac{1}{2m_3}\Delta_{\xx_3}+V_{12}(\xx_1-\xx_2)+V_{23}(\xx_2-\xx_3)+V_{13}(\xx_1-\xx_3)
\end{equation}
(in units $\hbar=1$) unambiguously realised as a self-adjoint operator on the three-body Hilbert space, and thus giving rise to a well-posed set of Faddeev equations. At zero range, instead, the model is formally thought of as 
\begin{equation}\label{eq:fistFormalHamilt}
 -\frac{1}{2m_1}\Delta_{\xx_1}-\frac{1}{2m_2}\Delta_{\xx_2}-\frac{1}{2m_3}\Delta_{\xx_3}+\mu_{12}\delta(\xx_1-\xx_2)+\mu_{23}\delta(\xx_2-\xx_3)+\mu_{13}\delta(\xx_1-\xx_3)
\end{equation}
(for some coupling constants $\mu_{ij}$): as \eqref{eq:fistFormalHamilt} is not an ordinary Schr\"{o}dinger operator, for it Faddeev components of the three-body eigenfunctions and the corresponding Faddeev equations do not make sense strictly speaking, but for a formal limit of zero interaction range.

In short, physical zero-range methods determine eigenfunctions and eigenvalues of a \emph{formal Hamiltonian that otherwise remains unspecified}.

The signature of a possible remaining ambiguity of the physical approach is the emergence of an unphysical continuum of eigenvalues, an occurrence that depends on the masses, the attractive or repulsive nature of the interaction, and the bosonic or fermionic exchange symmetry in \eqref{eq:fistFormalHamilt}. When this happens, an (infinite) discrete set of bound states is selected by imposing an additional restriction to the admissible eigenfunctions. Such restriction may be suitably interpreted as a three-body short-range boundary condition. This occurrence was initially observed by Skornyakov \cite{Skornyakov-1959} right after his joint work \cite{TMS-1956} with Ter-Martirosyan, and was analysed by Danilov \cite{Danilov-1961} who selected the admissible solutions in the spirit of the additional experimental three-body parameter proposed at the same time by Gribov \cite{Gribov-1959}. That choice was soon after justified on more rigorous operator-theoretic grounds by Faddeev and Minlos \cite{Minlos-Faddeev-1961-1,Minlos-Faddeev-1961-2}. (It is actually remarkable that such  Russian key contributions all span a fistful of years, from the work \cite{TMS-1956} by Skornyakov and Ter-Martirosyan in 1956 to the period 1959-1961 with the works by Skornyakov \cite{Skornyakov-1959}, Gribov \cite{Gribov-1959}, Danilov \cite{Danilov-1961}, Faddeev \cite{Faddeev-scattering-1960}, and Faddeev and Minlos \cite{Minlos-Faddeev-1961-1,Minlos-Faddeev-1961-2}.) The possible necessity of an additional three-body parameter and its physical interpretation have become by now a standard picture in the physical literature of cold atoms in the zero-range regime \cite[Sect.~4]{Naidon-Endo-Review_Efimov_Physics-2017}.

\emph{Mathematical} investigations of the quantum three-body problem with zero-range interaction, on the other hand, have pursued over the decades a different programme: to \emph{qualify first the Hamiltonian} of the system, as an explicitly declared self-adjoint operator on Hilbert space, through its operator or form domain of self-adjointness and its action on each function of the domain, \emph{and only after to analyse the spectral properties}.

This conceptual scheme was brought up first in the already mentioned seminal works by Faddeev and Minlos \cite{Minlos-Faddeev-1961-1,Minlos-Faddeev-1961-2}, which were deeply mathematical in nature. There, rigorous Hamiltonians of contact interaction were proposed as suitable self-adjoint extensions of the symmetric operator
\begin{equation}\label{eq:freeHamiltRestricted}
 \Big(-\frac{1}{m_1}\Delta_{\xx_1}-\frac{1}{m_2}\Delta_{\xx_2}-\frac{1}{m_2}\Delta_{\xx_2}\Big)\Big|_{C^\infty_0((\mathbb{R}^3_{\xx_1}\times\mathbb{R}^3_{\xx_2}\times\mathbb{R}^3_{\xx_3})\setminus\Gamma)}\,,
\end{equation}
namely the free three-body Hamiltonian restricted on smooth functions that are compactly supported away from the `coincidence manifold'
\begin{equation}
 \Gamma\;:=\;\bigcup_{i,j}\Gamma_{ij}\,,\qquad \Gamma_{ij}\;:=\;\{(\xx_1,\xx_2,\xx_3)\,|\,\xx_i=\xx_j\}\,.
\end{equation}
The motivation is that any such extension encodes by construction a singular interaction only `supported' at the points of $\Gamma$. (Such a scheme lied on the very same footing as the analogous rigorous construction of \emph{two-body} zero-range interaction Hamiltonians, initially proposed in 1960 by Berezin and Faddeev \cite{Berezin-Faddeev-1961}.) In order for the analysis to produce physically meaningful results, the actual extensions of \eqref{eq:freeHamiltRestricted} to be considered are only those defined on domains of self-adjointness consisting of wave-functions that display the Bethe-Peierls short-range asymptotics \eqref{eq:preBP}.

All this has been then specialised among various lines, among which:
\begin{itemize}
 \item a more operator-theoretic line in the Faddeev-Minlos spirit, developed from the mid 1980's to the recent years by Minlos (also in collaboration with Menlikov, Mogilner, and Shermatov) \cite{Minlos-1987,Minlos-Shermatov-1989,mogilner-shermatov-PLA-1990,Menlikov-Minlos-1991,Menlikov-Minlos-1991-bis,Minlos-TS-1994,Shermatov-2003,Minlos-2011-preprint_May_2010,Minlos-2010-bis,Minlos-2012-preprint_30sett2011,Minlos-2014-I_RusMathSurv,Minlos-2014-II_preprint-2012}, with also recent contributions by Yoshitomi \cite{Yoshitomi_MathSlov2017}, and by the present author in collaboration with Ottolini \cite{MO-2016,MO-2017}, and with Becker and Ottolini \cite{BMO-2017};
 \item a line exploiting quadratic forms methods, initiated at the end of the 1980's by Dell'Antonio, Figari, and Teta and mainly developed in the following decades by an Italian community \cite{Teta-1989,dft-Nparticles-delta,DFT-proc1995,Finco-Teta-2012,CDFMT-2012,michelangeli-schmidbauer-2013,Correggi-Finco-Teta-2015_N+1,CDFMT-2015,Basti-Teta-2015,MP-2015-2p2,Basti-Figari-Teta-Rendiconti2018} (the works \cite{CDFMT-2012,michelangeli-schmidbauer-2013,CDFMT-2015} being co-authored by the present author), with also recent contributions by Moser and Seiringer \cite{Moser-Seiringer-2017,Moser-Seiringer-2018-2p2};
 \item a side line by Pavlov and his school \cite{Kuperin-Makarov-Merk-Motovilov-Pavlov-1989-JMP1990,Makarov-Melezhik-Motovilov-1995}, retaining the same ideas, but aimed at rigorously constructing variants of the formal Hamiltonian \eqref{eq:fistFormalHamilt} for particles with spin, and a spin-spin contact interaction;
 \item an extremely interesting, not-much-developed-yet line of constructing (three-dimensional) three-body Hamiltonians with zero-range interactions as rigorous limits, in the resolvent sense, of ordinary Schr\"{o}dinger operators with potentials that scale up to a delta profile -- an idea discussed first by Albeverio, H\o{}egh-Krohn, and Wu \cite{Albe-HK-Wu-1981} in the early 1980's (one-dimensional counterpart results have been recently established in \cite{BastiEtAl2018-1d-resLim,Griesemer-Hofacker-Linden-2019}). 
\end{itemize}

For what exposed so far, it is clear that the physical and the mathematical branches of the literature on the quantum three-body problem on point interaction, albeit very deeply cross-intersecting, are not immediately transparent to each other. The rigorous definition of the self-adjoint Hamiltonian is much more laborious than the formal diagonalisation made by physicists, and unavoidably requires the analysis of technical features of the Hamiltonian other than the `observable energy levels'. Besides, the Hamiltonians of interest not having the form of a Schr\"{o}dinger operator, the mathematical analysis faces the lack of various powerful tools from Schr\"{o}dinger operator theory.

Furthermore, the implementation of the Bethe-Peierls asymptotics \eqref{eq:preBP}, a crucial step of the mathematical modelling, yields various technical difficulties.

First, \eqref{eq:preBP} is a \emph{point-wise} asymptotics and need be understood as an expansion in a precise \emph{functional} sense in order to be meaningfully implemented in the operator-theoretic construction of the Hamiltonian.

Next, there is an arbitrariness in the modelling as to prescribing the Bethe-Peierls condition for \emph{all} the functions of the desired domain of self-adjointness, or possibly just for a meaningful \emph{subspace}, e.g., the eigenfunctions only.

In addition, once a realisation of the minimal operator \eqref{eq:freeHamiltRestricted} is found that fulfills the Bethe-Peierls condition, a possibility that one encounters is that this is only a symmetric operator with a variety of self-adjoint extensions, so that another parameter must be introduced to label each extension beside the given scattering length $a$, in complete analogy to the three-body parameter of the physicists.

Another possibility is that after implementing the Bethe-Peierls asymptotics, the resulting candidate Hamiltonian, be it already self-adjoint or not, is unbounded from below (beside being obviously unbounded above, as is the initial operator \eqref{eq:freeHamiltRestricted}). That multi-particle quantum models of zero-range interaction may be such is known since when Thomas in 1935 \cite{Thomas1935}, modelling the tritium as if the range of the interaction was exactly zero, showed that the scattering of the proton over the two neutrons would result in an infinity of bound states accumulating at minus infinity (`Thomas collapse', in the sense of `fall of the particles to the centre'), and this is well familiar in modern cold atom theoretical physics. Yet, this complicates the mathematical treatment, for instance making the quadratic form approach unsuited.

Related to that, one is then also concerned with producing meaningful regularisations of those models obtained along the conceptual path described above, where the spectral instability is removed and yet certain relevant features of the effective Hamiltonian are retained.

With this work we provide a comprehensive and up-to-date overview of all such instances, and in particular a systematic discussion of the technical procedures for the rigorous construction of self-adjoint Hamiltonians of physical relevance. This also allows us to clarify certain steps of the operator-theoretic construction that are notoriously subtle for the correct identification of a domain of self-adjointness.

Our main results, Theorems \ref{thm:generalclassification}, \ref{thm:globalTMSext}, \ref{thm:H0beta}, \ref{thm:spectralanalysis}, and \ref{thm:regularised-models}, present respectively:
\begin{itemize}
 \item the general classification of all self-adjoint realisations of the minimal operator \eqref{eq:freeHamiltRestricted} (a vast class that of course includes also physically non-relevant operators, i.e., realisations characterised by non-local boundary conditions), 
 \item the characterisation of all those extensions displaying the physical short-scale structure for the functions of their domains,
 \item the rigorous construction of a class of canonical models with the physical short-scale structure, and their spectral analysis,
 \item the counterpart for a class of regularised models where the instability is cured at an effective level.
\end{itemize}

The bosonic trimer with zero-range interaction has a natural parameter to be declared in the first place, the scattering length $a$ of the two-body interaction. It is the above-mentioned parameter governing the short-scale asymptotics \eqref{eq:preBP}. Whereas throughout our general discussion on physically relevant extensions we shall keep $a$ generic, for a sharper presentation the final construction of the canonical models is done in the regime $a=\infty$. In physics this is referred to as the `unitary regime', and many-body systems with two-body interaction of infinite scattering length are customarily called `unitary gases' \cite{Castin-Werner-2011_-_review} (for the connection with the optical theorem in which the choice $a=\infty$ maximises the scattering amplitude, and the fact that in turn the optical theorem is a consequence of the unitarity of the quantum evolution). The unitary regime is surely the physically most relevant one, for its applications in cold atom physics and its universality properties: we shall then stay in this regime for a large part of our analysis.

On a more technical level, appropriate self-adjoint extension schemes are needed along the discussion. As the minimal operator \eqref{eq:freeHamiltRestricted} is non-negative, it is natural to apply to it the Kre{\u\i}n-Vi\v{s}ik-Birman extension scheme for semi-bounded symmetric operators \cite{GMO-KVB2017}, and in fact for the specific problem under consideration this turns out to be more informative than the (equivalent) extension scheme a la von Neumann \cite[Sect.~X.1]{rs2}. Yet, at a later stage, when the implementation of the physical short-scale structure only produces symmetric extensions, their self-adjoint realisations, namely the final Hamiltonians of interest, are to be found via von Neumann's theory, because already the symmetric operator one starts from is unbounded from below, hence the Kre{\u\i}n-Vi\v{s}ik-Birman is not applicable.

Once mathematically well-posed (i.e., self-adjoint) and physically meaningful Hamiltonians are constructed, it is fairly manageable to express their quadratic forms, as we do in the sequel. Of course, as in several precursors of the present work, one can revert the order and study first a given quadratic form, typically selected by a physically grounded educated guess, proving that it actually represent a self-adjoint operator. What escapes such approach is the systematic classification of all extensions of interest: the standard classification theorems, indeed, are essentially formulated as operator classifications.

To conclude, there are surely various interesting directions along which it would be desirable to continue this study. To mention some of the most attractive ones, a more explicit theoretic dictionary between this mathematical approach and the physical zero-range methods, the characterisation of the quantum dynamics under the considered Hamiltonians, and an extension of such models to many-body systems with zero-range interaction.

\bigskip

\textbf{Notation.}  Beside an amount of fairly standard notation, as well as further convenient shorthand that will be introduced in due time, we shall adopt the following conventions throughout.

\hspace{-0.7cm}\begin{tabular}{ ccl } 
 $\pp$ & & three-dimensional variable (bold face) \\
 $p$ & & one-dimensional variable (italics) \\
 $\overline{z}$ & & complex conjugate of $z\in\mathbb{C}$ \\
 $\langle\cdot,\cdot\rangle$ & & Hilbert scalar product, or pairing $\int_{\mathbb{R}^d}\overline{f}g$, anti-linear in the first entry \\
 $\|v\|$ & & Hilbert space norm of the vector $v$ \\
 $H^s(\mathbb{R}^d)$ & & Sobolev space of order $s\in\mathbb{R}$ \\
 $C^\infty_0(\Omega)$ & & space of smooth functions with compact support inside the open $\Omega\subset\mathbb{R}^d$ \\
 $\mathcal{D}(S)$ & & operator domain of the operator $S$ \\
 $\mathcal{D}[S]$ & & quadratic form domain of the operator $S$ \\
 $\mathcal{D}[q]$ & & domain of the quadratic form $q$ \\
 $S[f]$ & & evaluation of the quadratic form of $S$ on the element $f\in\mathcal{D}[S]$ \\
 $\overline{S}$ & & operator closure of the operator $S$ \\
 $S^*$, $S^\star$ & & adjoint of $S$ (different notation depending on the reference Hilbert space) \\
 $\mathfrak{m}(S)$ & & bottom of the symmetric operator $S$: $\mathfrak{m}(S)=\inf_{f\in\mathcal{D}(S)}\,\langle f,Sf\rangle/\|f\|^2$ \\
 $\mathbbm{1}$ & & identity operator, acting on the space that is clear from the context \\
 $\mathbbm{O}$ & & zero operator, acting on the space that is clear from the context \\
 $\mathbf{1}_K$ & & characteristic function of the set $K$ \\
 $\delta(\xx)$ & & Dirac delta distribution centred at $\xx=\mathbf{0}$ \\
 $\widehat{f}$ & & Fourier transform of $f$, with convention $\widehat{f}(k)=(2\pi)^{-\frac{d}{2}}\!\int_{\mathbb{R}^d}e^{-\ii k \omega}f(\omega)\ud \omega$ \\
 $f^\vee$ & & inverse Fourier transform of $f$ \\
 $f\approx g$ & & $c^{-1}|g(\omega)|\leqslant|f(\omega)|\leqslant c|g(\omega)|$ for some $c>0$ and all admissible $\omega$ \\
 $\dotplus$ & & direct sum between vector spaces \\
 $\oplus$ & & (if referred to operators) reduced direct sum of operators \\
 $\oplus$ & & (if referred to vector spaces) Hilbert orthogonal direct sum \\
 $\boxplus$ & & Hilbert orthogonal direct sum of non-closed subspaces \\
\end{tabular}

Unless when it becomes relevant to emphasize that, we shall tacitly understand all identities $f=g$ between measurable functions in the sense of almost everywhere identities.

\section{General extension scheme and admissible Hamiltonians}\label{sec:generalextscheme}

\subsection{The minimal operator}~

In order to discuss realisations of the formal Hamiltonian \eqref{eq:fistFormalHamilt} as self-adjoint extensions of \eqref{eq:freeHamiltRestricted}, one factors out the translation invariance by introducing the centre of mass and the internal coordinates
\begin{equation}
 \yy_{\mathrm{c.m.}}\;:=\;\frac{\xx_1+\xx_2+\xx_3}{3}\,,\qquad \yy_1\;:=\;\xx_1-\xx_3\,,\qquad \yy_2\;:=\;\xx_2-\xx_3\,,
\end{equation}
and by re-writing
\begin{equation}
  -\frac{1}{2m}\Delta_{\xx_1}-\frac{1}{2m}\Delta_{\xx_2}-\frac{1}{2m}\Delta_{\xx_3}\;=\;-\frac{1}{6m}\Delta_{\yy_{\mathrm{c.m.}}}+\frac{1}{m}\mathring{H}\,,
\end{equation}
where
\begin{equation}\label{eq:Hring-initial}
 \mathring{H}\;:=\;-\Delta_{\yy_1}-\Delta_{\yy_2}-\nabla_{\yy_1}\cdot\nabla_{\yy_2}\,.
\end{equation}

In absolute coordinates, three-body wave-functions $\Psi(\xx_1,\xx_2,\xx_3)$ are bosonic, namely invariant under exchange of any pair of variables, hence under the corresponding transformation of the internal coordinates, according to the following scheme:
\begin{equation}\label{eq:bosonic_transformations}
 \begin{array}{lcl}
  \begin{cases}
   \xx_1\leftrightarrow\xx_2 \\
   \xx_3\textrm{ fixed}
  \end{cases} & \Leftrightarrow\;\;\; &
  \begin{cases}
   \yy_1\leftrightarrow\yy_2 \\
   \yy_2-\yy_1\leftrightarrow-(\yy_2-\yy_1)
  \end{cases} \\ \\
   \begin{cases}
   \xx_1\leftrightarrow\xx_3 \\
   \xx_2\textrm{ fixed}
  \end{cases} & \Leftrightarrow\;\;\; &
  \begin{cases}
   \yy_1\leftrightarrow -\yy_1 \\
   \yy_2\leftrightarrow \yy_2-\yy_1 
  \end{cases} \\ \\
   \begin{cases}
   \xx_2\leftrightarrow\xx_3 \\
   \xx_1\textrm{ fixed}
  \end{cases} & \Leftrightarrow\;\;\; &
  \begin{cases}
   \yy_1\leftrightarrow -(\yy_2-\yy_1) \\
   \yy_2\leftrightarrow -\yy_2\,. 
  \end{cases}
 \end{array}
\end{equation}
Transformations \eqref{eq:bosonic_transformations} clearly preserve the centre-of-mass variable.
Therefore, bosonic symmetry selects, within the Hilbert space of the internal coordinates
\begin{equation}
 \cH\;:=\;L^2(\mathbb{R}^3\times\mathbb{R}^3,\ud\yy_1\ud\yy_2)\,,
\end{equation}
the `bosonic sector', namely the Hilbert subspace
\begin{equation}\label{eq:Hbosonic}
 \begin{split}
 \cH_\mathrm{b}\;\equiv\;\;& L^2_\mathrm{b}(\mathbb{R}^3\times\mathbb{R}^3,\ud\yy_1\ud\yy_2) \\
  :=\;&
 \left\{\!\!\!
 \begin{array}{c}
  \psi\in L^2(\mathbb{R}^3\times\mathbb{R}^3,\ud\yy_1\ud\yy_2)\,\textrm{ such that} \\
  \psi(\yy_1,\yy_2)=\psi(\yy_2,\yy_1)=\psi(-\yy_1,\yy_2-\yy_1)
 \end{array}
 \!\!\!\right\}
 \end{split}
\end{equation}
in the sense of almost-everywhere identities between square-integrable functions.

The meaningful problem is then to characterise the self-adjoint extensions, with respect to $\cH_\mathrm{b}$ of the densely defined, closed, and symmetric operator
\begin{equation}\label{eq:domHring-initial}
\begin{split} 
 \mathcal{D}(\mathring{H})\;&:=\;\cH_\mathrm{b}\:\cap\: H^2_0\big( (\mathbb{R}^3_{\yy_1}\times\mathbb{R}^3_{\yy_2})\setminus\Gamma\big) \\
 \mathring{H}\;&:=\;-\Delta_{\yy_1}-\Delta_{\yy_2}-\nabla_{\yy_1}\cdot\nabla_{\yy_2}\,,
\end{split}
\end{equation}
where
\begin{equation}
\Gamma\;:=\;\bigcup_{j=1}^3\Gamma_j\qquad \textrm{with}\qquad
\begin{cases}\label{eq:hyperplanes}
 \Gamma_1\;:=\;\{\yy_2=\mathbf{0}\} \\
 \Gamma_2\;:=\;\{\yy_1=\mathbf{0}\} \\
 \Gamma_3\;:=\;\{\yy_1=\yy_2\} 
\end{cases}
\end{equation}
in the sense of hyperplanes in $\mathbb{R}^3\times\mathbb{R}^3$,
and 
\begin{equation}
 H^2_0\big( (\mathbb{R}^3_{\yy_1}\times\mathbb{R}^3_{\yy_2})\setminus\Gamma\big)\;:=\;\overline{C^\infty_0\big( (\mathbb{R}^3_{\yy_1}\times\mathbb{R}^3_{\yy_2})\setminus\Gamma\big)}^{\|\,\|_{H^2}}\,.
\end{equation}
In the notation \eqref{eq:hyperplanes}, the hyperplane $\Gamma_j$ is the set of configurations where the two particles different than the $j$-th one coincide.

In short, $\mathring{H}$ is the operator closure of $-\Delta_{\yy_1}-\Delta_{\yy_2}-\nabla_{\yy_1}\nabla_{\yy_2}$ initially defined on the bosonic smooth functions on $\mathbb{R}^3\times\mathbb{R}^3$ which are compactly supported away from the coincidence manifold $\Gamma$.

As the reasonings that will follow are somewhat more informative in the momentum representation, we shall often switch to the  variables $\pp_1,\pp_2$ that are Fourier conjugate to $\yy_1,\yy_2$ (yet, all our considerations can be straightforwardly re-phrased in position coordinates). One deduces from \eqref{eq:Hbosonic} that, for any $\psi\in\cH$,
\begin{equation}\label{eq:bosonicmomentum}
 \psi\in\cH_{\mathrm{b}}\qquad\Leftrightarrow\qquad\widehat{\psi}(\pp_1,\pp_2)\;=\;\widehat{\psi}(\pp_2,\pp_1)\;=\;\widehat{\psi}(\pp_1,-\pp_1-\pp_2)
\end{equation}
for almost every $\pp_1,\pp_2$.
Moreover, 
\begin{equation}\label{eq:psiy0}
 \psi|_{\Gamma_1}(\yy_1)\;=\;\psi(\yy_1,\mathbf{0})\;=\;\frac{1}{\;(2\pi)^{3}}\iint_{\mathbb{R}^3\times\mathbb{R}^3}\ud\pp_1\ud\pp_2\,e^{\ii \yy_1 \cdot\pp_1}\widehat{\psi}(\pp_1,\pp_2)\,,
\end{equation}
from which, using the identity
\[
 \frac{1}{\;(2\pi)^{3}}\int_{\mathbb{R}^3}\ud\yy_1\,e^{\ii\yy_1\cdot(\qq_1-\pp_1)}=\delta(\qq_1-\pp_1)\,,
\]
one deduces
\begin{equation}\label{eq:tracemomentum}
 \widehat{\psi|_{\Gamma_1}}(\pp_1)\;=\;\frac{1}{\;(2\pi)^{\frac{3}{2}}}\int_{\mathbb{R}^3}\widehat{\psi}(\pp_1,\pp_2)\,\ud\pp_2\,,
\end{equation}
and analogous expressions for $\widehat{\psi|_{\Gamma_2}}$ and $\widehat{\psi|_{\Gamma_3}}$.
This includes also the possibility that the evaluation of $\psi$ at a coincidence hyperplane makes \eqref{eq:psiy0}-\eqref{eq:tracemomentum} infinite for (almost) every value of the remaining variable.

By a standard trace theorem (see, e.g.,  \cite[Lemma 16.1]{Trtar-SobSpaces_interp}), if $\psi\in H^2(\mathbb{R}^3\times\mathbb{R}^3)$, then its evaluation $\psi|_{\Gamma_j}$ at the $j$-th coincidence hyperplane is a function in $H^{\frac{1}{2}}(\mathbb{R}^3)$, hence not necessarily continuous. Thus, when $f\in\mathcal{D}(\mathring{H})$ the vanishing ``$f|_{\Gamma_j}=0$'' in $H^{\frac{1}{2}}(\mathbb{R}^3)$ is to be understood by duality as 
\begin{equation}\label{eq:vanishingduality}
 0\;=\;\langle \eta, f|_{\Gamma_j}\rangle_{H^{-\frac{1}{2}},H^{\frac{1}{2}}}\;=\;\int_{\mathbb{R}^3}\overline{\widehat{\eta}(\pp)}\;\widehat{f|_{\Gamma_j}}(\pp)\,\ud\pp\qquad\forall\eta\in H^{-\frac{1}{2}}(\mathbb{R}^3)\,.
\end{equation}
This means that for any $f\in\mathcal{D}(\mathring{H})$ the vanishing at $\Gamma_1$, $\Gamma_2$, or $\Gamma_3$ corresponds, respectively, to
\begin{equation}\label{eq:triplevanishing}
 \begin{split}
  \iint_{\mathbb{R}^3\times\mathbb{R}^3}\widehat{f}(\pp_1,\pp_2)\,\widehat{\eta}(\pp_1)\,\ud\pp_1\ud\pp_2\;&=\;0\,,
 \\
 \iint_{\mathbb{R}^3\times\mathbb{R}^3}\widehat{f}(\pp_1,\pp_2)\,\widehat{\eta}(\pp_2)\,\ud\pp_1\ud\pp_2\;&=\;0\,, \\
 \iint_{\mathbb{R}^3\times\mathbb{R}^3}\widehat{f}(\pp_1,\pp_2)\,\widehat{\eta}(-\pp_1-\pp_2)\,\ud\pp_1\ud\pp_2\;&=\;0\,,
 \end{split}
\end{equation}
for each $\eta\in H^{-\frac{1}{2}}(\mathbb{R}^3)$,
as one may conclude combining \eqref{eq:bosonicmomentum}, \eqref{eq:tracemomentum} (and its counterparts by symmetry), and \eqref{eq:vanishingduality}.

The following is therefore proved.


\begin{lemma}\label{lem:Hminimal} The definition \eqref{eq:domHring-initial} is equivalent to
\begin{equation}\label{eq:Hringshort}
 \begin{split}
   \mathcal{D}(\mathring{H})\;&=\;\left\{ 
   \begin{array}{c}
   f\in\cH_{\mathrm{b}}\cap H^2(\mathbb{R}^3\times\mathbb{R}^3) \\
   \textrm{$f$ satisfies \eqref{eq:triplevanishing} }\forall\eta\in H^{-\frac{1}{2}}(\mathbb{R}^3)
   \end{array}
   \right\} \\
  \widehat{\mathring{H} f}(\pp_1,\pp_2)\;&=\;(\pp_1^2+\pp_2^2+\pp_1\cdot\pp_2)\widehat{f}(\pp_1,\pp_2)\,.
 \end{split}
\end{equation}
\end{lemma}

A convenient shorthand shall be
\begin{equation}
 H_\mathrm{b}^s(\mathbb{R}^3\times\mathbb{R}^3)\;:=\;\cH_{\mathrm{b}}\cap H^s(\mathbb{R}^3\times\mathbb{R}^3)\,,\qquad s\geqslant 0\,,
\end{equation}
for Sobolev spaces with bosonic symmetry.

Let us also observe that for any $\lambda>0$
\begin{equation}\label{eq:lambda-equiv-1}
 \pp_1^2+\pp_2^2+\pp_1\cdot\pp_2+\lambda\;\sim\;\pp_1^2+\pp_2^2+1\,,
\end{equation}
in the sense that each quantity controls the other from above and from below.

\subsection{Friedrichs extension}~

It is clear that $\mathring{H}$ is lower semi-bounded, with lower bound $\mathfrak{m}(\mathring{H})=0$. As such, it has a distinguished extension, the Friedrichs extension $\mathring{H}_F$.

\begin{lemma}\label{lem:Friedrichs}
 The Friedrichs extension $\mathring{H}_F$ of $\mathring{H}$ is the self-adjoint operator acting as
 \begin{equation}\label{eq:HFspace}
  (\mathring{H}_F \phi)(\yy_1,\yy_2) \;=\;-\Delta_{\yy_1}\phi(\yy_1,\yy_2)-\Delta_{\yy_2}\phi(\yy_1,\yy_2)-\nabla_{\yy_1}\cdot\nabla_{\yy_2}\phi(\yy_1,\yy_2)\,,
 \end{equation}
 or equivalently
 \begin{equation}\label{eq:HFmomentum}
   (\widehat{\mathring{H}_F \phi})(\pp_1,\pp_2)\;=\;(\pp_1^2+\pp_2^2+\pp_1\cdot\pp_2)\widehat{\phi}(\pp_1,\pp_2)\,,
 \end{equation}
 defined on the domain
 \begin{equation}\label{eq:HFdomain}
  \mathcal{D}(\mathring{H}_F)\;=\;H_\mathrm{b}^2(\mathbb{R}^3\times\mathbb{R}^3)\,.
 \end{equation}
 Its quadratic form is
 \begin{equation}\label{eq:HFform}
 \begin{split}
  \mathcal{D}[\mathring{H}_F]\;&=\;H_\mathrm{b}^1(\mathbb{R}^3\times\mathbb{R}^3) \\
  \mathring{H}_F[\phi]\;&=\;\frac{1}{2}\iint_{\mathbb{R}^3\times\mathbb{R}^3}\Big(\big|(\nabla_{\yy_1}+\nabla_{\yy_2})\phi\big|^2+\big|\nabla_{\yy_1}\phi\big|^2+\big|\nabla_{\yy_2}\phi\big|^2\Big)\,\ud\yy_1\,\ud\yy_2\,.
 \end{split}
 \end{equation}
\end{lemma}

\begin{proof}
The \emph{form} domain $\mathcal{D}[\mathring{H}]$ of $\mathring{H}$ is the completion of $\mathcal{D}(\mathring{H})=\cH_\mathrm{b}\cap H^2_0( (\mathbb{R}^3_{\yy_1}\times\mathbb{R}^3_{\yy_2})\setminus\Gamma)$ in the norm $\|f\|_F:=(\langle f,\mathring{H}f\rangle+\|f\|_{\cH}^2)^{\frac{1}{2}}$, and $\|f\|_F\approx\|f\|_{H^1}$ owing to \eqref{eq:lambda-equiv-1}. Reasoning as done in collaboration with Ottolini in \cite[Lemma 3(ii)]{MO-2016}, the above-mentioned completion is precisely $H^1_{\mathrm{b}}(\mathbb{R}^3\times\mathbb{R}^3)$. Thus, $\mathcal{D}[\mathring{H}]=H^1_{\mathrm{b}}(\mathbb{R}^3\times\mathbb{R}^3)$. Since \eqref{eq:HFspace}-\eqref{eq:HFdomain} obviously defines a self-adjoint extension of $\mathring{H}$ with domain entirely contained in $\mathcal{D}[\mathring{H}]$, necessarily such operator is the Friedrichs extension of $\mathring{H}$. The explicit formula for the evaluation of the quadratic form follows from the identity $\pp_1^2+\pp_2^2+\pp_1\cdot\pp_2=\frac{1}{2}(\pp_1+\pp_2)^2+\frac{1}{2}\pp_1^2+\frac{1}{2}\pp_2^2$.
\end{proof}

As $\mathring{H}$ is bounded from below, its self-adjoint realisations may be identified by means of the Kre{\u\i}n-Vi\v{s}ik-Birman extension scheme for semi-bounded symmetric operators \cite{GMO-KVB2017}. In this scheme each extension is conveniently parametrised with respect to a reference extension that has everywhere-defined bounded inverse.

Now, the Friedrichs extension $\mathring{H}_F$ has zero at the bottom of its spectrum and hence is not everywhere invertible in $\cH_\mathrm{b}$. One then searches for self-adjoint realisations of the shifted operator $\mathring{H}+\lambda\mathbbm{1}$, for some $\lambda>0$, since obviously
\begin{equation}
 \mathcal{D}(\mathring{H}+\lambda\mathbbm{1})\;=\;\mathcal{D}(\mathring{H})\,,\qquad (\mathring{H}+\lambda\mathbbm{1})_F=\mathring{H}_F+\lambda\mathbbm{1}\qquad \forall\lambda>0\,,
\end{equation}
and the latter operator is indeed everywhere invertible in $\cH_\mathrm{b}$. Once the self-adjoint extensions of $\mathring{H}+\lambda\mathbbm{1}$ are identified, the corresponding ones for $\mathring{H}$ are then read out from the former by removing the shift.

The data needed for the classification of the self-adjoint extensions of  $\mathring{H}+\lambda\mathbbm{1}$, according to the Kre{\u\i}n-Vi\v{s}ik-Birman theory, are the deficiency subspace $\ker(\mathring{H}^*+\lambda\mathbbm{1})$ and the action of $(\mathring{H}_F+\lambda\mathbbm{1})^{-1}$ on such space. We shall designate these data.

First of all, obviously,
\begin{equation}\label{eq:HFinverse}
 ((\mathring{H}_F+\lambda\mathbbm{1})^{-1}\psi)\,{\textrm{\large $\widehat{\,}$\normalsize}}\,(\pp_1,\pp_2)\;=\;(\pp_1^2+\pp_2^2+\pp_1\cdot\pp_2+\lambda)^{-1}\widehat{\psi}(\pp_1,\pp_2)
\end{equation}
for every $\psi\in\cH_\mathrm{b}$ and $\lambda>0$. Next we describe the adjoint.

\subsection{Adjoint}~

For given $\xi\in H^{-\frac{1}{2}}(\mathbb{R}^3)$ and $\lambda>0$ let $u_\xi^\lambda$ be the function defined by
\begin{equation}\label{eq:uxi}
 \widehat{u_\xi^\lambda}(\pp_1,\pp_2)\;:=\;\frac{\widehat{\xi}(\pp_1)+\widehat{\xi}(\pp_2)+\widehat{\xi}(-\pp_1-\pp_2)}{\pp_1^2+\pp_2^2+\pp_1\cdot\pp_2+\lambda}\,.
\end{equation}

\begin{lemma}\label{lem:uxiproperties}~

\begin{itemize}
 \item[(i)] For every $\lambda>0$ there exists a constant $c_\lambda>0$ such that for every $\xi\in H^{-\frac{1}{2}}(\mathbb{R}^3)$ one has
 \begin{equation}\label{eq:equivalencenorms}
  c_\lambda^{-1}\|\xi\|_{H^{-\frac{1}{2}}(\mathbb{R}^3)}\;\leqslant\;\|u_\xi^\lambda\|_{\cH}\;\leqslant c_\lambda\|\xi\|_{H^{-\frac{1}{2}}(\mathbb{R}^3)}\,.
 \end{equation}
 \item[(ii)] For every $\xi\in H^{-\frac{1}{2}}(\mathbb{R}^3)$, $u_\xi^\lambda\in\cH_{\mathrm{b}}$.
 \item[(iii)] If $u_\xi^\lambda=u_\eta^\lambda$ for some $\xi,\eta\in H^{-\frac{1}{2}}(\mathbb{R}^3)$ and $\lambda>0$, then $\xi=\eta$.
 \item[(iv)] For $\xi\in H^{-\frac{1}{2}}(\mathbb{R}^3)$ and $\lambda,\mu>0$ one has $u_\xi^\lambda-u^\mu_\xi\in H^2_\mathrm{b}(\mathbb{R}^3\times\mathbb{R}^3)$.
\end{itemize} 
\end{lemma}

\begin{proof}
 Part (i) can be proved by easily mimicking the very same argument of \cite[Lemma B.2]{CDFMT-2015}. Part (ii) follows from (i) and from the invariance of \eqref{eq:uxi} under the transformations \eqref{eq:bosonic_transformations}. Part (iii) follows from (i), owing to the linearity $\xi\mapsto u_\xi$. Part (iv) follows from the identity
 \[
  \widehat{u_\xi^\lambda}(\pp_1,\pp_2)-\widehat{u_\xi^\mu}(\pp_1,\pp_2)\;=\;\frac{(\mu-\lambda)\big(\widehat{\xi}(\pp_1)+\widehat{\xi}(\pp_2)+\widehat{\xi}(-\pp_1-\pp_2)\big)}{\,(\pp_1^2+\pp_2^2+\pp_1\cdot\pp_2+\lambda)\,(\pp_1^2+\pp_2^2+\pp_1\cdot\pp_2+\mu)\,}
 \]
 and from \eqref{eq:lambda-equiv-1} and \eqref{eq:equivalencenorms}.
\end{proof}

\begin{lemma}\label{lem:Hstaretc} Let $\lambda>0$.
\begin{itemize}
  \item[(i)] One has
 \begin{equation}\label{eq:kerHstarlambda}
  \ker(\mathring{H}^*+\lambda\mathbbm{1})\;=\;\big\{u_\xi^\lambda\,|\,\xi\in H^{-\frac{1}{2}}(\mathbb{R}^3)\big\}\,. 
 \end{equation}
 \item[(ii)] One has
 \begin{equation}\label{eq:DHFdecomposed}
 \begin{split}
 \mathcal{D}(\mathring{H}_F)\;&=\; H_\mathrm{b}^2(\mathbb{R}^3\times\mathbb{R}^3) \\
  &=\;\left\{
   \phi\in \cH_{\mathrm{b}}\,\left|
   \begin{array}{c}
    \displaystyle\widehat{\phi}\,=\,\widehat{f^\lambda}+\frac{\widehat{u_\eta^\lambda}}{\pp_1^2+\pp_2^2+\pp_1\cdot\pp_2+\lambda} \\
    \textrm{for } f^\lambda\in\mathcal{D}(\mathring{H})\,,\;\eta\in H^{-\frac{1}{2}}(\mathbb{R}^3)
   \end{array}
   \right.\right\}.
   \end{split}
 \end{equation}
 \item[(iii)] One has
 \begin{equation}\label{eq:DHstardecomposed}
 \begin{split}
   \mathcal{D}(\mathring{H}^*)\;&=\left\{
   g\in \cH_{\mathrm{b}}\,\left|
   \begin{array}{c}
    \displaystyle\widehat{g}\,=\,\widehat{\phi^\lambda}+\widehat{u_\xi^\lambda} \\
    \textrm{for } \phi^\lambda\in H_\mathrm{b}^2(\mathbb{R}^3\times\mathbb{R}^3) \,,\;\xi\in H^{-\frac{1}{2}}(\mathbb{R}^3)
   \end{array}
   \!\!\right.\right\} \\
   &=\left\{
   g\in \cH_{\mathrm{b}}\,\left|
   \begin{array}{c}
    \displaystyle\widehat{g}\,=\,\widehat{f^\lambda}+\frac{\widehat{u_\eta^\lambda}}{\pp_1^2+\pp_2^2+\pp_1\cdot\pp_2+\lambda}+\widehat{u_\xi^\lambda} \\
    \textrm{for } f^\lambda\in\mathcal{D}(\mathring{H})\,,\;\xi,\eta\in H^{-\frac{1}{2}}(\mathbb{R}^3)
   \end{array}
   \right.\right\}
 \end{split}
 \end{equation}
 and 
 \begin{equation}
  ((\mathring{H}^*+\lambda\mathbbm{1})g)\,{\textrm{\large $\widehat{\,}$\normalsize}}\,\;=\;(\pp_1^2+\pp_2^2+\pp_1\cdot\pp_2+\lambda)\widehat{\phi^\lambda}
 \end{equation}
 with
 \begin{equation}
  \widehat{\phi^\lambda}\;:=\;\widehat{f^\lambda}+\frac{\widehat{u_\eta^\lambda}}{\pp_1^2+\pp_2^2+\pp_1\cdot\pp_2+\lambda}\,,
 \end{equation}
 or equivalently
 \begin{equation}\label{eq:Hstarg}
  \begin{split}
   \widehat{(\mathring{H}^* g)}(\pp_1,\pp_2)\;&=\;(\pp_1^2+\pp_2^2+\pp_1\cdot \pp_2) \widehat{g}(\pp_1,\pp_2) \\
   &\qquad -\big(\widehat{\xi}(\pp_1)+\widehat{\xi}(\pp_2)+\widehat{\xi}(-\pp_1-\pp_2) \big)\,.
  \end{split}
 \end{equation}
\end{itemize}
The decompositions in \eqref{eq:DHFdecomposed} and \eqref{eq:DHstardecomposed} are unique, in the sense that for each $\phi\in\mathcal{D}(\mathring{H}_F)$ there exist unique $f^\lambda,\eta$ satisfying the decomposition in \eqref{eq:DHFdecomposed}, and for each $g\in \mathcal{D}(\mathring{H}^*)$ there exist unique $f^\lambda,\eta,\xi$ satisfying the decomposition in \eqref{eq:DHstardecomposed}.
\end{lemma}

 \begin{proof}
  Any $u\in\ker(\mathring{H}^*+\lambda\mathbbm{1})=\mathrm{ran}(\mathring{H}+\lambda\mathbbm{1})^\perp$ is characterised by
  \[
   \iint_{\mathbb{R}^3}\widehat{u}(\pp_1,\pp_2)\,(\pp_1^2+\pp_2^2+\pp_1\cdot \pp_2+\lambda)\,\widehat{f}(\pp_1,\pp_2)\ud\pp_1\ud\pp_2\;=\;0\qquad\forall f\in\mathcal{D}(\mathring{H})\,.
  \]
  Combining this with \eqref{eq:triplevanishing}, one deduces that $\widehat{u}(\pp_1,\pp_2)\,(\pp_1^2+\pp_2^2+\pp_1\cdot \pp_2+\lambda)$ must be a linear combination of $\widehat{\xi}(\pp_1)$, $\widehat{\xi}(\pp_2)$, and $\widehat{\xi}(-\pp_1-\pp_2)$ for a generic $\xi\in H^{-\frac{1}{2}}(\mathbb{R}^3)$; as $u\in\cH_\mathrm{b}$, this combination must be the sum (up to an overall multiplicative prefactor). Thus, $\widehat{u}(\pp_1,\pp_2)\,(\pp_1^2+\pp_2^2+\pp_1\cdot \pp_2+\lambda)=\widehat{\xi}(\pp_1)+\widehat{\xi}(\pp_2)+\widehat{\xi}(-\pp_1-\pp_2)$, which proves part (i). Parts (ii) and (iii) then follow from part (i) and from \eqref{eq:HFinverse} as an application of the standard formulas
  \[
   \begin{split}
    \mathcal{D}(\mathring{H}_F)\;&=\;\mathcal{D}(\mathring{H})\dotplus (\mathring{H}_F+\lambda\mathbbm{1})^{-1}\ker(\mathring{H}^*+\lambda\mathbbm{1}) \\
    \mathcal{D}(\mathring{H}^*)\;&=\;\mathcal{D}(\mathring{H})\dotplus (\mathring{H}_F+\lambda\mathbbm{1})^{-1}\ker(\mathring{H}^*+\lambda\mathbbm{1})\dotplus \ker(\mathring{H}^*+\lambda\mathbbm{1})
   \end{split}
  \]
 (see, e.g., \cite[Lemma 1 and Theorem 1]{GMO-KVB2017}).  
 \end{proof}

 \begin{remark}
 \eqref{eq:kerHstarlambda}-\eqref{eq:DHstardecomposed} show that functions in $\mathcal{D}(\mathring{H}^*)$ have a `regular' $H^2$-component and a `singular' $L^2$-component, with no constraint among the two. The regular part is the domain of $\mathring{H}_F$, the singular part is the kernel of $\mathring{H}^*+\lambda\mathbbm{1}$.  Because of the possible singularity of a generic $g\in\mathcal{D}(\mathring{H}^*)$, the action on $g$ of the differential operator $-\Delta_{\yy_1}-\Delta_{\yy_2}-\nabla_{\yy_1}\cdot\nabla_{\yy_2}$ produces in general a non-$L^2$ output. More precisely, \eqref{eq:Hstarg} shows that one has to subtract from $(-\Delta_{\yy_1}-\Delta_{\yy_2}-\nabla_{\yy_1}\cdot\nabla_{\yy_2})g$ the distribution
 \begin{equation}\label{eq:livingonhyperplanes}
  (2\pi)^{\frac{3}{2}}\big(\xi(\yy_1)\delta(\yy_2)+\delta(\yy_1)\xi(\yy_2)+\delta(\yy_1-\yy_2)\xi(-\yy_2)\big)
 \end{equation}
 (that is, the inverse Fourier transform of the second summand in \eqref{eq:Hstarg}), a distribution supported at the coincidence manifold $\Gamma$, in order to obtain the $L^2$-function $\mathring{H}^* g$.
 \end{remark}

 \begin{remark}
 In position coordinates $(\yy_1,\yy_2)$, each of the two functions $u_\eta^\lambda$ and $u_\xi^\lambda$ appearing in the expression \eqref{eq:DHstardecomposed} of a generic element $g\in\mathcal{D}(\mathring{H}^*)$ is obtained by taking the \emph{convolution} of the Green function $\mathcal{G}_\lambda$ relative to $-\Delta_{\yy_1}-\Delta_{\yy_2}-\nabla_{\yy_1}\cdot\nabla_{\yy_2}+\lambda$ with a distribution of the form \eqref{eq:livingonhyperplanes} for the two considered labelling functions $\eta,\xi$. This structure, and the fact that in \eqref{eq:livingonhyperplanes} $\xi$ (and $\eta$) is interpreted as a function on the union of the coincidence hyperplanes, is formally analogous to the familiar picture in electrostatics, where $u_\xi^\lambda$ is the `potential' relative to the `charge' $\xi$. For this reason, as has been customary since long in this context \cite{dft-Nparticles-delta}, we shall retain the nomenclature that $\xi$ and $\eta$ are the \emph{charges} for the function $g$. In this respect, by \emph{charges} we shall mean functions in $H^{-\frac{1}{2}}(\mathbb{R}^3)$.  
 \end{remark}

 For $g\in\mathcal{D}(\mathring{H}^*)$, there is a unique charge $\xi$ at each parameter $\lambda>0$ satisfying the decomposition \eqref{eq:DHstardecomposed}, but a priori $\xi$ might be $\lambda$-dependent. Let us show that this cannot be the case, and one can speak of \emph{the} charge $\xi$ of $g$ tout court (understanding $\xi$, as usual, as the charge of the singular part of $g$, not to confuse it with the charge $\eta$ appearing in the regular part of $g$).

 \begin{lemma}\label{lem:chargexiofg}
  Let $\lambda,\lambda'>0$ and $g\in\mathcal{D}(\mathring{H}^*)$. If, according to \eqref{eq:DHstardecomposed},
  \[
   \widehat{g}\;=\;\widehat{\phi^\lambda}+\widehat{u_\xi^\lambda} \;=\;\widehat{\phi^{\lambda'}}+\widehat{u_{\xi'}^{\lambda'}}
  \]
  for some $\phi^\lambda,\phi^{\lambda'}\in H_\mathrm{b}^2(\mathbb{R}^3\times\mathbb{R}^3)$ and $\xi,\xi'\in H^{-\frac{1}{2}}(\mathbb{R}^3)$, then $\xi=\xi'$.  
 \end{lemma}

 \begin{proof}
  Since $u_{\xi}^{\lambda'}-u_\xi^\lambda\in H^2_\mathrm{b}(\mathbb{R}^3\times\mathbb{R}^3)$ (Lemma \ref{lem:uxiproperties}(iv)), then
  \[
   F^{\lambda'}\;:=\;\phi^\lambda-\big(u_{\xi}^{\lambda'}-u_{\xi}^{\lambda} \big)
  \]
 defines a function in $H^2_\mathrm{b}(\mathbb{R}^3\times\mathbb{R}^3)$, and
 \[
  \widehat{g}\;=\;\widehat{\phi^\lambda}+\widehat{u_\xi^\lambda} \;=\;\widehat{F^{\lambda'}}+\widehat{u_{\xi}^{\lambda'}}\,.
 \]
 Comparing the latter identity with $\widehat{g}=\widehat{\phi^{\lambda'}}+\widehat{u_{\xi'}^{\lambda'}}$, the uniqueness of the decomposition \eqref{eq:DHstardecomposed} of $g$ with parameter $\lambda'$ implies 
 \[
  F^{\lambda'}\;=\;\phi^{\lambda'}\,,\qquad u_{\xi}^{\lambda'}\;=\;u_{\xi'}^{\lambda'}\,.
 \]
 In turn, the latter identity implies $\xi=\xi'$ (Lemma \ref{lem:uxiproperties}(iii)).  
 \end{proof}

 \subsection{Deficiency subspace}\label{sec:deficiencysubspace}~
 
 Lemma \ref{lem:Hstaretc} shows that $\mathring{H}$ has \emph{infinite} deficiency index, as is the dimensionality of the deficiency subspace $\ker(\mathring{H}^*+\lambda\mathbbm{1})$ independently of $\lambda>0$.

 In the self-adjoint extension problem under study, it is convenient to use a unitarily isomorphic version of  $\ker(\mathring{H}^*+\lambda\mathbbm{1})$, which we shall now characterise.

 To this aim, for $\eta\in H^{-\frac{1}{2}}(\mathbb{R}^3)$ and $\lambda>0$ we define (for a.e.~$\pp$)
 \begin{equation}\label{eq:Wlambda}
  (\widehat{W_\lambda\eta})(\pp)\;:=\;\frac{3\pi^2}{\sqrt{\frac{3}{4}\pp^2+\lambda}}\,\widehat{\xi}(\pp)+6\int_{\mathbb{R}^3}\frac{\widehat{\xi}(\qq)}{(\pp^2+\qq^2+\pp\cdot \qq+\lambda)^2}\,\ud\qq\,.
 \end{equation}

 \begin{lemma}\label{lem:Wlambdaproperties} Let $\lambda>0$.
 
 \begin{itemize}
  \item[(i)] For generic $\xi,\eta\in H^{-\frac{1}{2}}(\mathbb{R}^3)$ one has $W_\lambda\eta\in H^{\frac{1}{2}}(\mathbb{R}^3)$ and
   \begin{equation}\label{eq:scalar_products}
 \langle u_\xi,u_\eta\rangle_{\cH}\;=\;\langle \xi,W_\lambda\eta\rangle_{H^{-\frac{1}{2}}(\mathbb{R}^3),H^{\frac{1}{2}}(\mathbb{R}^3)}\,.
 \end{equation}
  \item[(ii)] Formula \eqref{eq:Wlambda} defines a positive, bounded, linear bijection $W_\lambda: H^{-\frac{1}{2}}(\mathbb{R}^3)\to H^{\frac{1}{2}}(\mathbb{R}^3)$.    
 \end{itemize}
 \end{lemma}

 \begin{proof}
  By suitably exploiting symmetry in exchanging the integration variables,
  \[
   \begin{split}
    &\langle u_\xi,u_\eta\rangle_{\cH}\;= \\
    &=\;\iint_{\mathbb{R}^3\times\mathbb{R}^3}\ud\pp_1\ud\pp_2\,\frac{\overline{\widehat{\xi}(\pp_1)}+\overline{\widehat{\xi}(\pp_2)}+\overline{\widehat{\xi}(-\pp_1-\pp_2)}}{\pp_1^2+\pp_2^2+\pp_1\cdot\pp_2+\lambda}\,\frac{\widehat{\eta}(\pp_1)+\widehat{\eta}(\pp_2)+\widehat{\eta}(-\pp_1-\pp_2)}{\pp_1^2+\pp_2^2+\pp_1\cdot\pp_2+\lambda} \\
    &=\;3\iint_{\mathbb{R}^3\times\mathbb{R}^3}\ud\pp_1\ud\pp_2\,\frac{\overline{\widehat{\xi}(\pp_1)}\,\widehat{\eta}(\pp_1)}{(\pp_1^2+\pp_2^2+\pp_1\cdot\pp_2+\lambda)^2} \\
    &\qquad\qquad+6\iint_{\mathbb{R}^3\times\mathbb{R}^3}\ud\pp_1\ud\pp_2\,\frac{\overline{\widehat{\xi}(\pp_1)}\,\widehat{\eta}(\pp_2)}{(\pp_1^2+\pp_2^2+\pp_1\cdot\pp_2+\lambda)^2}\,.
   \end{split}
  \]
  In the first summand in the r.h.s.~above one computes
  \[
   \int_{\mathbb{R}^3}\frac{\ud\pp_2}{(\pp_1^2+\pp_2^2+\pp_1\cdot\pp_2+\lambda)^2}\;=\;\frac{\pi^2}{\,\sqrt{\frac{3}{4}\pp_1^2+\lambda}\,}\,,
  \]
  which eventually yields, in view of the definition \eqref{eq:Wlambda},
  \[
   \langle u_\xi,u_\eta\rangle_{\cH}\;=\;\int_{\mathbb{R}^3}\overline{\widehat{\xi}(\pp)}\, (\widehat{W_\lambda\eta})(\pp)\,\ud\pp\,.
  \]
    From the latter identity and \eqref{eq:equivalencenorms} we deduce
  \[
\begin{split}
\|W_\lambda\eta\|_{H^{\frac{1}{2}}}\;&=\;\sup_{\|\xi\|_{H^{-\frac{1}{2}}}=1}\Big|\int_{\mathbb{R}^3}\overline{\widehat{\xi}(\pp)}\,\widehat{(W_\lambda\eta)}(\pp)\,\ud \pp\,\Big|\;=\;\sup_{\|\xi\|_{H^{-\frac{1}{2}}}=1}\big|\langle u_\xi,u_\eta\rangle_{\cH}\big| \\
&\leqslant\;\sup_{\|\xi\|_{H^{-\frac{1}{2}}}=1}\|u_\xi\|_{\cH}\|u_\eta\|_{\cH}\;\leqslant \;\textrm{const}\cdot\|\eta\|_{H^{-\frac{1}{2}}}\qquad\forall\eta\in H^{-\frac{1}{2}}(\mathbb{R}^3)\,,
\end{split}
\]
  which shows that $W_\lambda H^{-\frac{1}{2}}(\mathbb{R}^3)\subset H^{\frac{1}{2}}(\mathbb{R}^3)$ and that \eqref{eq:scalar_products} holds true. This completes the proof of part (i).

Concerning part (ii), we know from the reasoning above that the map $W_\lambda:H^{-\frac{1}{2}}(\mathbb{R}^3)\to H^{\frac{1}{2}}(\mathbb{R}^3)$ is bounded. By \eqref{eq:equivalencenorms} and \eqref{eq:scalar_products},
  \[
   \langle \eta,W_\lambda\eta\rangle_{H^{-\frac{1}{2}},H^{\frac{1}{2}}}\;=\;\|u_\eta\|^2_{\cH}\;\geqslant\;c_\lambda^{-2}\,\|\eta\|_{H^{-\frac{1}{2}}}^2\,,
  \]
  which implies coercivity
\[
 \|W_\lambda\eta\|_{H^{\frac{1}{2}}}\;\geqslant\;c_\lambda^{-2}\,\|\eta\|_{H^{-\frac{1}{2}}}\,.
\]
This shows that $W_\lambda$ is a positive, injective $H^{-\frac{1}{2}}\to H^{\frac{1}{2}}$ map. $W_\lambda$ is thus invertible on $\mathrm{ran}W_\lambda$ and by boundedness $\mathrm{ran}W_\lambda$ is closed in $H^{\frac{1}{2}}(\mathbb{R}^3)$. It only remains to show that $\mathrm{ran}W_\lambda$ is also dense in $H^{\frac{1}{2}}(\mathbb{R}^3)$ to conclude that $W_\lambda^{-1}$ is everywhere defined and bounded. Now, testing by duality an \emph{arbitrary} $\xi\in H^{-\frac{1}{2}}(\mathbb{R}^3)$ against  $\mathrm{ran}W_\lambda\subset H^{\frac{1}{2}}(\mathbb{R}^3)$ we see that
\[
 \begin{split}
  &\langle \xi,W_\lambda\eta\rangle_{H^{-\frac{1}{2}},H^{\frac{1}{2}}}\;=\;0\qquad\forall\eta\in H^{-\frac{1}{2}}(\mathbb{R}^3) \\
  &\Rightarrow\;\langle u_\xi,u_\eta\rangle_{\cH}\;=\;0\qquad\quad\:\,\forall u_\eta\in\ker (\mathring{H}^*+\lambda\mathbbm{1}) \\
  &\Rightarrow\;u_\xi\,=\,0 \\
  &\Rightarrow\;\xi\,=\,0\,,
 \end{split}
\]
and this implies that $\mathrm{ran}W_\lambda$ is dense in $H^{\frac{1}{2}}(\mathbb{R}^3)$. Part (ii) is proved.  
 \end{proof}

 As a direct consequence of Lemma \ref{lem:Wlambdaproperties}, the expression
\begin{equation}\label{eq:W-scalar-product}
\langle \xi,\eta\rangle_{H^{-\frac{1}{2}}_{W_\lambda}}\;:=\;\langle \xi,W_\lambda\,\eta\rangle_{H^{-\frac{1}{2}},H^{\frac{1}{2}}}\;=\;\langle u_\xi,u_\eta\rangle_{\cH}
\end{equation}
defines a scalar product in $H^{-\frac{1}{2}}(\mathbb{R}^3)$. It is \emph{equivalent} to the standard scalar product of $H^{-\frac{1}{2}}(\mathbb{R}^3)$, as follows by combining \eqref{eq:W-scalar-product} with \eqref{eq:equivalencenorms}.

We shall denote by $H^{-\frac{1}{2}}_{W_\lambda}(\mathbb{R}^3)$ the Hilbert space consisting of the $H^{-\frac{1}{2}}(\mathbb{R}^3)$-functions and equipped with the scalar product \eqref{eq:W-scalar-product}.
Then the map
\begin{equation}\label{eq:isomorphism_Ulambda}
\begin{split}
U_\lambda\,:\,\ker (\mathring{H}^*+\lambda\mathbbm{1})\;&\;\xrightarrow[]{\;\;\;\cong\;\;\;}\;H^{-\frac{1}{2}}_{W_\lambda}(\mathbb{R}^3)\,,\qquad u_\xi \longmapsto \,\xi
\end{split}
\end{equation}
is an isomorphism between Hilbert spaces, with $\ker (\mathring{H}^*+\lambda\mathbbm{1})$ equipped with  the standard scalar product inherited from $\cH$.

\subsection{Extensions classification}~

For an arbitrary Hilbert space $\mathcal{K}$ let us denote by $\mathcal{S}(\mathcal{K})$ the collection of all self-adjoint operators $A:\mathcal{D}(A)\subset\mathcal{K}'\to\mathcal{K}'$ acting on Hilbert subspaces $\mathcal{K}'$ of $\mathcal{K}$.

The Kre{\u\i}n-Vi\v{s}ik-Birman theory determines that, given $\lambda>0$ and hence the deficiency subspace $\ker (\mathring{H}^*+\lambda\mathbbm{1})$, the self-adjoint extensions of $\mathring{H}$ are in an explicit one-to-one correspondence with the elements in $\mathcal{S}(\ker (\mathring{H}^*+\lambda\mathbbm{1}))$.

Equivalently, by unitary isomorphism (Subsect.~\ref{sec:deficiencysubspace}), each self-adjoint extension of $\mathring{H}$ is labelled by an element of $\mathcal{S}(H^{-\frac{1}{2}}_{W_\lambda}(\mathbb{R}^3))$. In practice, the latter viewpoint is going to be more informative.

The extension classification takes the following form.

\begin{theorem}\label{thm:generalclassification}
 Let $\lambda>0$.
 \begin{itemize}
  \item[(i)] The self-adjoint extensions of $\mathring{H}$ in $\cH_\mathrm{b}$ constitute the family
  \begin{equation}\label{eq:family}
   \big\{\mathring{H}_{\mathcal{A}_\lambda}\,\big|\,\mathcal{A}_\lambda\in \mathcal{S}\big(H^{-\frac{1}{2}}_{W_\lambda}(\mathbb{R}^3)\big)\big\}\,,
  \end{equation}
  where
  \begin{equation}\label{eq:domDHAshort}
   \mathcal{D}(\mathring{H}_{\mathcal{A}_\lambda})\;=\;\left\{ g\in\mathcal{D}(\mathring{H}^*)\left|\!
   \begin{array}{c}
   \eta=\mathcal{A}_\lambda\xi+\chi \\
   \xi\in\mathcal{D}(\mathcal{A}_\lambda) \\
   \chi\in \mathcal{D}(\mathcal{A}_\lambda)^{\perp_{\lambda}}\cap H^{-\frac{1}{2}}_{W_\lambda}(\mathbb{R}^3)
   \end{array}
   \!\!\right.\right\}
  \end{equation}
  or equivalently
  \begin{equation}\label{eq:domDHA}
   \mathcal{D}(\mathring{H}_{\mathcal{A}_\lambda})\;=\;\left\{
   g\in\cH_{\mathrm{b}}\left|
   \begin{array}{c}
    \widehat{g}\,=\,\widehat{\phi^\lambda}+\widehat{u_\xi^\lambda}\;\textrm{ with} \\
    \widehat{\phi^\lambda}\,=\,\widehat{f^\lambda}+\displaystyle\frac{\widehat{u_\eta^\lambda}}{\pp_1^2+\pp_2^2+\pp_1\cdot\pp_2+\lambda}\,, \\
    f^\lambda\in\mathcal{D}(\mathring{H})\,, \\
    \eta=\mathcal{A}_\lambda\xi+\chi\,,\quad \xi\in\mathcal{D}(\mathcal{A}_\lambda)\,, \\
    \chi\in \mathcal{D}(\mathcal{A}_\lambda)^{\perp_{\lambda}}\cap H^{-\frac{1}{2}}_{W_\lambda}(\mathbb{R}^3)
   \end{array}
   \!\!\right.\right\},
  \end{equation}
 and 
 \begin{equation}\label{eq:domDHA-actionDHA}
  \big((\mathring{H}_{\mathcal{A}_\lambda}+\lambda\mathbbm{1})g\big)\,{\textrm{\large $\widehat{\,}$\normalsize}}\,\;=\;(\pp_1^2+\pp_2^2+\pp_1\cdot\pp_2+\lambda)\,\widehat{\phi^\lambda}\,.
 \end{equation}
 In \eqref{eq:domDHAshort}-\eqref{eq:domDHA} above $\perp_{\lambda}$ refers to the orthogonality in the $H^{-\frac{1}{2}}_{W_\lambda}$-scalar product. The Friedrichs extension $\mathring{H}_F$, namely the operator \eqref{eq:HFspace}-\eqref{eq:HFdomain}, corresponds to the formal choice `$\mathcal{A}_\lambda=\infty$' on  $\mathcal{D}(\mathcal{A}_\lambda)=\{0\}$.
  \item[(ii)] An extension $\mathring{H}_{\mathcal{A}_\lambda}$ is lower semi-bounded with
  \[
   \mathring{H}_{\mathcal{A}_\lambda}\;\geqslant\;-\Lambda\mathbbm{1}\qquad\textrm{for some $\Lambda>0$}
  \]
  if and only if, $\forall\xi\in\mathcal{D}(\mathcal{A}_\lambda)$,
  \[
  \begin{split}
   \langle\xi,\mathcal{A}_\lambda\xi&\rangle_{H^{-\frac{1}{2}}_{W_\lambda}}\;\geqslant \\
   &\geqslant\;(\lambda-\Lambda)\|\xi\|^2_{H^{-\frac{1}{2}}_{W_\lambda}} +(\lambda-\Lambda)^2\langle\xi,(\mathring{H}_F+\Lambda\mathbf{1})^{-1}\xi\rangle_{H^{-\frac{1}{2}}_{W_\lambda}}\,.
  \end{split}
  \]
  In particular,
   \begin{equation}\label{eq:positiveSBiffpositveB-1_Tversion}
 \begin{split}
 \mathfrak{m}(\mathring{H}_{\mathcal{A}_\lambda})\;\geqslant \;-\lambda\quad&\Leftrightarrow\quad \mathfrak{m}(\mathcal{A}_\lambda)\;\geqslant\; 0 \\
 \mathfrak{m}(\mathring{H}_{\mathcal{A}_\lambda})\;> \;-\lambda\quad&\Leftrightarrow\quad \mathfrak{m}(\mathcal{A}_\lambda)\;>\; 0\,.
 \end{split}
 \end{equation}
  Moreover, if $\mathfrak{m}(\mathcal{A}_\lambda)>-\lambda$, then
 \begin{equation}\label{eq:bounds_mS_mB_Tversion}
 \mathfrak{m}(\mathcal{A}_\lambda)\;\geqslant\; \mathfrak{m}(\mathring{H}_{\mathcal{A}_\lambda})+\lambda\;\geqslant\;\frac{\lambda \,\mathfrak{m}(\mathcal{A}_\lambda)}{\,\lambda+\mathfrak{m}(\mathcal{A}_\lambda)}\,.
  \end{equation}
 \item[(iii)] The quadratic form of any lower semi-bounded extension $\mathring{H}_{\mathcal{A}_\lambda}$ is given by
 \begin{equation}\label{eq:decomposition_of_form_domains_Tversion}
 \begin{split}
 \mathcal{D}[\mathring{H}_{\mathcal{A}_\lambda}]\;&=\;\mathcal{D}[\mathring{H}_F]\,\dotplus\,U_\lambda^{-1}\mathcal{D}[\mathcal{A}_\lambda] \\
 \mathring{H}_{\mathcal{A}_\lambda}[\phi^\lambda+u_\xi^\lambda]\;&=\;\mathring{H}_F[\phi^\lambda] +\lambda\Big(\|\phi^\lambda\big\|^2_{\cH}-\big\|\phi^\lambda+u_\xi^\lambda\big\|^2_{\cH} \Big)+\mathcal{A}_\lambda[\xi] \\
 &\forall \phi^\lambda\in\mathcal{D}[\mathring{H}_F]=H_\mathrm{b}^1(\mathbb{R}^3\times\mathbb{R}^3)\,,\;\forall \xi\in\mathcal{D}[\mathcal{A}_\lambda]\,,
 \end{split}
\end{equation}
 and the lower semi-bounded extensions are ordered in the sense of quadratic forms according to the analogous ordering of the labelling operators, that is,
\begin{equation}\label{eq:extension_ordering_Tversion}
\mathring{H}_{\mathcal{A}_\lambda^{(1)}}\,\geqslant\,\mathring{H}_{\mathcal{A}_\lambda^{(2)}}\qquad\Leftrightarrow\qquad \mathcal{A}_\lambda^{(1)}\,\geqslant\,\mathcal{A}_\lambda^{(2)}\,.
\end{equation}
\end{itemize}
\end{theorem}

We recall that the symbol $\mathfrak{m}$ in \eqref{eq:positiveSBiffpositveB-1_Tversion}-\eqref{eq:bounds_mS_mB_Tversion} denotes the bottom of the spectrum of the considered operator.

Theorem \ref{thm:generalclassification} is a direct application of the general extension scheme a la Kre{\u\i}n-Vi\v{s}ik-Birman (we refer, e.g., to \cite[Theorems 5-7]{GMO-KVB2017}) to the minimal operator $\mathring{H}+\lambda\mathbbm{1}$, given the data provided by Lemmas \ref{lem:Hminimal}, \ref{lem:Friedrichs}, and \ref{lem:Hstaretc}, and exploiting the Hilbert space isomorphism \eqref{eq:isomorphism_Ulambda} in order to re-phrase the classification formulas in terms of the unitarily isomorphic version $H^{-\frac{1}{2}}_{W_\lambda}(\mathbb{R}^3)$ of the deficiency subspace.

We shall customarily refer to each $\mathcal{A}_\lambda$ as the `labelling operator', or also the `(Vi\v{s}ik-)Birman operator', of the extension $H_{\mathcal{A}_\lambda}$. (Strictly speaking, the actual labelling operator originally introduced by Vi\v{s}ik \cite{Vishik-1952} and Birman \cite{Birman-1956,KM-2015-Birman} was rather the inverse of the present $\mathcal{A}_\lambda$ on $\mathrm{ran}\mathcal{A}_\lambda$ -- see, e.g., \cite[Sect.~3]{GMO-KVB2017}.)

\begin{remark}\label{rem:restrictions}
 The domain of each $\mathring{H}_{\mathcal{A}_\lambda}$, as indicated by \eqref{eq:domDHAshort}, is a suitable restriction of the domain of $\mathring{H}^*$ obtained by selecting only those functions $g$ whose charges $\eta_g$ and $\xi_g$ are constrained by the self-adjointness condition
\begin{equation}
  \eta_g\;=\;\mathcal{A}_\lambda\xi_g+\chi_g
\end{equation}
for some additional charge $\chi_g\in\mathcal{D}(\mathcal{A}_\lambda)^{\perp_\lambda}$, whence
\begin{equation}\label{eq:constraintetaxi}
 ( \eta_g-\mathcal{A}_\lambda\xi_g )\;\in\;\mathcal{D}(\mathcal{A}_\lambda)^{\perp_\lambda}\cap H^{-\frac{1}{2}}(\mathbb{R}^3)\,.
\end{equation}
\end{remark}

\begin{remark}\label{rem:samedomains}
 Fixed a self-adjoint extension $\mathscr{H}$ of $\mathring{H}$ and representing it as $\mathscr{H}=\mathring{H}_{\mathcal{A}_\lambda^{(\mathscr{H})}}$ for suitable labelling operators $\mathcal{A}_\lambda^{(\mathscr{H})}$ for each $\lambda>0$ according to Theorem \ref{thm:generalclassification}, one has
 \begin{equation}
  \mathcal{D}\big(\mathcal{A}_{\lambda}^{(\mathscr{H})}\big)\;=\;\mathcal{D}\big(\mathcal{A}_{\lambda'}^{(\mathscr{H})}\big)\qquad\forall\lambda,\lambda'>0\,.
 \end{equation}
 That is, the explicit action of each labelling operator changes with $\lambda$, but the domain stays fixed. This is an obvious consequence of the identity
 \begin{equation}
   \mathcal{D}\big(\mathcal{A}_{\lambda}^{(\mathscr{H})}\big)\;=\;\left\{\xi\in H^{-\frac{1}{2}}(\mathbb{R}^3)\,\left|\!
   \begin{array}{c}
    \xi\textrm{ is the singular-part charge of g} \\
    \textrm{ for some }g\in\mathcal{D}(\mathscr{H})
   \end{array}
   \!\!\right.\right\}
 \end{equation}
 that follows from the uniqueness of the charge $\xi$ for each $g$ (Lemma \ref{lem:chargexiofg}). 
\end{remark}

\section{Two-body short-scale singularity}\label{sec:two-body-short-scale-sing}

The domain of each self-adjoint extension $\mathring{H}_{\mathcal{A}_\lambda}$ of $\mathring{H}$ is a suitable restriction of the domain of $\mathring{H}^*$. Each restriction of self-adjointness must be a constraint of the form \eqref{eq:constraintetaxi} on the charges $\eta$ and $\xi$ (Remark \ref{rem:restrictions}). As such two charges characterise respectively the regular ($\phi^\lambda$) and the singular ($u_\xi^\lambda$) part of a generic $g\in\mathcal{D}(\mathring{H}^*)$, indirectly this constraint is a condition linking $\phi^\lambda$ and $u_\xi^\lambda$ (which otherwise would be independent). In practice this amounts to selecting those $g$'s from $\mathcal{D}(\mathring{H}^*)$ which display an admissible type of short-scale asymptotics as $|\yy_1|\to 0$, of $|\yy_2|\to 0$, or $|\yy_2-\yy_1|\to 0$, that is, when \emph{two} of the three particles of the trimer come on top of each other.

In this Section we elaborate on this perspective, as it is going to drive the identification of \emph{physically meaningful} self-adjoint extensions of $\mathring{H}$.

\subsection{Short-scale structure}~

For the functions $\psi\in L^2(\mathbb{R}^3\times\mathbb{R}^3,\ud\yy_1\ud\yy_2)$ of interest, let us highlight a convenient way to monitor the behaviour of $\psi(\yy_1,\yy_2)$ as $|\yy_2|\to 0$ at fixed $\yy_1$. 

Let us write $\yy_2\in\mathbb{R}^3$ in spherical coordinates as $\yy_2\equiv |\yy_2|\Omega_{\yy_2}$, with $\Omega_{\yy_2}\in\mathbb{S}^2$, and for $\rho>0$ and almost every $\yy_1\in\mathbb{R}^3$, let us define
\begin{equation}
 \psi_{\mathrm{av}}(\yy_1;\rho)\;:=\;\frac{1}{4\pi}\int_{\mathbb{S}^2}\psi(\yy_1,\rho\Omega)\,\ud\Omega\,.
\end{equation}
Thus, the function $\yy_1\mapsto\psi_{\mathrm{av}}(\yy_1;\rho)$ is the spherical average of the function $\yy_1\mapsto\psi(\yy_1,\yy_2)$ over the sphere with $|\yy_2|=\rho$.

For later purposes, we are concerned with certain meaningful behaviours of $\psi_{\mathrm{av}}(\yy_1;\rho)$ as $\rho\to 0$ at fixed $\yy_1$, namely when it either approaches a finite value or instead diverges as $\rho^{-1}$. With no pretension of full generality, let us adopt the following characterisation: we shall say that a measurable function $\varphi:[0,+\infty)\to\mathbb{C}$ displays `$\mathcal{Z}$-behaviour' (at zero) when $\varphi\in L^2(\mathbb{R}^+,\rho^2\ud \rho)$, $\varphi$ is continuous in a neighbourhood $(0,\varepsilon_\varphi)$ for some $\varepsilon_\varphi>0$,  and
\begin{equation}\label{eq:Z}
 \frac{\int_0^{+\infty}\!\ud\rho\,\frac{\,\sin\rho-\rho\cos\rho}{\rho}\,\varphi(\frac{\rho}{R})}{\varphi(\frac{1}{R})}\;\xrightarrow[]{\,R\to +\infty\,}\;\frac{\pi}{2}\,c_\varphi
\end{equation}
for some constant $c_\varphi\in\mathbb{C}$. In terms of the even extension $\phi(\rho):=\varphi(|\rho|)$, $\rho\in\mathbb{R}$, \eqref{eq:Z} is equivalent to
\begin{equation}\label{eq:Z2}
 \frac{\,\int_{-R}^{R}\widehat{\phi}(s)\,\ud s-2 R\,\widehat{\phi}(R)\,}{\,\int_\mathbb{R}e^{\ii s/R}\,\widehat{\phi}(s)\,\ud s\,}\;\xrightarrow[]{\,R\to +\infty\,}\;c_\varphi\,.
\end{equation}
The request that $\varphi\in L^2(\mathbb{R}^+,\rho^2\ud \rho)$ is made precisely with the function $\rho\mapsto\psi_{\mathrm{av}}(\yy_1;\rho)$ in mind, of course.

Observe that \eqref{eq:Z}-\eqref{eq:Z2} is just a convenient way to characterise the behaviour of $\varphi(\rho)$ as $\rho\to 0$. This is clear if one interprets \emph{separately} the two summands that emerge from the above expressions (but in general we do want to include the possible effect of \emph{compensation} between them). Thus, for instance, if $\phi$ is a Schwartz function with $\phi(0)=\varphi(0)\neq 0$, standard Riemann-Lebesgue and Fourier transform arguments yield
\[
 \begin{split}
  \int_0^{+\infty}\!\ud\rho\,\frac{\sin\rho}{\rho}\,\varphi({\textstyle\frac{\rho}{R}})\;&\xrightarrow[]{\,R\to +\infty\,}\;\varphi(0)\int_0^{+\infty}\!\ud\rho\,\frac{\sin\rho}{\rho}\;=\;\frac{\pi}{2}\,\varphi(0)\,, \\
  \int_0^{+\infty}\!\ud\rho\,\cos\rho\,\varphi({\textstyle\frac{\rho}{R}})\;&=\;\sqrt{\frac{\pi}{2}}\,R\,\widehat{\phi}(R)\;\xrightarrow[]{\,R\to +\infty\,}\;0\,,
 \end{split} 
\]
therefore in this case \eqref{eq:Z} is satisfied with $c_\varphi=1$. More generally, the asymptotic finiteness of the quantities
\[
 \frac{\,\int_{-R}^{R}\widehat{\phi}(s)\,\ud s\,}{\,\int_\mathbb{R}e^{\ii s/R}\,\widehat{\phi}(s)\,\ud s\,}\,,\qquad\frac{\,R\,\widehat{\phi}(R)\,}{\phi(\frac{1}{R})}
\]
as $R\to +\infty$ encodes a prescription on $\phi(\rho)$ as $\rho$ vanishes, including when $\phi$ (hence $\varphi$) is singular at $\rho=0$. In fact, \eqref{eq:Z}-\eqref{eq:Z2} encode in general a possible compensation among the above two summands. For instance, for the function $\varphi=\rho^{-1}\mathbf{1}_{(0,1)}$ one finds
\[
\frac{\int_0^{+\infty}\!\ud\rho\,\frac{\,\sin\rho-\rho\cos\rho}{\rho}\,\varphi(\frac{\rho}{R})}{\varphi(\frac{1}{R})}\,=\,\int_0^R\ud\rho\,\frac{\,\sin\rho-\rho\cos\rho}{\rho^2}\,=\,\Big[-\frac{\sin\rho}{\rho}\,\Big]_{0}^R\,\xrightarrow[]{\,R\to +\infty\,}\,1\,,
\]
meaning that in this case \eqref{eq:Z} is satisfied with $c_\varphi=\frac{2}{\pi}$.

Clearly, the $\mathcal{Z}$-behaviour is not the most general behaviour of $\rho\mapsto\psi_{\mathrm{av}}(\yy_1;\rho)$ when $\psi\in L^2(\mathbb{R}^3\times\mathbb{R}^3,\ud\yy_1\ud\yy_2)$ or even, for later applications, when $\psi$ belongs to the domain of self-adjoint operator of interest. It is generic enough, though, to comprise both functions $\varphi$ with sufficient regularity at $\rho=0$ and integrability over $[0,+\infty)$, and functions with enough integrability and local $\rho^{-1}$-singularity.

\begin{lemma}\label{lem:shortscalegeneric}
  Let $\psi\in\cH_\mathrm{b}$ such that for almost every $\yy_1$ the function $\rho\mapsto\psi_{\mathrm{av}}(\yy_1;\rho)$ has $\mathcal{Z}$-behaviour, for concreteness uniformly in $\yy_1$ (thus, with the same constant in the limit \eqref{eq:Z}).
 For $R>0$ and a.e.~$\pp_1$ let
 \begin{equation}\label{eq:ApsiR}
  \widehat{A}_{\psi,R}(\pp_1)\;:=\;\frac{1}{\;(2\pi)^{\frac{3}{2}}}\int_{\!\substack{ \\ \\ \pp_2\in\mathbb{R}^3 \\ |\pp_2|<R}}\widehat{\psi}(\pp_1,\pp_2)\,\ud \pp_2\,.
 \end{equation}
 Then, for a.e.~$\yy_1$, and for some constant $c_\psi\in\mathbb{C}$,
 \begin{equation}\label{eq:AR}
  A_{\psi,R}(\yy_1)\;=\; c_\psi\,\psi_{\mathrm{av}}(\yy_1;{\textstyle\frac{1}{R})}\,(1+o(1))\qquad \textrm{as }\;R\to +\infty\,.
 \end{equation}
\end{lemma}

\begin{remark}
 In the assumption of the Lemma $\psi$ may be singular at $\yy_2=0$, in which case both sides of \eqref{eq:AR} diverge with $R$. If instead $\psi$ is suitably regular at $\yy_2=0$, then the r.h.s.~converges to $\psi(\yy_1,\mathbf{0})$, consistently with \eqref{eq:psiy0}-\eqref{eq:tracemomentum} above. 
\end{remark}

\begin{proof}[Proof of Lemma \ref{lem:shortscalegeneric}]
 One has
 \[
  \begin{split}
    A_{\psi,R}(\yy_1)\;&=\;\frac{1}{\;(2\pi)^3}\iint_{\mathbb{R}^3\times\mathbb{R}^3}\ud\pp_1\ud\pp_2\,e^{\ii\pp_1\cdot\yy_1}\,\mathbf{1}_{\{|\pp_2|<R\}}(\pp_2)\,\widehat{\psi}(\pp_1,\pp_2) \\
    &=\;\iint_{\mathbb{R}^3\times\mathbb{R}^3}\ud\zz_1\ud\zz_2\,\delta(\zz_1+\yy_1)\,\delta_R(\zz_2)\,\psi(\zz_1,\zz_2) \\
    &=\;\int_{\mathbb{R}^3}\ud\zz_2\,\delta_R(\zz_2)\,\psi(-\yy_1,\zz_2)\;=\;\int_{\mathbb{R}^3}\ud\zz_2\,\delta_R(\zz_2)\,\psi(\yy_1,-\zz_2)\,,
  \end{split}
 \]
 where
 \[
  \begin{split}
   \delta_R(\zz_2)\;:=&\;\,\Big(\frac{\mathbf{1}_{\{|\pp_2|<R\}}}{(2\pi)^{\frac{3}{2}}}\Big)^{\!\vee}(\zz_2)\;=\;\frac{1}{\;(2\pi)^{3}}\int_{\!\substack{ \\ \\ \pp_2\in\mathbb{R}^3 \\ |\pp_2|<R}}e^{\ii \pp_2\cdot\zz_2}\,\ud\pp_2 \\
   =&\;\,\frac{2\pi}{\;(2\pi)^{3}}\int_0^R\ud r\,r^2\int_{-1}^1\,\ud t\,e^{\ii r |\zz_2| t} \\ 
   =&\;\,\frac{2 R^3}{\;(2\pi)^{2}}\,\frac{\,\sin R|\zz_2|-R|\zz_2|\cos R|\zz_2|\,}{(R|\zz_2|)^3}\,.
  \end{split}
 \]
 In fact, $\delta_R$ is a smooth, approximate delta-distribution in three dimensions. Thus,
 \[
  \begin{split}
   A_{\psi,R}(\yy_1)\;&=\;\frac{2}{\;(2\pi)^2}\int_{\mathbb{R}^3}\ud\zz_2\,\frac{\,\sin|\zz_2|-|\zz_2|\cos |\zz_2|\,}{|\zz_2|^3}\,\psi(\yy_1,{\textstyle\frac{1}{R}}\zz_2) \\
   &=\;\frac{2}{\;(2\pi)^2}\int_0^{+\infty}\!\ud\rho\,\frac{\,\sin\rho-\rho\cos\rho}{\rho}\Big(\int_{\mathbb{S}^2}\ud\Omega\,\psi(\yy_1,{\textstyle\frac{1}{R}}\rho\,\Omega)\Big) \\
   &=\;\frac{2}{\pi}\int_0^{+\infty}\!\ud\rho\,\frac{\,\sin\rho-\rho\cos\rho}{\rho}\,\psi_{\mathrm{av}}(\yy_1,\textstyle\frac{\rho}{R})\,.
  \end{split}  
 \]
 By assumption (see \eqref{eq:Z} above),
 \[
  \lim_{R\to +\infty}\frac{\int_0^{+\infty}\!\ud\rho\,\frac{\,\sin\rho-\rho\cos\rho}{\rho}\,\psi_{\mathrm{av}}(\yy_1,\textstyle\frac{\rho}{R})}{\psi_{\mathrm{av}}(\yy_1,\textstyle\frac{1}{R})}\;=\;\frac{\pi}{2}\,c_\psi
 \]
 for some constant $c_\psi\in\mathbb{C}$. Therefore,
 \[
  A_{\psi,R}(\yy_1)\,=\,\frac{2}{\pi}\int_0^{+\infty}\!\ud\rho\,\frac{\,\sin\rho-\rho\cos\rho}{\rho}\,\psi_{\mathrm{av}}(\yy_1,\textstyle\frac{\rho}{R})\,\stackrel{R\to +\infty}{=}  \,c_\psi\,\psi_{\mathrm{av}}(\yy_1,\textstyle\frac{1}{R})\,(1+o(1))\,,
 \]
which completes the proof. 
\end{proof}

\begin{remark}
 Should, more realistically, the function $\rho\mapsto\psi_{\mathrm{av}}(\yy_1,\rho)$ in the above proof display $\mathcal{Z}$-behaviour non-uniformly in $\yy_1$, the counterpart of the $c_\psi$-constant would be a function $c_\psi(\yy_1)$. We are not really interested in pushing such generality forward: we merely introduced the $\mathcal{Z}$-behaviour to visualize, in meaningful concrete cases, the correspondence between the two expressions \eqref{eq:ApsiR} and \eqref{eq:AR} with explicit dependence on the cut-off parameter $R$. 
\end{remark}

\subsection{The $T_\lambda$ operator}\label{sec:Tlambdaoperator}~

In practice, the computation of \eqref{eq:ApsiR} for elements of $\mathcal{D}(\mathring{H}^*)$ produces a quantity that for convenience we analyse separately in this Subsection, before resuming the discussion in the following Subsect.~\ref{sec:largemomentumasympt}.

For $\xi\in H^{-\frac{1}{2}}(\mathbb{R}^3)$ and $\lambda>0$ we define (for a.e.~$\pp$)
\begin{equation}\label{eq:Tlambda}
  (\widehat{T_\lambda\xi})(\pp)\;:=\;2\pi^2\sqrt{\frac{3}{4}\pp^2+\lambda}\;\widehat{\xi}(\pp)-2\int_{\mathbb{R}^3}\frac{\widehat{\xi}(\qq)}{\,\pp^2+\qq^2+\pp\cdot \qq+\lambda\,}\,\ud\qq\,.
 \end{equation}

At least for $\xi\in H^{-\frac{1}{2}+\varepsilon}(\mathbb{R}^3)$, $\varepsilon>0$, \eqref{eq:Tlambda} defines an almost-everywhere finite quantity, for
\[
\Big|\int_{\mathbb{R}^3}\frac{\widehat{\xi}(\qq)}{\,\pp^2+\qq^2+\pp\cdot \qq+\lambda\,}\,\ud\qq\Big|\;\lesssim\;\|\xi\|_{H^{-\frac{1}{2}+\varepsilon}}\Big(\int_{\mathbb{R}^3}\frac{(\qq^2+1)^{\frac{1}{2}-\varepsilon}}{(\pp^2+\qq^2+1)^2}\,\ud q\Big)^{1/2}\;<\;+\infty\,.
\]
Instead, the example $\widehat{\xi}_0(\qq):=\mathbf{1}_{\{|\qq|\geqslant 2\}}(|\qq|\log|\qq|)^{-1}$ shows that \eqref{eq:Tlambda} may be \emph{infinite} for generic $H^{-\frac{1}{2}}$-functions.

The map $\xi\mapsto T_\lambda\xi$ is central in this work. It commutes with the rotations in $\mathbb{R}^3$ and therefore, upon densely defining it over $H^s(\mathbb{R}^3)$, $s\geqslant-\frac{1}{2}$, one has
\begin{equation}\label{eq:decompTTell}
 T_\lambda\;=\;\bigoplus_{\ell=0}^\infty T_\lambda^{(\ell)}
\end{equation}
in the sense of direct sum of operators on Hilbert space with respect to the canonical decomposition
\begin{equation}\label{eq:bigdecomp}
 \begin{split}
 H^s(\mathbb{R}^3)\;&\cong\;\bigoplus_{\ell=0}^\infty \Big( L^2(\mathbb{R}^+,(1+p^2)^s p^2\ud p) \otimes \mathrm{span}\big\{ Y_{\ell,n}\,|\,n=-\ell,\dots,\ell \big\} \Big)  \\
 &\equiv\;\bigoplus_{\ell=0}^\infty \,H^s_{\ell}(\mathbb{R}^3)\,.
 \end{split}
\end{equation}
Here the $Y_{\ell,n}$'s form the usual orthonormal basis of $L^2(\mathbb{S}^2)$ of spherical harmonics and each $\xi\in H^s(\mathbb{R}^3)$ decomposes with respect to \eqref{eq:bigdecomp} as
\begin{equation}\label{eq:xihatangularexpansion}
 \begin{split}
  \widehat{\xi}(\pp)\;&=\;\sum_{\ell=0}^\infty\sum_{n=-\ell}^\ell f_{\ell,n}^{(\xi)}(|\pp|) Y_{\ell,n}(\Omega_{\pp})\;=\;\sum_{\ell=0}^\infty\widehat{\xi^{(\ell)}}(\pp) \\
  \widehat{\xi^{(\ell)}}(\pp)\;&\!:=\;\sum_{n=-\ell}^\ell f_{\ell,n}^{(\xi)}(|\pp|) Y_{\ell,n}(\Omega_{\pp})
 \end{split}
\end{equation}
in polar coordinates $\pp\equiv|\pp|\Omega_\pp$. Explicitly, 
\begin{equation}
 \langle\xi,\eta\rangle_{H^s}\;=\;\sum_{\ell=0}^\infty\langle\xi^{(\ell)},\eta^{(\ell)}\rangle_{H^s_\ell}\;=\;\sum_{\ell=0}^\infty\sum_{n=-\ell}^\ell\int_{\mathbb{R}^+}\overline{f_{\ell,n}^{(\xi)}(p)}\,f_{\ell,n}^{(\eta)}(p)\,(1+p^2)^s p^2\ud p 
\end{equation}
%
%
%
%
%
and 
\begin{equation}\label{eq:TlambdaTlambdaell}
 \widehat{T_\lambda\xi}\;=\;\sum_{\ell=0}^\infty \widehat{T_\lambda^{(\ell)}\xi^{(\ell)}}\,.
\end{equation}

Let us denote by $P_\ell$ the Legendre polynomial of order $\ell=0,1,2,\dots$, namely
\begin{equation}\label{def_Legendre}
P_\ell(t)\equiv\frac{1}{2^\ell \ell!}\,\frac{\ud^\ell}{\ud t^\ell}\,(t^2-1)^\ell\,.
\end{equation}

\begin{lemma}[Decomposition properties of $T_\lambda$]\label{lem:Tlambdadecomposition}
 Let $\lambda>0$ and $\xi^{(\ell)},\eta^{(\ell)}\in H^s_\ell(\mathbb{R}^3)$. 
 Then, with respect to the representation \eqref{eq:xihatangularexpansion},
 \begin{itemize}
  \item[(i)] $T_\lambda^{(\ell)}$ acts trivially (i.e., as the identity) on the angular components of $\widehat{\xi}^{(\ell)}$, and acts as
  \begin{equation}\label{eq:fellsector}
   f_{\ell,n}^{(\xi)}(p)\:\mapsto\:2\pi^2\sqrt{{\textstyle\frac{3}{4}}p^2+\lambda}\,f_{\ell,n}^{(\xi)}(p)-4\pi\!\!\int_{\mathbb{R}^+}\!\ud q\,q^2  f_{\ell,n}^{(\xi)}(q)\!\int_{-1}^1\frac{P_\ell(t)\,\ud t}{p^2+q^2+p \,q \,t+\lambda}
  \end{equation}
 on each radial component;
  \item[(ii)] one has
  \begin{equation}\label{eq:xiTxipre}
    \int_{\mathbb{R}^3} \overline{\,\widehat{\xi}(\pp)}\, \big(\widehat{T_\lambda\eta}\big)(\pp)\,\ud\pp\;=\;\sum_{\ell=0}^\infty \int_{\mathbb{R}^3} \overline{\,\widehat{\xi^{(\ell)}}(\pp)}\, \big(\widehat{T_\lambda^{(\ell)}\eta^{(\ell)}}\big)(\pp)\,\ud\pp
  \end{equation}
 and
 \begin{equation}\label{eq:xiTxi}
   \begin{split}
   \int_{\mathbb{R}^3}& \overline{\,\widehat{\xi^{(\ell)}}(\pp)}\, \big(\widehat{T_\lambda^{(\ell)}\eta^{(\ell)}}\big)(\pp)\,\ud\pp\;=\;2\pi^2\!\int_{\mathbb{R}^+}\!\ud p\,p^2\,\overline{f_{\ell,n}^{(\xi)}(p)}f_{\ell,n}^{(\eta)}(p)\sqrt{{\textstyle\frac{3}{4}}p^2+\lambda} \\
   &\qquad -4\pi\!\iint_{\mathbb{R}^+\times\mathbb{R}^+}\ud p\,\ud q\,p^2q^2\,\overline{f_{\ell,n}^{(\xi)}(p)}\,f_{\ell,n}^{(\eta)}(q)\!\int_{-1}^1\ud t\,\frac{P_\ell(t)}{p^2+q^2+p \,q \,t+\lambda}\,.
    \end{split}
  \end{equation}
 \end{itemize}
\end{lemma}

\begin{proof}
   The triviality of the action of $T_\lambda^{(\ell)}$ on the angular components is due to the invariance of $T_\lambda$ under rotations. All other formulas are then straightforwardly derived from \eqref{eq:Tlambda} by exploiting the following standard expansion in Legendre polynomials and the addition formula for spherical harmonics: 
\begin{equation}\label{expansion-Leg-poly}
\begin{split}
&\frac{1}{\pp^2+\qq^2+\pp\cdot\qq+\lambda} \;=\;\sum_{\ell=0}^\infty\frac{2\ell+1}{2}\!\int_{-1}^1\ud t\,\frac{P_\ell(t)\,P_\ell(\cos(\theta_{\pp,\qq}))}{\pp^2+\qq^2+|\pp|\,|\qq|\,t+\lambda} \\
&\quad =\; \sum_{\ell=0}^\infty 2\pi\int_{-1}^1\ud t\,\frac{P_\ell(t)}{\pp^2+\qq^2+|\pp|\,|\qq|\,t+\lambda}\sum_{r=-\ell}^\ell\overline{Y_{\ell r}(\Omega_{\pp})}\,Y_{\ell r}(\Omega_{\qq})
\end{split}
\end{equation}
    (see, e.g., \cite[Eq.~(8.814)]{Gradshteyn-tables-of-integrals-etc}).
  \end{proof}

\begin{lemma}[Mapping properties of $T_\lambda$]\label{lem:Tlambdaproperties} Let $\lambda>0$.

\begin{itemize}
 \item[(i)] For each $s\geqslant 1$ \eqref{eq:Tlambda} defines an operator
 \[
  T_\lambda:\mathcal{D}(T_\lambda)\subset L^2(\mathbb{R}^3)\to L^2(\mathbb{R}^3)\,,\qquad \mathcal{D}(T_\lambda)\;:=\;H^s(\mathbb{R}^3)
 \]
 that is densely defined and symmetric in $L^2(\mathbb{R}^3)$.
 \item[(ii)] One has
 \begin{equation}\label{eq:Tlambdamapping}
  \| T_\lambda \xi\|_{H^{s-1}}\;\lesssim\;\|\xi\|_{H^s}\qquad \forall\xi\in H^s(\mathbb{R}^3)\,,\qquad s\in\Big(-\frac{1}{2},\frac{3}{2}\Big)\,,
 \end{equation}
 i.e., \eqref{eq:Tlambda} defines a bounded operator $T_\lambda:H^s(\mathbb{R}^3)\to H^{s-1}(\mathbb{R}^3)$ for every $s\in(-\frac{1}{2},\frac{3}{2})$.
 \item[(iii)] One has
 \begin{equation}\label{eq:TlambdamappingN}
  \| T_\lambda^{(\ell)} \xi\|_{H^{s-1}}\;\lesssim\;\|\xi\|_{H^s}\qquad \forall\xi\in H^s_\ell(\mathbb{R}^3)\,,\qquad s\in\Big[-\frac{1}{2},\frac{3}{2}\Big]\,,\quad\ell\in\mathbb{N}\,,
 \end{equation}
 i.e., \eqref{eq:Tlambda} defines a bounded operator $T_\lambda:H^s_\ell(\mathbb{R}^3)\to H^{s-1}_\ell(\mathbb{R}^3)$ for every $s\in[-\frac{1}{2},\frac{3}{2}]$, provided that $\ell\in\mathbb{N}$. In the sector $\ell=0$ \eqref{eq:TlambdamappingN} fails in general at the endpoints in $s$ and only \eqref{eq:Tlambdamapping} is valid.
 \item[(iv)] For any other $\lambda'>0$ one has
 \begin{equation}
  \|(T_{\lambda'}-T_\lambda)\xi\|_{H^{\frac{1}{2}}}\;\lesssim\;|\lambda'-\lambda|\,\|\xi\|_{H^{-\frac{1}{2}}}\,.
 \end{equation}
 \item[(v)] For $s\geqslant\frac{1}{2}$ and $\xi,\eta\in H^s(\mathbb{R}^3)$
 one has
 \begin{equation}\label{eq:Tlambdaexchange}
 \int_{\mathbb{R}^3} \overline{\,\widehat{\xi}(\pp)}\, \big(\widehat{T_\lambda\eta}\big)(\pp)\,\ud\pp\;=\;\int_{\mathbb{R}^3} \overline{\,\widehat{T_\lambda\xi}(\pp)}\, \widehat{\eta}(\pp)\,\ud\pp\quad\qquad(s\geqslant\textstyle\frac{1}{2})
 \end{equation}
 and the quantity above is real and finite.
\end{itemize}
\end{lemma}

\begin{proof}
 All claims (i)-(iii) are obvious for the multiplicative part of $T_\lambda$, namely the first summand in the r.h.s.~of \eqref{eq:Tlambda}, and need only be proved for the integral part of $T_\lambda$. The latter, apart from an irrelevant multiplicative prefactor, is the same as the multiplicative part of the `fermionic' counterpart of $T_\lambda$, namely the analogous operator emerging in the analysis of a trimer consisting of two identical fermions and a third different particle. All the claimed properties were already demonstrated in that case in collaboration with Ottolini in \cite[Propositions 3 and 4, Corollary 2]{MO-2016}.

 Concerning part (iv),
 \[
 \begin{split}
  ((T_{\lambda'}-T_\lambda)&\xi)\,{\textrm{\large $\widehat{\,}$\normalsize}}\,(\pp)\;=\; \frac{2\pi^2(\lambda'-\lambda)}{\:\sqrt{\frac{3}{4}\pp^2+\lambda'}+\sqrt{\frac{3}{4}\pp^2+\lambda}\:}\,\widehat{\xi}(\pp)\\
  & +2(\lambda'-\lambda)\int_{\mathbb{R}^3}\,\frac{\widehat{\xi}(\qq)}{\,(\pp^2+\qq^2+\pp\cdot\qq+\lambda)\,(\pp^2+\qq^2+\pp\cdot\qq+\lambda')\,}\,\ud\qq
 \end{split}
 \]
 whence
 \[
   \big|(T_{\lambda'}-T_\lambda)\xi)\,{\textrm{\large $\widehat{\,}$\normalsize}}\,(\pp)\big|\;\lesssim\;|\lambda'-\lambda|\,\bigg( \frac{2\pi^2\,|\widehat{\xi}(\pp)|}{\,\sqrt{\frac{3}{4}\pp^2+\lambda}\,}+\int_{\mathbb{R}^3}\,\frac{|\widehat{\xi}(\qq)|}{\,(\pp^2+\qq^2+1)^2\,}\,\ud\qq\bigg)\,.
 \]
 Thus, $\xi\mapsto(T_{\lambda'}-T_\lambda)\xi$ has the same behaviour as $W_\lambda$, and hence the same $H^{-\frac{1}{2}}\to H^{\frac{1}{2}}$ boundedness.

 Concerning (v), the only non-trivial piece of the claim regards the integral part of $T_\lambda$, namely the identity
  \begin{equation*}
 \begin{split}
  \int_{\mathbb{R}^3}&\ud\pp\,\overline{\widehat{\xi}(\pp)}\,\Big(\int_{\mathbb{R}^3}\ud\qq\,\frac{\widehat{\eta}(\qq)}{\,\pp^2+\qq^2+\pp\cdot\qq+\lambda\,}\Big) \\
   &=\;\int_{\mathbb{R}^3}\ud\qq\,\Big(\int_{\mathbb{R}^3}\ud\pp\,\frac{\overline{\widehat{\xi}(\pp)}}{\,\pp^2+\qq^2+\pp\cdot\qq+\lambda\,}\Big)\,\widehat{\eta}(\qq)\,.
 \end{split}
 \end{equation*}
 The exchange of integration order above is indeed legitimate, as the assumptions on $\xi,\eta$ guarantee the applicability of Fubini-Tonelli theorem. More precisely,
 \[
 \begin{split}
  &\bigg|\int_{\mathbb{R}^3} \ud\pp\,\overline{\widehat{\xi}(\pp)}\,\Big(\int_{\mathbb{R}^3}\ud\qq\,\frac{\widehat{\eta}(\qq)}{\,\pp^2+\qq^2+\pp\cdot\qq+\lambda\,}\Big)\bigg| \\
  &\;\lesssim\;	 \int_{\mathbb{R}^3}\ud\pp\,|\widehat{\xi}(\pp)|\,\Big(\int_{\mathbb{R}^3}\ud\qq\,\frac{|\widehat{\eta}(\qq)|}{\,\pp^2+\qq^2+1\,}\Big) \;\leqslant\;\|\xi\|_{H^{\frac{1}{2}}}\bigg\|\int_{\mathbb{R}^3}\ud\qq\,\frac{|\widehat{\eta}(\qq)|}{\,\pp^2+\qq^2+1\,} \bigg\|_{H^{-\frac{1}{2}}}\\
  &\;\lesssim\;\|\xi\|_{H^{\frac{1}{2}}}\|\eta\|_{H^{\frac{1}{2}}}\;\leqslant\;\|\xi\|_{H^{s}}\|\eta\|_{H^{s}}\;<\;+\infty\qquad\forall s\geqslant\frac{1}{2}\,,
 \end{split}
 \]
 where we applied \eqref{eq:lambda-equiv-1} in the first inequality and \eqref{eq:Tlambdamapping} in the third (estimate \eqref{eq:Tlambdamapping} refers to the whole $T_\lambda$, but as commented above in the course of its proof it is actually established by demonstrating the only non-trivial piece of the estimate, namely the one involving the integral part of $T_\lambda$). 
\end{proof}


\begin{remark}\label{rem:Tl-failstomap}
 $T_\lambda$ fails to map $H^{\frac{3}{2}}(\mathbb{R}^3)$ into $H^{\frac{1}{2}}(\mathbb{R}^3)$ as is the case, for instance, for the action of $T_\lambda$ on the class of spherically symmetric functions in $\mathcal{F}^{-1} C^\infty_0(\mathbb{R}^3_\pp)$. Indeed, if $\xi$ has symmetry $\ell=0$ and $\widehat{\xi}\in C^\infty_0(\mathbb{R}^3_\pp)$, then the contribution from the integral part of $(\widehat{T_\lambda\xi})(\pp)$ is of the order of (see \eqref{eq:fellsector} above)
 \[
 \begin{split}
  \int_{\mathrm{supp}\,f^{(\xi)}}&\ud q\,q^2\,f^{(\xi)}(q)\int_{-1}^1\frac{\ud t}{\,|\pp|^2+q^2+|\pp|qt+\lambda\,} \\
  &=\;\frac{1}{|\pp|}\int_{\mathrm{supp}\,f^{(\xi)}}\ud q\,q\,f^{(\xi)}(q)\,\log\Big(1+\frac{2|\pp|q}{\,|\pp|^2+q^2-|\pp|qt+\lambda\,}\Big)\,,
 \end{split}
 \]
 which, both in the limit $|\pp|\to 0$ and $|\pp|\to +\infty$ is of the order of
 \[
  \int_{\mathrm{supp}\,f^{(\xi)}}\ud q\,\frac{q^2}{\,|\pp|^2+q^2-|\pp|qt+\lambda\,}\,f^{(\xi)}(q)\;\sim\;\frac{1}{\,\pp^2+1\,}\,.
 \]
 The contribution from the multiplicative part of $(\widehat{T_\lambda\xi})(\pp)$ is obviously a compactly supported function, the conclusion therefore is $(\widehat{T_\lambda\xi})(\pp)\sim (\pp^2+1)^{-1}$, and the latter is a $H^{\frac{1}{2}-\varepsilon}$-function $\forall\varepsilon>0$ not belonging to $H^{\frac{1}{2}}(\mathbb{R}^3)$. 
\end{remark}

\begin{remark}\label{rem:whensymmetric}
 Parts (i) and (v) of Lemma \ref{lem:Tlambdaproperties} present two regimes of validity of the identity \eqref{eq:Tlambdaexchange} when $\xi,\eta\in H^s(\mathbb{R}^3)$ for $s\geqslant\frac{1}{2}$. In the regime $s\geqslant 1$, each side of the \eqref{eq:Tlambdaexchange} is a product of two $L^2$-functions and such identity amounts to the symmetry of $T_\lambda$ in $L^2(\mathbb{R}^3)$ with domain $H^s(\mathbb{R}^3)$. For $\frac{1}{2}\leqslant s<1$, instead, $T_\lambda$ does not make sense any longer as an operator on $L^2(\mathbb{R}^3)$, and yet \eqref{lem:Tlambdaproperties} still expresses the symmetry of the action of $T_\lambda$ on $H^s$-functions, and hence also the reality of the considered integrals. 
\end{remark}

 Additional relevant properties $T_\lambda$ are discussed in Subsect.~\ref{sec:Tlambdaestimates}.

\subsection{Large momentum asymptotics}\label{sec:largemomentumasympt}~

\begin{lemma}\label{lem:largepasympt-star}
 Let $g\in\mathcal{D}(\mathring{H}^*)$ and let $\lambda>0$. Then, decomposing $\widehat{g}=\widehat{\phi^\lambda}+\widehat{u_\xi^\lambda}$ with $\widehat{\phi^\lambda}=\widehat{f^\lambda}+(\pp_1^2+\pp_2^2+\pp_1\cdot\pp_2+\lambda)^{-1}\,\widehat{u_\eta^\lambda}$ as demonstrated in Lemma \ref{lem:Hstaretc}, in the limit $R\to +\infty$ one has the asymptotics
 \begin{equation}\label{eq:g-largep2-star}
  \int_{\!\substack{ \\ \\ \pp_2\in\mathbb{R}^3 \\ |\pp_2|<R}}\widehat{g}(\pp_1,\pp_2)\,\ud\pp_2\;=\;4\pi R\,\widehat{\xi}(\pp_1)+\Big({\textstyle\frac{1}{3}}(\widehat{W_\lambda\eta})(\pp_1)-(\widehat{T_\lambda\xi})(\pp_1)\Big)+o(1)
 \end{equation}
 as well as the identity
 \begin{equation}\label{eq:phi-largep2-star}
  \int_{\mathbb{R}^3}\widehat{\phi^\lambda}(\pp_1,\pp_2)\,\ud\pp_2\;=\;{\textstyle\frac{1}{3}}(\widehat{W_\lambda\eta})(\pp_1)\,.
 \end{equation} 
\end{lemma}

An immediate corollary of Lemma \ref{lem:largepasympt-star}, obtained by means of Lemma \ref{lem:shortscalegeneric} taking $R=|\yy_2|^{-1}\to +\infty$, is the following.

\begin{corollary}\label{cor:largepasympt-star}
 Under the assumptions of Lemmas \ref{lem:largepasympt-star} and \ref{lem:shortscalegeneric} one has
 \begin{equation}\label{eq:g-largep2-star-yversion}
 (2\pi)^{\frac{3}{2}} c_g\,g_{\mathrm{av}}(\yy_1;|\yy_2|)\,\stackrel{|\yy_2|\to 0}{=}\,\frac{4\pi}{|\yy_2|}\xi(\yy_1)+\Big( {\textstyle\frac{1}{3}}(W_\lambda\eta)(\yy_1)-(T_\lambda\xi)(\yy_1)\Big) + o(1)  
 \end{equation}
 for some constant $c_g\in\mathbb{C}$, and 
 \begin{equation}\label{eq:phi-largep2-star-yversion}
   \phi^\lambda(\yy_1,\mathbf{0})\;=\;\frac{1}{\,3\,(2\pi)^{\frac{3}{2}}}\,(W_\lambda\eta)(\yy_1)
 \end{equation}
 for a.e.~$\yy_1$. 
\end{corollary}

\begin{remark}\label{rem:gexpansionfinite}
 Not for all $g\in \mathcal{D}(\mathring{H}^*)$ are \eqref{eq:g-largep2-star} and \eqref{eq:g-largep2-star-yversion} \emph{finite} quantities, but surely they are if the charge $\xi$ of $g$ has at $H^{-\frac{1}{2}+\varepsilon}$-regularity for some $\varepsilon>0$ (as argued right after the definition \eqref{eq:Tlambda}). 
\end{remark}

\begin{proof}[Proof of Lemma \ref{lem:largepasympt-star}]
 For what observed in Remark \ref{rem:gexpansionfinite}, we tacitly restrict the computations to those $\xi$'s making the following integrals finite (e.g., all $\xi$'s with $H^{-\frac{1}{2}+\varepsilon}$-regularity), for otherwise the corresponding identities to prove are all identities between infinites.
 
 One has
 \[
  \begin{split}
    &\int_{\!\substack{ \\ \\ \pp_2\in\mathbb{R}^3 \\ |\pp_2|<R}}\widehat{u_\xi^\lambda}(\pp_1,\pp_2)\ud\pp_2\;=\;\widehat{\xi}(\pp_1) \int_{\!\substack{ \\ \\ \pp_2\in\mathbb{R}^3 \\ |\pp_2|<R}}\frac{\ud \pp_2}{\,\pp_1^2+\pp_2^2+\pp_1\cdot\pp_2+\lambda\,} \\
    &\quad +\int_{\!\substack{ \\ \\ \pp_2\in\mathbb{R}^3 \\ |\pp_2|<R}}\frac{\widehat{\xi}(\pp_2)}{\,\pp_1^2+\pp_2^2+\pp_1\cdot\pp_2+\lambda\,}\,\ud\pp_2+\int_{\!\substack{ \\ \\ \pp_2\in\mathbb{R}^3 \\ |\pp_2|<R}}\frac{\widehat{\xi}(-\pp_1-\pp_2)}{\,\pp_1^2+\pp_2^2+\pp_1\cdot\pp_2+\lambda\,}\,\ud\pp_2\,.
  \end{split}
 \]
 Both last two summands in the r.h.s.~above converge as $R\to +\infty$ to
 \[
  \int_{\mathbb{R}^3}\frac{\widehat{\xi}(\pp_2)}{\,\pp_1^2+\pp_2^2+\pp_1\cdot\pp_2+\lambda\,}\,\ud\pp_2
 \]
 (for the third one this follows after an obvious change of the integration variable). Moreover,
 \[
\begin{split}
\int_{\!\substack{ \\ \\ \pp_2\in\mathbb{R}^3 \\ |\pp_2|<R}}&\frac{\ud \pp_2}{\,\pp_1^2+\pp_2^2+\pp_1\cdot\pp_2+\lambda\,}\;=\;2\pi\int_0^R\ud r\,r^2\int_{-1}^1\frac{\ud t}{\pp_1^2+r^2+|\pp_1| r t+\lambda} \\
&=\;\frac{2\pi}{|\pp_1|}\int_0^R r\log\frac{r^2+\pp_1^2+|\pp_1|r+\lambda}{r^2+\pp_1^2-|\pp_1|r+\lambda}\,\ud r \\
&=\;2\pi R\,\Big(1+\frac{R}{2|\pp_1|}\log\frac{R^2+\pp_1^2+|\pp_1|R+\lambda}{R^2+\pp_1^2-|\pp_1| R+\lambda}\Big) \\
&\qquad\quad+2\pi\sqrt{\frac{3}{4} \pp_1^2+\lambda}\,\Big(\!\arctan\frac{|\pp_1|-2R}{2\sqrt{\frac{3}{4} \pp_1^2+\lambda}}-\arctan\frac{|\pp_1|+2R}{2\sqrt{\frac{3}{4} \pp_1^2+\lambda}}\,\Big) \\
&\qquad\quad+\pi\frac{\pp_1^2+\lambda}{4\sqrt{\frac{3}{4} \pp_1^2+\lambda}}\,\log\frac{R^2+\pp_1^2+|\pp_1|R+\lambda}{R^2+\pp_1^2-|\pp_1| R+\lambda} \\
&=\;4\pi R-2\pi^2\sqrt{\frac{3}{4} \pp_1^2+\lambda}+o(1)\qquad\textrm{as }\;R\to +\infty\,.
\end{split}
\]
Thus,
\[
\begin{split}
 &\int_{\!\substack{ \\ \\ \pp_2\in\mathbb{R}^3 \\ |\pp_2|<R}}\widehat{u_\xi^\lambda}(\pp_1,\pp_2)\ud\pp_2 \\
 &=\;4\pi R\,\widehat{\xi}(\pp_1)-2\pi^2\sqrt{\frac{3}{4} \pp_1^2+\lambda}\,\widehat{\xi}(\pp_1)+2 \int_{\mathbb{R}^3}\frac{\widehat{\xi}(\pp_2)}{\,\pp_1^2+\pp_2^2+\pp_1\cdot\pp_2+\lambda\,}\,\ud\pp_2+o(1) \\
 &=\;4\pi R\,\widehat{\xi}(\pp_1)-(\widehat{T_\lambda}\xi)(\pp_1)+o(1)\,.
\end{split}
\]

Next, we compute (using $\int_{\mathbb{R}^3}f^\lambda(\pp_1,\pp_2)\ud\pp_2=0$)
\[
 \begin{split}
 &\int_{\mathbb{R}^3}\widehat{\phi^\lambda}(\pp_1,\pp_2)\,\ud\pp_2\;=\;\int_{\mathbb{R}^3}\frac{\widehat{u}_\eta(\pp_1,\pp_2)}{\,\pp_1^2+\pp_2^2+\pp_1\cdot\pp_2+\lambda\,} \\
 &=\;\widehat{\eta}(\pp_1) \int_{\mathbb{R}^3}\frac{\ud \pp_2}{\,(\pp_1^2+\pp_2^2+\pp_1\cdot\pp_2+\lambda)^2} \\
    &\quad +\int_{\mathbb{R}^3}\frac{\widehat{\eta}(\pp_2)}{\,(\pp_1^2+\pp_2^2+\pp_1\cdot\pp_2+\lambda)^2}\,\ud\pp_2+\int_{\mathbb{R}^3}\frac{\widehat{\eta}(-\pp_1-\pp_2)}{\,(\pp_1^2+\pp_2^2+\pp_1\cdot\pp_2+\lambda)^2}\,\ud\pp_2\,.
 \end{split}
\]
By an obvious change of variable one sees that the last two summands are the same. Moreover, 
\[
 \int_{\mathbb{R}^3}\frac{\ud \pp_2}{\,(\pp_1^2+\pp_2^2+\pp_1\cdot\pp_2+\lambda)^2}\;=\;\frac{\pi^2}{\,\sqrt{\frac{3}{4} \pp_1^2+\lambda}\,}\,.
\]
Therefore,
\[
\begin{split}
 \int_{\mathbb{R}^3}\widehat{\phi^\lambda}(\pp_1,\pp_2)\,\ud\pp_2\;&=\;\frac{\pi^2}{\,\sqrt{\frac{3}{4} \pp_1^2+\lambda}\,}+2\int_{\mathbb{R}^3}\frac{\widehat{\eta}(\pp_2)}{\,(\pp_1^2+\pp_2^2+\pp_1\cdot\pp_2+\lambda)^2}\,\ud\pp_2 \\
 &=\;\frac{1}{3}\,(\widehat{W_\lambda\eta})(\pp_1)
\end{split}
\]
This proves \eqref{eq:phi-largep2-star}, and combining this with the above results for $\widehat{u_\xi^\lambda}$ one proves \eqref{eq:g-largep2-star}. 
\end{proof}

By exploiting the bosonic symmetry and repeating the above arguments with respect the other coincidence hyperplanes, one finally obtains the following picture:
\begin{itemize}
 \item a function $f\in \mathcal{D}(\mathring{H})$ vanishes by definition in a neighbourhood of the coincidence manifold $\Gamma$;
 \item at each hyperplane, away from the configuration of triple coincidence, a function $\phi\in\mathcal{D}(\mathring{H}_F)$ is finite, as shown by \eqref{eq:phi-largep2-star-yversion};
 \item a generic $g\in\mathcal{D}(\mathring{H}^*)$ display the $|\yy|^{-1}$ singularity, as shown by \eqref{eq:g-largep2-star-yversion}.
\end{itemize}
Actually, \eqref{eq:g-largep2-star-yversion} and  \eqref{eq:phi-largep2-star-yversion} express the short-scale behaviour counterpart of the large momentum asymptotics \eqref{eq:g-largep2-star} and \eqref{eq:phi-largep2-star}, respectively -- and always with the caveat that the asymptotics for $g$ have finite coefficients only for a subclass of charges $\xi$ (which includes all charges with $H^{-\frac{1}{2}+\varepsilon}$-regularity).

The leading singularity of $g$ is of order $|\yy|^{-1}$ in the relative variable with respect to the considered coincidence hyperplane. Explicitly, in terms of the charges $\xi$ and $\eta$ of $g$,
\begin{equation}\label{eq:shortscalegeneric2}
 g_{\mathrm{av}}(\yy_1;|\yy_2|)\,\stackrel{|\yy_2|\to 0}{=}\,c_g^{-1}\sqrt{\frac{2}{\pi}\,}\,\Big(\frac{\xi(\yy_1)}{|\yy_2|}+\omega_{\xi,\eta}(\yy_1)\Big)+o(1)\,,
\end{equation}
\begin{equation}
 \omega_{\xi,\eta}(\yy_1)\;:=\; {\textstyle\frac{1}{4\pi}}\big({\textstyle\frac{1}{3}}(W_\lambda\eta)(\yy_1)-(T_\lambda\xi)(\yy_1)\big)\,,
\end{equation}
point-wise almost-everywhere in $\yy_1$. Analogous expressions hold with respect to the other coincidence hyperplanes, with the same $\xi$ and $\omega_{\xi,\eta}$.

 $\xi$ and $\omega_{\xi,\eta}$  are interpreted in \eqref{eq:shortscalegeneric2} as functions supported on the coincidence hyperplane ($\{\yy_2=\mathbf{0}\}$ in this case). The leading singularity's coefficient $\xi$ has some $H^{-\frac{1}{2}}$-regularity (in fact, more than that). The next-to-leading singularity's coefficient $\omega_{\xi,\eta}$, in general, is not even $H^{-\frac{1}{2}}$-regular, owing to the mapping properties of $W_\lambda$ (Lemma \ref{lem:Wlambdaproperties}) and $T_\lambda$ (Lemma \ref{lem:Tlambdaproperties}). But if the charge $\xi$ is absent in $g$, and hence $g\in H^2(\mathbb{R}^3\times\mathbb{R}^3,\ud\yy_1\ud\yy_2)$, then $\omega_{\xi,\eta}$ has the same regularity of $W_\lambda\eta$, namely the very $H^{\frac{1}{2}}(\mathbb{R}^3)$-regularity prescribed by the trace theorem.

When a self-adjoint extension $\mathring{H}_{\mathcal{A}_\lambda}$ is considered, and hence the subspace $\mathcal{D}(\mathring{H}_{\mathcal{A}_\lambda})$ is selected out of $\mathcal{D}(\mathring{H}^*)$ by means of the constraint \eqref{eq:constraintetaxi} on the charges $\xi$ and $\eta$, in practice one makes a choice in the class of leading coefficients $\xi$ and subleading coefficients $\omega_{\xi,\eta}$ of the short-scale expansion \eqref{eq:shortscalegeneric2}  which amounts to taking
\begin{equation}\label{eq:selection}
 \begin{split}
  \xi \;&\in\;\mathcal{D}(\mathcal{A}_\lambda)\,, \\
  4\pi\,\omega_{\xi,\eta}\;&=\; {\textstyle\frac{1}{3}} W_\lambda(\mathcal{A}_\lambda\xi+\chi)-T_\lambda\xi\qquad\textrm{for some }\chi\in\mathcal{D}(\mathcal{A}_\lambda)^{\perp_\lambda}\cap H^{-\frac{1}{2}}(\mathbb{R}^3)\,.
 \end{split}
\end{equation}


\section{Ter-Martirosyan Skornyakov extensions}\label{sec:TMSextension-section}

\subsection{TMS and BP asymptotics}\label{sec:TMS-BP-asympt}~

Given $g\in\cH_\mathrm{b}$ such that $\int_{|\pp_2|<R}\widehat{g}(\pp_1,\pp_2)\ud\pp_2<+\infty$ for all $R>0$, we shall say that $g$ satisfies the Ter-Martirosyan Skornyakov (TMS) condition with parameter $a\in(\mathbb{R}\setminus\{0\})\cup\{\infty\}$ if there exists a function $\xi_\circ$ such that
 \begin{equation}\label{eq:g-TMS-generic}
  \int_{\!\substack{ \\ \\ \pp_2\in\mathbb{R}^3 \\ |\pp_2|<R}}\widehat{g}(\pp_1,\pp_2)\,\ud\pp_2\,\stackrel{R\to +\infty}{=}\,4\pi \Big(R-\frac{1}{a}\Big)\widehat{\xi_\circ}(\pp_1)+o(1)\,.
 \end{equation}
An operator $K$ on $\cH_\mathrm{b}$ for which all $g$'s of $\mathcal{D}(K)$ with $\int_{|\pp_2|<R}\widehat{g}(\pp_1,\pp_2)\ud\pp_2<+\infty$ $\forall R>0$ satisfy \eqref{eq:g-TMS-generic} shall be called a Ter-Martirosyan Skornyakov operator. (For the time being this definition is kept deliberately general: $K$ may or may not be densely defined, symmetric, self-adjoint, etc., and nothing is said about the class $\xi_\circ$ belongs to.)

For those $g$'s of $\mathcal{D}(K)$ for which it is possible to repeat the arguments of Lemma \ref{lem:shortscalegeneric}, \eqref{eq:g-TMS-generic} amounts to
 \begin{equation}\label{eq:g-TMS-BP-generic}
 g_{\mathrm{av}}(\yy_1;|\yy_2|)\,\stackrel{|\yy_2|\to 0}{=}\,c_g^{-1}\sqrt{\frac{2}{\pi}\,}\,\Big(\frac{1}{|\yy_2|}-\frac{1}{a}\Big)\,\xi_\circ(\yy_1) + o(1)  
 \end{equation}
   point-wise almost everywhere in $\yy_1\in\mathbb{R}^3$. (The numerical pre-factors appearing in the r.h.s.~of \eqref{eq:g-TMS-generic} and \eqref{eq:g-TMS-BP-generic} are merely prepared for the forthcoming application to the analysis of the extensions of $\mathring{H}$.)

The case $a=0$ in \eqref{eq:g-TMS-generic}-\eqref{eq:g-TMS-BP-generic} would correspond to $\xi_\circ\equiv 0$ and hence to the fact that the quantity $\int_{\mathbb{R}^3}\widehat{g}(\pp_1,\pp_2)\ud\pp_2$ is finite. As the TMS condition is meant to pinpoint an actual \emph{singularity} of $g$ at each coincidence hyperplane, one conventionally excludes $a=0$ from the above definition.

 In fact, \eqref{eq:g-TMS-generic}-\eqref{eq:g-TMS-BP-generic} describe a short-scale structure of $g$ in the vicinity of each coincidence hyperplane (but away from the triple coincidence point) which in spatial coordinates has precisely the form of the Bethe Peierls contact condition \eqref{eq:preBP} expected on physical grounds for the eigenfunctions of a quantum trimer with zero-range interaction: in this interpretation, $a$ is the two-body $s$-wave scattering length of the interaction.

 We shall refer to \eqref{eq:g-TMS-BP-generic} too as the Bethe Peierls (BP) condition and we shall equivalently say that in the TMS condition \eqref{eq:g-TMS-generic}, resp., the BP condition \eqref{eq:g-TMS-BP-generic}, the quantity
  \begin{equation}\label{eq:a-alpha}
  \alpha\;:=\;-\frac{4\pi}{a}\;\in\;\mathbb{R}
 \end{equation}
  is the `inverse (negative) scattering length' (in suitable units).

 This indicates that within the huge variety of self-adjoint extensions of $\mathring{H}$ (Theorem \ref{thm:generalclassification}), the physically meaningful ones are those displaying the TMS condition for functions of their domain.

 De facto some arbitrariness in the modelling still remains, as we shall elaborate further on in due time (Subsect.~\ref{sec:definite-ell-general} and \ref{sec:variants}), for one could deem an extension `physically meaningful'
 \begin{itemize}
  \item in the restrictive sense that \emph{all} functions in the domain of the extension satisfy the TMS asymptotics (meaning, all functions $g$ for which, at any $R>0$, the quantity $\int_{|\pp_2|<R}\widehat{g}(\pp_1,\pp_2)\ud\pp_2$ is finite),
  \item in the milder sense that only \emph{some} relevant functions do, for instance declaring the physical asymptotics for functions with given symmetry, or for certain \emph{eigenfunctions} of the extension.
 \end{itemize}

  Let us examine first the possibility that \emph{at least one} function in the domain of a self-adjoint extension of $\mathring{H}$ satisfies the TMS condition.

 \begin{lemma}\label{eq:oneTMSfunction}
  Let $\mathscr{H}$ be a self-adjoint extension of $\mathring{H}$ and let $\alpha\in\mathbb{R}$. Assume that there exists $g\in\mathcal{D}(\mathscr{H})$ satisfying the TMS condition \eqref{eq:g-TMS-generic} for the given $\alpha$ and for some function $\xi_\circ$. One has the following.
  \begin{itemize}
   \item[(i)] $\xi_\circ$ must coincide with the charge $\xi$ of (the singular part of) $g$ (Lemma \ref{lem:chargexiofg}):
   \begin{equation}\label{eq:xixicirc}
     \xi\;=\;\xi_\circ \,.
   \end{equation}
   \item[(ii)] \emph{For every} shift parameter $\lambda>0$ with respect to which the canonical representation \eqref{eq:domDHA} of $g$ is written, the charges $\xi\in\mathcal{D}(\mathcal{A}_\lambda)$ and $\chi\in\mathcal{D}(\mathcal{A}_\lambda)^{\perp_\lambda}\cap H^{-\frac{1}{2}}(\mathbb{R}^3)$ of $g$ must satisfy
   \begin{eqnarray}
     {\textstyle\frac{1}{3}} W_\lambda(\mathcal{A}_\lambda\xi+\chi)\!\!&=&\!\!T_\lambda\xi+\alpha\xi \label{eq:TMSoncharges} \\
     T_\lambda\xi+\alpha\xi\!\!&\in&\!\! H^{\frac{1}{2}}(\mathbb{R}^3)\,. \label{eq:Ta12}
   \end{eqnarray}
   \item[(iii)] \emph{For every} shift parameter $\lambda>0$, $g$ and its regular part $\phi^\lambda$ must satisfy 
 \begin{eqnarray}
  \int_{\!\substack{ \\ \\ \pp_2\in\mathbb{R}^3 \\ |\pp_2|<R}}\widehat{g}(\pp_1,\pp_2)\,\ud\pp_2\!\!&=&\!\!(4\pi R+\alpha)\,\widehat{\xi}(\pp_1)+o(1)\,, \label{eq:g-largep2-TMS0} \\
  \int_{\mathbb{R}^3}\widehat{\phi^\lambda}(\pp_1,\pp_2)\,\ud\pp_2\!\!&=&\!\!(\widehat{T_\lambda\xi})(\pp_1)+\alpha\,\widehat{\xi}(\pp_1) \label{eq:phi-largep2-TMS0}\,, \\
  \phi^\lambda(\yy_1,\mathbf{0})\!\!&=&\!\!\frac{1}{\,(2\pi)^{\frac{3}{2}}}\,\Big((T_\lambda\xi)(\yy_1)+\alpha\,\xi(\yy_1)\Big) \label{eq:phi-largep2-star-yversion-TMS0}\,,
 \end{eqnarray}
 \eqref{eq:g-largep2-TMS0}, \eqref{eq:phi-largep2-TMS0}, \eqref{eq:phi-largep2-star-yversion-TMS0} being equivalent.
 \end{itemize}
 \end{lemma}

 \begin{proof} Let $\lambda>0$ and write
  \[
   \widehat{g}\,=\,\widehat{f^\lambda}+\displaystyle\frac{\widehat{u_{\mathcal{A}_\lambda\xi+\chi}^\lambda}}{\pp_1^2+\pp_2^2+\pp_1\cdot\pp_2+\lambda}+\widehat{u_\xi^\lambda}
  \]
 according to \eqref{eq:domDHA}. 
 Comparing the asymptotics \eqref{eq:g-largep2-star} valid for such $g$ with the asymptotics \eqref{eq:g-TMS-generic} assumed in the hypothesis, one deduces
 \[
  \begin{split}
   \xi\;&=\;\xi_\circ \\
   {\textstyle\frac{1}{3}} W_\lambda(\mathcal{A}_\lambda\xi+\chi)-T_\lambda\xi\; &=\;\alpha\,\xi_\circ\,.
  \end{split}
 \]
 From the arbitrariness of $\lambda$ one concludes that the identity \eqref{eq:TMSoncharges} holds true \emph{irrespective of $\lambda$}. Moreover, $T_\lambda\xi+\alpha\xi$ must make sense as a function in $\mathrm{ran}W_\lambda=H^{\frac{1}{2}}(\mathbb{R}^3)$ (Lemma \ref{lem:Wlambdaproperties}) irrespective of $\lambda$. This completes the proof of parts (i) and (ii).
  Plugging \eqref{eq:TMSoncharges} into \eqref{eq:g-largep2-star} and \eqref{eq:phi-largep2-star} yields \eqref{eq:g-largep2-TMS0} and \eqref{eq:phi-largep2-TMS0}, and in fact one can be derived one from the other, by comparison with the corresponding general identities \eqref{eq:g-largep2-star} and \eqref{eq:phi-largep2-star}. In turn, \eqref{eq:phi-largep2-TMS0} and \eqref{eq:phi-largep2-star-yversion-TMS0} correspond to each other via Fourier transform. Parts (ii) and (iii) too is proved.
 \end{proof}

 \begin{remark}
  Owing to Lemma \ref{eq:oneTMSfunction}(ii), \eqref{eq:phi-largep2-star-yversion-TMS0} must be an identity in $H^{\frac{1}{2}}(\mathbb{R}^3)$. This is indeed consistent with the $H^2(\mathbb{R}^3\times\mathbb{R}^3)\to H^{\frac{1}{2}}(\mathbb{R}^3)$ trace properties.
 \end{remark}


 \begin{remark}
  The `lesson' from Lemma \ref{eq:oneTMSfunction} is that TMS condition \emph{and} self-adjointness of the extension impose strong restrictions. A function $g$ fulfilling the TMS condition \emph{inside the domain of a self-adjoint extension of $\mathring{H}$} must satisfy the restrictions \eqref{eq:xixicirc}-\eqref{eq:Ta12}  $\forall\lambda>0$. 
 \end{remark}

 \subsection{Generalities on TMS extensions}\label{sec:generalitiesTMSext}~
 
 There are two relevant types of operators related with $\mathring{H}$ and compatible with the emergence of large-momenta / short-scale asymptotics of Ter-Martirosyan Skornyakov / Bethe Peierls type:
 
 \noindent \emph{an operator $\mathscr{H}$ on $\cH_\mathrm{b}$ such that $\mathring{H}\subset\mathscr{H}=\mathscr{H}^*$ (respectively, $\mathring{H}\subset\mathscr{H}\subset\mathscr{H}^*$)  and such that \emph{every} $g\in\mathcal{D}(\mathscr{H})$ satisfies the TMS condition \eqref{eq:g-largep2-TMS0} with the same given $\alpha$ shall be called a Ter-Martirosyan Skornyakov self-adjoint extension (respectively, Ter-Martirosyan Skornyakov symmetric extension) of $\mathring{H}$ with parameter $\alpha$}.

 \begin{remark}\label{rem:TMSsym}
 While the above definition in the \emph{self-adjoint} case is self-explanatory, based on the preceding analysis, as the self-adjoint extensions of $\mathring{H}$ are classified in Theorem \ref{thm:generalclassification} and the circumstance that $g\in\mathcal{D}(\mathscr{H})$ satisfies the TMS condition is analysed in Lemma \ref{eq:oneTMSfunction}, a clarification is in order for the \emph{symmetric} case.
 In fact, formulas \eqref{eq:domDHA}-\eqref{eq:domDHA-actionDHA} above make sense also when $\mathcal{A}_\lambda$ is simply symmetric (not necessarily self-adjoint) in $H^{-\frac{1}{2}}_{W_\lambda}(\mathbb{R}^3)$, in which case the operator $H_{\mathcal{A}_\lambda}$ thus defined is evidently still an extension of $\mathring{H}$. Let us show that $H_{\mathcal{A}_\lambda}$ is also symmetric. For generic $g\in\mathcal{D}(\mathscr{H})$,
 \[
  \begin{split}
   &\langle g , (H_{\mathcal{A}_\lambda} +\lambda\mathbbm{1})g\rangle_{\cH_\mathrm{b}}\;=\;\big\langle \phi^\lambda+ u_\xi^\lambda , (H_{\mathcal{A}_\lambda} +\lambda\mathbbm{1}) (\phi^\lambda+ u_\xi^\lambda)\big\rangle_{\cH_\mathrm{b}} \\
   &=\;\big\langle \phi^\lambda , (H_{F} +\lambda\mathbbm{1}) \phi^\lambda\big\rangle_{\cH_\mathrm{b}}+ \big\langle u_\xi^\lambda , (H_{F} +\lambda\mathbbm{1}) \phi^\lambda\big\rangle_{\cH_\mathrm{b}} \\
   &=\;\big\langle \phi^\lambda , (H_{F} +\lambda\mathbbm{1}) \phi^\lambda\big\rangle_{\cH_\mathrm{b}}+ \big\langle u_\xi^\lambda , u_{\mathcal{A}_\lambda\xi+\chi}^\lambda\big\rangle_{\cH_\mathrm{b}} \\
    &=\;\big\langle \phi^\lambda , (H_{F} +\lambda\mathbbm{1}) \phi^\lambda\big\rangle_{\cH_\mathrm{b}}+ \langle \xi,\mathcal{A}_\lambda\xi+\chi\rangle_{H^{-\frac{1}{2}}_{W_\lambda}} \\
    &=\;\big\langle \phi^\lambda , (H_{F} +\lambda\mathbbm{1}) \phi^\lambda\big\rangle_{\cH_\mathrm{b}}+ \langle \xi,\mathcal{A}_\lambda\xi\rangle_{H^{-\frac{1}{2}}_{W_\lambda}}\;\in\;\mathbb{R}\,.
  \end{split}
 \]
 (We used $\langle  u_\xi^\lambda , (\mathring{H}+\lambda\mathbbm{1})f^\lambda\rangle_{\cH_\mathrm{b}}=0$ in the third step, \eqref{eq:W-scalar-product} in the fourth, and $\chi\perp_\lambda\xi$ in the fifth.) For the reality of the above expression it indeed suffices $\mathcal{A}_\lambda$ to be symmetric in $H^{-\frac{1}{2}}_{W_\lambda}(\mathbb{R}^3)$. The proof of Lemma \ref{eq:oneTMSfunction} can be just repeated for the symmetric $\mathcal{A}_\lambda$ and the same conclusions hold for the symmetric extension $\mathscr{H}$ considered now.
 \end{remark}

 TMS \emph{symmetric} extensions of $\mathring{H}$ will play a crucial role in Sect.~\ref{sec:lzero}. For the time being, let us focus on TMS \emph{self-adjoint} extensions, and comment on their symmetric counterpart at the end of this Subsection (Remark \ref{rem:remonsym}).

 The requirement that a self-adjoint extension of $\mathring{H}$ \emph{as a whole} be a Ter-Martirosyan Skornyakov operator imposes a precise choice of the corresponding Birman operators $\mathcal{A}_\lambda$.

 \begin{lemma}\label{lem:generalTMSext}
  Let $\lambda>0$ and $\alpha\in\mathbb{R}$. Let $\mathcal{A}_\lambda\in\mathcal{K}(H^{-\frac{1}{2}}_{W_\lambda}(\mathbb{R}^3))$ and let $\mathring{H}_{\mathcal{A}_\lambda}$ be corresponding self-adjoint extension of $\mathring{H}$.
  The following two conditions are equivalent.
  \begin{itemize}
   \item[(i)] Every $g\in\mathcal{D}(\mathring{H}_{\mathcal{A}_\lambda})$ satisfies the TMS condition \eqref{eq:TMSoncharges} with the given $\alpha$.
   \item[(ii)] $\mathcal{D}(\mathcal{A}_\lambda)$ is dense in $H^{-\frac{1}{2}}(\mathbb{R}^3)$,
      $(T_\lambda+\alpha\mathbbm{1})\mathcal{D}(\mathcal{A}_\lambda)\subset H^{\frac{1}{2}}(\mathbb{R}^3)$, and 
   \begin{equation}\label{eq:Alambdaexpression}
     \mathcal{A}_\lambda\;=\;3 W_\lambda^{-1}(T_\lambda+\alpha\mathbbm{1})\,.
   \end{equation}
  \end{itemize}
  \end{lemma}

 \begin{proof}
  The implication (ii) $\Rightarrow$ (i) is obvious from Lemma \ref{lem:largepasympt-star}. Conversely, if \eqref{eq:TMSoncharges} is to be satisfied by every $g\in\mathcal{D}(\mathring{H}_{\mathcal{A}_\lambda})$, then owing to Lemma \ref{eq:oneTMSfunction}
  \[
   \begin{cases}
    T_\lambda\xi+\alpha\,\xi\in H^{\frac{1}{2}}(\mathbb{R}^3) \\
    \chi=3W_\lambda^{-1}(T_\lambda\xi+\alpha\,\xi)-\mathcal{A}_\lambda\xi
   \end{cases}
   \quad \forall\xi\in\mathcal{D}(\mathcal{A}_\lambda)\,,\;\forall\chi\in\mathcal{D}(\mathcal{A}_\lambda)^{\perp_\lambda}\cap H^{-\frac{1}{2}}(\mathbb{R}^3)\,.
  \]
  The first condition means precisely $(T_\lambda+\alpha\mathbbm{1})\mathcal{D}(\mathcal{A}_\lambda)\subset H^{\frac{1}{2}}(\mathbb{R}^3)$, and the second condition can only be satisfied if $\mathcal{D}(\mathcal{A}_\lambda)^{\perp_\lambda}\cap H^{-\frac{1}{2}}(\mathbb{R}^3)$ is trivial, namely when $\mathcal{D}(\mathcal{A}_\lambda)$ is dense in $H^{-\frac{1}{2}}(\mathbb{R}^3)$ and $\mathcal{A}_\lambda=3 W_\lambda^{-1}(T_\lambda+\alpha\mathbbm{1})$.
  \end{proof}

 \begin{theorem}\label{thm:globalTMSext}
  Let $\alpha\in\mathbb{R}$ and let $\mathscr{H}$ be an operator on $\cH_{\mathrm{b}}$. The following two possibilities are equivalent.
  \begin{itemize}
   \item[(i)] $\mathscr{H}$ is a Ter-Martirosyan Skornyakov self-adjoint extension of $\mathring{H}$ with inverse scattering length $\alpha$.
   \item[(ii)] There exists a subspace $\mathcal{D}\subset H^{-\frac{1}{2}}(\mathbb{R}^3)$ such that, for one and hence for all $\lambda>0$,
   \begin{itemize}
   \item[1.] $\mathcal{D}$ is dense in $H^{-\frac{1}{2}}(\mathbb{R}^3)$,
   \item[2.] $(T_\lambda+\alpha\mathbbm{1})\mathcal{D}\subset H^{\frac{1}{2}}(\mathbb{R}^3)$,
   \item[3.] the operator
   \begin{equation}\label{eq:Alfinally}
    \begin{split}
     \mathcal{A}_\lambda\;&:=\;3 W_\lambda^{-1}(T_\lambda+\alpha\mathbbm{1}) \\
     \mathcal{D}(\mathcal{A}_\lambda)\;&:=\;\mathcal{D}
    \end{split}
   \end{equation}
   is self-adjoint in $H^{-\frac{1}{2}}_{W_\lambda}(\mathbb{R}^3)$,
   \item[4.] $\mathscr{H}=\mathring{H}_{\mathcal{A}_\lambda}$.
   \end{itemize}   
  \end{itemize}
   When (i) or (ii) are matched, for one and hence for all $\lambda>0$ one has
   \begin{equation}\label{eq:HAl-tms-dom}
  \mathcal{D}(\mathscr{H})\;=\;
  \left\{g=\phi^\lambda+u_\xi^\lambda\left|\!
  \begin{array}{c}
   \phi^\lambda\in H^2_\mathrm{b}(\mathbb{R}^3\times\mathbb{R}^3)\,,\;\xi\in\mathcal{D}\,, \\
   \displaystyle\int_{\mathbb{R}^3}\widehat{\phi^\lambda}(\pp_1,\pp_2)\,\ud\pp_2\,=\,(\widehat{T_\lambda\xi})(\pp_1)+\alpha\,\widehat{\xi}(\pp_1)
  \end{array}
  \!\!\!\right.\right\}
  \end{equation}
  where for each $\lambda$ the above decomposition of $g$ in terms of $\phi^\lambda$ and $\xi$ is unique, and 
  \begin{equation}\label{eq:HAl-tms-act}
   (\mathscr{H}+\lambda\mathbbm{1})g\;=\;(\mathring{H}_F+\lambda\mathbbm{1})\phi^\lambda\,.
  \end{equation}	
   \end{theorem}

   \begin{proof}
    Assume that every $g\in\mathcal{D}(\mathscr{H})$ satisfies the TMS asymptotics \eqref{eq:g-largep2-TMS0} with the same $\alpha$.  Applying Lemma \ref{lem:generalTMSext} one obtains all four conditions 1.~through 4.~listed in part (ii), for every $\lambda>0$, except that $\mathcal{D}$ is replaced by $\mathcal{D}(\mathcal{A}_\lambda)$ for each considered $\lambda$. But all such $\mathcal{D}(\mathcal{A}_\lambda)$'s are in fact the same subspace (Remark \ref{rem:samedomains}). 
%
    The proof of (i) $\Rightarrow$ (ii) is completed.
    
    Conversely, assume that (ii) holds true for \emph{one} $\lambda_\circ>0$. Applying Lemma \ref{lem:generalTMSext} one deduces that $\mathscr{H}$ is a Ter-Martirosyan Skornyakov self-adjoint extension of $\mathring{H}$ with inverse scattering length $\alpha$. Since we know already that (i) $\Rightarrow$ (ii), then condition (ii) holds true for \emph{any} other $\lambda>0$ as well. This establishes the full implication (ii) $\Rightarrow$ (i).
    
    Under condition (i), or equivalently (ii), \eqref{eq:HAl-tms-dom}-\eqref{eq:HAl-tms-act} then follow from \eqref{eq:domDHA}-\eqref{eq:domDHA-actionDHA} of Theorem \ref{thm:generalclassification}(i) and from \eqref{eq:phi-largep2-TMS0} and \eqref{eq:Alfinally}. 
   \end{proof}

  It is worth stressing that formulas \eqref{eq:HAl-tms-dom}-\eqref{eq:HAl-tms-act} alone, considered for some subspace $\mathcal{D}$ of $H^{-\frac{1}{2}}(\mathbb{R}^3)$, evidently define an extension $\mathscr{H}$ of $\mathring{H}$; however, they do not necessarily make $\mathscr{H}$ a self-adjoint extension. We formulate this point in the form of a separate corollary for later purposes.

  \begin{corollary}\label{cor:globalTMSext}
   Let $\alpha\in\mathbb{R}$, $\lambda>0$, let $\mathcal{D}$ be a subspace of $H^{-\frac{1}{2}}(\mathbb{R}^3)$, and let $\mathscr{H}$ be the operator defined by \eqref{eq:HAl-tms-dom}-\eqref{eq:HAl-tms-act}. Then $\mathscr{H}$ is self-adjoint in $\cH_\mathrm{b}$ if and only if $\mathcal{D}$ is dense in $H^{-\frac{1}{2}}(\mathbb{R}^3)$, $(T_\lambda+\alpha\mathbbm{1})\mathcal{D}\subset H^{\frac{1}{2}}(\mathbb{R}^3)$, and the operator \eqref{eq:Alfinally} is self-adjoint in $H^{-\frac{1}{2}}_{W_\lambda}(\mathbb{R}^3)$.   
  \end{corollary}

  Thus, the quest of Ter-Martirosyan Skornyakov self-adjoint extensions of $\mathscr{H}$ in $\cH_{\mathrm{b}}$ is boiled down to the self-adjointness problem of $W_\lambda^{-1}(T_\lambda+\alpha\mathbbm{1})$ in $H^{-\frac{1}{2}}_{W_\lambda}(\mathbb{R}^3)$ with domain $\mathcal{D}$, hence in practice to the \emph{problem of finding a domain of self-adjointness for the formal action $\xi\mapsto W_\lambda^{-1}(T_\lambda+\alpha\mathbbm{1})\xi$}. This task actually constitutes the hard part of the rigorous modelling of physically meaningful Hamiltonians of zero-range interactions for the considered bosonic trimer.

  For an operator satisfying either condition of Theorem \ref{thm:globalTMSext} we shall use the natural notation $\mathscr{H}_\alpha$, so as to emphasize the only relevant parameter of the considered Ter-Martirosyan Skornyakov (self-adjoint) extension of $\mathring{H}$. This must be done keeping in mind that in principle for the same $\alpha\in\mathbb{R}$ there could be distinct operators of the form $\mathscr{H}_\alpha$, that is, distinct domains of self-adjointness for $W_\lambda^{-1}(T_\lambda+\alpha\mathbbm{1})$ in $H^{-\frac{1}{2}}_{W_\lambda}(\mathbb{R}^3)$ (Corollary \ref{cor:globalTMSext}), in analogy with the familiar existence of a variety of distinct domains of self-adjointness in $L^2(0,1)$ for the same differential operator $-\frac{\ud^2}{\ud x^2}$.

  \begin{remark}\label{rem:remonsym}
   The reasonings that led to Theorem \ref{thm:globalTMSext} have an obvious counterpart for Ter-Martirosyan Skornyakov \emph{symmetric} extensions of $\mathring{H}$. 
   \begin{itemize}
    \item[(i)] Lemma \ref{lem:generalTMSext} is equally valid when $\mathcal{A}_\lambda$ is only assumed to be symmetric in $H^{-\frac{1}{2}}_{W_\lambda}(\mathbb{R}^3)$ and the corresponding $\mathring{H}_{\mathcal{A}_\lambda}$ is a Ter-Martirosyan Skornyakov symmetric extension of $\mathring{H}$, based on the observations made in Remark \ref{rem:TMSsym}. 
    \item[(ii)] By means of such `symmetric version' of Lemma \ref{lem:generalTMSext}, the proof of Theorem \ref{thm:globalTMSext} can be straightforwardly adjusted so as to establish that:
    
    \noindent \emph{$\mathscr{H}$ is a Ter-Martirosyan Skornyakov symmetric extension of $\mathring{H}$ with inverse scattering length $\alpha\in\mathbb{R}$ if and only if there exists a subspace $\mathcal{D}\subset H^{-\frac{1}{2}}(\mathbb{R}^3)$ such that, for one and hence for all $\lambda>0$, $\mathcal{D}$ is dense in $H^{-\frac{1}{2}}(\mathbb{R}^3)$, $(T_\lambda+\alpha\mathbbm{1})\mathcal{D}\subset H^{\frac{1}{2}}(\mathbb{R}^3)$, the operator $\mathcal{A}_\lambda:=3 W_\lambda^{-1}(T_\lambda+\alpha\mathbbm{1})$ is symmetric in $H^{-\frac{1}{2}}_{W_\lambda}(\mathbb{R}^3)$ on the domain $\mathcal{D}$, and $\mathscr{H}=\mathring{H}_{\mathcal{A}_\lambda}$.}    
%
   \end{itemize}
  \end{remark}

  \subsection{Symmetry and self-adjointness of the TMS parameter}~

  As emerged in Subsect.~\ref{sec:generalitiesTMSext}, the operator \eqref{eq:Alfinally} is the correct Birman operator labelling symmetric or self-adjoint TMS extensions of $\mathring{H}$ in terms of the general parametrisation provided by Theorem \ref{thm:generalclassification} (and Remarks \ref{rem:TMSsym} and \ref{rem:remonsym}).

%
%

  The symmetry or self-adjointness, in the respective Hilbert spaces, of the auxiliary operators $\mathcal{A}_\lambda$ and $T_\lambda$ on the domain $\mathcal{D}$ are closely related (albeit deceptively, in a sense), as we shall now discuss.

  To avoid ambiguities, let us reserve the standard notation $T^*$, $\overline{T}$, etc., for the adjoint of $T$, its operator closure, and so on, with respect to the underlying $L^2$-space (as done for $\mathring{H}^*$ as an operator on $\cH_{\mathrm{b}}$), which in this context shall be $L^2(\mathbb{R}^3)$, and let us write instead $\mathcal{A}_\lambda^\star$, $\overline{\mathcal{A}_\lambda}^\lambda$, $\perp_\lambda$, etc., with reference to $H^{-\frac{1}{2}}_{W_\lambda}(\mathbb{R}^3)$.

  \begin{lemma}\label{lem:symsym}
   Let $\lambda>0$, $\alpha\in\mathbb{R}$, and let $\mathcal{D}$ be a dense subspace of $L^2(\mathbb{R}^3)$ such that $(T_\lambda+\alpha\mathbbm{1})\mathcal{D}\subset H^{\frac{1}{2}}(\mathbb{R}^3)$. Consider both $\mathcal{A}_\lambda:=3W_\lambda^{-1}(T_\lambda+\alpha\mathbbm{1})$ and $T_\lambda$ as operators with domain $\mathcal{D}$. Then
   \[
    \mathcal{A}_\lambda\;\subset\;\mathcal{A}_\lambda^\star\qquad \Leftrightarrow \qquad T_\lambda\;\subset\;T_\lambda^*\,,
   \]
  that is, the symmetry of $\mathcal{A}_\lambda$ in $H^{-\frac{1}{2}}_{W_\lambda}(\mathbb{R}^3)$ is equivalent to the symmetry of $T_\lambda$ in $L^2(\mathbb{R}^3)$.  
  \end{lemma}

  \begin{proof}
   $\mathcal{D}$ is dense in $L^2(\mathbb{R}^3)$ and hence in $H^{-\frac{1}{2}}(\mathbb{R}^3)\cong H^{-\frac{1}{2}}_{W_\lambda}(\mathbb{R}^3)$.  
   Owing to \eqref{eq:W-scalar-product},
   \[
    \langle \xi, \mathcal{A}_\lambda\xi\rangle_{H^{-\frac{1}{2}}_{W_\lambda}}\;=\;3\langle\xi,(T_\lambda+\alpha\mathbbm{1})\xi\rangle_{L^2}\qquad\forall\xi\in\mathcal{D}\,.
   \]
  Therefore, the reality of the l.h.s.~is equivalent to the reality of the r.h.s.
  \end{proof}

  \begin{lemma}\label{lem:two-selfadj-problems}
   Let $\lambda>0$, $\alpha\in\mathbb{R}$, and let $\mathcal{D}$ be a dense subspace of $L^2(\mathbb{R}^3)$ such that $(T_\lambda+\alpha\mathbbm{1})\mathcal{D}\subset H^{\frac{1}{2}}(\mathbb{R}^3)$. Consider both $\mathcal{A}_\lambda:=3W_\lambda^{-1}(T_\lambda+\alpha\mathbbm{1})$ and $T_\lambda$ as operators with domain $\mathcal{D}$. Assume that $T_\lambda\,=\,T_\lambda^*$. Then,
   \begin{itemize}
    \item[(i)] $\mathcal{D}(\mathcal{A}_\lambda^\star)\cap L^2(\mathbb{R}^3)\,=\,\mathcal{D}(\mathcal{A}_\lambda)=\mathcal{D}$;
    \item[(ii)] $\mathcal{A}_\lambda\,=\,\mathcal{A}_\lambda^\star$ if and only if $\mathcal{D}(\mathcal{A}_\lambda^\star)\subset L^2(\mathbb{R}^3)$.
   \end{itemize}
  \end{lemma}

  \begin{proof} Clearly (ii) follows from (i). Concerning (i), the inclusion $\mathcal{D}(\mathcal{A}_\lambda^\star)\cap L^2(\mathbb{R}^3)\supset\mathcal{D}(\mathcal{A}_\lambda)$ is obvious. Let now $\eta\in\mathcal{D}(\mathcal{A}_\lambda^\star)\cap L^2(\mathbb{R}^3)$. Then, for some $c_\eta>0$,
   \[
   \Big| \langle\eta,\mathcal{A}_\lambda\xi\rangle_{H^{-\frac{1}{2}}_{W_\lambda}}\Big|\;\leqslant\; c_\eta\,\|\xi\|_{L^2}^2\qquad\forall\xi\in\mathcal{D}=\mathcal{D}(\mathcal{A}_\lambda)\,.
  \]
    Equivalently, owing to \eqref{eq:W-scalar-product},
  \[
   \big| \langle\eta,(T_\lambda+\alpha\mathbbm{1})\xi\rangle_{L^2}\big|\;\leqslant\; {\textstyle\frac{1}{3}}c_\eta\,\|\xi\|_{L^2}^2\qquad\forall\xi\in\mathcal{D}=\mathcal{D}(T_\lambda)\,.
  \]
   Therefore, $\eta\in\mathcal{D}(T_\lambda^*)=\mathcal{D}(T_\lambda)=\mathcal{D}(\mathcal{A}_\lambda)$.  
  \end{proof}

 \subsection{TMS extensions in sectors of definite angular momentum}\label{sec:definite-ell-general}~

 As the maps $\xi\mapsto T_\lambda\xi$, $\xi\mapsto W_\lambda\xi$, $\xi\mapsto W_\lambda^{-1}\xi$ all commute with the rotations in $\mathbb{R}^3$ (Subsect.~\ref{sec:Tlambdaoperator}), and so too does therefore the map $\xi\mapsto W_\lambda^{-1}T_\lambda$, then the TMS parameter $\mathcal{A}_\lambda=3W_\lambda^{-1}(T_\lambda+\alpha\mathbbm{1})$ is naturally reduced in each sector of definite angular momentum.

 More precisely, with respect to the decomposition \eqref{eq:bigdecomp}-\eqref{eq:xihatangularexpansion}, and following the same reasoning therein, one then has
\begin{equation}\label{eq:WlambdaWlambdaell}
 \widehat{W_\lambda\xi}\;=\;\sum_{\ell=0}^\infty \widehat{W_\lambda^{(\ell)}\xi^{(\ell)}}\,,
\end{equation}
where each $W_\lambda^{(\ell)}$ is non-trivial only radially. Moreover,
 \begin{equation}
  \begin{split}
   \langle \xi,\eta\rangle_{H^{-\frac{1}{2}}_{W_\lambda}}\;&=\; \int_{\mathbb{R}^3} \overline{\,\widehat{\xi}(\pp)}\, \big(\widehat{W_\lambda\eta}\big)(\pp)\,\ud\pp \\
   &=\;\sum_{\ell=0}^{\infty}\int_{\mathbb{R}^3} \overline{\,\widehat{\xi^{(\ell)}}(\pp)}\, \big(\widehat{W_\lambda^{(\ell)}\eta^{(\ell)}}\big)(\pp)\,\ud\pp\,,
  \end{split}
 \end{equation}
where
 \begin{equation}\label{eq:Wellsp}
   \begin{split}
   &\int_{\mathbb{R}^3} \overline{\,\widehat{\xi^{(\ell)}}(\pp)}\, \big(\widehat{W_\lambda^{(\ell)}\eta^{(\ell)}}\big)(\pp)\,\ud\pp \\
   &\;=\;\sum_{n=-\ell}^{\ell}\bigg(\int_{\mathbb{R}^+}\!\ud p\,p^2\,\frac{3\pi^2}{\sqrt{{\textstyle\frac{3}{4}}p^2+\lambda}\,}\,\overline{f_{\ell,n}^{(\xi)}(p)}\,f_{\ell,n}^{(\eta)}(p) \\
   &\quad +12\pi\!\iint_{\mathbb{R}^+\times\mathbb{R}^+}\ud p\,\ud q\,p^2q^2\,\overline{f_{\ell,n}^{(\xi)}(p)}\,f_{\ell,n}^{(\eta)}(q)\!\int_{-1}^1\ud t\,\frac{P_\ell(t)}{(p^2+q^2+p\,q \,t+\lambda)^2}\bigg) \\
   &\;\equiv\;\sum_{n=-\ell}^{\ell}\big\langle f_{\ell,n}^{(\xi)},f_{\ell,n}^{(\eta)} \big\rangle_{W_\lambda^{(\ell)}}\,,
    \end{split}
  \end{equation}
in complete analogy to \eqref{eq:xiTxipre}-\eqref{eq:xiTxi}.

As the scalar product $\langle\cdot,W_\lambda\cdot\rangle_{H^{-\frac{1}{2}},H^{\frac{1}{2}}}$ is equivalent to the ordinary $H^{-\frac{1}{2}}$-scalar product, so is the scalar product $\langle\cdot,\cdot\rangle_{W_\lambda^{(\ell)}}$ defined by \eqref{eq:Wellsp} equivalent to the ordinary scalar product in $ L^2(\mathbb{R}^+,(1+p^2)^{-\frac{1}{2}} p^2\ud p) $. The latter is therefore a Hilbert space also when equipped with $\langle\cdot,\cdot\rangle_{W_\lambda^{(\ell)}}$, in which case we shall denote it with $L^2_{W_\lambda^{(\ell)}}(\mathbb{R}^+)$. One thus has the canonical Hilbert space isomorphism
\begin{equation}\label{eq:RpRpW}
 L^2(\mathbb{R}^+,(1+p^2)^{-\frac{1}{2}} p^2\ud p)\;\cong\;L^2_{W_\lambda^{(\ell)}}(\mathbb{R}^+)\,.
\end{equation}
By means of \eqref{eq:RpRpW} one re-writes \eqref{eq:bigdecomp} as
\begin{equation}\label{eq:bigdecompW}
 \begin{split}
 H^{-\frac{1}{2}}_{W_\lambda}(\mathbb{R}^3)\;\cong\;H^{-\frac{1}{2}}(\mathbb{R}^3)\;&\cong\;\bigoplus_{\ell=0}^\infty \Big( L^2_{W_\lambda^{(\ell)}}(\mathbb{R}^+) \otimes \mathrm{span}\big\{ Y_{\ell,n}\,|\,n=-\ell,\dots,\ell \big\} \Big)  \\
 &\equiv\;\bigoplus_{\ell=0}^\infty \,H^{-\frac{1}{2}}_{W_\lambda,\ell}(\mathbb{R}^3) \,.
 \end{split}
\end{equation}
The expansion \eqref{eq:xihatangularexpansion} of a generic $\xi\in H^{-\frac{1}{2}}(\mathbb{R}^3)$ is equivalently referred to the ordinary decomposition or the $\lambda$-decomposition of the space \eqref{eq:bigdecompW}.

 With respect to \eqref{eq:bigdecompW} $\mathcal{A}_\lambda$ is reduced as
 \begin{equation}\label{eq:Alamdbareducedell}
  \mathcal{A}_{\lambda}\;=\;\bigoplus_{\ell=0}^{\infty}\,\mathcal{A}_{\lambda}^{(\ell)}\;=\;\bigoplus_{\ell=0}^{\infty}\,3 W_\lambda^{-1}(T_\lambda^{(\ell)}+\alpha\mathbbm{1})\,.
 \end{equation}
 The problem of finding a domain $\mathcal{D}$ of symmetry or of self-adjointness for $\mathcal{A}_\lambda$ with respect to $H^{-\frac{1}{2}}_{W_\lambda}(\mathbb{R}^3)$ is tantamount as finding a domain $\mathcal{D}_\ell$ of symmetry or of self-adjointness for $\mathcal{A}_\lambda^{(\ell)}$ with respect to $H^{-\frac{1}{2}}_{W_\lambda,\ell}(\mathbb{R}^3)$ for each $\ell\in\mathbb{N}_0$. This is the object of Sect.~\ref{sec:higherell} and \ref{sec:lzero}.

 The subspace of $\mathcal{D}(\mathcal{A}_\lambda)$ consisting of elements $g$ with charge $\xi\in\mathcal{D}_\ell$ is sometimes referred to as the charge domain of the TMS (symmetric or self-adjoint) extension of $\mathring{H}$ in the $\ell$-th sector of definite angular momentum.

 \section{Sectors of higher angular momenta}\label{sec:higherell}

 In the modelling of the bosonic trimer with zero-range interaction, all the relevant physics is expected in the sector of zero angular momentum, since in each two-body channel particles undergo a low-energy, and hence essentially an $s$-wave scattering.

 In this respect, the qualification of the quantum Hamiltonian is somewhat arbitrary in the sectors of higher (non-zero) angular momentum, as in practice experimental observations do not involve states in such sectors. For instance one could simply consider a Hamiltonian where for each $\ell\in\mathbb{N}$ the Birman parameter of formula \eqref{eq:domDHA} has domain $\mathcal{D}_\ell=\{0\}$ and value `$\mathcal{A}_\lambda=\infty$' on it, namely the Friedrichs extension of $\mathring{H}$ in those sectors, whereas only the $\ell=0$ is defined non-trivially. This would model a total absence of interaction at higher angular momenta.

 A more typical choice is to define a model that in all $\ell$-sectors, not only $\ell=0$, is characterised by the physical TMS asymptotics with inverse scattering length, and to do so by making a somewhat canonical construction for $\ell\neq 0$, and a non-trivial one for $\ell=0$. We present such programme in this Section for non-zero $\ell$. The analysis of $\ell=0$ is deferred to Sect.~\ref{sec:lzero}.
 

 \subsection{$T_\lambda$-estimates}\label{sec:Tlambdaestimates}~
 
 We import here a set of useful estimates established in the already mentioned work \cite{CDFMT-2012} by Correggi, Dell'Antonio, Finco, Michelangeli, and Teta.

 For given $\lambda>0$ and $\ell\in\mathbb{N}$ let us introduce the shorthands 
 \begin{equation}
   \begin{split}
    \Phi_\lambda[f,g]\;&:=\;2\pi^2\!\int_{\mathbb{R}^+}\!\ud p\,p^2\,\sqrt{{\textstyle\frac{3}{4}}p^2+\lambda}\,\overline{f(p)}\,g(p) \\
    \Psi_{\lambda,\ell}[f,g]\;&:=\;2\pi\!\iint_{\mathbb{R}^+\times\mathbb{R}^+}\ud p\,\ud q\,p^2q^2\,\overline{f(p)}\,g(q)\!\int_{-1}^1\ud t\,\frac{P_\ell(t)}{p^2+q^2+p \,q \,t+\lambda}
   \end{split}
 \end{equation}
 and 
  \begin{equation}
    \begin{split}
   \Phi_\lambda[f]\;&:=\;\Phi_\lambda[f,f] \\
   \Psi_{\lambda,\ell}[f]\;&:=\;\Psi_{\lambda,\ell}[f,f]
  \end{split}
  \end{equation}
 so that \eqref{eq:xiTxipre}-\eqref{eq:xiTxi} now read
  \begin{equation}\label{eq:xiTxi-updated}
   \int_{\mathbb{R}^3} \overline{\,\widehat{\xi}(\pp)}\, \big(\widehat{T_\lambda\eta}\big)(\pp)\,\ud\pp\;=\;\sum_{\ell=0}^\infty\sum_{n=-\ell}^\ell\Big( \Phi_\lambda\big[f_{\ell,n}^{(\xi)},f_{\ell,n}^{(\eta)}\big]-2\Psi_{\lambda,\ell}\big[f_{\ell,n}^{(\xi)},f_{\ell,n}^{(\eta)}\big]\Big)\,.
  \end{equation}

  \begin{lemma}
   Let $\lambda>0$ and $\ell\in\mathbb{N}$. Let $f:\mathbb{R}\to\mathbb{C}$ make the quantities below finite.
   \begin{itemize}
    \item[(i)] One has
    \begin{equation}\label{eq:Psilambdaordering}
     \begin{split}
    0\;\leqslant\;\Psi_{\lambda,\ell}[f]\;\leqslant\;\Psi_{0,\ell}[f] & \qquad\textrm{for even $\ell$} \\
    \Psi_{0,\ell}[f]\;\leqslant\;\Psi_{\lambda,\ell}[f]\;\leqslant\; 0 & \qquad\textrm{for odd $\ell$}.
   \end{split}
    \end{equation}
    \item[(ii)] One has 
    \begin{equation}\label{eq:PhiPsilambdazero}
     \begin{split}
      \Phi_0[f]\;&=\;\pi^2\sqrt{3}\int_\mathbb{R}\ud s\,|f^\sharp(s)|^2 \\
      \Psi_{0,\ell}[f]\;&=\;\int_\mathbb{R}\ud s\,S_\ell(s)\,|f^\sharp(s)|^2
     \end{split}
    \end{equation}
    where
    \begin{equation}
     f^\sharp(s)\;:=\;\frac{1}{\sqrt{2\pi}}\int_{\mathbb{R}}\,\ud x\,e^{-\ii k x}\,e^{2x}f(e^x)
    \end{equation} 
    and 
    \begin{equation}\label{eq:Sell}
     S_\ell(s)\;:=\;2\pi^2\int_{-1}^1\ud t\,P_\ell(t)\,\frac{\sinh(s\arccos\frac{t}{2})}{\sin(\arccos\frac{t}{2})\,\sinh\pi s}\,.
    \end{equation}
    \item[(iii)] $S_\ell:\mathbb{R}\to\mathbb{R}$ is a smooth even function, strictly monotone on $\mathbb{R}^\pm$, and such that
    \begin{equation}\label{eq:Slzeroordering}
     \begin{split}
      0\;\leqslant\;S_{\ell+2}(s)\;\leqslant\; S_{\ell}(s)\;\leqslant\;S_{\ell}(0) & \qquad\textrm{for even $\ell$} \\
      S_{\ell}(0)\;\leqslant\;S_{\ell}(s)\;\leqslant\;S_{\ell+2}(s)\;\leqslant\;0   & \qquad\textrm{for odd $\ell$}
     \end{split}
    \end{equation}
   \end{itemize}      
  \end{lemma}

  \begin{proof}
   Part (i) follows from \cite[Lemma 3.2]{CDFMT-2012}. Part (ii) from \cite[Lemma 3.3]{CDFMT-2012}. Part (iii) from \cite[Lemma 3.5]{CDFMT-2012}.   
  \end{proof}

  The values of $S_\ell(s)$ that will be relevant in the present analysis are
  \begin{equation}\label{eq:Sellzerospecialvalues}
   \begin{split}
    S_0(0)\;&=\;\frac{\,2\pi^3}{3}\;>\;0 \\
    S_1(0)\;&=\;-8\pi\Big(1-\frac{\pi}{2\sqrt{3}}\Big)\;<\;0 \\
    S_2(0)\;&=\;\frac{\,\pi^2}{3}(5\pi-9\sqrt{3})\;>\;0\,,
   \end{split}
  \end{equation}
 as one easily computes from \eqref{eq:Sell} and \eqref{def_Legendre}.

  By means of the estimates above, one obtains the following important bounds.

  \begin{lemma}\label{lem:xiTxi-equiv-H12}
   Let $\lambda>0$ and let $\xi\in H^{\frac{1}{2}}(\mathbb{R}^3)$.
   Then
   \begin{eqnarray}
    \int_{\mathbb{R}^3} \overline{\,\widehat{\xi}(\pp)}\, \big(\widehat{T_\lambda\xi}\big)(\pp)\,\ud\pp\!\!&\leqslant&\!\!\kappa^+\cdot 2\pi^2\!\int_{\mathbb{R}^3}\sqrt{\frac{3}{4}\pp^2+\lambda}\,|\widehat{\xi}(\pp)|^2\,\ud\pp	\label{eq:xiTxiFromAbove} \\
    \int_{\mathbb{R}^3} \overline{\,\widehat{\xi}(\pp)}\, \big(\widehat{T_\lambda\xi}\big)(\pp)\,\ud\pp\!\!&\geqslant&\!\!\kappa^-\cdot 2\pi^2\!\int_{\mathbb{R}^3}\sqrt{\frac{3}{4}\pp^2+\lambda}\,|\widehat{\xi}(\pp)|^2\,\ud\pp	\label{eq:xiTxiFromBelow}
   \end{eqnarray}
  where
  \begin{equation}
  \begin{split}
   \kappa^+\;&=\;\frac{16}{\pi\sqrt{3}}-\frac{5}{3} \\
    \kappa^-\;&=\;
   \begin{cases}
    -\displaystyle\Big(\frac{4\pi}{3\sqrt{3}}-1\Big) & \textrm{ if $\xi$ is non-trivial on $H^{\frac{1}{2}}_{\ell=0}(\mathbb{R}^3)$} \\
    \quad 7-\frac{10\pi}{3\sqrt{3}} & \textrm{ if }\;\xi\in\bigoplus_{\ell=1}^\infty H^{\frac{1}{2}}_\ell(\mathbb{R}^3)\,.
   \end{cases}
  \end{split}
  \end{equation}
 In particular, if $\xi\perp  H^{\frac{1}{2}}_{\ell=0}(\mathbb{R}^3)$, then $\kappa^->0$ and
 \begin{equation}
  \int_{\mathbb{R}^3} \overline{\,\widehat{\xi}(\pp)}\, \big(\widehat{T_\lambda\xi}\big)(\pp)\,\ud\pp\;\approx\;\|\xi\|_{H^{\frac{1}{2}}}^2\qquad (\ell\neq 0)
 \end{equation}
  in the sense of equivalence of norms (with $\lambda$-dependent multiplicative constants).
  \end{lemma}

  \begin{proof}
   Expanding $\xi$ as in \eqref{eq:xihatangularexpansion} and using \eqref{eq:xiTxi-updated} one has   
    \[\tag{*}\label{eq:xtxisums}
   \int_{\mathbb{R}^3} \overline{\,\widehat{\xi}(\pp)}\, \big(\widehat{T_\lambda\xi}\big)(\pp)\,\ud\pp\;=\;\sum_{\ell=0}^\infty\sum_{n=-\ell}^\ell\Big( \Phi_\lambda\big[f_{\ell,n}^{(\xi)}\big]-2\Psi_{\lambda,\ell}\big[f_{\ell,n}^{(\xi)}\big]\Big)\,.
  \]
   Owing to \eqref{eq:Psilambdaordering}, \eqref{eq:PhiPsilambdazero}, \eqref{eq:Slzeroordering}, and \eqref{eq:Sellzerospecialvalues},
   \[
    \begin{split}
     &\sum_{\ell=0}^\infty\sum_{n=-\ell}^\ell\Psi_{\lambda,\ell}\big[f_{\ell,n}^{(\xi)}\big]\;\;\geqslant\;\sum_{\substack{  \ell\in\mathbb{N}_0 \\ \ell\textrm{ odd} }}\sum_{n=-\ell}^\ell\Psi_{\lambda,\ell}\big[f_{\ell,n}^{(\xi)}\big]\;\geqslant\;\sum_{\substack{  \ell\in\mathbb{N}_0 \\ \ell\textrm{ odd} }}\sum_{n=-\ell}^\ell\Psi_{0,\ell}\big[f_{\ell,n}^{(\xi)}\big] \\
     &\;=\;\sum_{\substack{  \ell\in\mathbb{N}_0 \\ \ell\textrm{ odd} }}\sum_{n=-\ell}^\ell\int_\mathbb{R}\ud s\,S_\ell(s)\,\big|\big(f_{\ell,n}^{(\xi)}\big)^\sharp(s)\big|^2\;\geqslant\;S_1(0)\sum_{\substack{  \ell\in\mathbb{N}_0 \\ \ell\textrm{ odd} }}\sum_{n=-\ell}^\ell\int_\mathbb{R}\ud s\,\big|\big(f_{\ell,n}^{(\xi)}\big)^\sharp(s)\big|^2 \\
     &=\;-\frac{\,8\pi(1-\frac{\pi}{2\sqrt{3}})\,}{\pi^2\sqrt{3}}\sum_{\substack{  \ell\in\mathbb{N}_0 \\ \ell\textrm{ odd} }}\sum_{n=-\ell}^\ell\Phi_0\big[f_{\ell,n}^{(\xi)}\big]\;\geqslant\;-\frac{4}{3}\Big(\frac{2\sqrt{3}}{\pi}-1\Big)\sum_{\ell=0}^\infty\sum_{n=-\ell}^\ell\Phi_0\big[f_{\ell,n}^{(\xi)}\big] \\
     &\geqslant\;-\frac{4}{3}\Big(\frac{2\sqrt{3}}{\pi}-1\Big)\sum_{\ell=0}^\infty\sum_{n=-\ell}^\ell\Phi_\lambda\big[f_{\ell,n}^{(\xi)}\big]\,.
    \end{split}
   \]
    Plugging this into \eqref{eq:xtxisums} yields \eqref{eq:xiTxiFromAbove}. Analogously,
      \[
    \begin{split}
     &\sum_{\ell=0}^\infty\sum_{n=-\ell}^\ell\Psi_{\lambda,\ell}\big[f_{\ell,n}^{(\xi)}\big]\;\;\leqslant\;\sum_{\substack{  \ell\in\mathbb{N}_0 \\ \ell\textrm{ even} }}\sum_{n=-\ell}^\ell\Psi_{\lambda,\ell}\big[f_{\ell,n}^{(\xi)}\big]\;\leqslant\;\sum_{\substack{  \ell\in\mathbb{N}_0 \\ \ell\textrm{ even} }}\sum_{n=-\ell}^\ell\Psi_{0,\ell}\big[f_{\ell,n}^{(\xi)}\big] \\
     &\;=\;\sum_{\substack{  \ell\in\mathbb{N}_0 \\ \ell\textrm{ even} }}\sum_{n=-\ell}^\ell\int_\mathbb{R}\ud s\,S_\ell(s)\,\big|\big(f_{\ell,n}^{(\xi)}\big)^\sharp(s)\big|^2\;\leqslant\;S_{0}(0)\sum_{\substack{  \ell\in\mathbb{N}_0 \\ \ell\textrm{ even} }}\sum_{n=-\ell}^\ell\int_\mathbb{R}\ud s\,\big|\big(f_{\ell,n}^{(\xi)}\big)^\sharp(s)\big|^2 \\
     &=\;\frac{\,\frac{2\pi^3}{3}\,}{\pi^2\sqrt{3}}\sum_{\substack{  \ell\in\mathbb{N}_0 \\ \ell\textrm{ even} }}\sum_{n=-\ell}^\ell\Phi_0\big[f_{\ell,n}^{(\xi)}\big]\;\leqslant\;\frac{2\pi}{3\sqrt{3}}\sum_{\ell=0}^\infty\sum_{n=-\ell}^\ell\Phi_0\big[f_{\ell,n}^{(\xi)}\big]\;\leqslant\;\frac{2\pi}{3\sqrt{3}}\sum_{\ell=0}^\infty\sum_{n=-\ell}^\ell\Phi_\lambda\big[f_{\ell,n}^{(\xi)}\big]\,,
    \end{split}
   \]
    which combined with \eqref{eq:xtxisums} yields \eqref{eq:xiTxiFromBelow} in the general case. 
    In the particular case when $\xi$ has no $\ell=0$ component the previous computation becomes
    \[
     \begin{split}
       &\sum_{\ell=1}^\infty\sum_{n=-\ell}^\ell\Psi_{\lambda,\ell}\big[f_{\ell,n}^{(\xi)}\big]\;\;\leqslant\;S_{2}(0)\sum_{\substack{  \ell\in\mathbb{N} \\ \ell\textrm{ even} }}\sum_{n=-\ell}^\ell\int_\mathbb{R}\ud s\,\big|\big(f_{\ell,n}^{(\xi)}\big)^\sharp(s)\big|^2 \\
       &=\;\frac{\,\frac{\:\pi^2}{3}(5\pi-9\sqrt{3})\,}{\pi^2\sqrt{3}}\sum_{\substack{  \ell\in\mathbb{N} \\ \ell\textrm{ even} }}\sum_{n=-\ell}^\ell\Phi_0\big[f_{\ell,n}^{(\xi)}\big]\;\leqslant\;\frac{\,5\pi-9\sqrt{3}\,}{3\sqrt{3}}\sum_{\ell=1}^\infty\sum_{n=-\ell}^\ell\Phi_\lambda\big[f_{\ell,n}^{(\xi)}\big]
     \end{split}
    \]
    and plugging the latter estimate into \eqref{eq:xtxisums}, where now the $\ell=0$ summands are absent, one obtains \eqref{eq:xiTxiFromBelow} for this case.
    \end{proof}

  \subsection{Self-adjointness for $\ell\geqslant 1$}\label{sec:selfadj-ellnotzero}~

  Let us discuss a domain of self-adjointness for the TMS parameter $\mathcal{A}_\lambda$ in $H^{-\frac{1}{2}}_{W_\lambda,\ell}(\mathbb{R}^3)$ when $\ell\in\mathbb{N}$.

  It is convenient to realise first $\mathcal{A}_\lambda$ as a symmetric operator and then construct canonically a self-adjoint realisation of it.

  Let us set
  \begin{equation}\label{eq:lambdaalpha}
    \lambda_\alpha\;:=\;
    \begin{cases}
     \;\;0 & \textrm{ if }\;\alpha\geqslant 0 \\
     \alpha^2/(2\pi^2\kappa^-)^2& \textrm{ if }\;\alpha<0\qquad \big(\kappa^-=7-\frac{10\pi}{3\sqrt{3}}\big)\,.
    \end{cases}
  \end{equation}

  \begin{lemma}\label{lem:Atildenot0}
   For $\lambda>0$, $\alpha\in\mathbb{R}$, and $\ell\in\mathbb{N}$, let
  \begin{equation}\label{eq:Dtildeell}
   \widetilde{\mathcal{D}}_\ell\;:=\; H_\ell^{\frac{3}{2}}(\mathbb{R}^3)
  \end{equation}
  and 
  \begin{equation}\label{eq:Alambdatildenot0}
   \begin{split}
       \widetilde{\mathcal{A}_{\lambda}^{(\ell)}}\;&:=\;3W_\lambda^{-1}\big(T_\lambda^{(\ell)}+\alpha\mathbbm{1}\big) \\
       \mathcal{D}\big( \widetilde{\mathcal{A}_{\lambda}^{(\ell)}}\big)\;&:=\; \widetilde{\mathcal{D}}_\ell\,.
   \end{split}
  \end{equation}
  One has the following.
  \begin{itemize}
   \item[(i)] $ \widetilde{\mathcal{A}_{\lambda}^{(\ell)}}$ is a densely defined symmetric operator in $H^{-\frac{1}{2}}_{W_\lambda,\ell}(\mathbb{R}^3)$.
   \item[(ii)] If $\lambda>\lambda_\alpha$, then $\mathfrak{m}\big(\widetilde{\mathcal{A}_{\lambda}^{(\ell)}}\big)>0$, i.e., $\widetilde{\mathcal{A}_{\lambda}^{(\ell)}}$ has strictly positive lower bound.
  \end{itemize}
  \end{lemma}

  \begin{proof}
   (i) Obviously $\widetilde{\mathcal{D}}_\ell$ is dense in $H^{-\frac{1}{2}}_{W_\lambda,\ell}(\mathbb{R}^3)$. Moreover, $(T_\lambda^{(\ell)}+\alpha\mathbbm{1}\big)\widetilde{\mathcal{D}}_\ell\subset H^{\frac{1}{2}}_\ell(\mathbb{R}^3)$ (Lemma \ref{lem:Tlambdaproperties}(iii)) and $H^{\frac{1}{2}}_\ell(\mathbb{R}^3)=\mathrm{ran}W_{\lambda}^{(\ell)}$ (Lemma \ref{lem:Wlambdaproperties}(ii)), therefore \eqref{eq:Alambdatildenot0} is a well-posed definition for a densely defined operator in $H^{-\frac{1}{2}}_{W_\lambda,\ell}(\mathbb{R}^3)$. The map $\widetilde{\mathcal{D}}_\ell\ni\xi\mapsto T_\lambda^{(\ell)}\xi$ is densely defined and symmetric in $L^2(\mathbb{R}^3)$ (Lemma \ref{lem:Tlambdaproperties}(i)). All assumptions of Lemma \ref{lem:symsym} are then satisfied in the $\ell$-th sector: one then concludes that $ \widetilde{\mathcal{A}_{\lambda}^{(\ell)}}$ is symmetric in $H^{-\frac{1}{2}}_{W_\lambda,\ell}(\mathbb{R}^3)$.   
   
   (ii) For $\xi\in\widetilde{\mathcal{D}}_\ell$ we find
   \[
    \begin{split}
     \frac{1}{3}\big\langle\xi, \widetilde{\mathcal{A}_{\lambda}^{(\ell)}}\xi\big\rangle_{H^{-\frac{1}{2}}_{W_\lambda}}\;&=\;\big\langle\xi,(T_\lambda^{(\ell)}+\alpha\mathbbm{1}\big)\big\rangle_{L^2} \\
     &\geqslant\;2\pi^2\kappa^-\!\int_{\mathbb{R}^3}\sqrt{\frac{3}{4}\pp^2+\lambda}\,|\widehat{\xi}(\pp)|^2\,\ud\pp+\alpha\|\xi\|_{L^2}^2 \\
     &\geqslant\;(2\pi^2\kappa^-\sqrt{\lambda}+\alpha)\|\xi\|_{L^2}^2\;\geqslant\;c_\lambda(2\pi^2\kappa^-\sqrt{\lambda}+\alpha)\|\xi\|^2_{H^{-\frac{1}{2}}_{W_\lambda}}
    \end{split}
   \]
   for some $c_\lambda>0$, having used \eqref{eq:W-scalar-product} in the first step, \eqref{eq:xiTxiFromBelow} in the second, and the isomorphism $H^{-\frac{1}{2}}(\mathbb{R}^3)\cong H^{-\frac{1}{2}}_{W_\lambda}(\mathbb{R}^3)$ in the last. Thus, $\mathfrak{m}(\widetilde{\mathcal{A}_{\lambda}^{(\ell)}})\geqslant 3c_\lambda(2\pi^2\kappa^-\sqrt{\lambda}+\alpha)=6\pi^2c_\lambda\kappa^-(\sqrt{\lambda}-\sqrt{\lambda_\alpha})$ and the thesis follows.   
  \end{proof}

  Being densely defined, symmetric, and lower semi-bounded,  $\widetilde{\mathcal{A}_{\lambda}^{(\ell)}}$ has its Friedrichs self-adjoint extension. That will be our final TMS parameter $\mathcal{A}_{\lambda}^{(\ell)}$. 
%
  
  \begin{proposition}\label{prop:Alambdaellnot0}
   Let $\alpha\in\mathbb{R}$, $\lambda>\lambda_\alpha$, and $\ell\in\mathbb{N}$. Define 
   \begin{equation}\label{eq:domainDell}
    \mathcal{D}_\ell\;:=\;\big\{\xi\in H_\ell^{\frac{1}{2}}(\mathbb{R}^3)\,\big|\,T_\lambda^{(\ell)}\xi\in H_\ell^{\frac{1}{2}}(\mathbb{R}^3)\big\}\,.
   \end{equation}
   The operator 
   \begin{equation}\label{AFop-ellnot0}
    \begin{split}
     \mathcal{D}\big(\mathcal{A}_{\lambda}^{(\ell)}\big)\;&:=\;\mathcal{D}_\ell \\
     \mathcal{A}_{\lambda}^{(\ell)}\;&:=\;3 W_\lambda^{-1}\big(T_\lambda^{(\ell)}+\alpha\mathbbm{1}\big)\,.
    \end{split}
   \end{equation}
   is the Friedrichs extension of $\widetilde{\mathcal{A}_{\lambda}^{(\ell)}}$ with respect to $H^{-\frac{1}{2}}_{W_\lambda,\ell}(\mathbb{R}^3)$ and therefore is self-adjoint in such space.  
   Its sesquilinear form is
   \begin{equation}\label{AFform-ellnot0}
    \begin{split}
     \mathcal{D}\big[\mathcal{A}_{\lambda}^{(\ell)}\big]\;&=\;H_\ell^{\frac{1}{2}}(\mathbb{R}^3) \\
     \mathcal{A}_{\lambda}^{(\ell)}[\eta,\xi]\;&=\;3\big\langle\eta,\big(T_\lambda^{(\ell)}+\alpha\mathbbm{1}\big)\xi \big\rangle_{H^{\frac{1}{2}},H^{-\frac{1}{2}}}\,.
    \end{split}
   \end{equation}
  \end{proposition}

  \begin{proof}
   The definition \eqref{AFop-ellnot0} is well posed, as $\big(T_\lambda^{(\ell)}+\alpha\mathbbm{1}\big)\mathcal{D}_\ell\subset H^{\frac{1}{2}}_\ell(\mathbb{R}^3)=\mathrm{ran} W_\lambda$.  
   Lemma \ref{lem:xiTxi-equiv-H12} and the fact that $\widetilde{\mathcal{A}_{\lambda}^{(\ell)}}$ has strictly positive lower bound (Lemma \ref{lem:Atildenot0}(ii)) imply that the map
   \[
    \xi\;\mapsto\;\|\xi\|_{\mathcal{A}}\;:=\;\big\langle \xi,\widetilde{\mathcal{A}_{\lambda}^{(\ell)}}\xi\big\rangle_{H^{-\frac{1}{2}}_{W_\lambda}}^{\frac{1}{2}}\;=\;\Big( 3\big\langle\xi,\big(T_\lambda^{(\ell)}+\alpha\mathbbm{1}\big)\xi\big\rangle_{L^2}\Big)^{\frac{1}{2}}
   \]
   is a norm, and is actually equivalent to the $H^{\frac{1}{2}}$-norm. Let us temporarily denote by $\mathcal{A}_F$ the Friedrichs extension of $\widetilde{\mathcal{A}_{\lambda}^{(\ell)}}$ with respect to $H^{-\frac{1}{2}}_{W_\lambda,\ell}(\mathbb{R}^3)$.

   As prescribed by the Friedrichs construction, $\mathcal{A}_F$ has form domain
   \[
    \mathcal{D}[\mathcal{A}_F]\;=\;\overline{\mathcal{D}\big( \widetilde{\mathcal{A}_{\lambda}^{(\ell)}}\big)}^{\|\,\|_{\mathcal{A}	}}\;=\;\overline{H^{\frac{3}{2}}(\mathbb{R}^3)}^{\|\,\|_{H^{\frac{1}{2}}}}\;=\;H^{\frac{1}{2}}(\mathbb{R}^3)
   \]
   and for $\xi,\eta\in H_\ell^{\frac{1}{2}}(\mathbb{R}^3)$
   \[
    \mathcal{A}_F[\eta,\xi]\;=\;\lim_{n\to\infty}\big\langle\eta_n, \widetilde{\mathcal{A}_{\lambda}^{(\ell)}}\xi_n\big\rangle_{H^{-\frac{1}{2}}_{W_\lambda}}\;=\;3\lim_{n\to\infty}\big\langle\eta_n,\big(T_\lambda^{(\ell)}+\alpha\mathbbm{1}\big)\xi_n \big\rangle_{L^2}
   \]
  for any two sequences $(\xi_n)_n$ and $(\eta_n)_n$ in $H_\ell^{\frac{3}{2}}(\mathbb{R}^3)$ such that $\xi_n\to\xi$ and $\eta_n\to\eta$ in the $\|\,\|_{\mathcal{A}}$-norm, namely in $H_\ell^{\frac{1}{2}}(\mathbb{R}^3)$. Now, interpreting
  \[
   \begin{split}
    \big\langle\eta_n,\big(T_\lambda^{(\ell)}+\alpha\mathbbm{1}\big)\xi_n \big\rangle_{L^2}\;=\;\big\langle\eta_n,\big(T_\lambda^{(\ell)}+\alpha\mathbbm{1}\big)\xi_n \big\rangle_{H^{\frac{1}{2}},H^{-\frac{1}{2}}}
   \end{split}
  \]
  and using the fact that $T_\lambda^{(\ell)}+\alpha\mathbbm{1}$ is a bounded $H_\ell^{\frac{1}{2}}\to H_\ell^{-\frac{1}{2}}$ map (Lemma \ref{lem:Tlambdaproperties}(ii)), we see that $(T_\lambda^{(\ell)}+\alpha\mathbbm{1}\big)\xi_n\to (T_\lambda^{(\ell)}+\alpha\mathbbm{1}\big)\xi$ in $H_\ell^{-\frac{1}{2}}(\mathbb{R}^3)$ and therefore
  \[
   \mathcal{A}_F[\eta,\xi]\;=\;3\lim_{n\to\infty}\big\langle\eta_n,\big(T_\lambda^{(\ell)}+\alpha\mathbbm{1}\big)\xi_n \big\rangle_{H^{\frac{1}{2}},H^{-\frac{1}{2}}}\;=\;3\big\langle\eta,\big(T_\lambda^{(\ell)}+\alpha\mathbbm{1}\big)\xi \big\rangle_{H^{\frac{1}{2}},H^{-\frac{1}{2}}}\,.
  \]
 Formula \eqref{AFform-ellnot0} is thus proved.

 Next, the operator $\mathcal{A}_F$ is derived from its quadratic form in the usual manner, that is,
 \[
  \begin{split}
   \mathcal{D}(\mathcal{A}_F)\;&=\;\left\{
   \xi\in\mathcal{D}[\mathcal{A}_F]\,\left|\!
   \begin{array}{c}
    \exists\,\zeta_\xi\in H^{-\frac{1}{2}}_{W_\lambda,\ell}(\mathbb{R}^3)\textrm{ such that } \\
    \langle\eta,\zeta_\xi\rangle_{H^{-\frac{1}{2}}_{W_\lambda}}\,=\,\mathcal{A}_F[\eta,\xi]\;\;\;\forall\eta\in\mathcal{D}[\mathcal{A}_F]
   \end{array}
   \!\!\right.\right\} \\
   \mathcal{A}_F\,\xi\;&=\;\zeta_\xi
  \end{split}
 \]
 This means, owing to \eqref{eq:W-scalar-product} and \eqref{AFform-ellnot0}, that $\xi\in \mathcal{D}(\mathcal{A}_F)$ if and only if $\xi$ is a $H_\ell^{\frac{1}{2}}$-function with
 \[
  \big\langle\eta, W_\lambda^{(\ell)}\zeta_\xi-3\big(T_\lambda^{(\ell)}+\alpha\mathbbm{1}\big)\xi\big\rangle_{H^{\frac{1}{2}},H^{-\frac{1}{2}}}\;=\;0\qquad \forall \eta\in H_\ell^{\frac{1}{2}}(\mathbb{R}^3)\,.
 \]
 for some $\zeta_\xi\in H_\ell^{-\frac{1}{2}}(\mathbb{R}^3)$, and therefore equivalently
 \[
  \xi\in H_\ell^{\frac{1}{2}}(\mathbb{R}^3)\qquad\textrm{ and }\qquad 3\big(T_\lambda^{(\ell)}+\alpha\mathbbm{1}\big)\xi\;=\;W_\lambda^{(\ell)}\zeta_\xi\,.
 \]
 The second condition, owing to the $H^{-\frac{1}{2}}_\ell\to H^{\frac{1}{2}}_\ell$ bijectivity  of $W_\lambda^{(\ell)}$ (Lemma \ref{lem:Wlambdaproperties}(ii)), is tantamount as $\big(T_\lambda^{(\ell)}+\alpha\mathbbm{1}\big)\xi\in H_\ell^{\frac{1}{2}}(\mathbb{R}^3)$, and moreover $\zeta_\xi=3 W_\lambda^{-1}\big(T_\lambda^{(\ell)}+\alpha\mathbbm{1}\big)\xi$. Formula \eqref{AFop-ellnot0} is proved.   
  \end{proof}

 It is instructive to remark that whereas on the domain $\mathcal{D}_\ell$ the operator $\mathcal{A}_\lambda^{(\ell)}=3W_\lambda^{-1}\big(T_\lambda^{(\ell)}+\alpha\mathbbm{1}\big)$ is self-adjoint with respect to $H^{-\frac{1}{2}}_{W_\lambda,\ell}(\mathbb{R}^3)$, and therefore $T_\lambda^{(\ell)}$ on the same domain is symmetric with respect to $L^2(\mathbb{R}^3)$ (Lemma \ref{lem:symsym}), \emph{however} $T_\lambda^{(\ell)}$ is \emph{not} self-adjoint in $L^2_\ell(\mathbb{R}^3)$.

 \begin{lemma}\label{lem:exampleMinloswrong}
   Let $\lambda>0$, and $\ell\in\mathbb{N}$. The operator
   \begin{equation}
    \begin{split}
     \mathcal{D}\big(\mathsf{T}^{(\ell)}_\lambda\big)\;&:=\;\mathcal{D}_\ell\;=\;\big\{\xi\in H_\ell^{\frac{1}{2}}(\mathbb{R}^3)\,\big|\, T_{\lambda}^{(\ell)}\xi\in H_\ell^{\frac{1}{2}}(\mathbb{R}^3)\big\} \\
     \mathsf{T}^{(\ell)}_\lambda\xi\;&:=\; T_\lambda\xi\qquad\forall\xi\in\mathcal{D}\big(\mathsf{T}^{(\ell)}_\lambda\big)
    \end{split}
   \end{equation}
 is densely defined and symmetric with respect to the Hilbert space $L^2_\ell(\mathbb{R}^3)$. However, it is not self-adjoint.
  \end{lemma}

  \begin{proof}
   We already argued prior to stating the Lemma that $\mathsf{T}^{(\ell)}_\lambda$ is densely defined and symmetric in $L^2_\ell(\mathbb{R}^3)$. In fact, the symmetry property
   \[
    \big\langle\eta,\mathsf{T}^{(\ell)}_\lambda\xi \big\rangle_{L^2}\;=\;\big\langle\mathsf{T}^{(\ell)}_\lambda\eta,\xi \big\rangle_{L^2}\qquad\forall\xi,\eta\in \mathcal{D}_\ell
   \]
  also follows directly from Lemma \ref{lem:Tlambdaproperties}(v), because $\xi,\eta\in H^{\frac{1}{2}}(\mathbb{R}^3)$ and $\mathsf{T}^{(\ell)}_\lambda\xi,\mathsf{T}^{(\ell)}_\lambda\eta\in  H^{\frac{1}{2}}(\mathbb{R}^3)\subset L^2(\mathbb{R}^3)$.

  With respect to $L^2_\ell(\mathbb{R}^3)$ the quadratic form
   \[
   \begin{split}
    \mathcal{D}\big(q^{(\ell)}_\lambda\big)\;&:=\;H^{\frac{1}{2}}_\ell(\mathbb{R}^3) \\
    q^{(\ell)}_\lambda[\xi]\;&:=\;\big\langle\xi,T_\lambda^{(\ell)}\xi\big\rangle_{H^{\frac{1}{2}},H^{-\frac{1}{2}}}\;\approx\;\|\xi\|_{H^{\frac{1}{2}}}^2
   \end{split}
   \]
  is densely defined, coercive and hence lower semi-bounded with strictly positive lower bound (as follows from Lemma \ref{lem:xiTxi-equiv-H12}), and closed (because $ \mathcal{D}\big(q^{(\ell)}_\lambda\big)$ is obviously closed with respect to the norm induced by the form, namely the $H^{\frac{1}{2}}$-norm). As such, $q^{(\ell)}$ is the quadratic form of the self-adjoint operator
    \[
  \begin{split}
   \mathcal{D}(\mathsf{Q}_\lambda^{(\ell)})\;&=\;\left\{
   \xi\in\mathcal{D}\big(q^{(\ell)}_\lambda\big)\,\left|\!
   \begin{array}{c}
    \exists\,\zeta_\xi\in L^2_\ell(\mathbb{R}^3)\textrm{ such that } \\
    \langle\eta,\zeta_\xi\rangle_{L^2}\,=\,q^{(\ell)}_\lambda[\eta,\xi]\;\;\;\forall\eta\in\mathcal{D}\big(q^{(\ell)}_\lambda\big)
   \end{array}
   \!\!\right.\right\} \\
   \mathsf{Q}_\lambda^{(\ell)}\,\xi\;&=\;\zeta_\xi\,.
  \end{split}
 \]
  Equivalently, $\xi\in \mathcal{D}\big(\mathsf{Q}_\lambda^{(\ell)}\big)$ if and only if $\xi$ is an $H^{\frac{1}{2}}_\ell$-function such that
  \[
   \big\langle\eta,\zeta_\xi-T_{\lambda}^{(\ell)}\xi\big\rangle_{H^{\frac{1}{2}},H^{-\frac{1}{2}}}\;=\;0\qquad\forall\eta\in H_\ell^{\frac{1}{2}}(\mathbb{R}^3)
  \]
 for some $\zeta_\xi\in L^2_\ell(\mathbb{R}^3)$, and therefore equivalently
 \[
  \xi\in H_\ell^{\frac{1}{2}}(\mathbb{R}^3)\qquad\textrm{ and }\qquad T_{\lambda}^{(\ell)}\xi\,=\,\zeta_\xi\,.
 \]
 The second condition above is tantamount as $ T_{\lambda}^{(\ell)}\xi\in L^2_\ell(\mathbb{R}^3)$. In conclusion,
    \[
  \begin{split}
   \mathcal{D}\big(\mathsf{Q}_\lambda^{(\ell)}\big)\;&=\;\big\{\xi\in H_\ell^{\frac{1}{2}}(\mathbb{R}^3)\,\big|\, T_{\lambda}^{(\ell)}\xi\in L^2_\ell(\mathbb{R}^3)\big\}\\
   \mathsf{Q}_\lambda^{(\ell)}\,\xi\;&=\;T_{\lambda}^{(\ell)}\xi\,.
  \end{split}
 \]

 At this point it is clear that
 \[
  \mathsf{T}^{(\ell)}_\lambda\;\subset\;\mathsf{Q}_\lambda^{(\ell)}\;=\;\big(\mathsf{Q}_\lambda^{(\ell)}\big)^*\,.
 \]
 The lack of self-adjointness of $\mathsf{T}^{(\ell)}_\lambda$ is then evident from the strict inclusion $\mathcal{D}\big(\mathsf{T}^{(\ell)}_\lambda\big)\varsubsetneq\mathcal{D}\big(\mathsf{Q}_\lambda^{(\ell)}\big)$.  
\end{proof}

 \section{Sector of zero angular momentum}\label{sec:lzero}

  The problem of finding a domain $\mathcal{D}_0$ of self-adjointness in the Hilbert space $H^{-\frac{1}{2}}_{W_\lambda,\ell=0}(\mathbb{R}^3)$ for the operator $3W_\lambda^{-1}\big(T_\lambda^{(\ell=0)}+\alpha\mathbbm{1}\big)$ (in the following we shall shorten the full `$\ell=0$' superscript), is more subtle than the analogous problem for $\ell\in\mathbb{N}$ (Subsect.~\ref{sec:selfadj-ellnotzero}), and so too is the quest for a domain $\widetilde{\mathcal{D}_0}$ of sole symmetry.

  This is related with the fact that no Sobolev space $H^s_{\ell=0}(\mathbb{R}^3)$ is entirely mapped by $T_\lambda$ into $H^{\frac{1}{2}}(\mathbb{R}^3)$ (Remark \ref{rem:Tl-failstomap}), so $\widetilde{\mathcal{D}_0}$ cannot be a standard Sobolev space (as opposite to when $\ell\neq 0$: Lemma \ref{lem:Atildenot0}). A related difficulty, that emerges indirectly from the discussion of Lemma \ref{lem:xiTxi-equiv-H12}, is the fact that when $\ell=0$ the map
  \[
   \xi\;\longmapsto\;\int_{\mathbb{R}^3} \overline{\,\widehat{\xi}(\pp)}\, \big(\widehat{T_\lambda\xi}\big)(\pp)\,\ud\pp
  \]
  does not induce any longer an equivalent $H^{\frac{1}{2}}$-norm (see Remark \ref{rem:noequivH12norm} below). In fact, we shall see that any reasonable choice of a domain $\widetilde{\mathcal{D}_0}$ of symmetry for  $W_\lambda^{-1}T_\lambda^{(\ell=0)}$ makes it an unbounded below operator, unlike the lower semi-boundedness of $\mathcal{A}_\lambda^{(\ell)}$ when $\ell\neq 0$ (Lemma \ref{lem:Atildenot0}(ii)).

  These difficulties require an improved analysis that will be presented in this Section. They are also the source of various past mistakes leading to ill-posed models: Section \ref{sec:illposed} discusses such perspective.

  We shall follow the same conceptual path as in Sect.~\ref{sec:higherell}. First we discuss the symmetric case (symmetric realisation of $W_\lambda^{-1}T_\lambda^{(\ell=0)}$ and hence Ter-Martirosyan Skornyakov symmetric extension of $\mathring{H}$), then the self-adjoint case  (self-adjoint $W_\lambda^{-1}T_\lambda^{(\ell=0)}$ and hence self-adjoint TMS extension). For each two steps, an amount of technical preparation is needed.

  Here we opt to discuss explicitly only a special scenario, in fact the physically most relevant one: zero-range interaction with infinite scattering length, hence $\alpha=0$. This is the regime of \emph{unitarity} that we presented in the introduction. 

  \subsection{Mellin-like transformations}~

  For fixed $\lambda>0$, to each charge of interest $\xi\in H_{\ell=0}^{s}(\mathbb{R}^3)$, written according to \eqref{eq:xihatangularexpansion} as
  \begin{equation}\label{eq:0xi}
   \widehat{\xi}(\pp)\;=\;\frac{1}{\sqrt{4\pi}}f(|\pp|)\,,\qquad \pp\equiv|\pp|\Omega_{\pp}\,,\qquad f\in L^2(\mathbb{R}^+,(1+p^2)^sp^2\ud p)\,,
  \end{equation}
   we shall associate an odd, measurable function $\theta:\mathbb{R}\to\mathbb{C}$ defined by
  \begin{equation}\label{ftheta-1}
   \begin{split}
    \theta(x)\;&:=\;
    \begin{cases}
     \lambda f\big(\frac{2\sqrt{\lambda}}{\sqrt{3}}\sinh x\big)\sinh x\cosh x & \textrm{ if }x\geqslant 0 \\
     -\theta(-x) & \textrm{ if }x< 0
    \end{cases} \\
    x\;&:=\;\log\bigg(\sqrt{\frac{3p^2}{4\lambda}}+\sqrt{\frac{3p^2}{4\lambda}+1}\bigg)\,,\qquad p\,:=\,|\pp|\,.
   \end{split}
  \end{equation}
  The inverse transformation is
  \begin{equation}\label{ftheta-2}
   \begin{split}
    f(p)\;&=\;\frac{\,\theta\Big(\log\Big(\sqrt{\frac{3p^2}{4\lambda}}+\sqrt{\frac{3p^2}{4\lambda}+1}\Big)\Big)\,}{\sqrt{\frac{3}{4}p^2\,}\,\sqrt{\frac{3}{4}p^2+\lambda}} \\
    p\;&=\;\frac{2\sqrt{\lambda}}{\sqrt{3}}\,\sinh x\qquad \textrm{ for }x\geqslant 0\,.
   \end{split}
  \end{equation}
   The above change of variable $p\leftrightarrow x$ is a homeomorphism on $\mathbb{R}^+$, with also
   \begin{equation}\label{ftheta-3-eq:pxchangevar}
    \sinh x\;=\;\sqrt{\frac{3}{4}p^2\,}\,,\quad\sqrt{\lambda}\cosh x\;=\;\sqrt{\frac{3}{4}p^2+\lambda\,}\,,\quad\ud p\;=\;\frac{2\sqrt{\lambda}}{\sqrt{3}}\,\cosh x\,\ud x\,.
   \end{equation}
 It is for later convenience that the induced function $\theta$ on $\mathbb{R}^+$ has been extended by odd parity over the whole real line. 

 We shall refer to the function $\theta$ defined by \eqref{eq:0xi}-\eqref{ftheta-1} as the \emph{re-scaled radial component associated with the charge $\xi$ and with parameter $\lambda$}. When such correspondence need be emphasized, we shall write $\theta^{(\xi)}$.

   Using \eqref{ftheta-1}-\eqref{ftheta-3-eq:pxchangevar} and the fact that $1+p^2\sim\frac{3}{4}p^2+\lambda$, in the sense that each side is controlled from above and from below by the other with some $\lambda$-dependent constant, a straightforward computation gives
   \begin{equation}\label{eq:mellinnorms}
    \big\|\xi\big\|_{H^s(\mathbb{R}^3)}^2\;\approx\; 
    \big\|(\cosh x)^{s-\frac{1}{2}}\theta\big\|_{L^2(\mathbb{R})}^2
   \end{equation}
   in the sense of equivalence of norms (with $\lambda$-dependent multiplicative constants).

  Let us introduce further definitions and properties that are going to be useful in the course of the present discussion.

   The function $\theta$ having odd parity on $\mathbb{R}$, one has the identities
   \begin{equation}\label{eq:h-identities}
    \begin{split}
     &\int_{\mathbb{R}^+}\ud x\,\theta(x)\Big(\log \frac{\,2\cosh(x+y)-1\,}{\,2\cosh(x+y)+1\,}+\log \frac{\,2\cosh(x-y)+1\,}{\,2\cosh(x-y)-1\,}\Big) \\
     &=\;\int_{\mathbb{R}}\ud x\,\theta(x) \log \frac{\,2\cosh(x+y)-1\,}{\,2\cosh(x+y)+1\,}\;=\;\int_{\mathbb{R}}\ud x\,\theta(x) \log \frac{\,2\cosh(x-y)+1\,}{\,2\cosh(x-y)-1\,}\,.
    \end{split}
   \end{equation}

  Moreover (see, e.g., \cite[I.1.9.(50)]{Erdelyi-Tables1}),
  \begin{equation}\label{eq:logfourier}
   \Big(\log \frac{\,2\cosh x+1\,}{\,2\cosh x-1\,}\Big){\!\!\textrm{\huge ${\,}^{\widehat{\,}}$\normalsize}}\,(s)\;=\;\sqrt{2\pi}\,\frac{\sinh\frac{\pi}{6}s}{\,s\,\cosh \frac{\pi}{2}s}\,.
  \end{equation}
  By means of \eqref{eq:logfourier}, taking the Fourier transform in the following convolution yields
   \begin{equation}\label{eq:fourierconvolution}
     \bigg(\int_{\mathbb{R}}\ud x\,\theta(y) \log \frac{\,2\cosh(x-y)+1\,}{\,2\cosh(x-y)-1\,}\bigg){\!\!\textrm{\huge ${\,}^{\widehat{\,}}$\normalsize}}\,(s)\;=\;2\pi\,\widehat{\theta}(s)\,\frac{\sinh\frac{\pi}{6}s}{\,s\,\cosh \frac{\pi}{2}s}\,.
   \end{equation}
   Here and in the following the $s$-dependence in $\widehat{\theta}(s)$ is only symbolic, to indicate that the object $\widehat{\theta}$ is a distribution on test functions of $s\in\mathbb{R}$. Of course in special cases $\widehat{\theta}(s)$ may well be an ordinary function.

  Let $\gamma$ be the distribution on $\mathbb{R}$ defined by
  \begin{equation}\label{eq:gamma-distribution}
   \widehat{\gamma}(s)\;:=\;1-\frac{8}{\sqrt{3}}\,\frac{\sinh\frac{\pi}{6}s}{\,s\,\cosh \frac{\pi}{2}s}\,.
  \end{equation}
  The function $\mathbb{R}\ni s\mapsto\widehat{\gamma}(s)$ is smooth, even, strictly monotone increasing (resp., decreasing) for $s>0$ (resp., $s<0$), with values in $[1-\frac{4\pi}{3\sqrt{3}},1)$, asymptotically approaching $1$ as $s\to\pm\infty$, and with absolute minimum $ \widehat{\gamma}(0)=-(\frac{4\pi}{3\sqrt{3}}-1)$. The equation $\widehat{\gamma}(s)=0$ has thus simple roots $s=\pm s_0$, with $s_0\approx 1.0062$. We also define
  \begin{equation}\label{eq:gammaplus}
   \widehat{\gamma}_+(s)\;:=\;\frac{1}{\,(s-s_0)(s+s_0)\,}\, \widehat{\gamma}(s)\,.	
  \end{equation}
  $\widehat{\gamma}_+$ is therefore strictly positive, smooth, even, monotone to zero decreasing for $s>0$ with $s^{-2}$ decay, and with absolute maximum  $\widehat{\gamma}_+(0)=s_0^{-2}(\frac{4\pi}{3\sqrt{3}}-1)$ (Figure \ref{fig:gammagammaplus}).

\begin{figure}[t!]
\includegraphics[width=8cm]{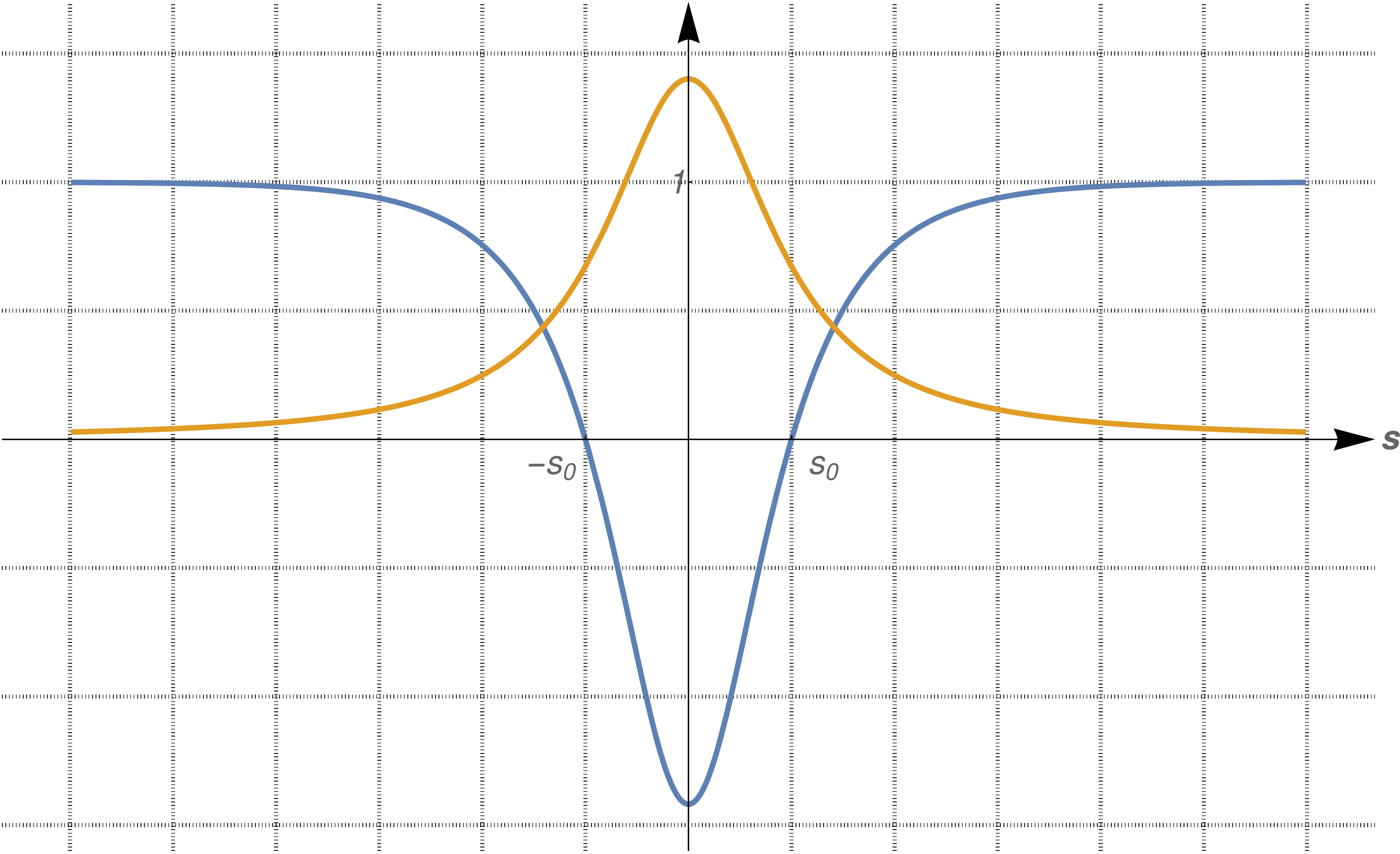}
\caption{Plot of the functions $\widehat{\gamma}(s)$ (blue) and $\widehat{\gamma}_+(s)$ (orange), defined respectively in \eqref{eq:gamma-distribution} and \eqref{eq:gammaplus}.}\label{fig:gammagammaplus}
\end{figure}

   Further quantities of interest involving $\xi$ are conveniently expressed in terms of the auxiliary function $\theta$ or its (one-dimensional) Fourier transform $\widehat{\theta}$.

   \begin{lemma}\label{lem:xithetaidentities}
    Let $\lambda>0$, $s\in\mathbb{R}$, and $\xi$ be as in \eqref{eq:0xi}. One has the identities  
    \begin{equation}\label{eq:Tlxi-theta}
     \big(\widehat{T_\lambda^{(0)}\xi}\big)(\pp)\;=\;\frac{1}{\sqrt{4\pi}\,|\pp|}\,\frac{4\pi^2}{\sqrt{3}\,}\Big(\theta(x)-\frac{4}{\pi\sqrt{3}}\int_{\mathbb{R}}\ud y\,\theta(y) \log \frac{\,2\cosh(x-y)+1\,}{\,2\cosh(x-y)-1\,}\Big),
    \end{equation}
        \begin{equation}\label{eq:Tlambda0xi-with-theta}
     \begin{split}
      \big\|T_\lambda^{(0)}\xi\big\|_{H^{s}(\mathbb{R}^3)}^2\;&\approx\;\int_{\mathbb{R}}\ud x\,(\cosh x)^{1+2s}\,\bigg|\,\theta(x)-\frac{4}{\pi\sqrt{3}}\int_{\mathbb{R}}\ud y\,\theta(y) \log \frac{\,2\cosh(x-y)+1\,}{\,2\cosh(x-y)-1\,}\bigg|^2,
     \end{split}
    \end{equation}
   and 
    \begin{equation}\label{eq:xiTxi-with-theta}
     \int_{\mathbb{R}^3} \overline{\,\widehat{\xi}(\pp)}\, \big(\widehat{T_\lambda^{(0)}\xi}\big)(\pp)\,\ud\pp\;=\;\frac{\,8\pi^2}{3\sqrt{3}}\int_{\mathbb{R}}\ud s\, \widehat{\gamma}(s)\,|\widehat{\theta}(s)|^2\,,
    \end{equation}
    with $x$ and $\theta$ given by \eqref{ftheta-1}, and $\gamma$ given by \eqref{eq:gamma-distribution}.
    In \eqref{eq:Tlxi-theta} it is understood that $x\geqslant 0$, and \eqref{eq:Tlambda0xi-with-theta} is meant as an equivalence of norms (with $\lambda$-dependent multiplicative constant).    
   \end{lemma}

 \begin{proof}
  Specialising formula \eqref{eq:fellsector} of Lemma \ref{lem:Tlambdadecomposition}(i) with the Legendre polynomial $P_0\equiv 1$ gives
  \[
   \big(\widehat{T_\lambda^{(0)}\xi}\big)(\pp)\;=\;\frac{1}{\sqrt{4\pi}\,}\,\frac{1}{p}\Big(2\pi^2pf(p)\sqrt{\frac{3}{4}p^2+\lambda\,}\,-\,4\pi\int_{\mathbb{R}^+}\!\ud q\,q f(q)\,\log\frac{\,p^2+q^2+pq+\lambda\,}{p^2+q^2-pq+\lambda}\Big)\,.
  \]
  With $p=\frac{2\sqrt{\lambda}}{\sqrt{3}}\sinh x$ and $q=\frac{2\sqrt{\lambda}}{\sqrt{3}}\sinh y$ one has 
  \[\tag{*}\label{logppqq}
   \begin{split}
    \frac{\,p^2+q^2+pq+\lambda\,}{p^2+q^2-pq+\lambda}\;&=\;\frac{\,\sinh^2 x+\sinh^2y+\sinh x\,\sinh y+\frac{3}{4}\,}{\,\sinh^2 x+\sinh^2y-\sinh x\,\sinh y+\frac{3}{4}\,} \\
    &=\;\frac{\,2\cosh(x+y)-1\,}{\,2\cosh(x+y)+1\,}\,\frac{\,2\cosh(x-y)+1\,}{\,2\cosh(x-y)-1\,}\,.
   \end{split}
  \]
 Using the latter identity and \eqref{ftheta-2} one then finds
  \[
   \sqrt{4\pi}\,|\pp|\,\big(\widehat{T_\lambda^{(0)}\xi}\big)(\pp)\;=\;\frac{4\pi^2}{\sqrt{3}\,}\Big(\theta(x)-\frac{4}{\pi\sqrt{3}}\int_{\mathbb{R}^+}\!\ud y\,\theta(y)\big(A_x(y)+B_x(y)\big)\Big)\,,
  \]
  where
  \[
   A_x(y)\;:=\;\log \frac{\,2\cosh(x+y)-1\,}{\,2\cosh(x+y)+1\,}\,,\qquad B_x(y)\;:=\;\log \frac{\,2\cosh(x-y)+1\,}{\,2\cosh(x-y)-1\,}\,.
  \]
  Combining this with \eqref{eq:h-identities} yields \eqref{eq:Tlxi-theta}. 
  
  Next, by means of \eqref{ftheta-3-eq:pxchangevar} and \eqref{eq:Tlxi-theta} we find
   \[
    \begin{split}
      &\big\|T_\lambda^{(0)}\xi\big\|_{H^{s}(\mathbb{R}^3)}^2 \\
      &\approx\;\int_{\mathbb{R}^+}\ud p\,p^2\,\Big(\frac{3}{4}p^2+\lambda\Big)^s\,\bigg|\frac{1}{p}\,\Big(\theta(x)-\frac{4}{\pi\sqrt{3}}\int_{\mathbb{R}}\ud y\,\theta(y) \log \frac{\,2\cosh(x-y)+1\,}{\,2\cosh(x-y)-1\,}\Big)\bigg|^2 \\
      &=\frac{\,2\lambda^{\frac{1}{2}+s}}{\sqrt{3}}\int_{\mathbb{R}^+}\ud x\,(\cosh x)^{1+2s}\,\bigg|\,\theta(x)-\frac{4}{\pi\sqrt{3}}\int_{\mathbb{R}}\ud y\,\theta(y) \log \frac{\,2\cosh(x-y)+1\,}{\,2\cosh(x-y)-1\,}\bigg|^2.
    \end{split}
   \]
   Owing to the odd parity of $\theta$ and to \eqref{eq:h-identities}, the integrand function above is invariant under change of variable $x\mapsto -x$, therefore the last line can be re-written as
   \[
    \frac{\,\lambda^{\frac{1}{2}+s}}{\sqrt{3}}\int_{\mathbb{R}}\ud x\,(\cosh x)^{1+2s}\,\bigg|\,\theta(x)-\frac{4}{\pi\sqrt{3}}\int_{\mathbb{R}}\ud y\,\theta(y) \log \frac{\,2\cosh(x-y)+1\,}{\,2\cosh(x-y)-1\,}\bigg|^2 
   \]
  This gives \eqref{eq:Tlambda0xi-with-theta}.

   Concerning \eqref{eq:xiTxi-with-theta}, specialising formulas \eqref{eq:xiTxipre}-\eqref{eq:xiTxi} of Lemma \ref{lem:Tlambdadecomposition}(ii) with the Legendre polynomial $P_0\equiv 1$ gives
    \begin{equation*}
   \begin{split}
    &\frac{1}{2\pi^2}\int_{\mathbb{R}^3} \overline{\,\widehat{\xi}(\pp)}\, \big(\widehat{T_\lambda^{(0)}\xi}\big)(\pp)\,\ud\pp\;=\;\int_{\mathbb{R}^+}\!\ud p\,\sqrt{{\textstyle\frac{3}{4}}p^2+\lambda}\,|p f(p)|^2 \\
   &\qquad -\frac{2}{\pi}\iint_{\mathbb{R}^+\times\mathbb{R}^+}\ud p\,\ud q\, (p \overline{f(p)})\,(q f(q))\,\log\frac{\,p^2+q^2+pq+\lambda\,}{p^2+q^2-pq+\lambda}
    \end{split}
  \end{equation*}
   The first summand in the r.h.s.~above can be re-written as
   \[
    \int_{\mathbb{R}^+}\!\ud p\,\sqrt{{\textstyle\frac{3}{4}}p^2+\lambda}\,|p f(p)|^2\;=\;\frac{8}{\,3\sqrt{3}}\int_{\mathbb{R}^+}\ud x\,|\theta(x)|^2\;=\;\frac{4}{\,3\sqrt{3}}\int_{\mathbb{R}}\ud x\,|\theta(x)|^2\,,
   \]
  having used \eqref{ftheta-2}-\eqref{ftheta-3-eq:pxchangevar} in the first step and the odd parity of $\theta$ in the second. Analogously, and using also \eqref{logppqq} and \eqref{eq:h-identities}, the second summand becomes
  \[
   \begin{split}
    & \frac{2}{\pi}\iint_{\mathbb{R}^+\times\mathbb{R}^+}\ud p\,\ud q\, (p \overline{f(p)})\,(q f(q))\,\log\frac{\,p^2+q^2+pq+\lambda\,}{p^2+q^2-pq+\lambda} \\
    &=\;\frac{32}{9\pi}\iint_{\mathbb{R}^+\times\mathbb{R}^+}\ud x\,\ud y\,\overline{\theta(x)}\,\theta(y)\Big(\log \frac{\,2\cosh(x+y)-1\,}{\,2\cosh(x+y)+1\,}+\log \frac{\,2\cosh(x-y)+1\,}{\,2\cosh(x-y)-1\,}\Big) \\
     &=\;\frac{16}{9\pi}\iint_{\mathbb{R}\times\mathbb{R}}\ud x\,\ud y\,\overline{\theta(x)}\,\theta(y)\log \frac{\,2\cosh(x-y)+1\,}{\,2\cosh(x-y)-1\,}\,.
   \end{split}
  \]
  Thus,
      \begin{equation*}
   \begin{split}
    \frac{1}{2\pi^2}\int_{\mathbb{R}^3} \overline{\,\widehat{\xi}(\pp)}\, &\big(\widehat{T_\lambda^{(0)}\xi}\big)(\pp)\,\ud\pp\;=\;\frac{4}{\,3\sqrt{3}}\bigg(\int_{\mathbb{R}}\ud x\,|\theta(x)|^2-\\
   &\qquad -\frac{4}{\pi\sqrt{3}}\iint_{\mathbb{R}\times\mathbb{R}}\ud x\,\ud y\,\overline{\theta(x)}\,\theta(y)\log \frac{\,2\cosh(x-y)+1\,}{\,2\cosh(x-y)-1\,}\bigg)\,.
    \end{split}
  \end{equation*}
  Applying Parseval's identity in both summands of the above r.h.s., and using \eqref{eq:fourierconvolution} in the second summand, one gets
  \[
   \frac{1}{2\pi^2}\int_{\mathbb{R}^3} \overline{\,\widehat{\xi}(\pp)}\, \big(\widehat{T_\lambda^{(0)}\xi}\big)(\pp)\,\ud\pp\;=\;\frac{4}{\,3\sqrt{3}}\bigg(\int_{\mathbb{R}}\ud s\,|\widehat{\theta}(s)|^2-\frac{8}{\sqrt{3}}\int_{\mathbb{R}}\ud s\,|\widehat{\theta}(s)|^2\frac{\sinh\frac{\pi}{6}s}{\,s\,\cosh \frac{\pi}{2}s}\bigg).
  \]
  This, and the definition \eqref{eq:gamma-distribution}, finally prove \eqref{eq:xiTxi-with-theta}.   
  \end{proof}

  \begin{remark}\label{rem:noequivH12norm}
   As the function $\widehat{\gamma}$ attains both positive and negative values, formula \eqref{eq:xiTxi-with-theta} shows that the pairing $\int_{\mathbb{R}^3} \overline{\,\widehat{\xi}(\pp)}\, \big(\widehat{T_\lambda^{(0)}\xi}\big)(\pp)\,\ud\pp$ is \emph{not equivalent} to the $L^2$-norm of the associated re-scaled radial function $\theta^{(\xi)}$, and therefore (owing to \eqref{eq:mellinnorms}) does not induce an equivalent $H^{\frac{1}{2}}$-norm, in contrast with the analogous properties in the sectors with $\ell\neq 0$ (Lemma \ref{lem:xiTxi-equiv-H12}).   
  \end{remark}

  \begin{lemma}\label{lem:xiWxi-zero}
   Let $\lambda>0$ and let $\xi_1,\xi_2$ be spherically symmetric functions with re-scaled radial components $\theta_1,\theta_2$ respectively, according to the definition \eqref{eq:0xi}-\eqref{ftheta-1}. Then 
   \begin{equation}\label{eq:xiWxi-zero}
    \begin{split}
      &\int_{\mathbb{R}^3} \overline{\,\widehat{\xi}_1(\pp)}\, \big(\widehat{W_\lambda^{(0)}\xi_2}\big)(\pp)\,\ud\pp\;=\;\frac{4\pi^2}{\,\lambda\sqrt{3}\,}\int_{\mathbb{R}}\frac{\,\overline{\theta_1}(x)\,\theta_2(x)}{\,(\cosh x)^2}\,\ud x \\
      &\qquad +\frac{32\pi}{\lambda}\iint_{\mathbb{R}\times\mathbb{R}}\frac{\overline{\theta_1(x)}\,\theta_2(y)}{\,(2\cosh(x+y)+1)\,(2\cosh(x-y)-1)\,}\,\ud x\,\ud y
    \end{split}
   \end{equation}
   and also
   \begin{equation}\label{eq:xiWxi-zero-FOURIER}
    \begin{split}
      &\int_{\mathbb{R}^3} \overline{\,\widehat{\xi}_1(\pp)}\, \big(\widehat{W_\lambda^{(0)}\xi_2}\big)(\pp)\,\ud\pp\;=\;\frac{2\pi^2}{\,\lambda\sqrt{3}\,}\int_{\mathbb{R}}\Big(\frac{s}{\,\sinh\frac{\pi}{2}s}*\overline{\widehat{\theta}_1}\Big)(s)\,\widehat{\theta}_2(s)\,\ud s \\
      &\qquad +\frac{16\pi}{3\lambda}\iint_{\mathbb{R}\times\mathbb{R}}\overline{\widehat{\theta}_1(s)}\,\widehat{\theta}_2(t)\,\frac{\sinh\frac{\pi}{6}(s+t)}{\sinh \frac{\pi}{2}(s+t)}\,\frac{\sinh\frac{\pi}{3}(s-t)}{\sinh \frac{\pi}{2}(s-t)}\,\ud s\,\ud t\,.
    \end{split}
   \end{equation}
   In the double integrals above the order of integration is not specified, tacitly understanding that $\theta_1
   \theta_2$ (or $\widehat{\theta}_1\widehat{\theta}_2$) is sufficiently integrable, depending on the applications.
  \end{lemma}

  \begin{proof}
   Specialising formula \eqref{eq:Wellsp} with the Legendre polynomial $P_0\equiv 1$ gives
   \begin{equation*}
   \begin{split}
   &\int_{\mathbb{R}^3} \overline{\,\widehat{\xi_1}(\pp)}\, \big(\widehat{W_\lambda^{(0)}\xi_2}\big)(\pp)\,\ud\pp \\
   &\;=\;\int_{\mathbb{R}^+}\!\ud p\,\frac{3\pi^2}{\sqrt{{\textstyle\frac{3}{4}}p^2+\lambda}\,}\,(p\overline{f_1(p)})\,(pf_2(p)) \\
   &\qquad +12\pi\!\iint_{\mathbb{R}^+\times\mathbb{R}^+}\ud p\,\ud q\,(p\overline{f_1(p)})\,(q f_2(q))\,\frac{2 p q}{\,(p^2+q^2+\lambda)^2-(p q)^2}\,.
    \end{split}
  \end{equation*}
    The first summand in the r.h.s.~above can be re-written as
   \[
   \begin{split}
    \int_{\mathbb{R}^+}\!\ud p\,\frac{3\pi^2}{\sqrt{{\textstyle\frac{3}{4}}p^2+\lambda}\,}\,(p\overline{f_1(p)})\,(pf_2(p))\;&=\;\frac{8\pi^2}{\,\lambda\sqrt{3}\,}\int_{\mathbb{R}^+}\frac{\,\overline{\theta_1}(x)\,\theta_2(x)}{\,(\cosh x)^2}\,\ud x \\
    &=\;\frac{4\pi^2}{\,\lambda\sqrt{3}\,}\int_{\mathbb{R}}\frac{\,\overline{\theta_1}(x)\,\theta_2(x)}{\,(\cosh x)^2}\,\ud x\,,
   \end{split}
   \]
  having used \eqref{ftheta-2}-\eqref{ftheta-3-eq:pxchangevar} in the first step and the odd parity of $\theta_1$ and $\theta_2$ in the second. Next, with $p=\frac{2\sqrt{\lambda}}{\sqrt{3}}\sinh x$ and $q=\frac{2\sqrt{\lambda}}{\sqrt{3}}\sinh y$, we re-write
  \[
   \begin{split}
    \frac{2 p q}{\,(p^2+q^2+\lambda)^2-(p q)^2}\;&=\;\frac{1}{\,p^2+q^2-pq+\lambda\,}-\frac{1}{\,p^2+q^2+pq+\lambda\,} \\
     &=\;\frac{3}{\lambda}\big( a(x,y)-b(x,y)\big)\,,
   \end{split}
  \]
   where
   \[
    \begin{split}
     a(x,y)\;& :=\;\frac{1}{\,(2\cosh(x+y)+1)\,(2\cosh(x-y)-1)\,} \\
     b(x,y)\;& :=\;\frac{1}{\,(2\cosh(x+y)-1)\,(2\cosh(x-y)+1)\,}
    \end{split}
   \]
  This and \eqref{ftheta-2}-\eqref{ftheta-3-eq:pxchangevar} then imply
  \[
   \begin{split}
    & 12\pi\!\iint_{\mathbb{R}^+\times\mathbb{R}^+}\ud p\,\ud q\,(p\overline{f_1(p)})\,(q f_2(q))\,\frac{2 p q}{\,(p^2+q^2+\lambda)^2-(p q)^2} \\
    &=\;\frac{64\pi}{\lambda}\iint_{\mathbb{R}^+\times\mathbb{R}^+}\ud x\,\ud y\,\overline{\theta_1(x)}\,\theta_2(y)\,\big( a(x,y)-b(x,y)\big) \\
    &=\;\frac{32\pi}{\lambda}\iint_{\mathbb{R}\times\mathbb{R}}\ud x\,\ud y\,\overline{\theta_1(x)}\,\theta_2(y)\,a(x,y) \\
    &=\;\frac{32\pi}{\lambda}\iint_{\mathbb{R}\times\mathbb{R}}\frac{\overline{\theta_1(x)}\,\theta_2(y)}{\,(2\cosh(x+y)+1)\,(2\cosh(x-y)-1)\,}\,\ud x\,\ud y
   \end{split}
  \]
  the second identity being due to the odd parity of $\theta_1$ and $\theta_2$ and to the obvious relations $a(-x,-y)=a(x,y)$, $a(-x,y)=b(x,y)=a(x,-y)$. Adding up the two summands we have thus worked out yields finally  \eqref{eq:xiWxi-zero}.

  Concerning \eqref{eq:xiWxi-zero-FOURIER},
  \[
    \Big(\frac{1}{(\cosh x)^2}\Big){}^{\textrm{\LARGE $\widehat{\,}$\normalsize}}\;(s)\;=\;\sqrt{\frac{\pi}{2}}\,\frac{s}{\,\sinh\frac{\pi}{2}s}\,,
  \]
  whence
  \[
   \frac{4\pi^2}{\,\lambda\sqrt{3}\,}\int_{\mathbb{R}}\frac{\,\overline{\theta_1}(x)\,\theta_2(x)}{\,(\cosh x)^2}\,\ud x\;=\;\frac{2\pi^2}{\,\lambda\sqrt{3}\,}\int_{\mathbb{R}}\Big(\frac{s}{\,\sinh\frac{\pi}{2}s}*\overline{\widehat{\theta}_1}\Big)(s)\,\widehat{\theta}_2(s)\,\ud s
  \]
%
%
  Moreover (see, e.g., \cite[I.1.9.(6)]{Erdelyi-Tables1}),
  \[
   \begin{split}
    \Big(\frac{1}{\,2\cosh x -1}\Big){}^{\textrm{\LARGE $\widehat{\,}$\normalsize}}\;(s)\;&=\;\sqrt{\frac{2\pi}{3}}\,\frac{\sinh\frac{2\pi}{3}s}{\sinh \pi s} \\
    \Big(\frac{1}{\,2\cosh x +1}\Big){}^{\textrm{\LARGE $\widehat{\,}$\normalsize}}\;(s)\;&=\;\sqrt{\frac{2\pi}{3}}\,\frac{\sinh\frac{\pi}{3}s}{\sinh \pi s}\,,
   \end{split}
  \]
  whence
  \[
   \begin{split}
    &\frac{1}{2\pi}\iint_{\mathbb{R}\times\mathbb{R}}e^{-\ii s x}\,e^{-\ii t y}\frac{1}{\,(2\cosh(x+y)+1)\,(2\cosh(x-y)-1)\,}\,\ud x\,\ud y \\
    &=\;\frac{1}{4\pi}\iint_{\mathbb{R}\times\mathbb{R}}\,e^{-\ii\frac{s+t}{2} x}\,e^{-\ii\frac{s-t}{2} x}\,\frac{1}{\,2\cosh x +1}\,\frac{1}{\,2\cosh x -1}\,\ud x\,\ud y \\
    &=\;\frac{1}{6}\,\frac{\sinh\frac{\pi}{6}(s+t)}{\sinh \frac{\pi}{2}(s+t)}\,\frac{\sinh\frac{\pi}{3}(s-t)}{\sinh \frac{\pi}{2}(s-t)}\,,
   \end{split}
  \]
  and
  \[
   \begin{split}
    &\frac{32\pi}{\lambda}\iint_{\mathbb{R}\times\mathbb{R}}\frac{\overline{\theta_1(x)}\,\theta_2(y)}{\,(2\cosh(x+y)+1)\,(2\cosh(x-y)-1)\,}\,\ud x\,\ud y \\
    &=\;\frac{16\pi}{3\lambda}\iint_{\mathbb{R}\times\mathbb{R}}\overline{\widehat{\theta}_1(s)}\,\widehat{\theta}_2(t)\,\frac{\sinh\frac{\pi}{6}(s+t)}{\sinh \frac{\pi}{2}(s+t)}\,\frac{\sinh\frac{\pi}{3}(s-t)}{\sinh \frac{\pi}{2}(s-t)}\,\ud s\,\ud t\,.
   \end{split}
  \]
  Adding up all together yields \eqref{eq:xiWxi-zero-FOURIER}. 
  \end{proof}

  \subsection{Radial Ter-Martirosyan Skornyakov equation}~
  
  The next technical tool is the solution formula for the equation
  \begin{equation}\label{eq:TMSeq0}
   T_\lambda^{(0)}\xi\;=\;\eta
  \end{equation}
  in the unknown $\xi$ and with datum $\eta$, spherically symmetric on $\mathbb{R}^3$. \eqref{eq:TMSeq0} is customarily referred to as the \emph{Ter-Martirosyan Skornyakov equation} for the sector of zero angular momentum. It appeared for the first time in \cite[Eq.~(12)]{TMS-1956}, the already-mentioned work by Ter-Martirosyan and Skornyakov, whence the name.

  Representing as usual (see \eqref{eq:xihatangularexpansion} and \eqref{eq:0xi} above)
  \[
   \widehat{\xi}(\pp)\;=\;\frac{1}{\sqrt{4\pi}}f^{(\xi)}(|\pp|)\,,\qquad  \widehat{\eta}(\pp)\;=\;\frac{1}{\sqrt{4\pi}}f^{(\eta)}(|\pp|)
  \]
  in terms of the corresponding radial components, switching to the \emph{re-scaled} radial components $\theta^{(\xi)}$ and $\theta^{(\eta)}$ defined in \eqref{ftheta-1}, and representing the l.h.s.~of \eqref{eq:TMSeq0} by means of \eqref{eq:Tlxi-theta}, equation \eqref{eq:TMSeq0} takes the form
  \begin{equation}\label{radialTMS0}
    \theta^{(\xi)}(x)-\frac{4}{\pi\sqrt{3}}\int_{\mathbb{R}}\ud y\,\theta^{(\xi)}(y) \log \frac{\,2\cosh(x-y)+1\,}{\,2\cosh(x-y)-1\,}\;=\;\frac{1}{2\pi^2\sqrt{\lambda}}\frac{\theta^{(\eta)}(x)}{\cosh x}
    \end{equation}
  as an identity for $x\geqslant 0$. Let us simply re-write
  \[
   \theta\;\equiv\; \theta^{(\xi)}\,,\qquad \vartheta\;\equiv\;\frac{1}{2\pi^2\sqrt{\lambda}}\frac{\theta^{(\eta)}}{\cosh x}\,,
  \]
  and hence
  \begin{equation}\label{radialTMS0radial}
    \theta(x)-\frac{4}{\pi\sqrt{3}}\int_{\mathbb{R}}\ud y\,\theta(y) \log \frac{\,2\cosh(x-y)+1\,}{\,2\cosh(x-y)-1\,}\;=\;\vartheta(x)
    \end{equation}
   in the unknown $\theta$. 
   Should one like to interpret \eqref{radialTMS0radial} as an identity on the whole real line, one has to assume for consistency that also $\vartheta$, as $\theta^{(\eta)}$, is prolonged by odd parity. We shall refer to \eqref{radialTMS0radial} as the \emph{radial Ter-Martirosyan Skornyakov equation} (for the $\ell=0$ sector).

   Applying  \eqref{eq:fourierconvolution}-\eqref{eq:gamma-distribution} one sees that \eqref{radialTMS0radial} is equivalently re-written in the Fourier-transformed version
  \begin{equation}\label{radialTMS0radial-FOURIER}
   \widehat{\gamma}(s)\,\widehat{\theta}(s)\;=\;\widehat{\vartheta}(s)\,,
  \end{equation}
  understood in general as a distribution equation in the distribution unknown $\widehat{\theta}$.

   In  \eqref{radialTMS0radial} the functional space for the unknown $\theta$ is determined by the space for the original unknown $\xi$ through formula \eqref{eq:mellinnorms}. The same holds for the functional space for $\vartheta$, recalling that $\vartheta=\frac{1}{2\pi^2\sqrt{\lambda}}\frac{\theta^{(\eta)}(x)}{\cosh x}$ for some datum $\eta$. In our applications (Sect.~\ref{sec:symmTMSubdd}-\ref{sec:adjointBirman}) we shall need $\eta\in H^{s}_{\ell=0}(\mathbb{R}^3)$ for some $s>-\frac{1}{2}$, in which case \eqref{eq:mellinnorms} implies that $\theta^{(\eta)}/(\cosh x)^{\frac{1}{2}-s}$ is an $L^2$-function and hence, by the H\"{o}lder inequality,
   \[
    \vartheta\;\sim\;\frac{1}{(\cosh x)^{\frac{1}{2}+s}}\,\frac{\theta^{(\eta)}}{(\cosh x)^{\frac{1}{2}-s}}\;\in\; L^p(\mathbb{R})\quad\forall p\in[1,2]\,.
   \]
   Therefore, as $|x|\to\infty$, the above $\vartheta$ has an $L^2$-behaviour dumped by a multiplicative exponential decay.

   When $\vartheta$ is smooth and has rapid decrease, we can translate the solution to \eqref{radialTMS0radial-FOURIER} back to the $x$-coordinate.

   \begin{lemma}\label{lem:0tms-generalsol}
    The general solution to \eqref{radialTMS0radial} when $\vartheta$ is smooth and with rapid decrease is 
    \begin{equation}\label{eq:0tms-generalsol}
      \theta(x)\;=\;c\,\sin s_0 x - \frac{1}{2s_0}\,\sin (s_0 |x|)*\Big(\frac{\widehat{\vartheta}}{\:\widehat{\gamma}_+}\Big)^{\!\vee}\,,\qquad c\in\mathbb{C}\,.
    \end{equation}
    The function \eqref{eq:0tms-generalsol} belongs to $L^\infty(\mathbb{R})$ and is asymptotically $\cos$-periodic as $|x|\to +\infty$ with period $2\pi/s_0$.
%
%
   \end{lemma}

   In the proof, as well as in the sequel, we shall make use of the identities
   \begin{eqnarray}
    \big(\delta(s-s_0)-\delta(s+s_0)\big)^\vee(x)\!\!&=&\!\!\frac{2\ii}{\sqrt{2\pi}}\,\sin s_0 x\,, \label{eq:distri-deltaFouriersin}\\
    \Big(PV\frac{1}{s-s_0}-PV\frac{1}{s+s_0}\Big)^\vee(x)\!\!&=&\!\!
    -\sqrt{2\pi}\sin s_0|x|\,, \label{eq:distriPVsignsin}
   \end{eqnarray}
  where $PV$ stands for the principal value distribution. \eqref{eq:distriPVsignsin} in particular follows from $\big(PV\frac{1}{s-s_0}\big)^\vee(x)=e^{\ii s_0 x}\big(PV\frac{1}{s}\big)^\vee(x) =\ii\sqrt{\frac{\pi}{2}}\,e^{\ii s_0 x}\,\mathrm{sign}\,x$.

   \begin{proof}[Proof of Lemma \ref{lem:0tms-generalsol}]
    Let us consider the Fourier-transformed version \eqref{radialTMS0radial-FOURIER} of \eqref{radialTMS0radial}. This is the same as
   \[
    (s-s_0)(s+s_0)\,\widehat{\theta}(s)\;=\;\frac{\widehat{\vartheta}(s)}{\widehat{\gamma}_+(s)}\,,\qquad s\in\mathbb{R}\,,
   \]
   owing to \eqref{eq:gammaplus}. Dividing by $\widehat{\gamma}_+(s)$ has not altered the set of solutions because $\widehat{\gamma}_+(s)>0$ $\forall s\in\mathbb{R}$. We observe that by construction $\widehat{\theta}(-s)=-\widehat{\theta}(s)$, $\widehat{\vartheta}(-s)=-\widehat{\vartheta}(s)$, and  $\widehat{\gamma}_+(-s)=\widehat{\gamma}_+(s)$, thus in the identity above both sides change sign when $s\mapsto -s$, consistently. Moreover, since $\widehat{\gamma}_+$ is smooth, $\widehat{\gamma}_+(s)\sim s^{-2}$ as $|s|\to\infty$, and $\vartheta$ is smooth and with rapid decrease, then so too is $\widehat{\vartheta}/\widehat{\gamma}_+$. Therefore, such distribution equation has general solution
  \[
   \widehat{\theta}(s)\;=\;c\big(\delta(s-s_0)-\delta(s+s_0)\big)+\frac{\widehat{\vartheta}(s)}{\widehat{\gamma}_+(s)}\,\frac{1}{2s_0}\Big(PV\frac{1}{s-s_0}-PV\frac{1}{s+s_0}\Big)\,,\quad c\in\mathbb{C}\,.
  \]
  The linear combination of $\delta(s-s_0)$ and $\delta(s+s_0)$ had to be anti-symmetric, owing to the odd parity of $\widehat{\theta}$. Taking the inverse Fourier transform by means of \eqref{eq:distri-deltaFouriersin}-\eqref{eq:distriPVsignsin}, one obtains
 \[
  \theta(x)\;=\;\frac{2\ii c}{\sqrt{2\pi}}\,\sin s_0 x - \frac{1}{2s_0}\,\sin (s_0 |x|)*\Big(\frac{\widehat{\vartheta}}{\:\widehat{\gamma}_+}\Big)^{\!\vee}\,.
 \]
 An obvious re-scaling of the arbitrary constant $c$ finally yields \eqref{eq:0tms-generalsol}. 

   The functions $\vartheta$, $\widehat{\vartheta}$, $\widehat{\vartheta}/\widehat{\gamma}_+$, and $(\widehat{\vartheta}/\widehat{\gamma}_+)^\vee$ are all functions of rapid decrease in the respective variable. Young's inequality then gives
   \begin{equation*}
    \Big\|\sin (s_0 |x|)*\Big(\frac{\widehat{\vartheta}}{\:\widehat{\gamma}_+}\Big)^\vee\Big\|_{L^\infty(\mathbb{R}_x)}\;\leqslant\;\Big\|\Big(\frac{\widehat{\vartheta}}{\:\widehat{\gamma}_+}\Big)^{\!\vee}\Big\|_{L^1(\mathbb{R}_x)}\;<\;+\infty.
   \end{equation*}
   The whole function \eqref{eq:0tms-generalsol} is therefore bounded.

   Next, let us set 
   \[
    \begin{split}
     h\;&:=\;(\widehat{\vartheta}/\widehat{\gamma}_+)^\vee \\
     A\;&:=\;\sin (s_0 |x|)*h\,.
    \end{split}
   \]
   Since $h$ is smooth and with rapid decrease, then $A$ is differentiable at any order, and a simple computation yields
   \[
     A''(x)\;=\;-s_0^2 A(x)+2s_0 h(x)\,.
   \]
   Thus, using again the rapid decrease of $h$, asymptotically when $|x|\to +\infty$ the function $A$ satisfies $A''(x)+s_0^2 A(x)=0$, whence its asymptotic periodicity
   \[
    A(x)\stackrel{|x|\to +\infty}{=} \Big(c_1^{(\vartheta)}\cos s_0 x + c_2^{(\vartheta)}\sin s_0 x\Big)\big(1 + o(1)\big)
   \]
   for some constants $c_1^{(\vartheta)},c_2^{(\vartheta)}\in\mathbb{C}$ that vanish when $\vartheta\equiv 0$.
   \end{proof}

  \subsection{Symmetric, unbounded below, TMS extension}\label{sec:symmTMSubdd}~

  Let us introduce the subspace
  \begin{equation}\label{eq:Dtilde0}
   \widetilde{\mathcal{D}}_0\;:=\;\left\{
   \xi\in H^{-\frac{1}{2}}_{\ell=0}(\mathbb{R}^3)\left|
   \begin{array}{c}
    \textrm{$\xi$ has re-scaled radial component} \\
     \theta=\sin (s_0 |x|)*\big(\widehat{\vartheta}/\widehat{\gamma}_+\big)^{\!\vee} \\
     \textrm{for }\;\vartheta\in C^\infty_{0,\mathrm{odd}}(\mathbb{R}_x)
   \end{array}
   \!\right.\right\}\,.
  \end{equation}
  In \eqref{eq:Dtilde0} the subscript `odd' indicates functions with odd parity. The constant $s_0\approx 1.0062$ is the unique positive root of $\widehat{\gamma}(s)=0$ as defined in \eqref{eq:gamma-distribution}, and $\widehat{\gamma}_+$ is defined in \eqref{eq:gammaplus}. The correspondence between $\xi$ and its re-scaled radial component $\theta$ is given by \eqref{eq:0xi}-\eqref{ftheta-1}.

  In the above definition of $\widetilde{\mathcal{D}}_0$ it is tacitly understood that the re-scaled radial components are all taken with the same parameter $\lambda>0$ in the definition \eqref{ftheta-1}. This does not mean $\widetilde{\mathcal{D}}_0$ is a $\lambda$-dependent subspace, as one can easily convince oneself: the choice of $\lambda$ only fixes the convention for representing the element of $\widetilde{\mathcal{D}}_0$ in terms of the corresponding $\theta$.

  \begin{lemma}\label{lem:D0tildedomainproperties}~
  
  \begin{itemize}
  \item[(i)] One has
  \begin{equation}\label{eq:thetavanishesatzero}
   \xi\,\in\,\widetilde{\mathcal{D}}_0\qquad\Rightarrow\qquad\theta^{(\xi)}(0)\;=\;0\,.
  \end{equation}
   \item[(ii)] One has
   \begin{equation}\label{eq:thetaxiD0}
   \xi\,\in\,\widetilde{\mathcal{D}}_0\qquad\Rightarrow\qquad\widehat{\gamma}(s)\,\widehat{\theta^{(\xi)}}(s)\;=\;-2 s_0\,\widehat{\vartheta}(s)\;\;\forall s\in\mathbb{R}\,.
  \end{equation}
   \item[(iii)] $\widetilde{\mathcal{D}}_0$ is dense in $H^{\frac{1}{2}-\varepsilon}_{\ell=0}(\mathbb{R}^3)$ for every $\varepsilon>0$.
   \item[(iv)] $\widetilde{\mathcal{D}}_0 \cap H^{\frac{1}{2}}_{\ell=0}(\mathbb{R}^3)=\{0\}$.
      \item[(v)] For every $\lambda>0$ and $s\in\mathbb{R}$, $T_\lambda^{(0)}\widetilde{\mathcal{D}}_0\subset H^{s}_{\ell=0}(\mathbb{R}^3)$.
      \item[(vi)] For $\lambda>0$, $T_\lambda^{(0)}$ is injective on $\widetilde{\mathcal{D}}_0$, that is,
      \begin{equation}
       \big(\;\xi\in \widetilde{\mathcal{D}}_0\;\;\textrm{ and }\;\;T_\lambda^{(0)}\xi\equiv 0\;\big)\qquad\Rightarrow\qquad \xi\equiv 0\,.
      \end{equation}
   \item[(vii)] For every $\lambda>0$ one has 
   \begin{equation}
    \big\langle\xi, T_\lambda^{(0)}\eta\big\rangle_{L^2}\;=\;\big\langle T_\lambda^{(0)}\xi, \eta\big\rangle_{L^2}\qquad\forall \xi,\eta\in\widetilde{\mathcal{D}}_0\,.
   \end{equation}
  \end{itemize} 
  \end{lemma}

  \begin{proof}
   Part (i) follows from the fact that in $\theta(0)=\int_{\mathbb{R}}\sin(s_0|y|)\big(\widehat{\vartheta}/\widehat{\gamma}_+\big)^{\!\vee}(y)$ the term $\big(\widehat{\vartheta}/\widehat{\gamma}_+\big)^{\!\vee}$ has odd parity, whereas $\sin(s_0|y|)$ is even.

   Part (ii) follows at once from the definition \eqref{eq:Dtilde0}, by means of Lemma \ref{lem:0tms-generalsol} and \eqref{radialTMS0radial-FOURIER}.
  
   (iii) Let $\xi\in\widetilde{\mathcal{D}}_0$ and $\varepsilon>0$. By means of \eqref{eq:mellinnorms} and Lemma \ref{lem:0tms-generalsol} one gets
   \[
    \big\|\xi\big\|_{H^{\frac{1}{2}-\varepsilon}(\mathbb{R}^3)}^2\;\approx\;\int_{\mathbb{R}}\,\frac{\,|\theta(x)|^2}{(\cosh x)^{2\varepsilon}}\,\ud x\;\leqslant\;\|\theta\|_{L^\infty}\!\int_{\mathbb{R}}\,\frac{\ud x}{(\cosh x)^{2\varepsilon}}\;<\;+\infty\,.
   \]
   In fact, \eqref{eq:mellinnorms} also implies that the density of the $\xi$'s of $\widetilde{\mathcal{D}}_0$ in $H^{\frac{1}{2}-\varepsilon}_{\ell=0}(\mathbb{R}^3)$ is equivalent to the density of the associated $\theta$'s in $L^2_{\mathrm{odd}}(\mathbb{R},(\cosh x)^{-2\varepsilon}\ud x)$. If in the latter Hilbert space a function $\theta_0$ was orthogonal to all such $\theta$'s, then
   \[
    \begin{split}
     0\;&=\;\int_{\mathbb{R}}\frac{\,\overline{\theta_0(x)}}{(\cosh x)^{2\varepsilon}}\,\,\theta(x)\,\ud x\;=\;\int_{\mathbb{R}}\frac{\,\overline{\theta_0(x)}}{(\cosh x)^{2\varepsilon}}\,\,\Big(\sin( s_0| x|)*\big(\widehat{\vartheta}/\widehat{\gamma}_+\big)^{\!\vee}\Big)(x)\,\ud x \\
     &=\;\int_{\mathbb{R}}\Big(\sin (s_0 |x|)*\frac{\,\overline{\theta_0(x)}}{(\cosh x)^{2\varepsilon}}\Big)(x)\,\big(\widehat{\vartheta}/\widehat{\gamma}_+\big)^{\!\vee}(x)\,\ud x\qquad\forall \vartheta\in C^\infty_{0,\mathrm{odd}}(\mathbb{R}_x)\,.
    \end{split}
   \]
   (The above change in the integration order is allowed, via Fubini-Tonelli, thanks to the rapid decrease of the functions $(\cosh x)^{-2\varepsilon}$ and $(\widehat{\vartheta}/\widehat{\gamma}_+)^\vee$.) This would imply 
   \[
    \int_{\mathbb{R}}\sin(x_0|x-y|) h_0(h)\,\ud y\;=\;0\qquad\textrm{for a.e.~} x\in\mathbb{R}\,,\qquad h_0\;:=\;\frac{\theta_0}{(\cosh x)^{2\varepsilon}}\,,
   \]
  whence $h_0\equiv 0$ and then $\theta_0\equiv 0$.
   
   (iv) Let $\xi\in\widetilde{\mathcal{D}}_0\setminus\{0\}$. Since $\theta$ is not identically zero and is asymptotically $\cos$-periodic (Lemma \ref{lem:0tms-generalsol}), then \eqref{eq:mellinnorms} implies
   \[
    \big\|\xi\big\|_{H^{\frac{1}{2}}(\mathbb{R}^3)}^2\;\approx\;\int_{\mathbb{R}}|\theta(x)|^2\,\ud x\;=\;+\infty\,.
   \]

   (v) Let $\xi\in\widetilde{\mathcal{D}}_0$ and $\lambda>0$. Owing to formula \eqref{eq:Tlambda0xi-with-theta} of Lemma \ref{lem:xithetaidentities},
    \begin{equation*}
     \begin{split}
      \big\|T_\lambda^{(0)}\xi\big\|_{H^{s}(\mathbb{R}^3)}^2\;&\approx\;\int_{\mathbb{R}}\ud x (\cosh x)^{1+2s}\bigg|\,\theta(x)-\frac{4}{\pi\sqrt{3}}\int_{\mathbb{R}}\ud y\,\theta(y) \log \frac{\,2\cosh(x-y)+1\,}{\,2\cosh(x-y)-1\,}\bigg|^2,
     \end{split}
    \end{equation*}
   and owing to Lemma \ref{lem:0tms-generalsol},
   \begin{equation*}
    \theta(x)-\frac{4}{\pi\sqrt{3}}\int_{\mathbb{R}}\ud y\,\theta(y) \log \frac{\,2\cosh(x-y)+1\,}{\,2\cosh(x-y)-1\,}\;=\;2s_0\,\vartheta(x)\,.
    \end{equation*}
    Therefore,
    \[
     \big\|T_\lambda^{(0)}\xi\big\|_{H^{s}(\mathbb{R}^3)}^2\;\lesssim\;\int_{\mathbb{R}}(\cosh x)^{1+2s}|\vartheta(x)|^2\,\ud x\;<\;+\infty\,,
    \]
   because $\vartheta$ is smooth and has compact support.

   (vi) Based on what argued for \eqref{eq:TMSeq0}-\eqref{radialTMS0radial}, and on Lemma \ref{lem:0tms-generalsol}, it is straightforward to see that if $\xi\in\widetilde{\mathcal{D}}_0$, and hence $\theta^{(\xi)}=\sin (s_0 |x|)*\big(\widehat{\vartheta}/\widehat{\gamma}_+\big)^{\!\vee}$ for some $\vartheta\in C^\infty_{0,\mathrm{odd}}(\mathbb{R}_x)$, then the re-scaled radial function $\theta^{(\eta)}$ of the charge $\eta:=T_\lambda^{(0)}\xi$ is
   \[
    \theta^{(\eta)}\;=\;-4\pi^2s_0\sqrt{\lambda}\,(\cosh x)\,\vartheta\,.
   \]
   Therefore, if $T_\lambda^{(0)}\xi\equiv 0$, then $\theta^{(\eta)}\equiv 0$, thus also $\vartheta\equiv 0$, implying that $\xi\equiv 0$.
   
   (vii) Owing to \eqref{eq:xiTxi-with-theta},
   \begin{equation*}
     \int_{\mathbb{R}^3} \overline{\,\widehat{\xi}(\pp)}\, \big(\widehat{T_\lambda^{(0)}\eta}\big)(\pp)\,\ud\pp\,=\,\frac{\,8\pi^2}{3\sqrt{3}}\int_{\mathbb{R}}\ud s\, \widehat{\gamma}(s)\,\overline{\widehat{\theta^{(\xi)}}(s)}\,\widehat{\theta^{(\eta)}}(s)\,=\,\int_{\mathbb{R}^3} \overline{\big(\widehat{T_\lambda^{(0)}\xi}\big)(\pp)}\,\widehat{\eta}(\pp)\,\ud\pp
    \end{equation*}
   as an identity between \emph{finite} quantities. Indeed,
   \[
   \begin{split}
    \int_{\mathbb{R}}\ud s\, \widehat{\gamma}(s)\,\overline{\widehat{\theta^{(\xi)}}(s)}\,\widehat{\theta^{(\eta)}}(s)\;&=\;-2s_0\int_{\mathbb{R}}\ud x\,\overline{\,\vartheta^{(\xi)}(x)}\,\theta^{(\eta)}(x) \\
    &=\;-2s_0\int_{\mathbb{R}}\ud x\,\overline{\,\vartheta^{(\xi)}(x)}\,\Big(\sin (s_0 |x|)*\big(\widehat{\vartheta^{(\eta)}}/\widehat{\gamma}_+\big)^{\!\vee}\Big)(x)
   \end{split}
   \]
  (having applied part (ii) in the first identity and \eqref{eq:Dtilde0} in the second), whence
  \[
   \begin{split}
    \Big| \int_{\mathbb{R}}\ud s\, \widehat{\gamma}(s)\,\overline{\widehat{\theta^{(\xi)}}(s)}\,\widehat{\theta^{(\eta)}}(s)\Big|\;&\leqslant\;\big\|\vartheta^{(\xi)}\big\|_{L^1}\Big\|\sin (s_0 |x|)*\big(\widehat{\vartheta^{(\eta)}}/\widehat{\gamma}_+\big)^{\!\vee}\Big\|_{L^\infty} \\
   &\leqslant\;\big\|\vartheta^{(\xi)}\big\|_{L^1}\Big\|\Big(\frac{\widehat{\vartheta^{(\eta)}}}{\:\widehat{\gamma}_+}\Big)^{\!\vee}\Big\|_{L^1}\;<\;+\infty
   \end{split}
  \]
 (again by Young's inequality).
%
  \end{proof}

  \begin{remark}\label{rem:whensymmetric2}
    Lemma \ref{lem:D0tildedomainproperties} supplements the picture previously emerged from Lemma \ref{lem:Tlambdaproperties}(i) and (v), and discussed in Remark \ref{rem:whensymmetric}, concerning the validity of the identity
     \begin{equation*}
 \int_{\mathbb{R}^3} \overline{\,\widehat{\xi}(\pp)}\, \big(\widehat{T_\lambda\eta}\big)(\pp)\,\ud\pp\;=\;\int_{\mathbb{R}^3} \overline{\,\widehat{T_\lambda\xi}(\pp)}\, \widehat{\eta}(\pp)\,\ud\pp\,.
 \end{equation*}
    Indeed, such a symmetry property was previously established for $\xi,\eta\in H^{\frac{1}{2}}(\mathbb{R}^3)$, whereas Lemma \ref{lem:D0tildedomainproperties}(v) now guarantees it also for $\xi,\eta\in\widetilde{\mathcal{D}}_0$, namely a non-$H^{\frac{1}{2}}$ domain (as shown by Lemma \ref{lem:D0tildedomainproperties}(iii)).     
  \end{remark}

  \begin{remark}
   Along the same line of the previous Remark, Lemma \ref{lem:D0tildedomainproperties} also supplements the analysis of the problem, considered in Lemma \ref{lem:Tlambdaproperties}(ii)-(iii) and Remark \ref{rem:Tl-failstomap}, of finding charges $\xi\in H^{-\frac{1}{2}}_{\ell=0}(\mathbb{R}^3)$ satisfying $T^{(0)}_\lambda\xi\in H^{\frac{1}{2}}(\mathbb{R}^3)$, which is in turn crucial to construct $W_\lambda^{-1}T^{(0)}_\lambda$. As observed already, high regularity of $\xi(\yy)$ (hence fast decay of $\widehat{\xi}(\pp)$) is of no avail in the sector $\ell=0$. Now Lemma \ref{lem:D0tildedomainproperties} implies that the good feature for having $T_\lambda\xi\in H^{\frac{1}{2}}(\mathbb{R}^3)$ is a suitable \emph{oscillation} of $\widehat{\xi}(\pp)$ combined with some $|\pp|$-decay compatible with $H^{-\frac{1}{2}^-}$-regularity (which, loosely speaking, corresponds to some localisation of $\xi(\yy)$ close to $\yy=0$). This is seen by the fact that $\theta^{(\xi)}$ must satisfy \eqref{eq:thetaxiD0}, hence (Lemma \ref{lem:0tms-generalsol}) $\theta^{(\xi)}(x)$ is $\cos$-periodic in $x$, with consequent (non-periodic) oscillation in $\widehat{\xi}(\pp)$ via \eqref{eq:0xi} and \eqref{ftheta-2}.
  \end{remark}

  \begin{lemma}\label{lem:expectationunbounded}
   Let $\lambda>0$. Then
   \begin{equation}\label{eq:expectationunbounded}
    \inf_{\xi\in \widetilde{\mathcal{D}}_0}\frac{\;\langle\xi,T_\lambda^{(0)}\xi\rangle_{L^2}}{\;\;\|\xi\|^2_{H^{-\frac{1}{2}}}}\;=\;-\infty\,.
   \end{equation}   
  \end{lemma}

  \begin{proof}
   Let $\xi\in \widetilde{\mathcal{D}}_0$, and let $\theta$ be its re-scaled radial component \eqref{eq:0xi}-\eqref{ftheta-1}. 
   The numerator in \eqref{eq:expectationunbounded} is indeed finite and real (Lemma \ref{lem:D0tildedomainproperties}(vii), Remark \ref{rem:whensymmetric2}).
   For $\varepsilon>0$ let $\xi_\varepsilon$ be the element of $\widetilde{\mathcal{D}}_0$ whose re-scaled radial component $\theta_\varepsilon$ is defined by $\theta_\varepsilon(x):=\theta(\varepsilon x)$. Then $\widehat{\theta}_\varepsilon(s)=\varepsilon^{-1}\widehat{\theta}(s/\varepsilon)$.
   Therefore,
   \[
    \begin{split}
     \lim_{\varepsilon\downarrow 0}\,\big\langle\xi_\varepsilon,T_\lambda^{(0)}\xi_\varepsilon\big\rangle_{L^2}\;&=\;\frac{\,8\pi^2}{3\sqrt{3}}\,\lim_{\varepsilon\downarrow 0}\,\frac{1}{\,\varepsilon^2}\int_{\mathbb{R}}\ud s\, \widehat{\gamma}(s)\,|\widehat{\theta}(s/\varepsilon)|^2 \\
     &=\;\frac{\,8\pi^2}{3\sqrt{3}}\,\lim_{\varepsilon\downarrow 0}\,\frac{1}{\varepsilon}\int_{\mathbb{R}}\ud s\, \widehat{\gamma}(\varepsilon s)\,|\widehat{\theta}(s)|^2\;=\;-\infty\,,
    \end{split}
   \]
%
   having used \eqref{eq:xiTxi-with-theta} in the first identity, and the fact that $\int_{\mathbb{R}}\ud s|\widehat{\theta}(s)|^2=+\infty$ (which follows from Lemma \ref{lem:D0tildedomainproperties}(iv) and \eqref{eq:mellinnorms}) and $ \widehat{\gamma}(0)=1-\frac{4\pi}{3\sqrt{3}}<0$. On the other hand,
    \[
     \|\xi_\varepsilon\|^2_{H^{-\frac{1}{2}}}\;\approx\;\int_{\mathbb{R}}\frac{|\theta_\varepsilon(x)|^2}{(\cosh x)^2}\,\ud x\;=\;\int_{\mathbb{R}}\frac{|\theta(\varepsilon x)|^2}{(\cosh x)^2}\,\ud x\;\xrightarrow[]{\varepsilon\downarrow 0}\;0\,,
    \]
   owing to \eqref{eq:mellinnorms}, \eqref{eq:thetavanishesatzero}, and dominated convergence. As a consequence,
   \[
    \lim_{\varepsilon\downarrow 0}\frac{\;\langle\xi_\varepsilon,T_\lambda^{(0)}\xi_\varepsilon\rangle_{L^2}}{\;\;\|\xi_\varepsilon\|^2_{H^{-\frac{1}{2}}}}\;=\;-\infty\,,
   \] 
    which proves \eqref{eq:expectationunbounded}.  
  \end{proof}

  For $\lambda>0$, let us now consider the operator
    \begin{equation}\label{eq:Alambdatilde-0}
   \begin{split}
       \widetilde{\mathcal{A}_{\lambda}^{(0)}}\;&:=\;3W_\lambda^{-1}T_\lambda^{(0)} \\
       \mathcal{D}\big( \widetilde{\mathcal{A}_{\lambda}^{(0)}}\big)\;&:=\; \widetilde{\mathcal{D}}_0
   \end{split}
  \end{equation}
  with respect to the Hilbert space $H^{-\frac{1}{2}}_{W_\lambda,\ell=0}(\mathbb{R}^3)$.
 The definition \eqref{eq:Alambdatilde-0} is well-posed, owing to Lemma \ref{lem:D0tildedomainproperties}(v) and Lemma \ref{lem:Wlambdaproperties}(ii). The operator $\widetilde{\mathcal{A}_{\lambda}^{(0)}}$ is the counterpart, for the sector of angular momentum $\ell=0$, of the operator $\widetilde{\mathcal{A}_{\lambda}^{(\ell)}}$ defined in \eqref{eq:Alambdatildenot0}.

  \begin{lemma}\label{lem:symmetricAtilde0} With respect to the Hilbert space $H^{-\frac{1}{2}}_{W_\lambda,\ell=0}(\mathbb{R}^3)$, the operator $\widetilde{\mathcal{A}_{\lambda}^{(0)}}$ is densely defined, symmetric, and not semi-bounded.
  \end{lemma}

  \begin{proof}
   The density of $\widetilde{\mathcal{A}_{\lambda}^{(0)}}$ follows from Lemma \ref{lem:D0tildedomainproperties}(iii) and the canonical Hilbert space isomorphism  $H^{-\frac{1}{2}}_{W_\lambda,\ell=0}(\mathbb{R}^3)\cong H^{-\frac{1}{2}}(\mathbb{R}^3)$. 
      Symmetry follows from the finiteness and reality of
   \[
    \big\langle \xi,\widetilde{\mathcal{A}_{\lambda}^{(0)}}\xi\big\rangle_{H^{-\frac{1}{2}}_{W_\lambda}}\;=\;3\big\langle\xi,T_\lambda^{(0)}\xi\big\rangle_{L^2}\;=\;\frac{\,8\pi^2}{\sqrt{3}}\int_{\mathbb{R}}\ud s\, \widehat{\gamma}(s)\,\big|\widehat{\theta}^{(\xi)}(s)\big|^2\qquad\forall\xi\in\widetilde{\mathcal{D}}_0\,,
   \]
   having used \eqref{eq:Alambdatilde-0} and \eqref{eq:W-scalar-product} in the first identity, and \eqref{eq:xiTxi-with-theta} in the second. (The finiteness of the above quantities is argued precisely as done in the proof of Lemma \ref{lem:D0tildedomainproperties}(vii).)   
      The unboundedness of $\widetilde{\mathcal{A}_{\lambda}^{(0)}}$ from above is obvious, and from below it follows directly from Lemma \ref{lem:expectationunbounded}, using again the isomorphism  $H^{-\frac{1}{2}}_{W_\lambda,\ell=0}(\mathbb{R}^3)\cong H^{-\frac{1}{2}}(\mathbb{R}^3)$.   
  \end{proof}

  Analogously to the operator $\widetilde{\mathcal{A}_{\lambda}^{(\ell)}}$ from \eqref{eq:Alambdatildenot0} when $\ell\in\mathbb{N}$, $\widetilde{\mathcal{A}_{\lambda}^{(0)}}$ is a well-defined labelling (Birman) operator for a Ter-Martirosyan Skornyakov \emph{symmetric} extension of our original operator of interest $\mathring{H}$, with inverse scattering length $\alpha=0$, of course only as far as the $\ell=0$ sector is concerned. Indeed, based on Lemma \ref{lem:symmetricAtilde0}, the considerations of Remark \ref{rem:remonsym}(ii) apply. Moreover, as found for $\widetilde{\mathcal{A}_{\lambda}^{(\ell)}}$, we shall see that  $\widetilde{\mathcal{A}_{\lambda}^{(0)}}$ is not self-adjoint in $H^{-\frac{1}{2}}_{W_\lambda,\ell=0}(\mathbb{R}^3)$, and therefore it does not identify a Ter-Martirosyan Skornyakov \emph{self-adjoint} extension of $\mathring{H}$ (with $\ell=0$).

  \emph{Unlike} $\widetilde{\mathcal{A}_{\lambda}^{(\ell)}}$, however, $\widetilde{\mathcal{A}_{\lambda}^{(0)}}$ is not bounded from below, hence does not admit a Friedrichs extension: the identification of its self-adjoint extension(s), if any, requires a different approach than Proposition \ref{prop:Alambdaellnot0}, discussed in the following Subsection \ref{sec:multiselfadj-0}.

  Moreover, the Ter-Martirosyan Skornyakov \emph{symmetric} extension of $\mathring{H}$ identified by $\widetilde{\mathcal{A}_{\lambda}^{(0)}}$ for the sector $\ell=0$ is itself unbounded from below (and above) on the bosonic space $\cH_\mathrm{b}$. This follows from the computation made in Remark \ref{rem:TMSsym}, namely
  \[
   \big\langle g , (H_{\mathcal{A}_{\lambda}^{(0)}} +\lambda\mathbbm{1})g\big\rangle_{\cH_\mathrm{b}}\;=\;\big\langle \xi,\widetilde{\mathcal{A}_{\lambda}^{(0)}}\xi\big\rangle_{H^{-\frac{1}{2}}_{W_\lambda}}
  \]
  for functions $g=u_\xi^\lambda$ (with an innocent abuse of notation in writing $H_{\mathcal{A}_{\lambda}^{(0)}}$, as we are only referring here to the Hamiltonian on functions with $\xi$-charge in the sector $\ell=0$): when $\xi\in\widetilde{\mathcal{D}}_0$ the expectation 
     \[
 \frac{\big\langle g , H_{\mathcal{A}_{\lambda}^{(0)}}g\big\rangle_{\cH_\mathrm{b}}}{\|g\|_{\cH_\mathrm{b}}^2}\;=\; \frac{\big\langle \xi,\widetilde{\mathcal{A}_{\lambda}^{(0)}}\xi\big\rangle_{H^{-\frac{1}{2}}_{W_\lambda}}}{\;\|\xi\|^2_{H^{-\frac{1}{2}}_{W_\lambda}}}-\lambda
 \]
  can be made, at fixed $\lambda$, arbitrarily negative, owing to \eqref{eq:expectationunbounded}.

  \begin{remark}\label{rem:merepurpose}
   For the mere purpose of realising $W_\lambda^{-1}T_\lambda^{(0)}$ as a densely defined and symmetric operator in $H^{-\frac{1}{2}}_{W_\lambda,\ell=0}(\mathbb{R}^3)$, one could have selected the \emph{larger} domain 
   \begin{equation}\label{eq:Dolarger}
    \widetilde{\mathcal{D}}_0^\prime\;:=\;\widetilde{\mathcal{D}}_0\dotplus\mathrm{span}\{\xi_{s_0}\}\,,
   \end{equation}
  where $\xi_{s_0}$ has re-scaled radial component
  \begin{equation}
   \theta_{s_0}(x)\;:=\;\sin s_0 x
  \end{equation}
  for some tacitly declared $\lambda>0$. All conclusions of Lemma \ref{lem:D0tildedomainproperties} remain straightforwardly valid for $\widetilde{\mathcal{D}}_0^\prime$ but for the injectivity of $T_\lambda^{(0)}$ on $\widetilde{\mathcal{D}}_0^\prime$, for 
  \begin{equation}
   T_\lambda^{(0)}\xi_{s_0}\;=\;0\,,
  \end{equation}
 as follows from (the proof of) Lemma \ref{lem:0tms-generalsol}. The conclusions of Lemma \ref{lem:expectationunbounded} apply to $\widetilde{\mathcal{D}}_0^\prime$ as well, and therefore one has a version of Lemma \ref{lem:symmetricAtilde0} also for the operator
       \begin{equation}\label{eq:Alambdatilde-0prime}
   \begin{split}
       \widetilde{\mathcal{B}_{\lambda}^{(0)}}\;&:=\;3W_\lambda^{-1}T_\lambda^{(0)} \\
       \mathcal{D}\big( \widetilde{\mathcal{B}_{\lambda}^{(0)}}\big)\;&:=\; \widetilde{\mathcal{D}}_0^\prime\,.
   \end{split}
  \end{equation}
  \emph{The issue with such $\widetilde{\mathcal{B}_{\lambda}^{(0)}}$}, as compared to $\widetilde{\mathcal{A}_{\lambda}^{(0)}}$, is that it is symmetric on a domain that is \emph{too large} for the problem of finding self-adjoint operators of Ter-Martirosyan Skornyakov type: the operator $\mathring{H}_{\widetilde{\mathcal{A}_{\lambda}^{(0)}}}$ (in the sector $\ell=0$) admits self-adjoint TMS extensions on $L^2_{\mathrm{b}}(\mathbb{R}^3\times\mathbb{R}^3\,\ud\yy_1,\ud\yy_2)$, the operator $\mathring{H}_{\widetilde{\mathcal{B}_{\lambda}^{(0)}}}$ does not. We shall complete this discussion in Remark \ref{rem:choicedomains}.  
  \end{remark}

  \subsection{Adjoint of the Birman parameter}\label{sec:adjointBirman}~

  Here we characterise the operator $ \big(\widetilde{\mathcal{A}_{\lambda}^{(0)}}\big)^\star$.  We recall that `$\star$' indicates the adjoint with respect to $H^{-\frac{1}{2}}_{W_\lambda,\ell=0}(\mathbb{R}^3)$, whereas `$*$' is reserved for the adjoint in $L^2(\mathbb{R}^3)$.

  \begin{lemma}\label{lem:deficiency1-1}
   Let $\lambda>0$. The densely defined and symmetric operator $\widetilde{\mathcal{A}_{\lambda}^{(0)}}$ defined in \eqref{eq:Alambdatilde-0} on the Hilbert space $H^{-\frac{1}{2}}_{W_\lambda,\ell=0}(\mathbb{R}^3)$ has deficiency indices $(1,1)$. Thus, for every $\mu>0$ there exist two non-zero functions $\xi_{\ii\mu},\xi_{-\ii\mu}\in H^{-\frac{1}{2}}_{W_\lambda,\ell=0}(\mathbb{R}^3)$ such that
   \begin{equation}\label{eq:twodefsub}
  \begin{split}
   \ker\Big(\big(\widetilde{\mathcal{A}_{\lambda}^{(0)}}\big)^\star-\ii\mu\mathbbm{1}\Big)\;&=\;\mathrm{span}\{\xi_{\ii\mu}\} \\
   \ker\Big(\big(\widetilde{\mathcal{A}_{\lambda}^{(0)}}\big)^\star+\ii\mu\mathbbm{1}\Big)\;&=\;\mathrm{span}\{\xi_{-\ii\mu}\}\,.
  \end{split}
  \end{equation}
  Moreover, $\big(\widetilde{\mathcal{A}_{\lambda}^{(0)}}\big)^\star$ acts on such eigenvectors as
  \begin{equation}\label{eq:A0tildestareigenv}
   \big(\widetilde{\mathcal{A}_{\lambda}^{(0)}}\big)^\star\xi_{\pm \ii \mu}\;=\;3 W_\lambda^{-1}T_\lambda^{(0)}\,	\xi_{\pm \ii \mu}\;=\;\pm \ii\mu\,\xi_{\pm \ii \mu}\,.
  \end{equation}
  \end{lemma}

  \begin{corollary}\label{cor:A0tildeadj}
   Under the above assumptions,
   \begin{equation}\label{eq:A0tildestar}
   \begin{split}
    \mathcal{D}\big(\big(\widetilde{\mathcal{A}_{\lambda}^{(0)}})^\star\big)\;&=\;\widetilde{\mathcal{D}}_{0}' \dotplus\mathrm{span}\{\xi_{\ii\mu}\}\dotplus \mathrm{span}\{\xi_{-\ii\mu}\} \\
     \big(\widetilde{\mathcal{A}_{\lambda}^{(0)}}\big)^\star\;&=\;3 W_\lambda^{-1}T_\lambda^{(0)}\,,
   \end{split}
  \end{equation}
  where $\widetilde{\mathcal{D}}_{0}'$ is the domain of the operator closure of $ \widetilde{\mathcal{A}_{\lambda}^{(0)}}$ with respect to $H^{-\frac{1}{2}}_{W_\lambda,\ell=0}(\mathbb{R}^3)$, that is, $\widetilde{\mathcal{D}}_{0}'$ is the closure of $\widetilde{\mathcal{D}}_{0}$ in the graph norm of  $\widetilde{\mathcal{A}_{\lambda}^{(0)}}$.
  \end{corollary}

  In practice we do not need an explicit characterisation of $\widetilde{\mathcal{D}}_{0}'$. This will be clear in due time from the usage of Lemma \ref{lem:adjointzero} in the proof of Theorem \ref{thm:spectralanalysis}.
   
    We also determine convenient asymptotics for the elements of $\big(\widetilde{\mathcal{A}_{\lambda}^{(0)}}\big)^\star$.

  \begin{lemma}\label{lem:adjointasymtotics}
   Let $\lambda>0$ and $z\in\mathbb{C}\setminus\mathbb{R}$ (in practice $z=\ii\mu$ for $\mu\in\mathbb{R}\setminus\{0\}$, so as to cover both deficiency subspaces \eqref{eq:twodefsub}). Let $\xi\in\ker\big(\big(\widetilde{\mathcal{A}_{\lambda}^{(0)}}\big)^\star-z\mathbbm{1}\big)$.
   Then, 
   \begin{equation}\label{eq:adjointasymtotics}
   \begin{split}
    \widehat{\xi}(\pp)\;&=\;A_\xi\frac{\,\sin({s_0\log|\pp|})+W_{\lambda,\xi}(z)\,\cos({s_0\log|\pp|})\,}{\pp^2}(1+o(1)) \\
    &\qquad\textrm{as }|\pp|\to+\infty\,,\;|z|/\lambda\to 0
   \end{split}
   \end{equation}
    for two constants $A_\xi,W_{\lambda,\xi}(z)\in\mathbb{C}$ with
    \begin{equation}\label{eq:Wconstproperty}
     W_{\lambda,\xi}(\overline{z})\;=\;\overline{W_{\lambda,\xi}(z)}\,.
    \end{equation}
 Here $s_0\approx 1.0062$ is the unique positive root of $\widehat{\gamma}(s)=0$ as defined in \eqref{eq:gamma-distribution}.
   \end{lemma}

  \begin{proof}[Proof of Lemma \ref{lem:deficiency1-1}] Let $\mu>0$. Let $\xi_+\in H^{-\frac{1}{2}}_{W_\lambda,\ell=0}(\mathbb{R}^3)$ with
  \[
   \xi_+\;\in\;\ker\Big(\big(\widetilde{\mathcal{A}_{\lambda}^{(0)}}\big)^\star-\ii\mu\mathbbm{1}\Big)\;=\;\mathrm{ran}\big(\widetilde{\mathcal{A}_{\lambda}^{(0)}})+\ii\mu\mathbbm{1}\big)^{\perp_\lambda}\,,
  \]
  and let $\theta_+$ be the re-scaled radial component of $\xi_+$. Then
  \[\tag{$\star$}\label{eq:tagstar}
  \begin{split}
   0\;&=\;\big\langle\xi_+,\big(\widetilde{\mathcal{A}_{\lambda}^{(0)}})+\ii\mu\mathbbm{1}\big)\xi\big\rangle_{H^{-\frac{1}{2}}_{W_\lambda}} \\
   &=\;3\big\langle \xi_+,T_\lambda^{(0)}\xi\big\rangle_{H^{-\frac{1}{2}},H^{\frac{1}{2}}}+\ii\mu\big\langle \xi_+,W_\lambda^{(0)}\xi\big\rangle_{H^{-\frac{1}{2}},H^{\frac{1}{2}}}\qquad\forall\xi\in\widetilde{\mathcal{D}}_0\,.
  \end{split}
  \]
  Here we applied \eqref{eq:W-scalar-product} as usual, together with \eqref{eq:Alambdatilde-0} and Lemmas \ref{lem:Wlambdaproperties}(ii) and  \ref{lem:D0tildedomainproperties}(v).

   The two duality products appearing in the r.h.s.~above have been computed in (the proof of) Lemma \ref{lem:xithetaidentities} and in Lemma \ref{lem:xiWxi-zero}:
   \begin{equation*}
   \begin{split}
    \big\langle \xi_+,&T_\lambda^{(0)}\xi\big\rangle_{H^{-\frac{1}{2}},H^{\frac{1}{2}}}\;=\;\frac{\;8\pi^2}{\,3\sqrt{3}}\bigg(\int_{\mathbb{R}}\overline{\theta_+(x)}\,\theta(x)\,\ud x\\
   &\quad -\frac{4}{\pi\sqrt{3}}\iint_{\mathbb{R}\times\mathbb{R}}\overline{\theta_+(x)}\,\theta(y)\log \frac{\,2\cosh(x-y)+1\,}{\,2\cosh(x-y)-1\,}\,\ud x\,\ud y\bigg) \\
   \big\langle \xi_+,&W_\lambda^{(0)}\xi\big\rangle_{H^{-\frac{1}{2}},H^{\frac{1}{2}}}\;=\;\frac{4\pi^2}{\,\lambda\sqrt{3}\,}\int_{\mathbb{R}}\frac{\,\overline{\theta_+}(x)\,\theta(x)}{\,(\cosh x)^2}\,\ud x \\
   &\quad +\frac{32\pi}{\lambda}\iint_{\mathbb{R}\times\mathbb{R}}\frac{\overline{\theta_+(x)}\,\theta(y)}{\,(2\cosh(x+y)+1)\,(2\cosh(x-y)-1)\,}\,\ud x\,\ud y\,.
    \end{split}
  \end{equation*}
  Here $\theta$ denotes the re-scaled radial component of $\xi\in\widetilde{\mathcal{D}}_0$.

  We can also re-write the above quantities after taking the Fourier transform, by means of \eqref{eq:xiTxi-with-theta} and \eqref{eq:xiWxi-zero-FOURIER}:
  \[
   \begin{split}
     \big\langle \xi_+,T_\lambda^{(0)}\xi\big\rangle_{H^{-\frac{1}{2}},H^{\frac{1}{2}}}\;&=\;\frac{\;8\pi^2}{\,3\sqrt{3}}\int_{\mathbb{R}}\widehat{\gamma}(s)\,\overline{\widehat{\theta}_+(s)}\,\widehat{\theta}(s)\,\ud s\, \\
     \big\langle \xi_+,W_\lambda^{(0)}\xi\big\rangle_{H^{-\frac{1}{2}},H^{\frac{1}{2}}}\;&=\;\frac{2\pi^2}{\,\lambda\sqrt{3}\,}\int_{\mathbb{R}}\Big(\frac{s}{\,\sinh\frac{\pi}{2}s}*\overline{\widehat{\theta}_+}\Big)(s)\,\widehat{\theta}(s)\,\ud s \\
     &\!\!\!\!\!\!\!\!\!\!\!\!\!\!\!\!\!\!+\frac{16\pi}{3\lambda}\iint_{\mathbb{R}\times\mathbb{R}}\overline{\widehat{\theta}_+(t)}\,\widehat{\theta}(s)\,\frac{\sinh\frac{\pi}{6}(s+t)}{\sinh \frac{\pi}{2}(s+t)}\,\frac{\sinh\frac{\pi}{3}(s-t)}{\sinh \frac{\pi}{2}(s-t)}\,\ud s\,\ud t\,.
   \end{split}
  \]

  Adding them up into the eigenvector equation \eqref{eq:tagstar}, and using the density $\widetilde{\mathcal{D}}_0$ together with Fubini-Tonelli theorem, leads to the following equation in the unknown $\theta_+$:
  \begin{equation}\label{eq:eigenequation-thetaplus}
   \begin{split}
    & \bigg(\theta_+(x)-\frac{4}{\pi\sqrt{3}}\int_{\mathbb{R}}\theta_+(y)\,\log \frac{\,2\cosh(x-y)+1\,}{\,2\cosh(x-y)-1\,}\,\ud y\bigg) \\
    & =\;\frac{ \ii \mu}{ 2 \lambda }\bigg(\frac{\theta_+(x)}{\:(\cosh x)^2}+\frac{8\sqrt{3}}{\pi}\int_{\mathbb{R}}\frac{\theta_+(y)}{\,(2\cosh(x+y)+1)\,(2\cosh(x-y)-1)\,}\,\ud y\bigg)
   \end{split}
  \end{equation}
  (for a.e.~$x\in\mathbb{R}$), or, equivalently, 
  \begin{equation}\label{eq:eigenequation-thetaplusF}
   \begin{split}
      \widehat{\gamma}(s)\,\widehat{\theta}_+(s)\;&= \;\frac{\,\ii\mu}{4\lambda}\Big(\Big(\frac{s}{\,\sinh\frac{\pi}{2}s}*\widehat{\theta}_+\Big)(s) \\
      &\qquad\qquad +\frac{8}{\,\pi\sqrt{3}}\int_{\mathbb{R}}\frac{\sinh\frac{\pi}{6}(s+t)}{\sinh \frac{\pi}{2}(s+t)}\,\frac{\sinh\frac{\pi}{3}(s-t)}{\sinh \frac{\pi}{2}(s-t)}\,\widehat{\theta}_+(t)\,\ud t\Big)\,.
   \end{split}
  \end{equation}

  The integration order's exchange in the double integrals was possible thanks to the fast decay of the integral kernels. This also demonstrates, unfolding \eqref{eq:eigenequation-thetaplusF} backwards, that $\xi_+$ satisfies
  \[
   3T_\lambda^{(0)}\xi_+\;=\;\ii \mu W_\lambda^{(0)}\xi_+\,,
  \]
  therefore $T_\lambda^{(0)}\xi_+\in\mathrm{ran}W_\lambda^{(0)}=H^{\frac{1}{2}}_{\ell=0}(\mathbb{R}^3)$ (Lemma \ref{lem:Wlambdaproperties}(ii)) and 
  \[
   3W_\lambda^{-1}T_\lambda^{(0)}\xi_+\;=\;\ii \mu\xi_+\;=\;\big(\widetilde{\mathcal{A}_{\lambda}^{(0)}}\big)^\star\xi_+\,.
  \]
  Thus, \emph{$\big(\widetilde{\mathcal{A}_{\lambda}^{(0)}}\big)^\star$ acts on the eigenvector $\xi_+$ precisely as $3W_\lambda^{-1}T_\lambda^{(0)}$}.

  Now, on the one hand it is well known that the dimension of the deficiency subspace considered at the beginning of this proof is independent of $\mu>0$, and therefore the dimension of the space of solutions to \eqref{eq:eigenequation-thetaplus} (equivalently, \eqref{eq:eigenequation-thetaplusF}) does not depend on $\mu>0$. (In fact, owing to von Neumann's conjugation criterion \cite[Theorem X.3]{rs2}, such dimension is the same even when one takes instead $\mu<0$, namely when one considers the other deficiency subspace.) On the other hand, mirroring the reasoning of Lemma \ref{lem:0tms-generalsol}'s proof, \eqref{eq:eigenequation-thetaplusF} can be re-written as
  \[
    \widehat{\theta}_+(s)\;=\;\widehat{\theta}_0(s)+ \ii\,\frac{\mu}{\lambda}\,\mathcal{L}\widehat{\theta}_+(s)
   \]
   with
   \[
    \widehat{\theta}_0(s)\;:=\;c\big(\delta(s-s_0)-\delta(s+s_0)\big)\,,\qquad c\in\mathbb{C}
   \]
   and 
   \[
    \begin{split}
     \mathcal{L}\widehat{\theta}_+(s)\;:=\;&\frac{1}{\,8 s_0\widehat{\gamma}_+(s)}\Big(\big({\textstyle\frac{s}{\,\sinh\frac{\pi}{2}s}}*\widehat{\theta}_+\big)(s)+{\textstyle\frac{8}{\,\pi\sqrt{3}}}\int_{\mathbb{R}}{\textstyle\frac{\sinh\frac{\pi}{6}(s+t)}{\sinh \frac{\pi}{2}(s+t)}\,\frac{\sinh\frac{\pi}{3}(s-t)}{\sinh \frac{\pi}{2}(s-t)}}\,\widehat{\theta}_+(t)\,\ud t\Big)\;\times \\
     &\;\times\Big(PV\frac{1}{s-s_0}-PV\frac{1}{s+s_0}\Big)
    \end{split}
   \]
   ($s_0\approx 1.0062$ being the unique positive root of $\widehat{\gamma}(s)=0$).
   
  From this we see that the dimension of the space of solutions to \eqref{eq:eigenequation-thetaplusF} is precisely equal to one, namely it is dictated by the dimensionality of the solution space for the pivot equation $ \widehat{\gamma}\,\widehat{\theta}_+=0$. The previous formulas provide also an iterative expansion of the form
  \[
    \widehat{\theta}_+\;=\;\sum_{k=0}^{N-1}\Big(\frac{\ii\mu}{\lambda}\Big)^k\mathcal{L}^k\widehat{\theta}_0+\Big(\frac{\ii\mu}{\lambda}\Big)^N\mathcal{L}\widehat{\theta}_+\,,\qquad N\in\mathbb{N}
  \]
  in the regime $\mu/\lambda\ll 1$, which does not alter the dimension of the space of solutions.

  In conclusion,
  \[
   \dim \ker\Big(\big(\widetilde{\mathcal{A}_{\lambda}^{(0)}}\big)^\star-\ii\mu\mathbbm{1}\Big)\;=\;\dim \ker\Big(\big(\widetilde{\mathcal{A}_{\lambda}^{(0)}}\big)^\star+\ii\mu\mathbbm{1}\Big)\;=\;1\,,
  \]
  the analysis of the second deficiency subspace being clearly the same as above, upon exchanging $\mu$ with $-\mu$. This proves that $\widetilde{\mathcal{A}_{\lambda}^{(0)}}$ has deficiency indices $(1,1)$.  
  \end{proof}

  \begin{proof}[Proof of Corollary \ref{cor:A0tildeadj}]
   Formula \eqref{eq:A0tildestar} is an immediate consequence of Lemma \ref{lem:deficiency1-1} and \eqref{eq:Alambdatilde-0}, as an application of a standard formula by von Neumann (see, e.g., \cite[Lemma on page 138]{rs2}. For sure $\big(\widetilde{\mathcal{A}_{\lambda}^{(0)}})^\star$ does act on $\widetilde{\mathcal{D}}_{0}'$ as $3 W_\lambda^{-1}T_\lambda^{(0)}$ because this is the action of $\big(\widetilde{\mathcal{A}_{\lambda}^{(0)}})^\star$ both on  $\widetilde{\mathcal{D}}_{0}$ (owing to the definition \eqref{eq:Alambdatilde-0}) and on the deficiency subspaces (owing to \eqref{eq:A0tildestareigenv}).  
  \end{proof}

  \begin{proof}[Proof of Lemma \ref{lem:adjointasymtotics}]
   Let $\theta$ be the re-scaled radial function associated with $\xi$ through \eqref{eq:0xi}-\eqref{ftheta-1}. 
   
   Continuing the discussion from Lemma \ref{lem:deficiency1-1}'s proof (where the present $\theta$ was denoted by $\theta_+$), $\theta$ is the unique solution, up to complex multiplicative constant, to
   \[\tag{i}\label{eq:asymptag2}
    \widehat{\theta}(s)\;=\;\widehat{\theta}_0(s)+ \frac{z}{\lambda}\,\mathcal{L}\widehat{\theta}(s)\,.
   \]
   In the iteration
     \[
    \widehat{\theta}\;=\;\sum_{k=0}^{N-1}\Big(\frac{z}{\lambda}\Big)^k\mathcal{L}^k\widehat{\theta}_0+\Big(\frac{z}{\lambda}\Big)^N\mathcal{L}\widehat{\theta}\,,\qquad N\in\mathbb{N}\,,
  \]
  and in the considered regime $|z|/\lambda\ll 1$, the leading expression for the solution is
  \[\tag{ii}\label{eq:asymptag3new}
   \widehat{\theta}\;=\;\widehat{\theta}_0+ \frac{z}{\lambda}\,\mathcal{L}\widehat{\theta}_0\,,
  \]
  up to $O((z/\lambda)^2)$-corrections as $|z|/\lambda\to 0$.
  
%

  Let us work out \eqref{eq:asymptag3new} choosing explicitly $c=\frac{\sqrt{2\pi}}{2\ii}$ for the convenience of having
  \[
   \widehat{\theta}_0(s)\;=\;\frac{\sqrt{2\pi}}{2\ii}(\delta(s-s_0)-\delta(s+s_0))\qquad\textrm{ and hence }\qquad \theta_0(x)=\sin s_0 x
  \]
  (see \eqref{eq:distri-deltaFouriersin}). This will fix $\theta$, and hence $\xi$, up to a complex multiplicative constant. The computation for $\mathcal{L}\widehat{\theta}_0$ then gives
  \[
   \mathcal{L}\widehat{\theta}_0\;=\;-\widehat{\Lambda}(s)\Big(PV\frac{1}{s-s_0}-PV\frac{1}{s+s_0}\Big)
  \]
  with
  \[
   \begin{split}
    \widehat{\Lambda}(s)\;:=\;&-\frac{1}{\,8 s_0\widehat{\gamma}_+(s)}\Big(\big({\textstyle\frac{s}{\,\sinh\frac{\pi}{2}s}}*\widehat{\theta}_0\big)(s)+{\textstyle\frac{8}{\,\pi\sqrt{3}}}\int_{\mathbb{R}}{\textstyle\frac{\sinh\frac{\pi}{6}(s+t)}{\sinh \frac{\pi}{2}(s+t)}\,\frac{\sinh\frac{\pi}{3}(s-t)}{\sinh \frac{\pi}{2}(s-t)}}\,\widehat{\theta}_0(t)\,\ud t\Big) \\
    =\;&\frac{\ii\sqrt{\pi}}{\,8s_0\widehat{\gamma}_+(s)\sqrt{2}\,}\Big(\frac{s-s_0}{\,\sinh\frac{\pi}{2}(s-s_0)}-\frac{s+s_0}{\,\sinh\frac{\pi}{2}(s+s_0)} \\
    &\qquad\qquad\qquad\qquad-\frac{32}{\pi\sqrt{3}}\,\frac{\sinh\frac{\pi}{6}s_0\,\sinh\frac{\pi}{6}s}{\,\,(1+2\cosh\frac{\pi}{3}(s-s_0))\,(1+2\cosh\frac{\pi}{3}(s+s_0))}\,\Big).
   \end{split}
  \]
  We observe that $\widehat{\Lambda}$ is a smooth and rapidly decreasing function that is purely imaginary and has odd parity (Figure \ref{fig:Lambdafunction}). Then $\Lambda$ is also smooth and rapidly decreasing, and is real-valued and with odd parity.  Thus \eqref{eq:asymptag3new} takes the form
 \[\tag{iii}\label{eq:asymptag3}
   \widehat{\theta}\;=\; \widehat{\theta}_0-\frac{z}{\lambda}\,\widehat{\Lambda}(s)\,\Big(PV\frac{1}{s-s_0}-PV\frac{1}{s+s_0}\Big)\,.
 \]
  
  \begin{figure}[t!]
\includegraphics[width=8cm]{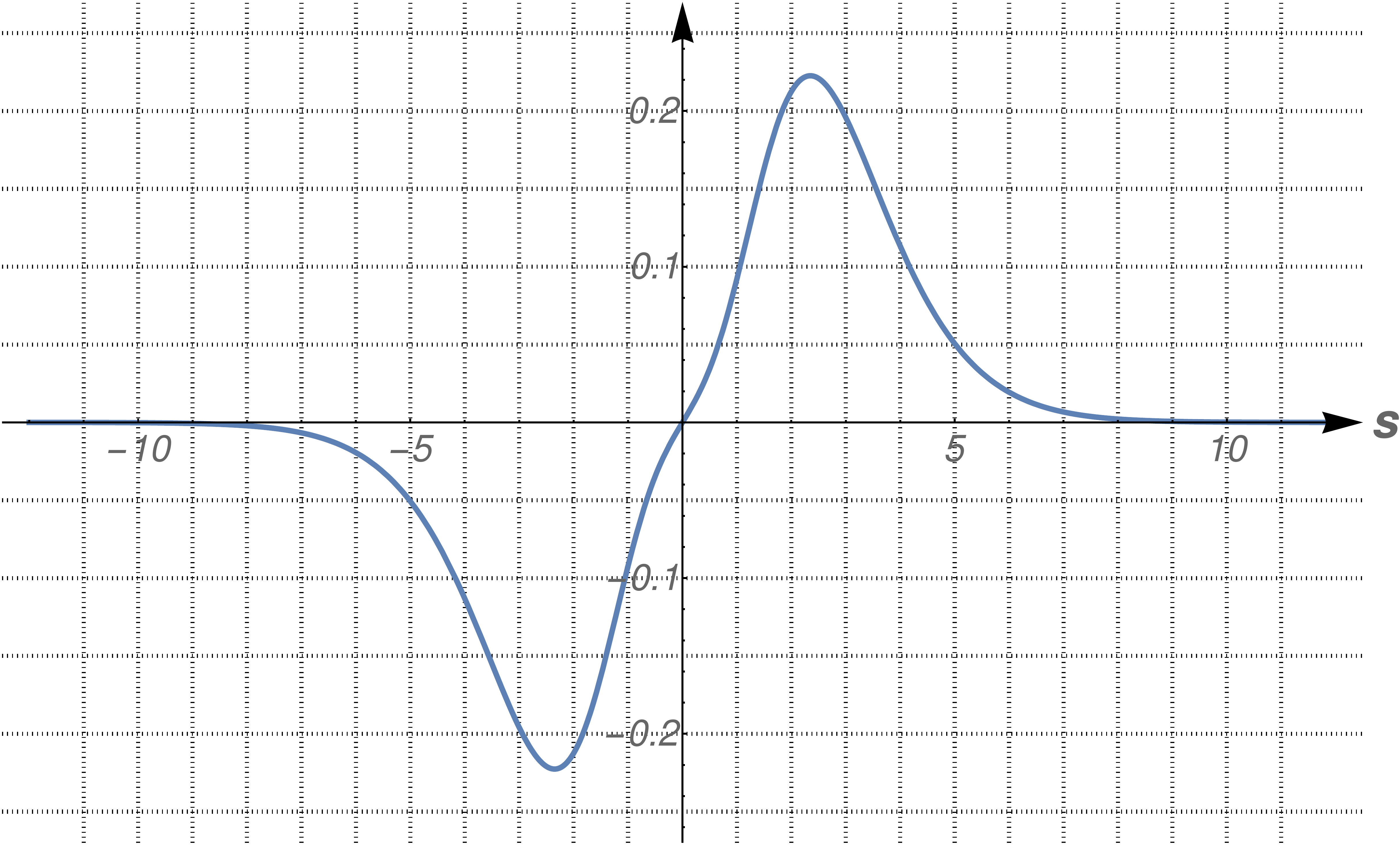}
\caption{Plot of the imaginary part of the function $\widehat{\Lambda}(s)$.}\label{fig:Lambdafunction}
\end{figure}

 The inverse Fourier transform of \eqref{eq:asymptag3} and \eqref{eq:distri-deltaFouriersin}-\eqref{eq:distriPVsignsin} then yield 
 \[\tag{iv}\label{eq:asymptag4}
  \theta(x)\;=\;\sin s_0 x+\frac{z}{\lambda}\,\big(\sin (s_0|\cdot|) *\Lambda\big)(x)\,.
 \]
 For such expression we can repeat the very same reasoning of Lemma \ref{lem:0tms-generalsol}'s proof and deduce that asymptotically as $|x|\to +\infty$ (and $|z|/\lambda\to 0$) 
 \[\tag{v}\label{eq:asymptag5}
  \theta(x)\;=\;\big(\sin s_0 x+\frac{z}{\lambda}\,a_\xi\cos( s_0 x+\sigma_\xi)\big)(1+o(1))
 \]
 for some constants $a_\xi\in\mathbb{R}$ and $\sigma_\xi\in[0,2\pi)$ depending on $\xi$. In particular $a_\xi$ is surely real, because \eqref{eq:asymptag4} expresses a real-valued function.

 With the asymptotics \eqref{eq:asymptag5} for the re-scaled radial function $\theta$, we reconstruct the asymptotics for $\xi$ by means of \eqref{eq:0xi}-\eqref{ftheta-1}. Up to $o(1)$-corrections as $|x|\to +\infty$ (hence $|\pp|\to +\infty)$,
 \[
  \begin{split}
   \sin s_0 x\;&=\;\sin \Big(s_0\log\Big({\textstyle\sqrt{\frac{3\pp^2}{4\lambda}}+\sqrt{\frac{3\pp^2}{4\lambda}+1}}\Big)\Big) \\
   &\approx\;\sin \Big(s_0\log{|\pp|\textstyle\sqrt{\frac{3}{\lambda}}}\Big)\\
   &=\;\sin\big(s_0\log{\textstyle\sqrt{\frac{3}{\lambda}}}\big)\,\cos(s_0\log|\pp|)+\cos\big(s_0\log{\textstyle\sqrt{\frac{3}{\lambda}}}\big)\,\sin(s_0\log|\pp|)
  \end{split}
 \]
 and 
\[
  \begin{split}
   \cos (s_0 x+\sigma_\xi)\;&=\;\cos \Big(s_0\log\Big({\textstyle\sqrt{\frac{3\pp^2}{4\lambda}}+\sqrt{\frac{3\pp^2}{4\lambda}+1}}\Big)+\sigma_\xi\Big) \\
   &\approx\;\cos \Big(s_0\log{|\pp|\textstyle\sqrt{\frac{3}{\lambda}}}\Big) \\
   &=\;\cos\big(s_0\log{\textstyle\sqrt{\frac{3}{\lambda}}}\big)\,\cos(s_0\log|\pp|)-\sin\big(s_0\log{\textstyle\sqrt{\frac{3}{\lambda}}}\big)\,\sin(s_0\log|\pp|)\,.
  \end{split}
 \]
 Plugging the latter identities into \eqref{eq:asymptag5} yields
 \[
  \begin{split}
   \theta(x)\;&=\;\Big(\cos\big(s_0 \log{\textstyle\sqrt{\frac{3}{\lambda}}}\big)- \frac{z}{\lambda} \,a_\xi\,\sin\big(s_0 \log{\textstyle\sqrt{\frac{3}{\lambda}}}\big)\Big)\sin (s_0\log|\pp|) \\
   &\qquad +\Big(\sin\big(s_0 \log{\textstyle\sqrt{\frac{3}{\lambda}}}\big)+ \frac{z}{\lambda} \,a_\xi\,\cos\big(s_0 \log{\textstyle\sqrt{\frac{3}{\lambda}}}\big)\Big)\cos (s_0\log|\pp|)
  \end{split}
 \]
 up to $o(1)$-corrections as $|\pp|\to +\infty$ and $|z|/\lambda\to 0$. Inserting this into \eqref{eq:0xi}-\eqref{ftheta-1}, that is,
  \[
  \widehat{\xi}(\pp)\;=\;\frac{\,\theta\Big(\log\Big(\sqrt{\frac{3\pp^2}{4\lambda}}+\sqrt{\frac{3\pp^2}{4\lambda}+1}\Big)\Big)}{\,\sqrt{3\pi}\,|\pp|\sqrt{\frac{3}{4}\pp^2+\lambda}}\,,
 \]
 yields finally \eqref{eq:adjointasymtotics} with
    \begin{equation*}
    W_{\lambda,\xi}(z)\;:=\;\frac{\,\sin\big(s_0 \log\sqrt{\frac{3}{\lambda}}\big)+\frac{z}{\lambda} \,a_\xi\,\cos\big(s_0 \log\sqrt{\frac{3}{\lambda}}\big)\,}{\,\cos\big(s_0 \log\sqrt{\frac{3}{\lambda}}\big)- \frac{z}{\lambda} \,a_\xi\,\sin\big(s_0 \log\sqrt{\frac{3}{\lambda}}\big)\,}\,.
   \end{equation*}
 In \eqref{eq:adjointasymtotics} we also re-instated the overall multiplicative constant.

 The above expression for $W_{\lambda,\xi}(z)$ shows that $\overline{W_{\lambda,\xi}(z)}=W_{\lambda,\xi}(\overline{z})$, thanks to the fact that $a_\xi\in\mathbb{R}$. Had one expressed the asymptotic periodicity \eqref{eq:asymptag5} above in terms of the sinus function, which amounts in practice to re-define the shift $\sigma_\xi$, one would have come up with an analogous expression for $W_{\lambda,\xi}(z)$ with the same symmetry property \eqref{eq:Wconstproperty}. The proof is thus completed. 
  \end{proof}

  Based on the discussion of this Subsection, it is convenient to introduce the following nomenclature for the charges belonging to the domain of $\big(\widetilde{\mathcal{A}_{\lambda}^{(0)}}\big)^\star$. Fixed $\mu>0$ and representing $\xi\in\mathcal{D}\Big(\widetilde{\mathcal{A}_{\lambda}^{(0)}}\Big)^\star$ through \eqref{eq:A0tildestar} as 
  \[
   \xi\;=\;\widetilde{\xi}+ c_+\xi_{\ii\mu}+c_-\xi_{-\ii\mu}
  \]
  for some $\widetilde{\xi}\in\widetilde{\mathcal{D}}_{0}'$ and $c_\pm\in\mathbb{C}$, we shall refer to
  \begin{equation}\label{eq:xiregxising}
   \xi_{\mathrm{reg}}\;:=\;\widetilde{\xi}\,,\qquad\xi_{\mathrm{sing}}\;:=\;c_+\xi_{\ii\mu}+c_-\xi_{-\ii\mu}
  \end{equation}
 as, respectively, the \emph{regular} and the \emph{singular} component of $\xi$. Regular and singular parts are unambiguously defined because the sum in \eqref{eq:A0tildestar} is direct. 

 To visualize the actual difference in behaviour, take for concreteness $\xi_{\mathrm{reg}}\in\widetilde{\mathcal{D}}_{0}$: then its re-scaled radial function $\theta_{\mathrm{reg}}$ satisfies
 \[
  \widehat{\theta}_{\mathrm{reg}}(s)\;=\;-\frac{\widehat{\vartheta}(s)}{\widehat{\gamma}_+(s)}\,\Big(PV\frac{1}{s-s_0}-PV\frac{1}{s+s_0}\Big)
 \]
 for some $\vartheta\in C^\infty_{0,\mathrm{odd}}(\mathbb{R}_x)$. On the contrary, $\xi_{\mathrm{sing}}$ has the behaviour of $\xi_{\pm\ii\mu}$ and in the course of the proof of Lemma \ref{lem:adjointasymtotics} we showed that the associated re-scaled radial functions $\theta_{\pm\ii\mu}$ satisfy
 \[
   \widehat{\theta}_{\pm\ii\mu}(s)\;\approx\; \frac{\sqrt{2\pi}\,}{2\ii}\big(\delta(s-s_0)-\delta(s+s_0)\big)\mp\frac{\ii\mu}{\lambda}\,\widehat{\Lambda}(s)\Big(PV\frac{1}{s-s_0}-PV\frac{1}{s+s_0}\Big)
 \]
 up to an overall multiplicative constant and up to $O((z/\lambda)^2)$-corrections as $|z|/\lambda\to 0$.

    \subsection{Multiplicity of TMS self-adjoint realisations}\label{sec:multiselfadj-0}~

    The operator $\widetilde{\mathcal{A}_{\lambda}^{(0)}}$ defined in \eqref{eq:Alambdatilde-0} admits a one-real-parameter family of self-adjoint extensions with respect to the Hilbert space $H^{-\frac{1}{2}}_{W_\lambda,\ell=0}(\mathbb{R}^3)$. They are qualified as follows.

    \begin{proposition}\label{prop:Aellzero-selfadj}
     Let $\lambda>0$. The self-adjoint extensions of $\widetilde{\mathcal{A}_{\lambda}^{(0)}}$ on $H^{-\frac{1}{2}}_{W_\lambda,\ell=0}(\mathbb{R}^3)$ form the family
     \begin{equation}
      \big\{\mathcal{A}_{\lambda,\beta}^{(0)}\,\big|\,\beta\in\mathbb{R}\big\}
     \end{equation}
    with
    \begin{equation}\label{eq:Azerolambda}
    \begin{split}
     \mathcal{D}\big(\mathcal{A}_{\lambda,\beta}^{(0)}\big)\;&:=\;\mathcal{D}_{0,\beta} \\
     \mathcal{A}_{\lambda,\beta}^{(0)}\;&:=\;3 W_\lambda^{-1} T_\lambda^{(0)}\,,
   \end{split}
   \end{equation}
   where
   \begin{equation}\label{eq:domainD0beta}
    \mathcal{D}_{0,\beta}\;=\;\left\{
    \begin{array}{c}
     \xi\in\mathcal{D}\Big(\widetilde{\mathcal{A}_{\lambda}^{(0)}}\Big)^{\!\star} \;\;\textrm{with singular part satisfying} \\
     \widehat{\xi}_{\mathrm{sing}}(\pp)\,=\,c\,\displaystyle\frac{\,\cos({s_0\log|\pp|})+\beta \sin({s_0\log|\pp|})\,}{\pp^2}\,(1+o(1)) \\
     \textrm{as $|\pp|\to +\infty$\quad for some $c\in\mathbb{C}$}
    \end{array}
    \right\}.
   \end{equation}     
    \end{proposition}

     Each of the $\mathcal{A}_{\lambda,\beta}^{(0)}$ is a legitimate TMS parameter in the sector of zero angular momentum, according to the discussion of Sect.~\ref{sec:TMSextension-section}, precisely as the operator $\mathcal{A}_{\lambda}^{(\ell)}$ defined in \eqref{AFop-ellnot0} is a TMS parameter in the sector $\ell\in\mathbb{N}$.

    \begin{proof}[Proof of Proposition \ref{prop:Aellzero-selfadj}]
     Let $\mu>0$.
     By a standard application of von Neumann's extension theory to the operator $\widetilde{\mathcal{A}_{\lambda}^{(0)}}$ with deficiency subspaces \eqref{eq:twodefsub} and adjoint \eqref{eq:A0tildestareigenv}-\eqref{eq:A0tildestar}, the self-adjoint extensions of $\widetilde{\mathcal{A}_{\lambda}^{(0)}}$ form the family 
      \[
       \big\{\mathcal{A}_{\lambda,U_\nu}^{(0)}\,\big|\,\nu\in[0,2\pi)\big\}\,,
      \]
 where
   \[
    \begin{split}
     \mathcal{D}\big(\mathcal{A}_{\lambda,U_\nu}^{(0)}\big)\;&:=\;\left\{\xi\:=\:\widetilde{\xi}+c(\xi_{\ii\mu}+e^{\ii \nu}\,\xi_{-\ii\mu})\left|
     \begin{array}{c}
      \widetilde{\xi}\in\widetilde{\mathcal{D}}_0' \\
      c\in\mathbb{C}
     \end{array}\!\!
     \right.\right\} \\
     \mathcal{A}_{\lambda,U_\nu}^{(0)}\,\xi\;&:=\;3 W_\lambda^{-1} T_\lambda^{(0)}\xi\;=\;3 W_\lambda^{-1} T_\lambda^{(0)}\widetilde{\xi}+c\big(\ii\mu\xi_{\ii\mu}-\ii\mu e^{\ii\nu}\xi_{-\ii\mu}\big)\,.
    \end{split}
   \]
   We have tacitly and non-restrictively assumed that the functions $\xi_{\pm\ii\mu}$ are normalised in $H^{-\frac{1}{2}}_{W_\lambda,\ell=0}(\mathbb{R}^3)$. The notation $U_\nu$ is to remind that the  map $\xi_{\ii\mu}\mapsto e^{\ii\nu}\xi_{-\ii\mu}$ induces a unitary isomorphism $U_\nu$ between the two deficiency subspaces \eqref{eq:twodefsub}.

   Let us now characterise the $\xi$'s in $ \mathcal{D}\big(\mathcal{A}_{\lambda,U_\nu}^{(0)}\big)$ in terms of the large-$|\pp|$ asymptotics of the corresponding $\widehat{\xi}_{\mathrm{sing}}=c(\widehat{\xi}_{\ii\mu}+e^{\ii \nu}\,\widehat{\xi}_{-\ii\mu})$ (see definition \eqref{eq:xiregxising}). Owing to Lemma \eqref{lem:adjointasymtotics}, at the leading order as $\mu/\lambda\to 0$ and $|\pp|\to +\infty$, and up to an overall multiplicative constant, one has
   \[
   \begin{split}
    |\pp|^{-2}\,\widehat{\xi}_{\mathrm{sing}}(\pp)\;&=\;\sin({s_0\log|\pp|})+w_{\xi,\lambda,\mu}\cos({s_0\log|\pp|}) \\
    &\qquad +e^{\ii\nu}\big(\sin({s_0\log|\pp|})+\overline{w_{\xi,\lambda,\mu}}\,\cos({s_0\log|\pp|})\big)\,,
   \end{split}
   \]
   having set
   \[
    w_{\xi,\lambda,\mu}\;:=\;W_{\lambda,\xi}(\ii\mu)
   \]
 from formula \eqref{eq:adjointasymtotics} and having used the property
  \[
   W_{\lambda,\xi}(-\ii\mu)\;=\;\overline{W_{\lambda,\xi}(\ii\mu)}\;=\;\overline{w_{\xi,\lambda,\mu}}
  \]
 from \eqref{eq:Wconstproperty}. Thus, within such approximation, and suitably re-defining the overall multiplicative constant,
  \[
   \begin{split}
    |\pp|^{-2}\,\widehat{\xi}_{\mathrm{sing}}(\pp)\;&=\;\cos({s_0\log|\pp|})+\frac{1+e^{\ii\nu}}{\,w_{\xi,\lambda,\mu}+e^{\ii\nu}\overline{w_{\xi,\lambda,\mu}}\,}\,\sin({s_0\log|\pp|})\,.
   \end{split}
  \]
 At fixed $\nu$, the above (asymptotic) condition selects all charges from $\mathcal{D}\Big(\widetilde{\mathcal{A}_{\lambda}^{(0)}}\Big)^\star$ that constitute the domain of the `$\nu$-th extension'.
 
 As
 \[
  \beta\;:=\;\frac{1+e^{\ii\nu}}{\,w_{\xi,\lambda,\mu}+e^{\ii\nu}\overline{w_{\xi,\lambda,\mu}}\,}\;\in\mathbb{R}\,,
 \]
 one can switch from $\nu$-parametrisation to the $\beta$-parametrisation, re-defining
 \[
  \mathcal{A}_{\lambda,\beta}^{(0)}\;:=\;\mathcal{A}_{\lambda,U_\nu}^{(0)}\,.
 \]
 This leads to the final thesis.    
    \end{proof}

   \begin{remark}
    It is worth underlying that the overall construction so far has involved \emph{two distinct extension schemes}: the Kre{\u\i}n-Vi\v{s}ik-Birman scheme for the self-adjoint extensions of the operator $\mathring{H}$, classified in Theorem \ref{thm:generalclassification}, and von Neumann's scheme for the self-adjoint extensions of the operator $\widetilde{\mathcal{A}_{\lambda}^{(0)}}$, classified in Proposition \ref{prop:Aellzero-selfadj}. (As $\widetilde{\mathcal{A}_{\lambda}^{(0)}}$ is not semi-bounded, the Kre{\u\i}n-Vi\v{s}ik-Birman scheme is not applicable.) In either case one speaks of singular component of a generic element of the adjoint as that component belonging to the deficiency subspaces, and in either case each extension corresponds to a suitable restriction of the domain of the adjoint. However, in the Kre{\u\i}n-Vi\v{s}ik-Birman scheme such a restriction selects a subspace of the adjoint's domain by means of a constraint between the singular and the regular component of its elements (as commented at the beginning of Sect.~\ref{sec:two-body-short-scale-sing}), whereas in von Neumann's scheme the restriction of self-adjointness is a constraint within the singular components only (see \eqref{eq:xiregxising} and \eqref{eq:domainD0beta} above).    
   \end{remark}

  \section{The canonical model and other well-posed variants}\label{sec:canonicalmodel}

  Merging the findings of Sect.~\ref{sec:higherell} and \ref{sec:lzero} within the general scheme of Sect.~\ref{sec:TMSextension-section} we can finally present a class of models for the bosonic trimer 
  which are mathematically well-posed (i.e., self-adjoint) and physically meaningful (i.e., of Ter-Martirosyan Skornyakov type), and which in a sense are canonical, as we shall comment further.

  \subsection{Canonical model at unitarity and at given three-body parameter}\label{sec:constructioncanonical}~

  Let $\beta\in\mathbb{R}$ and $\lambda>0$.
  With respect to the decomposition \eqref{eq:bigdecompW}, namely
  \begin{equation}\label{eq:bigdecompW2}
 \begin{split}
 H^{-\frac{1}{2}}_{W_\lambda}(\mathbb{R}^3)\;\cong\;\bigoplus_{\ell=0}^\infty \,H^{-\frac{1}{2}}_{W_\lambda,\ell}(\mathbb{R}^3) \,,
 \end{split}
\end{equation}
  let us consider the operator
  \begin{equation}\label{eq:globalAlambda-1}
   \mathcal{A}_{\lambda,\beta}\;:=\; \mathcal{A}_{\lambda,\beta}^{(0)}\:\oplus\: \bigoplus_{\ell=1}^\infty\mathcal{A}_\lambda^{(\ell)}
  \end{equation}
 in the usual sense of direct sum of operators on an orthogonal direct sum of Hilbert spaces. $\mathcal{A}_\lambda^{(\ell)}$, with $\ell\in\mathbb{N}$, is defined in \eqref{AFop-ellnot0}, taking here $\alpha=0$, and $\mathcal{A}_{\lambda,\beta}^{(0)}$ is defined in \eqref{eq:Azerolambda}-\eqref{eq:domainD0beta}. Observe that the condition $\lambda>\lambda_\alpha$ required in the definition \eqref{AFop-ellnot0} is automatically satisfied here, as $\alpha=0$.

  The self-adjointness of each summand in \eqref{eq:globalAlambda-1} with respect to the corresponding Hilbert space $H^{-\frac{1}{2}}_{W_\lambda,\ell}(\mathbb{R}^3)$ is proved, respectively, in Propositions \ref{prop:Alambdaellnot0} and \ref{prop:Aellzero-selfadj}. Therefore, altogether $\mathcal{A}_{\lambda,\beta}$ is self-adjoint on $H^{-\frac{1}{2}}_{W_\lambda}(\mathbb{R}^3)$.

  Upon setting
  \begin{equation}\label{eq:finalDbeta}
   \mathcal{D}_\beta\;:=\;\mathcal{D}_{0,\beta}\;\boxplus\;\op_{k=1}^\infty\mathcal{D}_\ell\,,
  \end{equation}
  the definition \eqref{eq:globalAlambda-1} is equivalent to
  \begin{equation}\label{eq:globalAlambda-2} 
   \begin{split}
    \mathcal{D}(\mathcal{A}_{\lambda,\beta})\;&:=\;\mathcal{D}_\beta \\
    \mathcal{A}_{\lambda,\beta}\;&:=\;3 W_\lambda^{-1}T_\lambda\,.
   \end{split}
  \end{equation}
  Here we use `$\boxplus$' instead of `$\oplus$' to indicate that the sum is orthogonal with respect to the Hilbert space orthogonal direct sum \eqref{eq:bigdecompW2}, but the summands are non-closed subspaces of $H^{-\frac{1}{2}}_{W_\lambda}(\mathbb{R}^3)$. Actually \eqref{eq:finalDbeta} is nothing but the explicit expression for the domain of the direct sum operator \eqref{eq:globalAlambda-1} (with respect to the decomposition \eqref{eq:bigdecompW2}): the domain $\mathcal{D}_{0,\beta}$ of $\mathcal{A}_{\lambda,\beta}^{(0)}$ is defined in \eqref{eq:domainD0beta}, and the $\mathcal{D}_\ell$ of $\mathcal{A}_{\lambda}^{(\ell)}$ is defined in \eqref{eq:domainDell}. Observe that $\mathcal{D}_\beta$ is $\lambda$-independent, because so are its $\ell$-components. The second line of \eqref{eq:globalAlambda-2} is due to the fact that each of the summands in \eqref{eq:globalAlambda-1} is an operator acting on the corresponding $\ell$-sector as $W_\lambda^{(\ell)}T_\lambda^{(\ell)}$ (Propositions \ref{prop:Alambdaellnot0} and \ref{prop:Aellzero-selfadj}), and in turn $T_\lambda$ and $W_\lambda$ are reduced with respect to the decomposition \eqref{eq:bigdecompW2} with component, respectively, $T_\lambda^{(\ell)}$ and $W_\lambda^{(\ell)}$(as seen in \eqref{eq:decompTTell} and \eqref{eq:WlambdaWlambdaell}).

  It is instructive to re-cap what $\mathcal{D}_\beta$ altogether is:
  \begin{equation}\label{eq:Dbetaaltogether}
   \mathcal{D}_\beta\;=\;\left\{ 
   \begin{array}{c}
    \displaystyle\xi=\sum_{\ell=0}^\infty\xi^{(\ell)}\in\;\bigoplus_{\ell=0}^\infty \,H^{-\frac{1}{2}}_{W_\lambda,\ell}(\mathbb{R}^3)\,\cong\, H^{-\frac{1}{2}}_{W_\lambda}(\mathbb{R}^3) \\
    \textrm{such that} \\
    \xi^{(\ell)}\in H^{\frac{1}{2}}_\ell(\mathbb{R}^3)\;\textrm{ and }\; T_\lambda^{(\ell)}\xi^{(\ell)}\in H^{\frac{1}{2}}_\ell(\mathbb{R}^3)\;\textrm{ for }\;\ell\in\mathbb{N}\,, \\
    \xi^{(0)}\in\mathcal{D}\Big(\widetilde{\mathcal{A}_{\lambda}^{(0)}}\Big)^{\!\star} \;\;\textrm{with singular part satisfying} \\
    \widehat{\xi}_{\mathrm{sing}}(\pp)\,=\,c\,\displaystyle\frac{\,\cos({s_0\log|\pp|})+\beta \sin({s_0\log|\pp|})\,}{\pp^2}\,(1+o(1)) \\
     \textrm{as $|\pp|\to +\infty$\quad for some $c\in\mathbb{C}$}
   \end{array}
   \right\},
  \end{equation}
   where the subspace $\mathcal{D}\Big(\widetilde{\mathcal{A}_{\lambda}^{(0)}}\Big)^{\!\star}\subset H^{-\frac{1}{2}}_{W_\lambda,\ell=0}(\mathbb{R}^3)$ is defined in \eqref{eq:A0tildestar}. Moreover, following from the analogous properties of each $\ell$-component,
   \begin{equation}\label{eq:Dbetaproperties}
    \textrm{$\mathcal{D}_\beta$ is dense in $H^{-\frac{1}{2}}_{W_\lambda}(\mathbb{R}^3)$}\qquad\textrm{and}\qquad T_\lambda\mathcal{D}_\beta\;\subset\;H^{\frac{1}{2}}(\mathbb{R}^3)\,.
   \end{equation}

   Being self-adjoint on the deficiency subspace of $\mathring{H}+\lambda\mathbbm{1}$ (more precisely, on a unitarily equivalent version of it), $\mathcal{A}_{\lambda,\beta}$ identifies a self-adjoint extension $\mathring{H}_{\mathcal{A}_{\lambda,\beta}}$ of $\mathring{H}$ in the sense of the general classification of Theorem \ref{thm:generalclassification}.

   In turn, since $\mathcal{A}_{\lambda,\beta}$ acts as $3 W_\lambda^{-1}T_\lambda$, its domain $\mathcal{D}_\beta$ satisfies \eqref{eq:Dbetaproperties}, according to Theorem \ref{thm:globalTMSext} the operator $\mathring{H}_{\mathcal{A}_{\lambda,\beta}}$ is a Ter-Martirosyan Skornyakov extension of $\mathring{H}$, namely a physical extension.

   Such extension can be defined as follows.

   \begin{theorem}\label{thm:H0beta}
    Let $\beta\in\mathbb{R}$ and $\lambda>0$. Define
    \begin{equation}\label{eq:H0betadomaction}
     \begin{split}
      \mathcal{D}(\mathscr{H}_{0,\beta})\;&:=\;\left\{g=\phi^\lambda+u_\xi^\lambda\left|\!
  \begin{array}{c}
   \phi^\lambda\in H^2_\mathrm{b}(\mathbb{R}^3\times\mathbb{R}^3)\,,\;\xi\in\mathcal{D}_\beta\,, \\
   \displaystyle\phi^\lambda(\yy_1,\mathbf{0})\,=\,(2\pi)^{-\frac{3}{2}} (T_\lambda\xi)(\yy_1)
  \end{array}
  \!\!\!\right.\right\} \\
  (\mathscr{H}_{0,\beta}+\lambda\mathbbm{1})g\;&:=\;(-\Delta_{\yy_1}-\Delta_{\yy_2}-\nabla_{\yy_1}\cdot\nabla_{\yy_2}+\lambda\mathbbm{1})\phi^\lambda\,,
     \end{split}
    \end{equation}
    where the subspace $\mathcal{D}_\beta\subset H^{-\frac{1}{2}}(\mathbb{R}^3)$ is given by \eqref{eq:Dbetaaltogether}.
  \begin{itemize}
   \item[(i)] The decomposition of $g$ in terms of $\phi^\lambda$ and $\xi$ is unique, at fixed $\lambda$. The subspace $ \mathcal{D}(\mathscr{H}_{0,\beta})$ is $\lambda$-independent.
   \item[(ii)] $\mathscr{H}_{0,\beta}$ is self-adjoint on $L^2_\mathrm{b}(\mathbb{R}^3\times\mathbb{R}^3,\ud\yy_1,\ud\yy_2)$ and extends $\mathring{H}$ given in \eqref{eq:domHring-initial}.
   \item[(iii)] For each $g\in  \mathcal{D}(\mathscr{H}_{0,\beta})$ one has
   \begin{equation}\label{eq:allBPTMS}
    \begin{split}
     \phi^\lambda(\yy_1,\mathbf{0})\;&=\;(2\pi)^{-\frac{3}{2}} (T_\lambda\xi)(\yy_1)\,, \\
     \int_{\mathbb{R}^3}\widehat{\phi^\lambda}(\pp_1,\pp_2)\,\ud\pp_2\;&=\;(\widehat{T_\lambda\xi})(\pp_1)\,, \\
     \int_{\!\substack{ \\ \\ \pp_2\in\mathbb{R}^3 \\ |\pp_2|<R}}\widehat{g}(\pp_1,\pp_2)\,\ud\pp_2\;&=\;4\pi R\,\widehat{\xi}(\pp_1)+o(1)\qquad\textrm{as }R\to +\infty\,.
    \end{split}
   \end{equation}
  All such conditions are equivalent, and each of them expresses the Bethe-Peierls alias Ter-Martirosyan Skornyakov condition. In particular, the first version of \eqref{eq:allBPTMS} is an identity in $H^{\frac{1}{2}}(\mathbb{R}^3)$. 
  \item[(iv)] $\mathscr{H}_{0,\beta}$ is not semi-bounded.
  \end{itemize}
  \end{theorem}

   \begin{proof}
    As argued already, the operator $\mathring{H}_{\mathcal{A}_{\lambda,\beta}}$ matches the conditions of Theorem \ref{thm:globalTMSext}(ii) for the considered $\lambda$, therefore it is a Ter-Martirosyan Skornyakov self-adjoint extension of $\mathring{H}$ with inverse scattering length $\alpha=0$. 
    Renaming $\mathring{H}_{\mathcal{A}_{\lambda,\beta}}\equiv\mathscr{H}_{0,\beta}$, Theorem \ref{thm:globalTMSext} guarantees that such $\mathscr{H}_{0,\beta}$ is $\lambda$-independent (only the explicit decomposition of its domain's elements $g$ depends on $\lambda$), with
       \begin{equation*}
  \mathcal{D}(\mathscr{H}_{0,\beta})\;=\;
  \left\{g=\phi^\lambda+u_\xi^\lambda\left|\!
  \begin{array}{c}
   \phi^\lambda\in H^2_\mathrm{b}(\mathbb{R}^3\times\mathbb{R}^3)\,,\;\xi\in\mathcal{D}_\beta\,, \\
   \displaystyle\int_{\mathbb{R}^3}\widehat{\phi^\lambda}(\pp_1,\pp_2)\,\ud\pp_2\,=\,(\widehat{T_\lambda\xi})(\pp_1)
  \end{array}
  \!\!\!\right.\right\}.
  \end{equation*}
  The various BP/TMS conditions for $\mathscr{H}_{0,\beta}$ and their equivalence are then guaranteed by Lemma \ref{eq:oneTMSfunction}(iii), and the unboundedness from below (and above) of $\mathscr{H}_{0,\beta}$ follows from the fact that $\mathscr{H}_{0,\beta}$ extends, in the $\ell=0$ sector, a symmetric operator that is not semi-bounded (see Lemma \ref{lem:symmetricAtilde0} and the observations right after).  
   \end{proof}

   \begin{remark}
    Owing to the bosonic symmetry, if $g\in  \mathcal{D}(\mathscr{H}_{0,\beta})$, then \eqref{eq:allBPTMS} has equivalent versions in the other variables, e.g.,
    \begin{equation}
     \phi^\lambda(\yy,\mathbf{0})\;=\; \phi^\lambda(\mathbf{0},\yy)\;=\;\phi^\lambda(\yy,\yy)\;=\;(2\pi)^{-\frac{3}{2}} (T_\lambda\xi)(\yy)
    \end{equation}
    (see \eqref{eq:Hbosonic}).
   \end{remark}

   \begin{remark}
    In the sectors of definite angular momentum $\ell\in\mathbb{N}$ the Birman operator $\mathcal{A}_{\lambda,\beta}$ labelling the Hamiltonian $\mathring{H}_{\mathcal{A}_{\lambda,\beta}}\equiv\mathscr{H}_{0,\beta}$ is strictly positive, and correspondingly $\mathscr{H}_{0,\beta}$ is lower semi-bounded (Theorem \ref{thm:generalclassification}(ii)). In this case we can express the quadratic form of $\mathscr{H}_{0,\beta}$ according to Theorem \ref{thm:generalclassification}(iii). Explicitly, combining \eqref{eq:HFform}, \eqref{eq:decomposition_of_form_domains_Tversion}, and \eqref{AFform-ellnot0}, we find that for all $g$'s of the form
    \begin{equation}\label{eq:g-for-the-form}
     g\;=\;\phi^\lambda+u_\xi^\lambda\,,\qquad \phi^\lambda\in H^1_{\mathrm{b}}(\mathbb{R}^3\times\mathbb{R}^3)\,,\quad\xi\in H^{\frac{1}{2}}_\ell(\mathbb{R}^3)
    \end{equation}
    for some $\lambda>0$ and some $\ell\in\mathbb{N}$ (thus, excluding $\ell=0$) the evaluation of the quadratic form of $\mathscr{H}_{0,\beta}$ gives
    \begin{equation}\label{eq:formhighsectors}
    \begin{split}
     \mathscr{H}_{0,\beta}[g]\;&=\;\frac{1}{2}\Big(\big\|(\nabla_{\yy_1}+\nabla_{\yy_2})\phi^\lambda\big\|^2_{L^2}+\big\|\nabla_{\yy_1}\phi^\lambda\big\|^2_{L^2}+\big\|\nabla_{\yy_2}\phi^\lambda\big\|^2_{L^2}\Big) \\
     &\qquad+\lambda\Big(\|\phi^\lambda\big\|^2_{L^2}-\big\|\phi^\lambda+u_\xi^\lambda\big\|^2_{L^2}\Big)+3\int_{\mathbb{R}^3} \overline{\,\widehat{\xi}(\pp)}\, \big(\widehat{T_\lambda\xi}\big)(\pp)\,\ud\pp
    \end{split}
    \end{equation}
  (the $L^2$-norms being norms in $L^2(\mathbb{R}^3\times\mathbb{R}^3)$). Of course, the above expression is the same for all $\beta$'s, since the parameter $\beta$ only qualifies the properties of the Hamiltonian $\mathscr{H}_{0,\beta}$ in the sector $\ell=0$. On the $g$'s of \eqref{eq:g-for-the-form} one then has $\mathscr{H}_{0,\beta}[g]\geqslant 0$, and by self-adjointness the form \eqref{eq:formhighsectors} is closed. Through a quadratic form analysis, the form \eqref{eq:formhighsectors} was proposed and proved to be closed and semi-bounded in the recent work \cite{Basti-Teta-2015}. 
   \end{remark}

  The double index in $\mathscr{H}_{0,\beta}$ is to indicate that \emph{two parameters} have been selected in order to identify the operator within the general class of self-adjoint extensions of $\mathring{H}$, namely the parameter $\alpha=0$ in the Ter-Martirosyan Skornyakov condition, and the parameter $\beta\in\mathbb{R}$ in the choice of the charge domain $\mathcal{D}_\beta$. (It is surely of interest to repeat the same analysis for the analogous extensions $\mathscr{H}_{\alpha,\beta}$: as said, from this perspective we only focus here on the unitarity regime $\alpha=0$, which is the physically relevant one.)

  Explicitly, $\alpha=0$ and $\beta$ select the following prescriptions:
  \begin{eqnarray}
  &  &\int_{\!\substack{ \\ \\ \pp_2\in\mathbb{R}^3 \\ |\pp_2|<R}}\widehat{g}(\pp_1,\pp_2)\,\ud\pp_2\stackrel{R\to +\infty}{=} 4\pi R+\,\widehat{\xi}(\pp_1)+o(1)   \qquad\qquad\quad\;\; (\textrm{TMS}_{\alpha=0}) \label{eq:TMSalphazero}\\
  & & \widehat{\xi}^{(0)}_{\mathrm{sing}}(\pp)\stackrel{|\pp|\to +\infty}{=}c\,\displaystyle\frac{\,\cos({s_0\log|\pp|})+\beta \sin({s_0\log|\pp|})\,}{\pp^2}\,(1+o(1)) \quad (\textrm{III}_\beta) \label{eq:IIIbeta}\,.
  \end{eqnarray}
  As discussed in Subsect.~\ref{sec:symmTMSubdd}-\ref{sec:adjointBirman} and Proposition \ref{prop:Aellzero-selfadj}, the TMS condition alone, indicated here with the shorthand $\textrm{TMS}_{\alpha=0}$, is \emph{not} enough to qualify the self-adjointness of the model: an additional $\beta$-driven condition is needed, present only for charges in the $\ell=0$ sector.

  The shorthand $\textrm{III}_\beta$ in \eqref{eq:IIIbeta} is meant to express the following difference. \eqref{eq:TMSalphazero} is a \emph{two-body} condition, constraining the trimer's wave-function when two of the bosons come on top of each other, which is explicitly seen from the first version of \eqref{eq:allBPTMS} or also from its consequence
   \begin{equation*}
 g_{\mathrm{av}}(\yy_1;|\yy_2|)\,\stackrel{|\yy_2|\to 0}{\sim}\frac{1}{|\yy_2|}\,\xi(\yy_1) + o(1)  
 \end{equation*}
  (see \eqref{eq:g-TMS-BP-generic} above). Instead, \eqref{eq:IIIbeta} is interpreted as a \emph{three-body} condition, regulating the behaviour of the trimer's wave-function in the vicinity of the triple coincidence configuration. Some mathematical heuristics on such three-body interpretation is presented in \cite[Remark 2.8]{CDFMT-2015} and \cite[Sect.~8]{MO-2017}, and above all one can see from the spectral analysis that follows (Subsect.~\ref{sec:spectralThomas}) chat such $\beta$ has precisely the role of the three-body parameter introduced by the physicists, on which we commented in the introduction.

  From this perspective, each $\mathscr{H}_{0,\beta}$ is a \emph{canonical} Hamiltonian for the bosonic trimer with zero-range interaction: it is defined by a canonical choice, namely the Friedrichs extension of the TMS parameter, in all sectors $\ell\neq 0$, and by a $\beta$-extension of the TMS parameter in the sector $\ell=0$. In retrospect, also in the latter sector the construction was canonical, in that the choice of the initial domain of symmetry $\widetilde{\mathcal{D}}_0$ (formula \eqref{eq:Dtilde0}) is the natural one guaranteeing the well-posedness condition $T_\lambda^{(0)}\widetilde{\mathcal{D}}_0\subset H^{\frac{1}{2}}_{\ell=0}(\mathbb{R}^3)$ (Lemmas \ref{lem:0tms-generalsol} and \ref{lem:D0tildedomainproperties}(v)).


   \subsection{Spectral analysis and Thomas collapse}\label{sec:spectralThomas}~

   For the Hamiltonian $\mathscr{H}_{0,\beta}$, $\beta\in\mathbb{R}$, we now consider the eigenvalue problem
   \begin{equation}
    \mathscr{H}_{0,\beta}\,g\;=\;E\,g\,,\qquad E<0\,.
   \end{equation}
   As $\mathscr{H}_{0,\beta}$ is a non-trivial self-adjoint extension of the positive symmetric operator $\mathring{H}$, we are thus concerned with the \emph{negative bound states} of $\mathscr{H}_{0,\beta}$.

   \begin{theorem}\label{thm:spectralanalysis}
    Let $\beta\in\mathbb{R}$ and let $\mathscr{H}_{0,\beta}$ the operator introduced in Theorem \ref{thm:H0beta}. The negative eigenvalues of $\mathscr{H}_{0,\beta}$ relative to eigenfunctions with spherically symmetric singular charge constitute the sequence $(E_{\beta,n})_{n\in\mathbb{Z}}$ with
    \begin{equation}\label{eq:EVbeta}
     E_{\beta,n}\;=\;-3\,e^{-\frac{2}{\,s_0}\,\mathrm{arccot}\beta}\,e^{\frac{2\pi}{s_0}n}\,.
    \end{equation}
    The constant $s_0\approx 1.0062$ is the unique positive root of $\widehat{\gamma}(s)=0$ as defined in \eqref{eq:gamma-distribution}.
   Each such eigenvalue is simple and corresponds to an eigenfunction of the form $g_{\beta,n}=u_{\xi_{\beta,n}}^{(-E_{\beta,n})}$ with
    \begin{equation}\label{eq:EFbeta}
    \widehat{\xi}_{\beta,n}(\pp)\;=\;c_{\beta,n}\,\frac{\,\sin s_0\Big(\log\Big(\sqrt{\frac{3\pp^2}{\,4|E_{\beta,n}|\,}}+\sqrt{\frac{3\pp^2}{\,4|E_{\beta,n}|\,}+1}\Big)\Big)}{\,|\pp|\sqrt{\frac{3}{4}\pp^2+|E_{\beta,n}|}}
   \end{equation}
   with normalisation factor $c_{\beta,n}\in\mathbb{C}$.
   \end{theorem}

   We recall that the nomenclature `singular charge' is reserved for the function $\xi$ uniquely associated to $g$ in the general decomposition \eqref{eq:DHstardecomposed} (Lemmas \ref{lem:Hstaretc} and \ref{lem:chargexiofg}).

   \begin{corollary}\label{cor:spectralanalysis}~
    \begin{itemize}
     \item[(i)] Each $\mathscr{H}_{0,\beta}$ admits an infinite sequence of negative bound states with energies $E_{\beta,n}$ accumulating to $-\infty$ as $n\to+\infty$, and accumulating to zero from below as $n\to -\infty$.
     \item[(ii)] Different realisations $\mathscr{H}_{0,\beta_1}$ and $\mathscr{H}_{0,\beta_2}$, namely $\beta_1\neq\beta_2$, have disjoint sequences of negative bound states, but with the same universal geometric law
     \[
      \frac{E_{\beta,n+1}}{E_{\beta,n}}\;=\;\exp \frac{2\pi}{s_0}\;\approx\;515\qquad\forall n\in\mathbb{Z}
     \]
       irrespective of $\beta$.
     \item[(iii)] Denoting by
     \[
      \sigma_{\mathrm{p}}^-(\mathscr{H}_{0,\beta})\;:=\;\{ E_{\beta,n}\,|\,n\in\mathbb{Z}\}
     \]
    the negative point spectrum of $\mathscr{H}_{0,\beta}$ in the sector $\ell=0$, one has
    \[
     \bigcup_{\beta\in\mathbb{R}}\sigma_{\mathrm{p}}^-(\mathscr{H}_{0,\beta})\;=\;\mathbb{R}^-\,.
    \]
    \end{itemize}
   \end{corollary}

  The presence of an infinite sequence of eigenvalues for the three-body Hamiltonian which accumulate to $-\infty$ is referred to as the `\emph{Thomas effect}', or `\emph{Thomas collapse}', with reference to the phenomenon that, as mentioned in the introduction, was first discovered theoretically by Thomas in 1935 \cite{Thomas1935} through an analysis of the three-body problem in which the Bethe-Peierls contact condition was formally implemented in each two-body channel. The collapse, or `\emph{fall to the centre}', refers to the circumstance that the corresponding three-body wave-function was showed to shrink around the triple coincidence point. This is precisely what can be seen from the eigenfunctions \eqref{eq:EFbeta} (Remark \ref{rem:eigenfunctions} below).

  The presence of an infinite sequence of negative eigenvalues for the three-body Hamiltonian which accumulate  to zero is referred to as the `\emph{Efimov effect}', with reference to the same phenomenon predicted theoretically in the early 1970's by Efimov \cite{Efimov-1971,Efimov-1973} for three-body quantum systems with two-body \emph{resonant} interaction of \emph{finite range}.

  Each Hamiltonian $\mathscr{H}_{0,\beta}$ thus displays both the Thomas and the Efimov effect.

  Moreover, the negative point spectra of the  $\mathscr{H}_{0,\beta}$'s fibre the whole negative half line and their disjoint union fills $\mathbb{R}^-$. This is an indirect signature of the fact that the $\mathscr{H}_{0,\beta}$'s are a one-parameter family of extensions of the same symmetric operator.

  The above properties of the negative point spectra of the $\mathscr{H}_{0,\beta}$'s, significantly formula \eqref{eq:EVbeta}, coincide with those emerging from the formal diagonalisation argument of physicists' `zero-range methods' \cite{Braaten-Hammer-2006,Naidon-Endo-Review_Efimov_Physics-2017} we surveyed in the introduction, \emph{where $\beta$ is precisely the physically grounded `three-body parameter'} \cite[Sect.~4]{Naidon-Endo-Review_Efimov_Physics-2017}. On this basis, as anticipated in the discussion of \eqref{eq:TMSalphazero}-\eqref{eq:IIIbeta}, we too shall refer to $\beta$ as the three-body parameter in the Hamiltonian. In Remark \ref{rem:eigenfunctions} below we will substantiate this nomenclature with rigorous mathematical arguments.

   Prior to proving Theorem \ref{thm:spectralanalysis}, let us single out this simple fact.

   \begin{lemma}\label{lem:adjointzero}
    Let $A$ be a densely defined and symmetric operator on a Hilbert space $\mathfrak{h}$ and assume that $A$ admits self-adjoint extensions, i.e., $\mathrm{dim}\ker(A^*-z\mathbbm{1})=\mathrm{dim}\ker(A^*-\overline{z}\mathbbm{1})>0$ for $z\in\mathbb{C}\setminus\mathbb{R}$. Let $A_U$ be the generic self-adjoint extension of $A$ with the notation of von Neumann's extension scheme, that is, $A_U=A^*|_{\mathcal{D}(A_U)}$ with
    \[
     \mathcal{D}(A_U)\;=\;\mathcal{D}(\overline{A})\dotplus(\mathbbm{1}+U)\ker(A^*-z\mathbbm{1})
    \]
    for some unitary $U:\ker(A^*-z\mathbbm{1})\stackrel{\cong}{\to}\ker(A^*-\overline{z}\mathbbm{1})$. Decompose a generic $\xi\in \mathcal{D}(A_U)$ accordingly as $\xi=\xi_0+c(\xi_+ + U\xi_+)$ for some $\xi_0\in\mathcal{D}(\overline{A})$, $\xi_+\in\ker(A^*-z\mathbbm{1})$, $c\in\mathbb{C}$.
    Assume in addition that $A$ is injective and that for some non-zero $\xi\in \mathcal{D}(A_U)$ one has $A^*\xi=0$. Then $\xi_0=0$.
   \end{lemma}

   \begin{proof}
    By assumption $0=A^*\xi=\overline{A}\xi_0+c(z\xi_+ + \overline{z} U\xi_+)$. Moreover,
    \[
     \langle \overline{A}\xi_0,\xi_+\rangle_{\mathfrak{h}}\;=\;\langle \overline{A}\xi_0,U\xi_+\rangle_{\mathfrak{h}}\;=\;0\,,
    \]
    meaning that $\overline{A}\xi_0$ and $c(z\xi_+ + \overline{z} U\xi_+)$ are orthogonal in $\mathfrak{h}$. Therefore, both such vectors must vanish, and in particular $\overline{A}\xi_0=0$. By injectivity of $A$ (and hence of $\overline{A}$) the conclusion follows.   
   \end{proof}

   \begin{proof}[Proof of Theorem \ref{thm:spectralanalysis}]
    Let us decompose $g\in\mathcal{D}(\mathscr{H}_{0,\beta})$ according to \eqref{eq:H0betadomaction} with decomposition parameter
    \[
     \lambda\;:=\;-E\,,
    \]
    that is, $g=\phi^\lambda+u_\xi^\lambda$. Then \eqref{eq:H0betadomaction}, combined with $ \mathscr{H}_{0,\beta}g=-\lambda g$, implies $\phi^\lambda\equiv 0$ and $T_\lambda\xi\equiv 0$. The eigenfunctions have then necessarily the form $g=u_\xi^\lambda$ for $\xi\in\mathcal{D}_\beta$ such that $T_\lambda\xi= 0$.

    As $T_\lambda$ is reduced with respect to the decomposition \eqref{eq:bigdecompW2} with components $T_\lambda^{(\ell)}$ (as seen in \eqref{eq:decompTTell}), the latter equation is equivalent to the collection of equations $T_\lambda^{(\ell)}\xi^{(\ell)}= 0$, $\ell\in\mathbb{N}_0$.
    We are concerned with eigenfunctions relative to charges $\xi$ belonging to the sector $\ell=0$, namely the physically relevant ones.

    Let us then focus on the problem
    \[
     T_\lambda^{(0)}\xi^{(0)}\;=\; 0\,,\qquad \xi^{(0)}\in\mathcal{D}_{0,\beta}\,,
    \]
    henceforth expressing the unknown $\xi^{(0)}$ simply as $\xi$.
    Such equation, owing to Lemma \ref{lem:0tms-generalsol}, is solved by those $\xi$'s in the subspace $\mathcal{D}_{0,\beta}$ such that the corresponding re-scaled radial function $\theta$'s, in the notation \eqref{eq:0xi}-\eqref{ftheta-1}, have the form
    \[
     \theta(x)\;=\;c\,\sin s_0 x\,,\qquad c\in\mathbb{C}.
    \]
   In this case \eqref{eq:0xi} and \eqref{ftheta-2} give
    \[
  \widehat{\xi}(\pp)\;=\;c\,\frac{\,\sin s_0\Big(\log\Big(\sqrt{\frac{3\pp^2}{4\lambda}}+\sqrt{\frac{3\pp^2}{4\lambda}+1}\Big)\Big)}{\,\sqrt{3\pi}\,|\pp|\sqrt{\frac{3}{4}\pp^2+\lambda}}\,.
 \]

   Now, in order for such $\xi$ to belong to $\mathcal{D}_{0,\beta}$, $\xi$ must only have singular component, that is, $\xi=\xi_{\mathrm{sing}}$ in the notation \eqref{eq:A0tildestar} and \eqref{eq:xiregxising}. This follows from Lemma \ref{lem:adjointzero} applied to the operator $\widetilde{\mathcal{A}_{\lambda}^{(0)}}$ defined in \eqref{eq:Alambdatilde-0} and to its extension $\mathcal{A}_{\lambda,\beta}^{(0)}$ defined in \eqref{eq:Azerolambda}. For the former, injectivity is proved in Lemma \ref{lem:D0tildedomainproperties}(vi) (using also the bijectivity property of $W_\lambda$, Lemma \ref{lem:Wlambdaproperties}(ii)). For the latter, $\Big(\widetilde{\mathcal{A}_{\lambda}^{(0)}}\Big)^\star\xi=\mathcal{A}_{\lambda,\beta}^{(0)}\xi=3W_\lambda^{-1}T_\lambda^{(0)}\xi=0$. Lemma \ref{lem:adjointzero} is then applicable, and yields $\xi-\xi_{\mathrm{sing}}=\xi_{\mathrm{reg}}=0$.

   It then remains to impose that the above generic solution $\xi$ satisfy the asymptotics \eqref{eq:IIIbeta} for the considered $\beta$.

   With simple computations analogous to those made in the proof of Lemma \ref{lem:adjointasymtotics} we find
   \[
    \widehat{\xi}(\pp)\,=\,c'\Big(\cos\big(s_0\log{\textstyle\sqrt{\frac{3}{\lambda}}}\big)\,\frac{\,\sin (s_0\log|\pp|)}{\pp^2}+\sin\big(s_0\log{\textstyle\sqrt{\frac{3}{\lambda}}}\big)\,\frac{\,\cos (s_0\log|\pp|)}{\pp^2}\Big)(1+o(1))
   \]
   as $|\pp|\to +\infty$, for some $c'\in\mathbb{C}$. The comparison with \eqref{eq:IIIbeta} then implies
   \[
    \cos\big(s_0\log{\textstyle\sqrt{\frac{3}{\lambda}}}\big)\;=\;\beta\,\sin\big(s_0\log{\textstyle\sqrt{\frac{3}{\lambda}}}\big)\,.	
   \]
   The latter condition selects the admissible values for $\lambda$, and hence $E=-\lambda$: explicitly, only the values $E_{\beta,n}=-\lambda_n$ with
   \[
    \lambda_n\;=\;3\,e^{-\frac{2}{\,s_0}\,\mathrm{arccot}\beta}\,e^{\frac{2\pi}{s_0}n}\,,\qquad n\in\mathbb{Z}\,.
   \]
   This establishes \eqref{eq:EVbeta}, and moreover it is clear from the above discussion that the corresponding eigenfunctions are all of the form $u_{\xi_n}^{(-\lambda_n)}$ and that each eigenvalue $E_{\beta,n}$ is non-degenerate.    
   \end{proof}

     \begin{proof}[Proof of Corollary \ref{cor:spectralanalysis}]
    Parts (i) and (iii), as well as the geometric formula of part (ii), all follow at once from \eqref{eq:EVbeta} of Theorem \ref{thm:spectralanalysis}. The fact that
    \[
     \sigma_{\mathrm{p}}^-(\mathscr{H}_{0,\beta})\cap\sigma_{\mathrm{p}}^-(\mathscr{H}_{0,\beta'})\;=\;\emptyset\,,\qquad \beta\neq\beta'\,,
    \]
   can be seen as follows. If $E_{\beta,n}=E_{\beta',n'}$ for some $n,n'\in\mathbb{Z}$, then
   \[
   \frac{1}{\pi}\big(\mathrm{arccot}\beta-\mathrm{arccot}\beta'\big)\;=\;k
   \]
  for some $k=n-m\in\mathbb{Z}$, as follows straightforwardly from \eqref{eq:EVbeta}. For the properties of the $\mathrm{arccot}$-function, this is only possible when $k=0$, in which case $\beta=\beta'$.   
   \end{proof}

   \begin{remark}\label{rem:eigenfunctions}
    At given $\beta\in\mathbb{R}$, the eigenfunctions $g_{\beta,n}=u_{\xi_{\beta,n}}^{(-E_{\beta,n})}$ have charges $\xi_{\beta,n}$ that tend more and more to be localised around $\yy=0$ as $E_{\beta,n}\to-\infty$, and on the contrary more and more delocalised in space as $E_{\beta,n}\uparrow 0$. In the former case $g_{\beta,n}(\yy_1,\yy_2)$ is generated by a `charge distribution'
    \begin{equation*}
  \xi_{\beta,n}(\yy_1)\delta(\yy_2)+\delta(\yy_1)\xi_{\beta,n}(\yy_2)+\delta(\yy_1-\yy_2)\xi_{\beta,n}(-\yy_2)
 \end{equation*}
   (up to a multiplicative constant, see \eqref{eq:livingonhyperplanes}) that tends to concentrate at the triple coincidence point $\yy_1=\yy_2=\mathbf{0}$ as $E_{\beta,n}\to-\infty$. This is precisely the fall-to-the-centre phenomenon associated with the Thomas effect. All this can be seen from the explicit expression of the eigenfunctions \eqref{eq:EFbeta}. To visualize it we may consider the radial distribution $\varrho_{\beta,n}$ of the charge $\xi_{\beta,n}$ in momentum coordinate, namely
   \[
    \varrho_{\beta,n}(p)\;=\;\frac{p^2\,|f_{\beta,n}(p)|^2}{\displaystyle\int_0^{+\infty}\!\ud p\,p^2\,|f_{\beta,n}(p)|^2\,}\,,
   \]
   where
   \[
   \begin{split}
    \widehat{\xi}_{\beta,n}(\pp)\;&=\;\frac{1}{\sqrt{4\pi}}\,f_{\beta,n}(|\pp|)\,,\\
    f_{\beta,n}(p)\;&=\;c_{\beta,n}\,\frac{\,2\,\sin s_0\Big(\log\Big(\sqrt{\frac{3\pp^2}{\,4|E_{\beta,n}|\,}}+\sqrt{\frac{3\pp^2}{\,4|E_{\beta,n}|\,}+1}\Big)\Big)}{\,\sqrt{3}\,|\pp|\sqrt{\frac{3}{4}\pp^2+|E_{\beta,n}|}}\,.
   \end{split}
   \]
  Figure \ref{fig:eigenfunctions} shows indeed that the more negative $E_{\beta,n}$ (namely, the larger $n>0$), the more flattened $\varrho_{\beta,n}(p)$, meaning the more localised in space $\xi_{\beta,n}(\yy)$.
   \end{remark}

   \begin{figure}[t!]
\includegraphics[width=8cm]{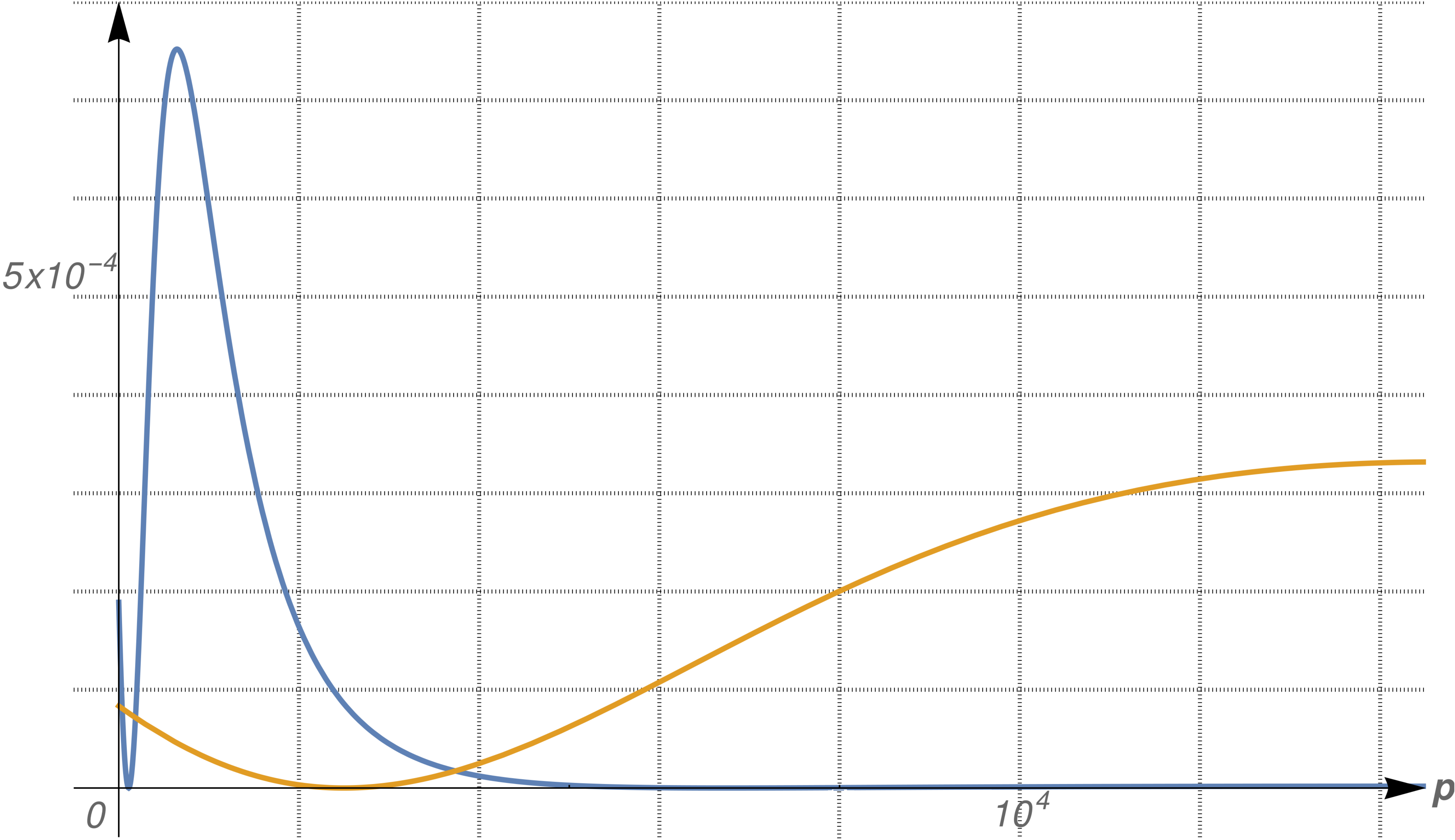}
\caption{Plot of the radial distribution profile $\varrho_{\beta,n}$ for the two charges of the eigenfunctions of $\mathscr{H}_{0,\beta}$, with $\beta=1$, relative to the quantum numbers $n=2$ (blue curve) and $n=3$ (orange curve: this has been multiplicatively magnified by a factor 10, for a clearer rendering). Correspondingly, $E_{\beta,3}<E_{\beta,2}<0$. The charge relative to the more negative eigenvalue is more delocalised in momentum, hence more localised around zero in space. Discussion in Rem.~\ref{rem:eigenfunctions}.}\label{fig:eigenfunctions}
\end{figure}

\begin{remark}
 Whereas Theorem \ref{thm:spectralanalysis} and Corollary \ref{cor:spectralanalysis} focus on the negative point spectrum of the Hamiltonian $\mathscr{H}_{0,\beta}$, it is not difficult to determine that its essential spectrum is precisely
 \begin{equation}
  \sigma_{\mathrm{ess}}(\mathscr{H}_{0,\beta})\;=\;[0,+\infty)\,.
 \end{equation}
 This can be obtained by suitably adjusting to the present setting the reasoning developed in \cite[Sect.~5 and 6]{BMO-2017} in collaboration with Becker and Ottolini. That $[0,+\infty)\subset\sigma_{\mathrm{ess}}(\mathscr{H}_{0,\beta})$ can be seen by means of a Weyl sequence of the same type as the standard Weyl sequences in $H^2_{\mathrm{b}}(\mathbb{R}^3\times\mathbb{R}^3)$ used to show that $\sigma_{\mathrm{ess}}(\mathring{H}_F)=[0,+\infty)$, suitably modified so as to vanish at the coincidence hyperplanes (see \cite[Prop.~5.1]{BMO-2017}). For the opposite inclusion, one can reproduce a version of the Minlos-Yoshitomi decomposition for the present $T_\lambda$, as we did in \cite[Sect.~6]{BMO-2017} for the fermionic $T_\lambda$, based on which by compactness arguments one can show that the spectral projection relative to each interval $[a,b]\subset(-\infty,0)$ is finite-dimensional. 
\end{remark}

 \begin{remark}\label{rem:choicedomains}
 As a follow-up of Remark \ref{rem:merepurpose}, in retrospect we can comment over the choice \eqref{eq:Dtilde0} of the initial charge domain $\widetilde{\mathcal{D}}_0$ used to realise $\widetilde{\mathcal{A}_\lambda^{(0)}}=W_\lambda^{-1}T_\lambda^{(0)}$ as a densely defined and symmetric operator on $H^{-\frac{1}{2}}_{W_\lambda,\ell=0}(\mathbb{R}^3)$. 
 \begin{itemize}
  \item[(i)] As observed already, also on the larger domain $\widetilde{\mathcal{D}}_0'$ from \eqref{eq:Dolarger} could one make $W_\lambda^{-1}T_\lambda^{(0)}$ symmetric on $H^{-\frac{1}{2}}_{W_\lambda,\ell=0}(\mathbb{R}^3)$. In either case, the resulting  $\widetilde{\mathcal{A}_\lambda^{(0)}}$ is an admissible Birman parameter for a symmetric Ter-Martirosyan  Skornyakov extension $\mathscr{H}_0\equiv\mathring{H}_{\widetilde{\mathcal{A}_\lambda^{(0)}}}$ of the initial operator $\mathring{H}$ (in the sector $\ell=0$). The difference is in the smaller ($\widetilde{\mathcal{D}}_0$) or larger ($\widetilde{\mathcal{D}}_0'$) domain of charges, with corresponding re-scaled radial functions of the form 
 \[
  \begin{split}
   \theta\;&=\;\sin (s_0 |x|)*\big(\widehat{\vartheta}/\widehat{\gamma}_+\big)^{\!\vee}\qquad\qquad\qquad\quad\!\textrm{for charges in $\widetilde{\mathcal{D}}_0$}, \\
   \theta\;&=\;c\,\sin s_0 x+\sin (s_0 |x|)*\big(\widehat{\vartheta}/\widehat{\gamma}_+\big)^{\!\vee}\qquad\textrm{for charges in $\widetilde{\mathcal{D}}_0'$},
  \end{split}
 \]
 with $\vartheta\in C^\infty_{0,\mathrm{odd}}(\mathbb{R}_x)$ and $c\in\mathbb{C}$.
 \item[(ii)] Choosing the \emph{larger} domain, one solves the eigenvalue problem
 \[
  \mathscr{H}_0 \,g\;=\;-\lambda g\,,\qquad \lambda>0
 \]
 (in the sector $\ell=0$) with the same reasoning as in the proof of Theorem \ref{thm:spectralanalysis}, and thus finds solutions $g=u_{\xi_\lambda}^{\lambda}$ with
 \[
 \begin{split}
  \theta(x)\;&=\;\sin s_0 x\,, \\
   \widehat{\xi}_\lambda(\pp)\;&=\;\frac{\,\sin s_0\Big(\log\Big(\sqrt{\frac{3\pp^2}{4\lambda}}+\sqrt{\frac{3\pp^2}{4\lambda}+1}\Big)\Big)}{\,\sqrt{3\pi}\,|\pp|\sqrt{\frac{3}{4}\pp^2+\lambda}}\,,
 \end{split}
 \]
 up to an overall multiplicative constant. All such solutions are now admissible, in that all the above $\xi_\lambda$'s belong to $\widetilde{\mathcal{D}}_0'$ irrespective of $\lambda>0$. This proves that $\mathscr{H}_0$ has a continuum of eigenvalues, which is \emph{incompatible with self-adjointness}.
 \item[(iii)] In fact, a laborious but instructive computation (originally alluded to in \cite[Sect.~III.3]{Flamand-Cargese1965}) shows that imposing the orthogonality of any two such $u_{\xi_{\lambda}}^{\lambda}$ and $u_{\xi_{\lambda'}}^{\lambda'}$ in $L^2_\mathrm{b}(\mathbb{R}^3\times\mathbb{R}^3,\ud\yy_1\ud\yy_2)$ does partition $\mathbb{R}^+$ into the disjoint union
 \[
  \mathbb{R}^+\;=\;\bigcup_{\beta\in\mathbb{R}}\sigma_\beta\,,\qquad \sigma_\beta\;:=\;\{\lambda_{\beta,n}=3\,e^{-\frac{2}{\,s_0}\,\mathrm{arccot}\beta}\,e^{\frac{2\pi}{s_0}n}\,|\,n\in\mathbb{Z}\}\,,
 \]
 where each sequence $(-\lambda_{\beta,n})_{n\in\mathbb{Z}}$ is an admissible sequence of simple eigenvalues, with orthogonal eigenfunctions by construction, for a self-adjoint operator on $L^2_\mathrm{b}(\mathbb{R}^3\times\mathbb{R}^3,\ud\yy_1\ud\yy_2)$ (precisely, for the operator $\mathscr{H}_{0,\beta}$ from Theorem \ref{thm:H0beta}).
 \item[(iv)] In connection to (ii), and in view of Theorems \ref{thm:H0beta} and \ref{thm:spectralanalysis}, we see that the inclusion
 \[
  \widetilde{\mathcal{D}}_0\;\subset\; \mathcal{D}_{0,\beta} \;\subset\;\widetilde{\mathcal{D}}_0'
 \] 
 involves three distinct admissible choices for the singular charges domain (in the sector $\ell=0$) for symmetric TMS extensions of $\mathring{H}$, only the second of which produces a self-adjoint extension. 
 \end{itemize}
 \end{remark}

 \subsection{Variants}\label{sec:variants}~
 
 The canonical model(s) $\mathscr{H}_{0,\beta}$, $\beta\in\mathbb{R}$, have variants that do not alter the $\ell=0$ sector construction, where the essential physics takes place.

 As already observed at the beginning of Sect.~\ref{sec:higherell}, there is an amount of arbitrariness in the definition of the trimer's Hamiltonian in sectors of higher angular momentum.

 The construction developed in Sect.~\ref{sec:higherell} is canonical in that it provides the Friedrichs realisation of an operator of Ter-Martirosyan Skornyakov type with inverse negative scattering length $\alpha$. (It has of course also a considerable degree of instructiveness, from the technical point of view.)

 Such construction, combined with the analysis in the sector $\ell=0$ (Sect.~\ref{sec:lzero}) led to the self-adjoint Ter-Martirosyan Skornyakov Hamiltonians of the form $\mathring{H}_{\mathcal{A}_{\lambda,\beta}}$ with Birman parameter $\mathcal{A}_{\lambda,\beta}$ given by \eqref{eq:globalAlambda-1}.

 Equally admissible (self-adjoint and TMS) alternatives are given by modified Birman parameters of the form 
  \begin{equation}
   \mathcal{A}_{\lambda,\beta}\;:=\; \mathcal{A}_{\lambda,\beta}^{(0)}\:\oplus\: \bigoplus_{\ell=1}^\infty\mathcal{A}_{\lambda,\alpha_\ell}^{(\ell)}
  \end{equation}
 where
 \begin{equation}\label{AFop-ellnot0_ell}
    \begin{split}
     \mathcal{D}\big(\mathcal{A}_{\lambda,\alpha_\ell}^{(\ell)}\big)\;&:=\;\mathcal{D}_\ell\;=\;\big\{\xi\in H_\ell^{\frac{1}{2}}(\mathbb{R}^3)\,\big|\,T_\lambda^{(\ell)}\xi\in H_\ell^{\frac{1}{2}}(\mathbb{R}^3)\big\} \\
     \mathcal{A}_{\lambda,\alpha_\ell}^{(\ell)}\;&:=\;3 W_\lambda^{-1}\big(T_\lambda^{(\ell)}+\alpha_\ell\mathbbm{1}\big)\,,
    \end{split}
   \end{equation}
 thus on the same charge domain $\mathcal{D}_\ell$ that guarantees self-adjointness (Proposition \ref{prop:Alambdaellnot0}), but with scattering lengths that depend on the angular sector.

 In fact, it would be physically acceptable also to ignore in the first place the interaction in sectors of non-zero angular momentum, thus focusing on the Hamiltonians of interest only as effective models in the sector $\ell=0$. This is obtained by taking the trivial (Friedrichs) extension of $\mathring{H}$ whenever $\ell\neq 0$: particles in a three-body state with charges that do not belong to the zero angular momentum sector just move with free dynamics. In this case the final Birman parameter's domain, instead of \eqref{eq:finalDbeta}, becomes
 \begin{equation}\label{eq:finalDbetaNOINT}
   \mathcal{D}_{0,\beta}\;\boxplus\;\op_{k=1}^\infty\{0\}\,.
  \end{equation}
 The TMS condition remains only in the sector $\ell=0$. The Hamiltonian is just the free kinetic operator on the other sectors.

\section{Ill-posed models}\label{sec:illposed}

The analysis developed in Sections \ref{sec:generalextscheme} through \ref{sec:canonicalmodel} is deeply inspired by many previous investigations we extensively referred to in the introduction, and yet it is novel in that a number of crucial steps are performed here by thoroughly working out a rigorous operator theoretic scheme.

In the introduction we argued that for three-body quantum systems with contact interaction physical zero-range methods determine eigenfunctions and eigenvalues of a formal Hamiltonian that otherwise remains unqualified. We also argued that mathematical approaches are aimed at constructing a self-adjoint Hamiltonian of Ter-Martirosyan Skornyakov type: first one declares the operator or its quadratic form, then one performs the subsequent spectral analysis on it. Of course, on the physical side there is the advantage of an ultimate agreement check with the experiments.

As a matter of fact, one can track down, through the mathematical literature on the subject, certain recurrent sources of ill-posed models, failing to provide a three-body Hamiltonian that at the same time be self-adjoint and exhibit the Bethe-Peierls / Ter-Martirosyan Skornyakov contact condition.

On the mathematical technical level, the model's well-posedness lies in the correct choice of the domain of self-adjointness, among those domains that in addition reproduce the desired short-scale physical asymptotics. A wrong choice of the (operator or form) domain fails to yield self-adjointness and produces incorrect spectral data. In this informal sense we speak of incomplete or ill-posed models.

In some circumstances an explicit signature of some sort of incompleteness of the mathematical model is the quantitative discordance in the spectral analysis with numerical and experimental evidence from physics. This has been the case significantly for three-body systems with a pair of identical fermions: the recent works \cite{CDFMT-2015,MO-2016,MO-2017} mentioned already in the introduction were essentially aimed at clarifying this perspective, on which we shall further comment in the course of this Section.

In other occurrences the ill-posedness of the model is more subtle and less evident, and the case of the bosonic trimer is typical in this sense.

For clarity of presentation, let us group such occurrences into two categories, discussed, respectively, in Subsect.~\ref{sec:illbc} and \ref{sec:illdomain}.

\subsection{Ill-posed boundary condition}\label{sec:illbc}~

The operator-theoretic programme aims at realising a Hamiltonian of zero-range interaction as a suitable self-adjoint extension of $\mathring{H}$, the free Hamiltonian initially restricted to wave-functions that do not meet the coincidence configuration $\Gamma$ (see \eqref{eq:domHring-initial}), by selecting an extension's domain where instead the wave-functions behave at $\Gamma$ with a precise, physically grounded boundary condition (BP/TMS).

As demonstrated in Sect.~\ref{sec:TMSextension-section} (Theorem \ref{thm:globalTMSext}), such two-fold requirement is possible if and only if one restricts $\mathring{H}^*$ to those functions $g\in\mathcal{D}(\mathring{H}^*)$ with singular charges $\xi$ from a distinguished subspace $\mathcal{D}\subset H^{-\frac{1}{2}}(\mathbb{R}^3)$:
\begin{itemize}
 \item[1.] $\mathcal{D}$ must be dense in $H^{-\frac{1}{2}}(\mathbb{R}^3)$ (for generic self-adjoint extensions of $\mathring{H}$ the charge domain need not be dense: Theorem \ref{thm:generalclassification});
 \item[2.] $\mathcal{D}$ must be mapped by $T_\lambda+\alpha\mathbbm{1}$ into $H^{\frac{1}{2}}(\mathbb{R}^3)$ for some (and hence for all) $\lambda>0$;
 \item[3.] $\mathcal{D}$ must be a domain of self-adjointness for $W_\lambda^{-1}(T_\lambda+\alpha\mathbbm{1})$ in the Hilbert space given by $H^{-\frac{1}{2}}(\mathbb{R}^3)$ equipped with the twisted (equivalent) scalar product $\langle\cdot,W_\lambda\cdot\rangle$ (see \eqref{eq:W-scalar-product}).
\end{itemize}

There are no other possibilities (Theorem \ref{thm:globalTMSext}).

Condition 2.~above makes $W_\lambda^{-1}(T_\lambda+\alpha\mathbbm{1})$ well-posed, because $W_\lambda$ is a \emph{bijection} of $H^{-\frac{1}{2}}(\mathbb{R}^3)$ \emph{onto} $H^{\frac{1}{2}}(\mathbb{R}^3)$ (Lemma \ref{lem:Wlambdaproperties}(ii)), and eventually leads to the desired boundary condition, namely 
\begin{equation}\label{eq:again3/2}
 \begin{split}
  & \textrm{for every $\xi\in\mathcal{D}$ there is $\phi^\lambda\in H^2_{\mathrm{b}}(\mathbb{R}^3\times\mathbb{R}^3)$ with} \\
  & \phi^\lambda(\yy,\mathbf{0})\;=\;(2\pi)^{-\frac{3}{2}} (T_\lambda+\alpha\mathbbm{1})\xi(\yy)\qquad\textrm{for a.e.~$\yy\in\mathbb{R}^3$}\,,
 \end{split}
\end{equation}
or any of the equivalent versions \eqref{eq:g-largep2-TMS0}-\eqref{eq:phi-largep2-star-yversion-TMS0}. From the perspective of \eqref{eq:again3/2} the requirement $(T_\lambda+\alpha\mathbbm{1})\mathcal{D}\subset H^{\frac{1}{2}}(\mathbb{R}^3)$ is needed because by standard trace arguments \eqref{eq:again3/2} is a $H^{\frac{1}{2}}$-identity and would not have sense if $(T_\lambda+\alpha\mathbbm{1})\xi$ had \emph{strictly less} than $H^{\frac{1}{2}}$-regularity.

In a number of past studies the choice of the charge domain $\mathcal{D}$ left instead the boundary condition \eqref{eq:again3/2} ambiguous.

The first semi-rigorous mathematical treatment of the bosonic trimer was given by Minlos and Faddeev in the work \cite{Minlos-Faddeev-1961-1}, and there the choice was (with our current notation) $\widetilde{\mathcal{D}}=\mathcal{F}^{-1}C^\infty_0(\mathbb{R}^3_{\pp})$. That is, a \emph{symmetric} extension of $\mathring{H}$ of Ter-Martirosyan Skornyakov type was suggested as follows: the extension's domain consists of those functions whose singular charges $\xi$ are all those with smooth and compactly supported Fourier transform $\widehat{\xi}$. In fact, the first two seminal works \cite{Minlos-Faddeev-1961-1,Minlos-Faddeev-1961-2} by Minlos and Faddeev had rather the form of very brief announcements with only sketches of the main reasoning and proofs; yet the space of charges was clearly declared therein and moreover, shortly after, Flamand \cite{Flamand-Cargese1965} presented a detailed review of \cite{Minlos-Faddeev-1961-1} with the same explicit domain declaration.

We also mention the subsequent choices $\widetilde{\mathcal{D}}=\mathcal{F}^{-1}C^\infty_0(\mathbb{R}^3_{\pp})$ in \cite{Minlos-1987,Minlos-Shermatov-1989,Minlos-2011-preprint_May_2010,Minlos-2014-I_RusMathSurv}, $\widetilde{\mathcal{D}}=H^1(\mathbb{R}^3)$ in \cite{Minlos-2012-preprint_30sett2011,Minlos-2014-II_preprint-2012,Moser-Seiringer-2017,Figari-Teta-2020}, and $\widetilde{\mathcal{D}}=H^{\frac{3}{2}-\varepsilon}(\mathbb{R}^3)$, $\varepsilon>0$, in \cite{Shermatov-2003}. (The above-mentioned works \cite{Minlos-Shermatov-1989,Shermatov-2003,Minlos-2011-preprint_May_2010,Minlos-2012-preprint_30sett2011,Minlos-2014-I_RusMathSurv,Minlos-2014-II_preprint-2012,Moser-Seiringer-2017} are actually 
focused on the \emph{fermionic} counterpart setting; yet, also in that case one has to face the very same technical problem of providing a well-posed definition of $W_\lambda^{-1}(T_\lambda+\alpha\mathbbm{1})$ and of making the boundary condition \eqref{eq:again3/2} unambiguous, up to non-essential changes of numerical coefficients in $T_\lambda$ and $W_\lambda$ from the bosonic to the fermionic analysis.)

Now, such proposals for $\widetilde{\mathcal{D}}$ are problematic. In the case $\widetilde{\mathcal{D}}=\mathcal{F}^{-1}C^\infty_0(\mathbb{R}^3_{\pp})$, hence $\widetilde{\mathcal{D}}\subset H^s(\mathbb{R}^3)$ $\forall s\in\mathbb{R}$, the sectors $\ell\in\mathbb{N}$ are unambiguously described through analogues of our Lemma \ref{lem:Atildenot0} and Proposition \ref{prop:Alambdaellnot0} (where our choice was $\widetilde{\mathcal{D}}=H^\frac{3}{2}_{\ell}(\mathbb{R}^3)$, namely the lowest Sobolev space that is entirely mapped with continuity by $T_\lambda^{(\ell)}$ into the desired $H^\frac{1}{2}_{\ell}(\mathbb{R}^3)$), and one realises $W_\lambda^{-1}(T_\lambda^{(\ell)}+\alpha\mathbbm{1})$ self-adjointly on the domain $\mathcal{D}_\ell=\big\{\xi\in H_\ell^{\frac{1}{2}}(\mathbb{R}^3)\,\big|\,T_\lambda^{(\ell)}\xi\in H_\ell^{\frac{1}{2}}(\mathbb{R}^3)\big\}$ (Proposition \ref{prop:Alambdaellnot0}). On the contrary, choosing $\widetilde{\mathcal{D}}=H^1_{\ell}(\mathbb{R}^3)$, $\ell\in\mathbb{N}$, poses the problem of whether $(T_\lambda^{(\ell)}+\alpha\mathbbm{1})\widetilde{\mathcal{D}}\subset H^\frac{1}{2}_{\ell}(\mathbb{R}^3)$, which is not true in general.

Moreover, even the most stringent choice $\widetilde{\mathcal{D}}=\mathcal{F}^{-1}C^\infty_0(\mathbb{R}^3_{\pp})$ does not guarantee the well-posedness of the sector $\ell=0$. We already observed (Remark \ref{rem:Tl-failstomap}) that if $\xi\in\mathcal{F}^{-1}C^\infty_0(\mathbb{R}^3_{\pp})$, then $T_\lambda^{(0)}\xi$ belongs to $H^{\frac{1}{2}-\varepsilon}(\mathbb{R}^3)$ $\forall\varepsilon>0$, but not to $H^{\frac{1}{2}}(\mathbb{R}^3)$.

\subsection{Incomplete criterion of self-adjointness}\label{sec:illdomain}~

The next source of ill-posedness may be tracked down in the problem of determining a domain $\mathcal{D}\supset\widetilde{\mathcal{D}}$ of self-adjointness for $W_\lambda^{-1}(T_\lambda+\alpha\mathbbm{1})$ with respect to the Hilbert space $H^{-\frac{1}{2}}_{W_\lambda}(\mathbb{R}^3)$, once a domain $\widetilde{\mathcal{D}}$ of symmetry is selected.

Because of the special form of the scalar product \eqref{eq:W-scalar-product} in $H^{-\frac{1}{2}}_{W_\lambda}(\mathbb{R}^3)$, it is straightforward to see (Lemma \ref{lem:symsym}) that, as long as $\widetilde{\mathcal{D}}$ is dense in $L^2(\mathbb{R}^3)$, the symmetry on $\widetilde{\mathcal{D}}$ of $W_\lambda^{-1}(T_\lambda+\alpha\mathbbm{1})$ with respect to $H^{-\frac{1}{2}}_{W_\lambda}(\mathbb{R}^3)$ is equivalent to the symmetry on $\widetilde{\mathcal{D}}$ of $T_\lambda$ with respect to $L^2(\mathbb{R}^3)$.

Based on such a suggestive property, an amount of previous investigations \cite{Minlos-Faddeev-1961-1,Minlos-Faddeev-1961-2,Flamand-Cargese1965,Minlos-1987,Minlos-Shermatov-1989,Menlikov-Minlos-1991,Menlikov-Minlos-1991-bis,Minlos-TS-1994,Shermatov-2003,Minlos-2011-preprint_May_2010,Minlos-2010-bis,Minlos-2012-preprint_30sett2011,Minlos-2014-I_RusMathSurv,Minlos-2014-II_preprint-2012,Figari-Teta-2020} adopted the claim that, for a dense subspace $\mathcal{D}$ of $L^2(\mathbb{R}^3)$, $W_\lambda^{-1}(T_\lambda+\alpha\mathbbm{1})$ on $\mathcal{D}$ is self-adjoint with respect to $H^{-\frac{1}{2}}_{W_\lambda}(\mathbb{R}^3)$ if and only if $T_\lambda$ on $\mathcal{D}$ is self-adjoint with respect to $L^2(\mathbb{R}^3)$.

In fact, this is not true (Lemma \ref{lem:exampleMinloswrong}) and the link between the two self-adjointness problems is more subtle (Lemma \ref{lem:two-selfadj-problems}).

That the emergent Hamiltonian obtained by realising the Birman parameter $W_\lambda^{-1}(T_\lambda+\alpha\mathbbm{1})$ self-adjointly on $L^2(\mathbb{R}^3)$ (instead of $H^{-\frac{1}{2}}_{W_\lambda}(\mathbb{R}^3)$) yields inconsistencies, has been known for a few years with reference to the \emph{fermionic} problem (a trimer consisting of two identical fermions of mass $m$ and a third particle of different type, and with inter-particle zero-range interaction). In that setting, a quantitative difference emerges between the mass thresholds of self-adjointness in the various $\ell$-sectors computed in \cite{Minlos-2011-preprint_May_2010,Minlos-2012-preprint_30sett2011,Minlos-2014-I_RusMathSurv,Minlos-2014-II_preprint-2012} by solving the self-adjointness problem in $L^2(\mathbb{R}^3)$, and certain spectral mass thresholds having the same conceptual meaning and obtained by formal theoretical computations and numerics within the physicists' zero-range methods \cite{Werner-Castin-2006-PRA,Kartavtsev-Malykh-2007,Castin-Tignone-2011}. The work \cite{CDFMT-2015} in collaboration with Correggi, Dell'Antonio, Figari, and Teta gave a first mathematical explanation of the situation, in the unitary regime $\alpha=0$, by means of a quadratic form construction of self-adjoint Hamiltonians of Ter-Martirosyan Skornyakov type, showing that certain non-$L^2$-charges in $H^{-\frac{1}{2}}(\mathbb{R}^3)$ were needed for a correct domain of self-adjointness. Right after, in our previous works \cite{MO-2016,MO-2017} in collaboration with Ottolini we addressed the same issue, recognising that indeed the correct self-adjointness problem for the Birman parameter is only with respect to the Hilbert space $H^{-\frac{1}{2}}_{W_\lambda}(\mathbb{R}^3)$.

For a three-body systems of \emph{three identical bosons} there is of course no mass parameter, hence no counterpart of the type of inconsistencies described above for the fermionic case.

Moreover, deceptively enough, the study of the self-adjoint extensions of $T_\lambda$ with respect to $L^2(\mathbb{R}^3)$, with initial domain, say, $H^1(\mathbb{R}^3)$, yields conclusions that are qualitatively very similar to the correct analysis of the self-adjoint realisations of $W_\lambda^{-1}T_\lambda^{(0)}$ with respect to $H^{-\frac{1}{2}}_{W_\lambda,\ell=0}(\mathbb{R}^3)$.

More precisely, in analogy to our discussion of Lemma \ref{lem:deficiency1-1}, we can easily check that the $L^2$-computation of the deficiency spaces, namely of the solutions $\xi$ to $T_\lambda^{(0)}\xi=\ii\mu\xi$ in $L^2(\mathbb{R}^3)$ for $\mu>0$, yields 
\begin{equation*}
    \theta_+(x)-\frac{4}{\pi\sqrt{3}}\int_{\mathbb{R}}\theta_+(y)\,\log \frac{\,2\cosh(x-y)+1\,}{\,2\cosh(x-y)-1\,}\,\ud y \;=\;\frac{\ii \mu}{\,2\pi^2\sqrt{\lambda}}\,\frac{\theta_+(x)}{\,\cosh x}
  \end{equation*}
 (see \eqref{radialTMS0} for a comparison). 
 The above homogeneous equation replaces \eqref{eq:eigenequation-thetaplus}, and is equivalent to
 \[
  \widehat{\gamma}(s)\,\widehat{\theta}_+(s)\;=\;\frac{\ii \mu}{\,4\pi^2\sqrt{\lambda}}\,\Big(\frac{1}{\,\cosh\frac{\pi}{2}s}*\widehat{\theta}_+\Big)(s)\,,
 \]
which replaces \eqref{eq:eigenequation-thetaplusF}. By the same reasoning of the proof of Lemma \ref{lem:deficiency1-1}, the latter equation has a unique solution, up to multiplicative prefactor, meaning that the deficiency indices are $(1,1)$. Then, mimicking the proof of Lemma \ref{lem:adjointasymtotics}, one finds a completely analogous large-momentum asymptotics for the singular elements of the adjoint, which leads to a structure of $L^2$-self-adjoint realisations of $T_\lambda^{(0)}$ that mirrors that of Proposition \ref{prop:Aellzero-selfadj}.

Nevertheless, each such domain of $L^2$-self-adjointness for $T_\lambda$ is not enough to guarantee that the corresponding three-body Hamiltonian is self-adjoint.

\section{Regularised models}

The Hamiltonian $\mathscr{H}_{0,\beta}$ constructed as canonical model in Theorem \ref{thm:H0beta} is regarded as \emph{instable}, owing to its infinite sequence of bound state energy levels accumulating to $-\infty$ (Thomas collapse).

In retrospect, this feature is due to the combination of the \emph{zero-range character} of the modelled interaction and the \emph{bosonic symmetry} of the model. As a comparison, the analogous construction for a three-body system with zero-range interaction consisting of two identical fermions and a particle of different type produces a Hamiltonian that is lower semi-bounded in a suitable regime of masses \cite{CDFMT-2012}.

Thus, in order to have a stable model one hypothesis must be removed, among the vanishing of the interaction range and the bosonic symmetry. Such an observation was made by Thomas himself in his work on the tritium  \cite{Thomas1935}, which is remarkable if one considers that at the time of \cite{Thomas1935} neither the precise nature of the nuclear interaction nor the connection between spin and statistics had been understood yet.

This poses the problem of constructing \emph{regularised} models for the bosonic trimer, which do not display the spectral instability analysed in Theorem \ref{thm:spectralanalysis}, and yet describe an interaction of zero range that retains certain spectral features such as the continuous spectrum all above some threshold, or the occurrence of negative eigenvalues accumulating at the continuum threshold (Efimof effect)), or the typical short-range profile of the wave-functions.

Abstractly speaking, one can perform a regularisation with an ad hoc energy cut-off on the canonical model $\mathscr{H}_{0,\beta}$, or also with a modified Hamiltonian in the form of a proper Schr\"{o}dinger operator with two-body potentials of small but finite (i.e., non-zero) effective range.

Somewhat intermediate between such two directions, we discuss here a construction, formerly contemplated by Minlos and Faddeev with no further analysis, of a contact interaction Hamiltonian similar to $\mathscr{H}_{0,\beta}$, but with a regularisation that has the overall effect of switching off the interaction in the vicinity of the triple coincidence configuration. The three identical bosons are allowed by the statistics to occupy that region, in which now the regularisation make them asymptotically free. This removes the instability of the canonical model. (Subsect.~\ref{sec:MFregularisation}-\ref{sec:MinFadzero}).

Further types of regularisations have been proposed, which are conceptually analogous to the idea of Minlos and Faddeev in that they introduce a non-constant, effective scattering length that tends to be suppressed (meaning, no interaction, particles are free) when the three bosons get close to the point of triple coincidence. Whereas the Minlos-Faddeev regularisation implements such idea in position coordinates, one can analogously work in momentum coordinates, making the effective scattering length vanish at large relative momenta. For comparison, we shall discuss such high energy cut-off in Subsect.~\ref{sec:highenergycutoff}

\subsection{Minlos-Faddeev regularisation}\label{sec:MFregularisation}~

This is the ultra-violet regularisation originally proposed in \cite[Sect.~6]{Minlos-Faddeev-1961-1} (see also \cite[Sect.~VI.2]{Flamand-Cargese1965} and \cite{Albe-HK-Wu-1981}), and on which a number of results with the quadratic form approach have been recently announced in \cite{Figari-Teta-2020}.

We shall study it within the operator-theoretic scheme of the present analysis.

In practice, this is a modification of the canonical model (Theorem \ref{thm:H0beta}) along the following line: the ordinary Birman parameter $3W_\lambda^{-1}(T_\lambda+\alpha\mathbbm{1})$, that selects (via Theorems \ref{thm:generalclassification} and \ref{thm:globalTMSext}) self-adjoint extensions of Ter-Martirosyan Skornyakov type of the minimal operator $\mathring{H}$ defined in \eqref{eq:domHring-initial}, is replaced by 
\begin{equation}\label{eq:newBirmanpar}
 3\,W_\lambda^{-1}(T_\lambda+\alpha\mathbbm{1}+K_\sigma)\,,\qquad \sigma\,>\,0\,,
\end{equation}
where, for generic $\sigma\in\mathbb{R}$,
\begin{equation}\label{eq:Ksigmaposition}
\begin{split}
  (K_\sigma\xi)(\yy)\;&:=\;\frac{\,\sigma_0+\sigma\,}{|\yy|}\,\xi(\yy) \\
 \sigma_0\;&:=\;2\pi\sqrt{3}\,\Big(\frac{4\pi}{3\sqrt{3}}-1\Big)\,.
\end{split}
\end{equation}

The motivation is clear from the large momentum asymptotics \eqref{eq:g-largep2-star} valid for generic $g\in\mathcal{D}(\mathring{H})$, namely
 \begin{equation*}
  \int_{\!\substack{ \\ \\ \pp_2\in\mathbb{R}^3 \\ |\pp_2|<R}}\widehat{g}(\pp_1,\pp_2)\,\ud\pp_2\;=\;4\pi R\,\widehat{\xi}(\pp_1)+\Big({\textstyle\frac{1}{3}}(\widehat{W_\lambda\eta})(\pp_1)-(\widehat{T_\lambda\xi})(\pp_1)\Big)+o(1)
 \end{equation*}
 as $R\to +\infty$. Indeed, when a self-adjoint extension is selected, out of the family \eqref{eq:family}, labelled by the Birman parameter \eqref{eq:newBirmanpar} densely defined in $H^{-\frac{1}{2}}(\mathbb{R}^3)$, then in the asymptotics above one has $\eta=3W_\lambda^{-1}(T_\lambda+\alpha\mathbbm{1}+K_\sigma)\xi$ (as prescribed by formula \eqref{eq:domDHA} of Theorem \ref{thm:generalclassification}), whence
 \begin{equation*}
  \int_{\!\substack{ \\ \\ \pp_2\in\mathbb{R}^3 \\ |\pp_2|<R}}\widehat{g}(\pp_1,\pp_2)\,\ud\pp_2\;=\;4\pi R\,\widehat{\xi}(\pp_1)+\alpha\widehat{\xi}(\pp_1)+(\widehat{K_\sigma\xi})(\pp_1)+o(1)\,,
 \end{equation*}
 and also (see Corollary \ref{cor:largepasympt-star})
  \begin{equation*}
 (2\pi)^{\frac{3}{2}} c_g\,g_{\mathrm{av}}(\yy_1;|\yy_2|)\,\stackrel{|\yy_2|\to 0}{=}\,\frac{4\pi}{|\yy_2|}\xi(\yy_1)+\Big(\alpha+ \frac{\,\sigma_0+\sigma\,}{|\yy_1|}\Big)\xi(\yy_1)+ o(1)\,.
 \end{equation*}
 Thus, the new self-adjoint Hamiltonian has a domain of functions that display a modified short-scale asymptotics, as compared to the zero-range Bethe-Peierls condition: the modification consists of the inverse negative scattering length $\alpha$ being replaced by a position-dependent value
 \begin{equation}
  \alpha_{\mathrm{eff}}\;:=\;\alpha+(\sigma_0+\sigma)/|\yy|\,,
 \end{equation}
 where $|\yy|$ is the distance of the third particle from the point towards which the other two are getting closer and closer. 
 Therefore, $\alpha_{\mathrm{eff}}\to +\infty$ when \emph{all three particles} collapse to the same spatial position, meaning that the scattering length vanishes in such limit. As vanishing scattering length means absence of interaction, the overall effect is a three-body regularisation that prevents the collapse of the system along an unbounded sequence of negative energy levels.

 By construction, $K_\sigma$ commutes with the rotations in $\mathbb{R}$ and therefore is reduced as
 \begin{equation}\label{eq:Ksigmarotations}
  K_\sigma\;=\;\bigoplus_{\ell\in\mathbb{N}}K_\sigma^{(\ell)}
 \end{equation}
 with respect to the orthogonal Hilbert space decomposition \eqref{eq:bigdecompW} of $H^{-\frac{1}{2}}_{W_\lambda}(\mathbb{R}^3)$. For all practical purposes (in view also of the discussion of Subsect.~\ref{sec:variants}) it suffices to implement the Minlos-Faddeev regularisation in the sector $\ell=0$, thus only inserting $K_\sigma^{(0)}$ in \eqref{eq:newBirmanpar}, as the canonical model is already stable in the sectors of higher angular momentum. That is the version of the regularisation that we shall study here.

 With definition \eqref{eq:Ksigmaposition} and the above considerations in mind, the same conceptual path of Section \ref{sec:lzero} can be now re-done, adapting it to the new Birman parameter \eqref{eq:newBirmanpar}.

 To this aim, we need updated expressions of the quantities of interest in terms of the re-scaled radial components associated to the charges.
 
 First of all, taking the Fourier transform in \eqref{eq:Ksigmaposition} yields
    \begin{equation}\label{eq:Ksigmamomentum}
     (\widehat{K_\sigma\xi})(\pp)\;=\;\frac{\sigma_0+\sigma}{2\pi^2}\,\int_{\mathbb{R}^3}\frac{\widehat{\xi}(\qq)}{\,|\pp-\qq|^2}\,\ud\qq\,.
    \end{equation}
 
 Let us also introduce two auxiliary functions, namely
 \begin{equation}\label{eq:bigTheta}
  \begin{split}
    \Theta^{(\theta)}_\sigma(x)\;&:=\;\theta(x)-\frac{4}{\pi\sqrt{3}}\int_{\mathbb{R}}\ud y\,\theta(y) \log \frac{\,2\cosh(x-y)+1\,}{\,2\cosh(x-y)-1\,}+ \\
     &\qquad\qquad\qquad +\frac{\,\sigma_0+\sigma\,}{\pi^3\sqrt{3}}\int_{\mathbb{R}}\ud y\,\theta(y)\,\log\Big|\coth\frac{x-y}{2}\Big|
  \end{split}
 \end{equation}
  for given $\theta$, and
  \begin{equation}\label{eq:gamma-s-distribution}
   \widehat{\gamma}_\sigma(s)\;:=\;1+\frac{1}{\,\sqrt{3}\,s\,\cosh\frac{\pi}{2}s}\Big(\frac{\,\sigma_0+\sigma\,}{\pi^2}\,\sinh\frac{\pi}{2}s-8\sinh\frac{\pi}{6}s\Big)\,.
  \end{equation}
  Observe that $\widehat{\gamma}_{-\sigma_0}$ is precisely the function $\widehat{\gamma}$ defined in \eqref{eq:gamma-distribution}. It is easy to see that $\widehat{\gamma}_\sigma$ is a smooth even function of $\mathbb{R}$, that converges asymptotically to 1 as $|s|\to +\infty$, and that for $\sigma\in[0,2\pi\sqrt{3})$ has \emph{absolute} minimum at $s=0$ of magnitude
  \begin{equation}\label{eq:gammasigmazero}
   \widehat{\gamma}_\sigma(0)\;=\;\widehat{\gamma}(0)+\frac{\,\sigma_0+\sigma}{\,2\pi\sqrt{3}\,}\;=\;\frac{\,\sigma}{\,2\pi\sqrt{3}\,}
  \end{equation}
  (see Figure \ref{fig:gammasigma}). Indeed, $\sigma_0=-2\pi\sqrt{3}\,\widehat{\gamma}(0)$.

 \begin{figure}[t!]
\includegraphics[width=8cm]{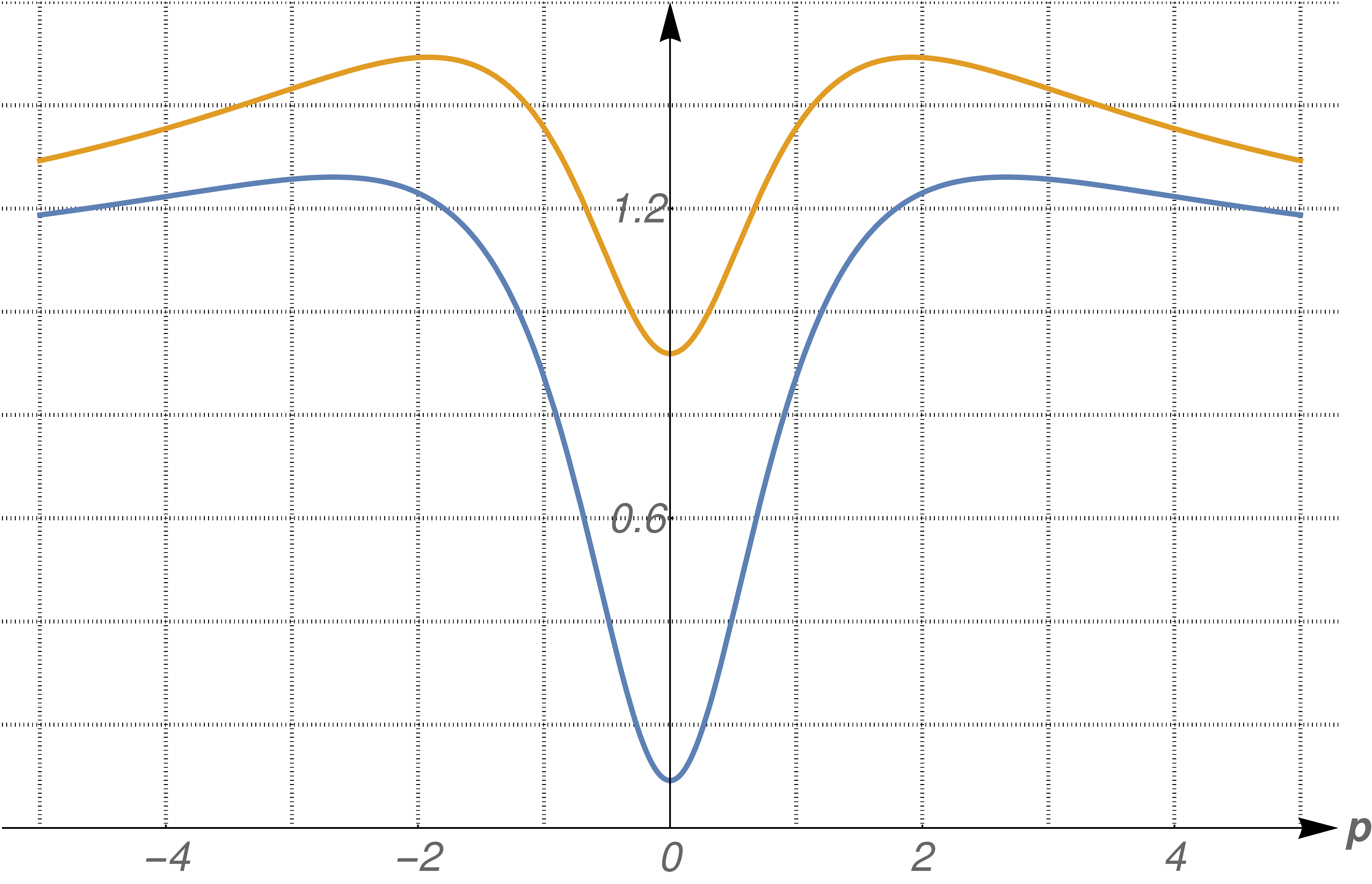}
\caption{Plot of the function $\widehat{\gamma}_\sigma(s)$ defined in \eqref{eq:gamma-s-distribution}, with parameter $\sigma=1$ (blue) and $\sigma=10$ (orange).}\label{fig:gammasigma}
\end{figure}

   Lemma \ref{lem:xithetaidentities} has the following analogue.

  \begin{lemma}\label{lem:regularisedlemma1}
   Let $\lambda>0$, $s,\sigma\in\mathbb{R}$, and $\xi$ be as in \eqref{eq:0xi}. One has the identities
         \begin{eqnarray}
         \big(\big(T_\lambda^{(0)}+K_\sigma^{(0)}\big)\xi\big)\,{\textrm{\large $\widehat{\,}$\normalsize}}\,(\pp)\!\!&=&\!\!\frac{1}{\sqrt{4\pi}\,|\pp|}\,\frac{4\pi^2}{\sqrt{3}\,}\,\Theta^{(\theta)}_\sigma(x)\,, \label{eq:Tlxi-theta-withK}\\
      \big\|\big(T_\lambda^{(0)}+K_\sigma^{(0)}\big)\xi\big\|_{H^{s}(\mathbb{R}^3)}^2\!\!&\approx&\!\!\int_{\mathbb{R}}\ud x\,(\cosh x)^{1+2s}\,\big|\Theta^{(\theta)}_\sigma(x)\big|^2\,, \label{eq:Tlambda0xi-with-theta-andK} \\
      \widehat{\Theta^{(\theta)}_\sigma}(s)\!\!&=&\!\!\widehat{\gamma}_\sigma(s)\,\widehat{\theta}(s)\,, \label{eq:ThetaTransformed} \\
       \int_{\mathbb{R}^3} \overline{\,\widehat{\xi}(\pp)}\, \big(\big(T_\lambda^{(0)}+K_\sigma^{(0)}\big)\xi\big)\,{\textrm{\large $\widehat{\,}$\normalsize}}\,(\pp)\,\ud\pp\!\!&=&\!\!\frac{\,8\pi^2}{3\sqrt{3}}\int_{\mathbb{R}}\ud s\, \widehat{\gamma}_\sigma(s)\,|\widehat{\theta}(s)|^2\,, \label{eq:xiTxi-with-theta-andK}
    \end{eqnarray}
    with $x$ and $\theta$ given by \eqref{ftheta-1}, $\Theta^{(\theta)}_{\sigma}$ given by \eqref{eq:bigTheta} with respect to the present $\theta$, and $\widehat{\gamma}_\sigma$ given by \eqref{eq:gamma-s-distribution}.
    In \eqref{eq:Tlxi-theta-withK} it is understood that $x\geqslant 0$, and \eqref{eq:Tlambda0xi-with-theta-andK} is meant as an equivalence of norms (with $\lambda$-dependent multiplicative constant).   
  \end{lemma}

  \begin{proof}
   Let us focus on the $K_\sigma$-term. From \eqref{eq:Ksigmamomentum}, in complete analogy with \eqref{eq:fellsector},
   \[
    (\widehat{K_\sigma^{(\ell)}\xi^{(\ell)}})(\pp)\;=\;\frac{\,\sigma_0+\sigma\,}{\pi}\sum_{n=-\ell}^\ell Y_{\ell,n}(\Omega_{\pp})\int_{\mathbb{R}^+}\ud q\,q^2\,f_{\ell,n}^{(\xi)}(q)\int_{-1}^1\ud t\,\frac{P_\ell(t)}{\,\pp^2+q^2-2|\pp|qt\,}\,,
   \]
   where for $\xi$ we took the expansion \eqref{eq:xihatangularexpansion}, namely
   \[
  \widehat{\xi}(\pp)\;=\;\sum_{\ell=0}^\infty\sum_{n=-\ell}^\ell f_{\ell,n}^{(\xi)}(|\pp|) Y_{\ell,n}(\Omega_{\pp})\;=\;\sum_{\ell=0}^\infty\widehat{\xi^{(\ell)}}(\pp)\,.
\]
  Thus, for a spherically symmetric charge $\widehat{\xi}(\pp)=\frac{1}{\sqrt{4\pi}}f(|\pp|)$,
   \[
    (\widehat{K_\sigma^{(0)}\xi})(\pp)\;=\;\frac{1}{\sqrt{4\pi}}\,\frac{1}{p}\,\frac{\,\sigma_0+\sigma\,}{\pi}\int_{\mathbb{R}^+}\ud q\,qf(q)\,\log\frac{p+q}{\,|p-q|\,}\,,\qquad p=|\pp|\,.
   \]
  Let $\theta$ be the re-scaled radial function \eqref{ftheta-1} associated with $\xi$ (and $f$) above. With $p=\frac{2\sqrt{\lambda}}{\sqrt{3}}\,\sinh x$ and $q=\frac{2\sqrt{\lambda}}{\sqrt{3}}\,\sinh y$ one has
  \[
   \frac{p+q}{\,|p-q|\,}\;=\;\frac{\sinh\frac{x+y}{2}\,\cosh\frac{x-y}{2}}{\,|\sinh\frac{x-y}{2}\,\cosh\frac{x+y}{2}|\,}\;=\;\Big|\coth\frac{x-y}{2}\Big|\,\tanh\frac{x+y}{2}\,,
  \]
  which, together with \eqref{ftheta-2}, gives
  \[
   \begin{split}
    \int_{\mathbb{R}^+}\ud q\,qf(q)\,\log\frac{p+q}{\,|p-q|\,}\;&=\;\frac{4}{3}\int_{\mathbb{R}^+}\ud y\,\theta(y)\Big(\log\Big|\tanh\frac{x+y}{2}\Big|+\log\Big|\coth\frac{x-y}{2}\Big|\Big) \\
    &=\;\frac{4}{3}\int_{\mathbb{R}}\ud y\,\theta(y)\,\log\Big|\coth\frac{x-y}{2}\Big|\,,
   \end{split}
  \]
  the last identity following from the odd-parity extension of $\theta$ over $\mathbb{R}$. Therefore,
  \[
    (\widehat{K_\sigma^{(0)}\xi})(\pp)\;=\;\frac{1}{\sqrt{4\pi}}\,\frac{1}{|\pp|}\,\frac{\,4(\sigma_0+\sigma)\,}{3\pi}\int_{\mathbb{R}}\ud y\,\theta(y)\,\log\Big|\coth\frac{x-y}{2}\Big|\,.
   \]
   Combining this with \eqref{eq:Tlxi-theta} yields \eqref{eq:Tlxi-theta-withK}.

  Formula \eqref{eq:Tlambda0xi-with-theta-andK} is proved from \eqref{eq:Tlxi-theta-withK} with the very same reasoning used for the analogous formula \eqref{eq:Tlambda0xi-with-theta} in Lemma \ref{lem:xithetaidentities}.

  Concerning \eqref{eq:ThetaTransformed}, we observe that the first two summands in the expression \eqref{eq:bigTheta} of $\Theta^{(\theta)}_\sigma$ have precisely Fourier transform $\widehat{\gamma}(s)\widehat{\theta}(s)$, as determined already in the proof of Lemma \ref{lem:xithetaidentities}, with $\widehat{\gamma}$ defined in \eqref{eq:gamma-distribution}. Let us focus on the third summand, namely the one with the $\sigma$-dependent pre-factor. One has
  \[
   \begin{split}
    \Big(\log\Big|\coth\frac{x}{2}\Big|\,\Big){\!\!\textrm{\huge ${\,}^{\widehat{\,}}$\normalsize}}\,(s)\;&=\;\sqrt{\frac{2}{\pi}}\int_{\mathbb{R}^+}\ud x\,\cos sx\,\log\coth\frac{x}{2} \\
    &=\;\sqrt{\frac{2}{\pi}}\int_{\mathbb{R}^+}\ud x\,\cos sx\,\big(\log(1+e^{-x})-\log(1-e^{-x})\big)\,,
   \end{split}
  \]
 and  (see, e.g., \cite[I.1.5.(13)-(14)]{Erdelyi-Tables1})
 \[
  \begin{split}
   \int_{\mathbb{R}^+}\ud x\,\cos sx\,\log(1+e^{-x})\;&=\;\frac{1}{\,2s^2}-\frac{\pi}{\,2s}\,\frac{1}{\sinh \pi s}\,, \\
   \int_{\mathbb{R}^+}\ud x\,\cos sx\,\log(1-e^{-x})\;&=\;\frac{1}{\,2s^2}-\frac{\pi}{\,2s}\,\coth \pi s\,,
  \end{split}
 \]
 whence
 \[
  \Big(\log\Big|\coth\frac{x}{2}\Big|\,\Big){\!\!\textrm{\huge ${\,}^{\widehat{\,}}$\normalsize}}\,(s)\;=\;\sqrt{\frac{\pi}{2}}\,\frac{\,\tanh\frac{\pi}{2}s}{s}\,.
 \]
 Therefore, the Fourier transform of the last summand in the expression \eqref{eq:bigTheta} of $\Theta^{(\theta)}_\sigma$ is
 \[
 \frac{\,\sigma_0+\sigma\,}{\pi^3\sqrt{3}}\,\Big(\theta*\log\Big|\coth\frac{x}{2}\Big|\Big){\!\!\textrm{\huge ${\,}^{\widehat{\,}}$\normalsize}}\,(s)\;=\;\frac{\,\sigma_0+\sigma\,}{\,\pi^2\sqrt{3}\,}\,\frac{\,\tanh\frac{\pi}{2}s}{s}\:\widehat{\theta}(s)\,.
 \]
 Adding this term to $\widehat{\gamma}(s)\widehat{\theta}(s)$ yields \eqref{eq:gamma-s-distribution} and hence \eqref{eq:ThetaTransformed}.

 Last, concerning \eqref{eq:xiTxi-with-theta-andK}, in complete analogy with the proof of \eqref{eq:xiTxi-with-theta}, we find
  \[
   \begin{split}
    \int_{\mathbb{R}^3} \overline{\,\widehat{\xi}(\pp)}\, \big(\big(T_\lambda^{(0)}+K_\sigma^{(0)}\big)\xi\big)\,{\textrm{\large $\widehat{\,}$\normalsize}}\,(\pp)\,\ud\pp\;&=\;\int_{\mathbb{R}^+}\ud p\,p^2\,\overline{f(p)}\,\frac{4\pi^2}{\,p\sqrt{3}\,}\,\Theta^{(\theta)}_\sigma(x(p)) \\
    &=\;\frac{8\pi^2}{3\sqrt{3}\,}\int_{\mathbb{R}}\ud x\,\overline{\theta(x)}\,\Theta^{(\theta)}_\sigma(x) \\
    &=\;\frac{\,8\pi^2}{3\sqrt{3}}\int_{\mathbb{R}}\ud s\, \widehat{\gamma}_\sigma(s)\,|\widehat{\theta}(s)|^2
   \end{split}
  \]
  having applied \eqref{eq:Tlxi-theta-withK} in the first identity, \eqref{ftheta-2}-\eqref{ftheta-3-eq:pxchangevar} and the odd parity in the second, and Parseval's identity in the third.
  \end{proof}

  \begin{corollary}\label{cor:equivalentH12norm}
   For $\lambda,\sigma>0$, the map
   \[
    \xi\;\longmapsto\;\bigg(\int_{\mathbb{R}^3} \overline{\,\widehat{\xi}(\pp)}\, \big(\big(T_\lambda^{(0)}+K_\sigma^{(0)}\big)\xi\big)\,{\textrm{\large $\widehat{\,}$\normalsize}}\,(\pp)\,\ud\pp\bigg)^{\!\frac{1}{2}}
   \]
   defines an equivalent norm in $H^{\frac{1}{2}}_{\ell=0}(\mathbb{R}^3)$.
  \end{corollary}

  \begin{proof}
   Because of \eqref{eq:xiTxi-with-theta-andK} and the fact that $\widehat{\gamma}_\sigma$ is uniformly bounded and strictly positive when $\sigma>0$, up to inessential pre-factors each $\xi$ is mapped to $\|\theta^{(\xi)}\|_{L^2(\mathbb{R})}$ and hence to $\|\xi\|_{H^{\frac{1}{2}}(\mathbb{R}^3)}$, owing to \eqref{eq:mellinnorms}.   
  \end{proof}

 \subsection{Regularisation in the sector $\ell=0$}\label{sec:MinFadzero}~

 Let us now make the Birman parameter explicit for a self-adjoint extension of $\mathring{H}$ when the Minlos-Faddeev regularisation is implemented with respect to the canonical construction of extensions of Ter-Martirosyan Skornyakov type. As argued in the previous Subsection, we only need to replace the Birman parameter $\mathcal{A}_{\lambda,\beta}^{(0)}$ (defined in \eqref{eq:Azerolambda}-\eqref{eq:domainD0beta}) of the sector $\ell=0$ with a \emph{regularised} version, that we shall denote by $\mathcal{R}_{\lambda,\sigma}^{(0)}$.

 Then, with such $\mathcal{R}_{\lambda,\sigma}^{(0)}$ at hand, and with all other canonical Birman parameters $\mathcal{A}_\lambda^{(\ell)}$, $\ell\in\mathbb{N}$, we construct the associated self-adjoint Hamiltonian by means of Theorem \ref{thm:generalclassification}.

 To begin with, for $\sigma>0$ we define (in analogy to \eqref{eq:Dtilde0}) the subspace
  \begin{equation}\label{eq:newD0regularised}
   \widetilde{\mathsf{D}}_{0,\sigma}\;:=\;\left\{
   \xi\in H^{-\frac{1}{2}}_{\ell=0}(\mathbb{R}^3)\left|
   \begin{array}{c}
    \textrm{$\xi$ has re-scaled radial component} \\
     \theta=\big(\widehat{\Theta}/\widehat{\gamma}_\sigma)\big)^{\!\vee} \\
     \textrm{for }\;\Theta\in C^\infty_{0,\mathrm{odd}}(\mathbb{R}_x)
   \end{array}
   \!\right.\right\}\,.
  \end{equation}
  Here the subscript `odd' indicates functions with odd parity and $\widehat{\gamma}_\sigma$ is defined in \eqref{eq:gamma-s-distribution}.  The correspondence between $\xi$ and its re-scaled radial component $\theta$ is given by \eqref{eq:0xi}-\eqref{ftheta-1}. It is tacitly understood that the re-scaled radial components are all taken with the same parameter $\lambda>0$ in the definition \eqref{ftheta-1}: this does not mean $\widetilde{\mathsf{D}}_{0,\sigma}$ is a $\lambda$-dependent subspace, as one can easily convince oneself, the choice of $\lambda$ only fixes the convention for representing its elements in terms of the corresponding $\theta$.

  \begin{lemma}\label{lem:DD0properties} Let $\lambda,\sigma>0$.
  \begin{itemize}
   \item[(i)] $\widetilde{\mathsf{D}}_{0,\sigma}$ is dense in $H^{\frac{1}{2}}_{\ell=0}(\mathbb{R}^3)$.
   \item[(ii)] $(T_\lambda^{(0)}+K_\sigma^{(0)})\widetilde{\mathsf{D}}_{0,\sigma}\subset H^{s}_{\ell=0}(\mathbb{R}^3)$ for every $s\in\mathbb{R}$.
  \end{itemize}
  \end{lemma}

  \begin{proof} (i) For $\xi\in  \widetilde{\mathsf{D}}_{0,\sigma}$, the identity $\widehat{\gamma}_\sigma\widehat{\theta}=\widehat{\Theta}$ and \eqref{eq:mellinnorms} imply
  \[
  \begin{split}
   \|\xi\|_{H^{\frac{1}{2}}(\mathbb{R}^3)}\;\approx\;\|\theta\|_{L^2(\mathbb{R})}\;&=\;\|\widehat{\theta}\|_{L^2(\mathbb{R})}\;\leqslant\;\|\widehat{\gamma}_\sigma^{-1}\|_{L^\infty(\mathbb{R})}\|\widehat{\Theta}\|_{L^2(\mathbb{R})} \\
   &=\;(\widehat{\gamma}_\sigma(0))^{-1}\|\Theta\|_{L^2(\mathbb{R})}\;<\;+\infty\,,
  \end{split}
  \]
  owing to \eqref{eq:gammasigmazero} and to the fact that $\Theta$ is smooth and with compact support.
  Because of \eqref{eq:mellinnorms}, the density of the $\xi$'s of $\widetilde{\mathsf{D}}_{0,\sigma}$ in $H^{\frac{1}{2}}_{\ell=0}(\mathbb{R}^3)$ is equivalent to the density of the associated $\theta$'s in $L^2_{\mathrm{odd}}(\mathbb{R})$. If in the latter Hilbert space a function $\theta_0$ was orthogonal to all such $\theta$'s, then
  \[
   0\;=\;\int_{\mathbb{R}}\overline{\theta_0(x)}\,\theta(x)\,\ud x\;=\;\int_{\mathbb{R}}\overline{\widehat{\theta}_0(s)}\,\frac{\widehat{\Theta}(s)}{\widehat{\gamma}_\sigma(s)}\,\ud s\;=\;\int_{\mathbb{R}}\overline{\big(\widehat{\theta}_0/\widehat{\gamma}_\sigma\big)^{\!\vee}\!(x)}\,\Theta(x)\,\ud x
  \]
  for all $\Theta\in C^\infty_{0,\mathrm{odd}}(\mathbb{R})$: as $\widehat{\gamma}_\sigma$ is uniformly bounded and strictly positive, this implies $\theta_0\equiv 0$.
  
  (ii) On the one hand $\widehat{\gamma}_\sigma\widehat{\theta}=\widehat{\Theta}$ by the assumption that $\xi\in  \widetilde{\mathsf{D}}_{0,\sigma}$, on the other hand $\widehat{\gamma}_\sigma\widehat{\theta}=\widehat{\Theta^{(\theta)}_\sigma}$, owing to \eqref{eq:ThetaTransformed}, whence $\Theta^{(\theta)}_\sigma=\Theta\in C^\infty_{0,\mathrm{odd}}(\mathbb{R})$. Plugging this information into \eqref{eq:Tlambda0xi-with-theta-andK} yields the conclusion.   
  \end{proof}

  Next, for $\lambda,\sigma>0$ we define
  \begin{equation}
   \begin{split}
    \mathcal{D}\big(\widetilde{\mathcal{R}}_{\lambda,\sigma}^{(0)}\big)\;&:=\;\widetilde{\mathsf{D}}_{0,\sigma} \\
    \widetilde{\mathcal{R}}^{(0)}_{\lambda,\sigma}\;&:=\; 3W_\lambda^{-1}\big(T_\lambda^{(0)}+K_\sigma^{(0)}\big)\,.
   \end{split}
  \end{equation}
 Lemma \ref{lem:DD0properties}, together with Corollary \ref{cor:equivalentH12norm}, guarantees that this is a well-posed definition for a densely defined, symmetric, and coercive operator on $H^{-\frac{1}{2}}_{W_\lambda,\ell=0}(\mathbb{R}^3)$. As such, having a strictly positive lower bound, $\widetilde{\mathcal{R}}^{(0)}_{\lambda,\sigma}$ has the Friedrichs extension. That will be our final Birman parameter. In analogy with Proposition \ref{prop:Alambdaellnot0}, we prove the following.

 \begin{proposition}\label{prop:NEW-Birman-param-selfadj}
  Let $\lambda,\sigma>0$. Define
  \begin{equation}
   \mathsf{D}_{0,\sigma}\;:=\;\big\{\xi\in H^{\frac{1}{2}}_{\ell=0}(\mathbb{R}^3)\,\big|\,\big(T_\lambda^{(0)}+K_\sigma^{(0)}\big)\xi\in H^{\frac{1}{2}}_{\ell=0}(\mathbb{R}^3)\big\}\,.
  \end{equation}
  The operator
  \begin{equation}\label{eq:RlambdasigmaOPERATOR}
   \begin{split}
    \mathcal{D}\big(\mathcal{R}_{\lambda,\sigma}^{(0)}\big)\;&:=\;\mathsf{D}_{0,\sigma} \\
    \mathcal{R}_{\lambda,\sigma}^{(0)})\;&:=\;3W_\lambda^{-1}\big(T_\lambda^{(0)}+K_\sigma^{(0)}\big)
   \end{split}
  \end{equation}
 is the Friedrichs extension of $\widetilde{\mathcal{R}}^{(0)}_{\lambda,\sigma}$ with respect to $H^{-\frac{1}{2}}_{W_\lambda,\ell=0}(\mathbb{R}^3)$. Its sesquilinear form is
 \begin{equation}\label{eq:RlambdasigmaFORM}
   \begin{split}
    \mathcal{D}\big[\mathcal{R}_{\lambda,\sigma}^{(0)}\big]\;&=\;H^{\frac{1}{2}}_{\ell=0}(\mathbb{R}^3) \\
    \mathcal{R}_{\lambda,\sigma}^{(0)}[\eta,\xi]\;&=\;3\int_{\mathbb{R}^3} \overline{\,\widehat{\xi}(\pp)}\, \big(\big(T_\lambda^{(0)}+K_{\sigma}^{(0)}\big)\xi\big)\,{\textrm{\large $\widehat{\,}$\normalsize}}\,(\pp)\,.
   \end{split}
\end{equation}
 \end{proposition}

 \begin{proof}
  Let us temporarily denote by $\mathcal{R}_F$ the Friedrichs extension of $\widetilde{\mathcal{R}}^{(0)}_{\lambda,\sigma}$ with respect to $H^{-\frac{1}{2}}_{W_\lambda,\ell=0}(\mathbb{R}^3)$, and let us set
  \[
   \|\xi\|_{\mathcal{R}}\;:=\; \Big(\big\langle\xi,\widetilde{\mathcal{R}}^{(0)}_{\lambda,\sigma}\xi\big\rangle_{H^{-\frac{1}{2}}_{W_\lambda}}\Big)^{\frac{1}{2}}\;=\;\big(3\big\langle \xi,\big(T_\lambda^{(0)}+K_\sigma^{(0)}\big)\xi\big\rangle_{L^2}\big)^{\frac{1}{2}}\,.
  \]
  Owing to Corollary \ref{cor:equivalentH12norm}, the latter induces an equivalent $H^{\frac{1}{2}}$-norm on $\widetilde{\mathsf{D}}_{0,\sigma} $. As prescribed by the Friedrichs construction, $\mathcal{R}_F$ has form domain
  \[
   \mathcal{D}[\mathcal{R}_F]\;=\;\overline{\mathcal{D}\big(\widetilde{\mathcal{R}}^{(0)}_{\lambda,\sigma}\big)\,}^{\|\,\|_{\mathcal{R}}}\;=\;\overline{\widetilde{\,\mathsf{D}}_{0,\sigma}\,}^{\|\,\|_{H^{\frac{1}{2}}}}\;=\;H^{\frac{1}{2}}_{\ell=0}(\mathbb{R}^3)
  \]
  (the last identity following from Lemma \ref{lem:DD0properties}(i)), and for $\xi,\eta\in H^{\frac{1}{2}}_{\ell=0}(\mathbb{R}^3)$
  \[
   \mathcal{R}_F[\eta,\xi]\;=\;\lim_{n\to\infty}\big\langle\eta_n,\widetilde{\mathcal{R}}^{(0)}_{\lambda,\sigma}\xi_n\big\rangle_{H^{-\frac{1}{2}}_{W_\lambda}}\;=\;3\lim_{n\to\infty}\big\langle \eta_n,\big(T_\lambda^{(0)}+K_\sigma^{(0)}\big)\xi_n\big\rangle_{L^2}
  \]
  for any two sequences $(\xi_n)_n$ and $(\eta_n)_n$ in $\widetilde{\mathsf{D}}_{0,\sigma} $ such that $\xi_n\to \xi$ and $\eta_n\to\eta$ in $H^{\frac{1}{2}}_{\ell=0}(\mathbb{R}^3)$. Since the pairing $(\eta_n,\xi_n)\mapsto \big\langle \eta_n,\big(T_\lambda^{(0)}+K_\sigma^{(0)}\big)\xi_n\big\rangle_{L^2}$ is an $H^{\frac{1}{2}}$-pairing (Corollary \ref{cor:equivalentH12norm}), one finds
  \[
   \mathcal{R}_F[\eta,\xi]\;=\;3\lim_{n\to\infty}\big\langle \eta_n,\big(T_\lambda^{(0)}+K_\sigma^{(0)}\big)\xi_n\big\rangle_{L^2}\;=\;3\,\big\langle \eta,\big(T_\lambda^{(0)}+K_\sigma^{(0)}\big)\xi\big\rangle_{L^2}\,.
  \]
  Formula \eqref{eq:RlambdasigmaFORM} is thus proved. The operator $\mathcal{R}_F$ is derived from its quadratic form in the usual matter: a straightforward adaptation of the analogous argument used in the proof of Proposition \ref{prop:Alambdaellnot0} shows that $\mathcal{R}_F$ is indeed the operator \eqref{eq:RlambdasigmaOPERATOR}. 
 \end{proof}

 With the new Birman parameter \eqref{eq:RlambdasigmaOPERATOR} for the sector $\ell=0$, the construction of the canonical model $\mathscr{H}_{0,\beta}$ is modified as follows (see the discussion of Subsect.~\ref{sec:constructioncanonical}).

 Instead of the self-adjoint extension $\mathring{H}_{\mathcal{A}_{\lambda,\beta}}$ obtained by means of Theorem \ref{thm:generalclassification} with Birman parameter
  \begin{equation*}
   \mathcal{A}_{\lambda,\beta}\;=\; \mathcal{A}_{\lambda,\beta}^{(0)}\:\oplus\: \bigoplus_{\ell=1}^\infty\mathcal{A}_\lambda^{(\ell)}
  \end{equation*}
 (see \eqref{eq:globalAlambda-1}), we consider another operator from the family \eqref{eq:family}, namely the self-adjoint extension $\mathring{H}_{\mathcal{R}_{\lambda,\sigma}}$ with modified Birman parameter
 \begin{equation}\label{eq:globalRlambda-1}
   \mathcal{R}_{\lambda,\sigma}\;:=\;\mathcal{R}_{\lambda,\sigma}^{(0)}\:\oplus\: \bigoplus_{\ell=1}^\infty\mathcal{A}_\lambda^{(\ell)}\,.
 \end{equation}
 By definition, the domain of $\mathcal{R}_{\lambda,\sigma}$ is
 \begin{equation}
  \mathcal{D}(\mathcal{R}_{\lambda,\sigma})\;=\;\left\{ 
   \begin{array}{c}
    \displaystyle\xi=\sum_{\ell=0}^\infty\xi^{(\ell)}\in\;\bigoplus_{\ell=0}^\infty \,H^{-\frac{1}{2}}_{W_\lambda,\ell}(\mathbb{R}^3)\,\cong\, H^{-\frac{1}{2}}_{W_\lambda}(\mathbb{R}^3) \\
    \textrm{such that} \\
    \xi^{(\ell)}\in H^{\frac{1}{2}}_\ell(\mathbb{R}^3)\;\textrm{ and }\; T_\lambda^{(\ell)}\xi^{(\ell)}\in H^{\frac{1}{2}}_\ell(\mathbb{R}^3)\;\textrm{ for }\;\ell\in\mathbb{N}\,, \\
    \xi^{(0)}\in H_{\ell=0}^{\frac{1}{2}}(\mathbb{R}^3)\;\textrm{ and }\;\big(T_\lambda^{(0)}+K_\sigma^{(0)}\big)\xi\in H^{\frac{1}{2}}_{\ell=0}(\mathbb{R}^3)
   \end{array}
   \right\},
 \end{equation}
 to be compared with the previous domain \eqref{eq:Dbetaaltogether}.

 For the sake of a more compact, unified notation, let us write
 \begin{equation}\label{eq:compactTTT}
  (\mathbf{T}_{\lambda,\sigma}\xi)(\yy)\;:=\;\big(\big(T_\lambda^{(0)}+K_\sigma^{(0)}\big)\xi^{(0)})(\yy)+\sum_{\ell=1}^\infty \big(T_\lambda^{(\ell)}\xi^{(\ell)}\big)(\yy)
 \end{equation}
 for $\xi=\sum_{\ell=0}^\infty\xi^{(\ell)}\in\bigoplus_{\ell=0}^\infty \,H^{-\frac{1}{2}}_{W_\lambda,\ell}(\mathbb{R}^3)\cong H^{-\frac{1}{2}}_{W_\lambda}(\mathbb{R}^3)$. Thus,
 \begin{equation}\label{eq:Tcompact1}
  \begin{split}
   \mathcal{D}(\mathcal{R}_{\lambda,\sigma})\;&=\;\big\{\xi\in H^{\frac{1}{2}}(\mathbb{R}^3)\,\big|\,\mathbf{T}_{\lambda,\sigma}\xi\in H^{\frac{1}{2}}(\mathbb{R}^3)\big\} \\
   \mathcal{R}_{\lambda,\sigma}\;&=\;3\,W_\lambda^{-1} \mathbf{T}_{\lambda,\sigma}
  \end{split}
 \end{equation}
 and 
 \begin{equation}\label{eq:Tcompact2}
  \big(\mathbf{T}_{\lambda,\sigma}-T_\lambda\big)\xi\;=\; K_\sigma^{(0)}\xi^{(0)}\,.
 \end{equation}

%
%

 \begin{theorem}\label{thm:regularised-models}
  Let $\sigma>0$ and $\lambda>0$. Define
 \begin{equation}\label{eq:Hsigmadomaction}
     \begin{split}
      \mathcal{D}(\mathscr{H}_\sigma)\;&:=\;\left\{g=\phi^\lambda+u_\xi^\lambda\left|\!
  \begin{array}{c}
   \phi^\lambda\in H^2_\mathrm{b}(\mathbb{R}^3\times\mathbb{R}^3)\,, \\
   \xi\in H^{\frac{1}{2}}(\mathbb{R}^3)\,\textrm{ with }\,\mathbf{T}_{\lambda,\sigma}\xi\in H^{\frac{1}{2}}(\mathbb{R}^3)\,,\\
   \displaystyle\phi^\lambda(\yy,\mathbf{0})\,=\,(2\pi)^{-\frac{3}{2}} (\mathbf{T}_{\lambda,\sigma}\xi)(\yy)
  \end{array}
  \!\!\!\right.\right\} \\
  (\mathscr{H}_\sigma+\lambda\mathbbm{1})g\;&:=\;(-\Delta_{\yy_1}-\Delta_{\yy_2}-\nabla_{\yy_1}\cdot\nabla_{\yy_2}+\lambda\mathbbm{1})\phi^\lambda\,,
     \end{split}
    \end{equation}
    \begin{itemize}
   \item[(i)] The decomposition of $g$ in terms of $\phi^\lambda$ and $\xi$ is unique, at fixed $\lambda$. The subspace $\mathcal{D}(\mathscr{H}_\sigma)$ is $\lambda$-independent.
   \item[(ii)] $\mathscr{H}_\sigma$ is self-adjoint on $L^2_\mathrm{b}(\mathbb{R}^3\times\mathbb{R}^3,\ud\yy_1,\ud\yy_2)$ and extends $\mathring{H}$ given in \eqref{eq:domHring-initial}.
   \item[(iii)] For each $g\in  \mathcal{D}(\mathscr{H}_\sigma)$ one has
   \begin{equation}\label{eq:allBPTMS-regularised}
    \begin{split}
     \phi^\lambda(\yy_1,\mathbf{0})\;&=\;(2\pi)^{-\frac{3}{2}} (\mathbf{T}_{\lambda,\sigma}\xi)(\yy_1)\,, \\
     \int_{\mathbb{R}^3}\widehat{\phi^\lambda}(\pp_1,\pp_2)\,\ud\pp_2\;&=\;(\widehat{\mathbf{T}_{\lambda,\sigma}\xi})(\pp_1)\,, \\
     \int_{\!\substack{ \\ \\ \pp_2\in\mathbb{R}^3 \\ |\pp_2|<R}}\widehat{g}(\pp_1,\pp_2)\,\ud\pp_2\;&=\;4\pi R\,\widehat{\xi}(\pp_1)+(\widehat{K_\sigma^{(0)}\xi^{(0)}})(\pp_1)+o(1)
    \end{split}
   \end{equation}
   (where $R\to +\infty$). All such conditions are equivalent. In particular, the first version of \eqref{eq:allBPTMS-regularised} is an identity in $H^{\frac{1}{2}}(\mathbb{R}^3)$. 
  \item[(iv)] $\mathscr{H}_\sigma$ is non-negative and with quadratic form
  \begin{equation}\label{eq:HsigmaQuadrForm}
   \begin{split}
    \mathcal{D}[\mathscr{H}_\sigma]\;&=\;\left\{g=\phi^\lambda+u_\xi^\lambda\left|\!
  \begin{array}{c}
   \phi^\lambda\in H^1_\mathrm{b}(\mathbb{R}^3\times\mathbb{R}^3)\,,\;\xi\in H^{\frac{1}{2}}(\mathbb{R}^3)
  \end{array}
  \!\!\!\right.\right\} \\
     \mathscr{H}_\sigma[g]\;&=\;\frac{1}{2}\Big(\big\|(\nabla_{\yy_1}+\nabla_{\yy_2})\phi^\lambda\big\|^2_{L^2}+\big\|\nabla_{\yy_1}\phi^\lambda\big\|^2_{L^2}+\big\|\nabla_{\yy_2}\phi^\lambda\big\|^2_{L^2}\Big) \\
     &\qquad+\lambda\Big(\|\phi^\lambda\big\|^2_{L^2}-\big\|\phi^\lambda+u_\xi^\lambda\big\|^2_{L^2}\Big)+3\int_{\mathbb{R}^3} \overline{\,\widehat{\xi}(\pp)}\, \big(\widehat{\mathbf{T}_{\lambda,\sigma}\xi}\big)(\pp)\,\ud\pp\,,
   \end{split}
  \end{equation}
  the $L^2$-norms being norms in $L^2(\mathbb{R}^3\times\mathbb{R}^3)$.
  \end{itemize}
 \end{theorem}

 \begin{proof}
  All claims follow from plugging \eqref{eq:Tcompact1}-\eqref{eq:Tcompact2} (hence, in particular, \eqref{eq:RlambdasigmaOPERATOR}-\eqref{eq:RlambdasigmaFORM}) into the general classification formulas of Theorem \ref{thm:generalclassification}, owing to the self-adjointness of the Birman parameter guaranteed by Proposition \ref{prop:NEW-Birman-param-selfadj}.  
 \end{proof}

 Each of the asymptotics \eqref{eq:allBPTMS-regularised} for $g\in \mathcal{D}(\mathscr{H}_\sigma)$ expresses an ultra-violet regularised Bethe-Peierls alias Ter-Martirosyan Skornyakov condition, that in view of Corollary \ref{cor:largepasympt-star} can be thought of as
 \begin{equation}\label{eq:allBPTMS-regularised-2}
 (2\pi)^{\frac{3}{2}} g_{\mathrm{av}}(\yy_1;|\yy_2|)\,\stackrel{|\yy_2|\to 0}{=}\,\frac{4\pi}{|\yy_2|}\xi(\yy_1)+\frac{\,\sigma_0+\sigma\,}{|\yy_1|}\,\xi(\yy_1)+ o(1)\,.
 \end{equation}
 As the simultaneous limit $|\yy_1|\to 0$, $|\yy_2|\to 0$ in the expression above suggests, the regularisation effectively amounts to distorting the ordinary Bethe-Peierls short-scale asymptotics by means of a position-dependent scattering length
 \begin{equation}
  a_{\mathrm{eff}}(\yy)\;:=\;-\frac{4\pi|\yy|}{\,\sigma_0+\sigma\,}
 \end{equation}
 (see \eqref{eq:a-alpha} above) that vanishes when all three bosons come to occupy the same point. At this effective level, a (small) three-body correction prevents the triple collision.

 In fact, this Minlos-Faddeev regularisation is rather radical, because it completely eliminates the negative spectrum (see also Remark \ref{rem:noEV} below). Yet, $\mathscr{H}_\sigma$ is not merely the reduced component of the canonical Hamiltonian $\mathscr{H}_{0,\beta}$ onto the positive spectral subspace: the signature of the physical short-scale behaviour is retained in \eqref{eq:allBPTMS-regularised} and \eqref{eq:allBPTMS-regularised-2}, with the $|\yy_2|^{-1}$ leading singularity as $|\yy_2|\to 0$ at fixed $\yy_1$, only with a distorted subleading singularity driven by an effective scattering length that vanishes as $|\yy_1|\to 0$.

  \begin{remark}\label{rem:noEV}
  In contrast with the computation of the negative eigenvalues of the Hamiltonian $\mathscr{H}_{0,\beta}$ (proof of Theorem \ref{thm:spectralanalysis}), the analogous computation for $\mathscr{H}_\sigma$ would lead to the equation $\big(T_\lambda^{(0)}+K_\sigma^{(0)}\big)\xi=0$ for some $\xi$ in the sector $\ell=0$. Owing to \eqref{eq:Tlxi-theta-withK} and \eqref{eq:ThetaTransformed}, this is the same as $\widehat{\gamma}_\sigma\widehat{\theta}=0$, where $\theta$ is the re-scaled radial function associated with $\xi$. The difference is thus
  \[
   \begin{array}{lcl}
    \;\,\widehat{\gamma}\,\widehat{\theta}\;=\;0 & & \textrm{for the \emph{canonical} eigenvalue problem} \\
    \widehat{\gamma}_\sigma\,\widehat{\theta}\;=\;0 & & \textrm{for the \emph{regularised} eigenvalue problem}\,.
   \end{array}
  \]
  Because of the presence of roots of $\widehat{\gamma}(s)=0$, the first equation turns out to have non-trivial admissible solutions. Instead, $\widehat{\gamma}_\sigma(s)\geqslant\widehat{\gamma}_\sigma(0)>0$ and for the second equation one necessarily has $\theta\equiv 0$ (absence of negative eigenvalues).  
 \end{remark}

 \begin{remark}
 The charge term in the quadratic form expression \eqref{eq:HsigmaQuadrForm} is explicitly given by
 \begin{equation}
  \begin{split}
   \int_{\mathbb{R}^3} \overline{\,\widehat{\xi}(\pp)}\, \big(\widehat{\mathbf{T}_{\lambda,\sigma}\xi}\big)(\pp)\,\ud\pp\;&=\;\sum_{\ell=1}^\infty\int_{\mathbb{R}^3} \overline{\,\widehat{\xi^{(\ell)}}(\pp)}\, \big(\widehat{T_{\lambda}^{(\ell)}\xi^{(\ell)}}\big)(\pp)\,\ud\pp \\
   &\qquad\quad +\int_{\mathbb{R}^3} \overline{\,\widehat{\xi^{(0)}}(\pp)}\, \big(\big(T_\lambda^{(0)}+K_\sigma^{(0)}\big)\xi^{(0)}\big)\,{\textrm{\large $\widehat{\,}$\normalsize}}\,(\pp)\,\ud\pp
  \end{split}
 \end{equation}
 (as follows from \eqref{eq:RlambdasigmaFORM} and \eqref{eq:compactTTT}). As $\mathscr{H}_\sigma$ is self-adjoint and non-negative, its quadratic form \eqref{eq:HsigmaQuadrForm} is obviously closed and non-negative.
  An announcement that the quadratic form \eqref{eq:HsigmaQuadrForm} is closed and lower semi-bounded on $L^2_\mathrm{b}(\mathbb{R}^3\times\mathbb{R}^3)$, and thus induces a self-adjoint Hamiltonian for the regularised three-body interaction in the bosonic trimer, has been recently made in \cite{Figari-Teta-2020}. 
 \end{remark}

 \subsection{High energy cut-off}\label{sec:highenergycutoff}~

 As mentioned already, an alternative, conceptually equivalent way of making the scattering length effectively vanish in the vicinity of the triple coincidence point is to realise this effect at large relative momenta.
 
 An example of this type of high energy cut-off has been recently proposed by Basti, Figari, and Teta \cite[Sect.~3]{Basti-Figari-Teta-Rendiconti2018} by means of quadratic form methods. We shall study this possibility in the general operator-theoretic framework of the present analysis, working out more precisely a modification of \cite{Basti-Figari-Teta-Rendiconti2018} that allows for explicit computations.
 
  For a clearer readability and comparison with Subsect.~\ref{sec:MFregularisation}-\ref{sec:MinFadzero}, we shall keep the same notation used therein for the counterpart regularised quantities, of course re-defined them now in a different way.

 The original proposal of \cite{Basti-Figari-Teta-Rendiconti2018} goes along the following line: the ordinary Birman parameter $3W_\lambda^{-1}(T_\lambda+\alpha\mathbbm{1})$, that selects (via Theorems \ref{thm:generalclassification} and \ref{thm:globalTMSext}) self-adjoint extensions of Ter-Martirosyan Skornyakov type of the minimal operator $\mathring{H}$ defined in \eqref{eq:domHring-initial}, is replaced by 
\begin{equation*}
\begin{split}
 & 3\,W_\lambda^{-1}(T_\lambda+\alpha\mathbbm{1}+K_{\sigma,\rho})\,,\qquad \sigma,\rho\,>\,0\,,  \\
 & (\widehat{K_{\sigma,\rho}\xi})(\pp)\;:=\;\sigma\,\mathbf{1}_{\{|\pp|\geqslant\rho\}}\,\pp^2\,\widehat{\xi}(\pp)\,.
\end{split}
\end{equation*}
 This gives rise, via \eqref{eq:g-largep2-star}, to the modified large momentum asymptotics
 \begin{equation*}
 \begin{split}
  \int_{\!\substack{ \\ \\ \pp_2\in\mathbb{R}^3 \\ |\pp_2|<R}}\widehat{g}(\pp_1,\pp_2)&\,\ud\pp_2\stackrel{R\to+\infty}{=}
  4\pi R\,\widehat{\xi}(\pp_1)+\alpha_{\mathrm{eff}}(\pp)\widehat{\xi}(\pp_1)+o(1) \\
  & \alpha_{\mathrm{eff}}(\pp)\;:=\;\alpha+\sigma\,\mathbf{1}_{\{|\pp|\geqslant\rho\}}\,\pp^2\,,
 \end{split}
 \end{equation*}
 for the elements $g$ in the domain of the corresponding self-adjoint extension, 
 again with the interpretation of an effective parameter $\alpha_{\mathrm{eff}}(\pp)\to +\infty$ as $|\pp|\to +\infty$.

 Reasoning in terms of quadratic forms, it is simple to check (as done in \cite[Sect.~3]{Basti-Figari-Teta-Rendiconti2018}) that $\sigma$ and $\rho$ can be adjusted on $\alpha$ and $\lambda$ so that the map
 \[
  \xi\;\longmapsto\;\bigg(\int_{\mathbb{R}^3}\overline{\widehat{\xi}(\pp)}\,\big((T_\lambda+\alpha\mathbbm{1}+K_{\sigma,\rho})\xi\big)\,{\textrm{\large $\widehat{\,}$\normalsize}}\,(\pp)\bigg)^{\!\frac{1}{2}}
 \]
 is an equivalent $H^1$-norm: in fact, the effect of the additional $K_{\sigma,\rho}$ is to rise the multiplicative part of $T_\lambda$ with an $H^1$-term (added to the original $H^{\frac{1}{2}}$-term), which controls the integral part of $T_\lambda$. This way, a quadratic form on $L^2_{\mathrm{b}}(\mathbb{R}^3_{\yy_1}\times\mathbb{R}^3_{\yy_2})$ of the same type of \eqref{eq:decomposition_of_form_domains_Tversion} can be constructed, namely with regular functions from $H^1_{\mathrm{b}}(\mathbb{R}^3_{\yy_1}\times\mathbb{R}^3_{\yy_2})$ and charges from $H^1(\mathbb{R}^3)$, in which the charge term (the quadratic form of the Birman parameter) is precisely
 \[
  \int_{\mathbb{R}^3}\overline{\widehat{\xi}(\pp)}\,\big((T_\lambda+\alpha\mathbbm{1}+K_{\sigma,\rho})\xi\big)\,{\textrm{\large $\widehat{\,}$\normalsize}}\,(\pp)\,.
 \]
 Standard arguments then show that the form is closed and non-negative, hence it is the energy form of a self-adjoint Hamiltonian for the bosonic trimer with zero-range interactions.

 The Friedrichs construction for the Birman parameter $3\,W_\lambda^{-1}(T_\lambda+\alpha\mathbbm{1}+K_{\sigma,\rho})$ associated to the above form would lead to a somewhat implicit expression for its operator domain.
 Furthermore, for the purpose of having a closed and semi-bounded charge form it suffices to add to $T_\lambda$ an additional  $H^{\frac{1}{2}}$-term, instead of the $H^1$-term proposed in \cite{Basti-Figari-Teta-Rendiconti2018}, which be large enough so as to shift the form up above zero. We shall study this version of the large momentum (high energy) cut-off, as the computations within the operator-theoretic scheme are explicit.

 For $\lambda,\sigma>0$ let us set 
 \begin{equation}\label{eq:heregularisationchoice}
 \begin{split}
  (\widehat{K_{\lambda,\sigma}\,\xi})(\pp)\;&:=\;2\pi^2(\sigma_0+\sigma)\sqrt{\frac{3}{4}\pp^2+\lambda}\,\widehat{\xi}(\pp) \\
  \sigma_0\;&:=\;-\widehat{\gamma}(0)\;=\;\frac{4\pi}{3\sqrt{3}}-1\,,
  \end{split}
 \end{equation}
 as well as (with $x,s\in\mathbb{R}$)
 \begin{equation}\label{eq:newregularisedquantities}
  \begin{split}
   \Theta^{(\theta)}_{\lambda,\sigma}(x)\;&:=\;(1+\sigma_0+\sigma)\theta(x)-\frac{4}{\pi\sqrt{3}}\int_{\mathbb{R}}\ud y\,\theta(y) \log \frac{\,2\cosh(x-y)+1\,}{\,2\cosh(x-y)-1\,} \\
   \widehat{\gamma}_\sigma(s)\;&:=\;\widehat{\gamma}(s)+\sigma_0+\sigma\,,
  \end{split}
 \end{equation}
 for given $\theta$,
 where $ \widehat{\gamma}$ is defined in \eqref{eq:gamma-distribution}. By construction $\widehat{\gamma}_\sigma$ is the shifted version of the $\widehat{\gamma}$-curve from Fig.~\ref{fig:gammagammaplus}, thus a smooth even function on $\mathbb{R}$ that is uniformly bounded and strictly positive, with absolute minimum $\widehat{\gamma}_\sigma(0)=\sigma$.

 Analogously to \eqref{eq:Ksigmarotations}, $K_{\lambda,\sigma}=\bigoplus_{\ell\in\mathbb{N}}K_{\lambda,\sigma}^{(\ell)}$ with respect to the usual decomposition in sectors of definite angular momentum, and arguing as in Subsect.~\ref{sec:variants} and \ref{sec:MFregularisation} let us only implement the regularisation in the meaningful sector $\ell=0$.

 A straightforward modification of the reasoning for Lemmas \ref{lem:xithetaidentities} and \ref{lem:regularisedlemma1} and for Corollary \ref{cor:equivalentH12norm} yields the following.

  \begin{lemma}\label{lem:regularisedlemma2}
   Let $\lambda,\sigma>0$, $s\in\mathbb{R}$, and $\xi$ be as in \eqref{eq:0xi}. One has the identities
         \begin{eqnarray}
         \big(\big(T_\lambda^{(0)}+K_{\lambda,\sigma}^{(0)}\big)\xi\big)\,{\textrm{\large $\widehat{\,}$\normalsize}}\,(\pp)\!\!&=&\!\!\frac{1}{\sqrt{4\pi}\,|\pp|}\,\frac{4\pi^2}{\sqrt{3}\,}\,\Theta^{(\theta)}_{\lambda,\sigma}(x)\,, \label{eq:Tlxi-theta-withK2}\\
      \big\|\big(T_\lambda^{(0)}+K_{\lambda,\sigma}^{(0)}\big)\xi\big\|_{H^{s}(\mathbb{R}^3)}^2\!\!&\approx&\!\!\int_{\mathbb{R}}\ud x\,(\cosh x)^{1+2s}\,\big|\Theta^{(\theta)}_{\lambda,\sigma}(x)\big|^2\,, \label{eq:Tlambda0xi-with-theta-andK2} \\
      \widehat{\Theta^{(\theta)}_{\lambda,\sigma}}(s)\!\!&=&\!\!\widehat{\gamma}_\sigma(s)\,\widehat{\theta}(s)\,, \label{eq:ThetaTransformed2} \\
       \int_{\mathbb{R}^3} \overline{\,\widehat{\xi}(\pp)}\, \big(\big(T_\lambda^{(0)}+K_{\lambda,\sigma}^{(0)}\big)\xi\big)\,{\textrm{\large $\widehat{\,}$\normalsize}}\,(\pp)\,\ud\pp\!\!&=&\!\!\frac{\,8\pi^2}{3\sqrt{3}}\int_{\mathbb{R}}\ud s\, \widehat{\gamma}_\sigma(s)\,|\widehat{\theta}(s)|^2\,, \label{eq:xiTxi-with-theta-andK2}
    \end{eqnarray}
    with $x$ and $\theta$ given by \eqref{ftheta-1}, and $\Theta^{(\theta)}_{\lambda,\sigma}$ and $\widehat{\gamma}_\sigma$ given by \eqref{eq:newregularisedquantities}, where the definition of $\Theta^{(\theta)}_{\lambda,\sigma}$ is now taken with respect to the present $\theta$.
    In \eqref{eq:Tlxi-theta-withK2} it is understood that $x\geqslant 0$, and \eqref{eq:Tlambda0xi-with-theta-andK2} is meant as an equivalence of norms (with $\lambda$-dependent multiplicative constant).
    Moreover, the map
   \begin{equation}
    \xi\;\longmapsto\;\bigg(\int_{\mathbb{R}^3} \overline{\,\widehat{\xi}(\pp)}\, \big(\big(T_\lambda^{(0)}+K_{\lambda,\sigma}^{(0)}\big)\xi\big)\,{\textrm{\large $\widehat{\,}$\normalsize}}\,(\pp)\,\ud\pp\bigg)^{\!\frac{1}{2}}
   \end{equation}
   defines an equivalent norm in $H^{\frac{1}{2}}_{\ell=0}(\mathbb{R}^3)$.
  \end{lemma}

  We can now define, for $\sigma>0$,
  \begin{equation}\label{eq:newD0regularised2}
   \begin{split}
    \widetilde{\mathsf{D}}_{0,\sigma}\;&:=\;\left\{
   \xi\in H^{-\frac{1}{2}}_{\ell=0}(\mathbb{R}^3)\left|
   \begin{array}{c}
    \textrm{$\xi$ has re-scaled radial component} \\
     \theta=\big(\widehat{\Theta}/\widehat{\gamma}_\sigma)\big)^{\!\vee} \\
     \textrm{for }\;\Theta\in C^\infty_{0,\mathrm{odd}}(\mathbb{R}_x)
   \end{array}
   \!\right.\right\}\,, \\
   \mathsf{D}_{0,\sigma}\;&:=\;\big\{\xi\in H^{\frac{1}{2}}_{\ell=0}(\mathbb{R}^3)\,\big|\,\big(T_\lambda^{(0)}+K_{\lambda,\sigma}^{(0)}\big)\xi\in H^{\frac{1}{2}}_{\ell=0}(\mathbb{R}^3)\big\}\,,
   \end{split}
  \end{equation}
  with the same remarks made for \eqref{eq:newD0regularised}, and prove the following, based on Lemma \ref{lem:regularisedlemma2}, and easily mimicking the reasoning that let to Lemma \ref{lem:DD0properties} and Proposition \ref{prop:NEW-Birman-param-selfadj}.

  \begin{proposition}\label{prop:againandagain} Let $\sigma>0$.  
  \begin{itemize}
    \item[(i)] $\widetilde{\mathsf{D}}_{0,\sigma}$ is dense in $H^{\frac{1}{2}}_{\ell=0}(\mathbb{R}^3)$.
   \item[(ii)] $(T_\lambda^{(0)}+K_{\lambda,\sigma}^{(0)})\widetilde{\mathsf{D}}_{0,\sigma}\subset H^{s}_{\ell=0}(\mathbb{R}^3)$ for every $s\in\mathbb{R}$ and $\lambda>0$.
    \item[(iii)] $\widetilde{\mathsf{D}}_{0,\sigma}\subset \mathsf{D}_{0,\sigma}$.
    \item[(iv)] For every $\lambda>0$ the operator
    \begin{equation}
   \begin{split}
    \mathcal{D}\big(\widetilde{\mathcal{R}}_{\lambda,\sigma}^{(0)}\big)\;&:=\;\widetilde{\mathsf{D}}_{0,\sigma} \\
    \widetilde{\mathcal{R}}^{(0)}_{\lambda,\sigma}\;&:=\; 3W_\lambda^{-1}\big(T_\lambda^{(0)}+K_{\lambda,\sigma}^{(0)}\big)\,.
   \end{split}
  \end{equation}
  is densely defined, symmetric, and coercive on $H^{-\frac{1}{2}}_{W_\lambda,\ell=0}(\mathbb{R}^3)$.
  \item[(v)] For every $\lambda>0$ the operator
    \begin{equation}\label{eq:RlambdasigmaOPERATOR2}
   \begin{split}
    \mathcal{D}\big(\mathcal{R}_{\lambda,\sigma}^{(0)}\big)\;&:=\;\mathsf{D}_{0,\sigma} \\
    \mathcal{R}_{\lambda,\sigma}^{(0)})\;&:=\;3W_\lambda^{-1}\big(T_\lambda^{(0)}+K_{\lambda,\sigma}^{(0)}\big)
   \end{split}
  \end{equation}
 is the Friedrichs extension of $\widetilde{\mathcal{R}}^{(0)}_{\lambda,\sigma}$ with respect to $H^{-\frac{1}{2}}_{W_\lambda,\ell=0}(\mathbb{R}^3)$. Its sesquilinear form is
 \begin{equation}\label{eq:RlambdasigmaFORM2}
   \begin{split}
    \mathcal{D}\big[\mathcal{R}_{\lambda,\sigma}^{(0)}\big]\;&=\;H^{\frac{1}{2}}_{\ell=0}(\mathbb{R}^3) \\
    \mathcal{R}_{\lambda,\sigma}^{(0)}[\eta,\xi]\;&=\;3\int_{\mathbb{R}^3} \overline{\,\widehat{\xi}(\pp)}\, \big(\big(T_\lambda^{(0)}+K_{\lambda,\sigma}^{(0)}\big)\xi\big)\,{\textrm{\large $\widehat{\,}$\normalsize}}\,(\pp)\,.
   \end{split}
\end{equation}
    \end{itemize}
  \end{proposition}

  Proposition \ref{prop:againandagain} finally shows that the operator
  \begin{equation}\label{eq:globalRlambda-2}
   \mathcal{R}_{\lambda,\sigma}\;:=\;\mathcal{R}_{\lambda,\sigma}^{(0)}\:\oplus\: \bigoplus_{\ell=1}^\infty\mathcal{A}_\lambda^{(\ell)}
 \end{equation}
 is an admissible Birman parameter labelling a self-adjoint extension $\mathring{H}_{\mathcal{R}_{\lambda,\sigma}}$ of $\mathring{H}$ according to the general classification and construction of Theorem \ref{thm:generalclassification}. With a more compact notation we can write
  \begin{equation}\label{eq:Tcompact1-2version}
  \begin{split}
   \mathcal{D}(\mathcal{R}_{\lambda,\sigma})\;&=\;\big\{\xi\in H^{\frac{1}{2}}(\mathbb{R}^3)\,\big|\,\mathbf{T}_{\lambda,\sigma}\xi\in H^{\frac{1}{2}}(\mathbb{R}^3)\big\} \\
   \mathcal{R}_{\lambda,\sigma}\;&=\;3\,W_\lambda^{-1} \mathbf{T}_{\lambda,\sigma}
  \end{split}
 \end{equation}
 and 
 \begin{equation}\label{eq:Tcompact2-2version}
  \big(\mathbf{T}_{\lambda,\sigma}-T_\lambda\big)\xi\;=\; K_{\lambda,\sigma}^{(0)}\xi^{(0)}\,,
 \end{equation}
 where
%
%
  \begin{equation}\label{eq:compactTTT-2version}
  (\mathbf{T}_{\lambda,\sigma}\xi)(\yy)\;:=\;\big(\big(T_\lambda^{(0)}+K_{\lambda,\sigma}^{(0)}\big)\xi^{(0)})(\yy)+\sum_{\ell=1}^\infty \big(T_\lambda^{(\ell)}\xi^{(\ell)}\big)(\yy)
 \end{equation}
 for $\xi=\sum_{\ell=0}^\infty\xi^{(\ell)}\in\bigoplus_{\ell=0}^\infty \,H^{-\frac{1}{2}}_{W_\lambda,\ell}(\mathbb{R}^3)\cong H^{-\frac{1}{2}}_{W_\lambda}(\mathbb{R}^3)$.

 All this leads to a new version of Theorem \ref{thm:regularised-models}, \emph{with exactly the same statement}, of course now referred to the present $\mathbf{T}_{\lambda,\sigma}$ defined in \eqref{eq:compactTTT-2version} (and not to its counterpart \eqref{eq:compactTTT} considered in the previous Subsection).

 We have thus identified \emph{two} distinct classes of regularised self-adjoint Hamiltonians of zero-range interaction for the bosonic trimer:
 \begin{itemize}
  \item the operator $\mathscr{H}_\sigma^{\mathrm{MF}}$, $\sigma>0$, namely the Hamiltonian with the \emph{Minlos-Faddeev regularisation}, obtained as $\mathscr{H}_\sigma$ from Theorem \ref{thm:regularised-models} with the modified $\mathbf{T}_{\lambda,\sigma}$ fixed in \eqref{eq:compactTTT};
  \item the operator $\mathscr{H}_\sigma^{\mathrm{he}}$, $\sigma>0$, namely the Hamiltonian with \emph{high energy regularisation}, obtained as $\mathscr{H}_\sigma$ from Theorem \ref{thm:regularised-models} with the modified $\mathbf{T}_{\lambda,\sigma}$ fixed in \eqref{eq:compactTTT-2version}.  
 \end{itemize}

 Both types of Hamiltonians are non-negative, with only essential spectrum given by $[0,+\infty)$, and both retain a physically meaningful short-scale structure. 
 With the Minlov-Faddeev regularisation, each $g\in\mathcal{D}(\mathscr{H}_\sigma^{\mathrm{MF}})$ and the corresponding regular part $\phi^\lambda$ of $g$ (for fixed $\lambda>0$) satisfy
 \begin{equation}\label{eq:allBPTMS-regularised-again1}
    \begin{split}
      \int_{\!\substack{ \\ \\ \pp_2\in\mathbb{R}^3 \\ |\pp_2|<R}}\widehat{g}(\pp_1,\pp_2)\,\ud\pp_2\;&=\;4\pi R\,\widehat{\xi}(\pp_1)+\frac{\sigma_0+\sigma}{2\pi^2}\,\int_{\mathbb{R}^3}\frac{\widehat{\xi^{(0)}}(\qq)}{\,|\pp_1-\qq|^2}\,\ud\qq+o(1)\,, \\
      \phi^\lambda(\yy_1,\mathbf{0})\;&=\;(2\pi)^{-\frac{3}{2}}\Big( (T_\lambda\xi)(\yy_1)+\frac{\,\sigma_0+\sigma\,}{|\yy_1|}\,\xi^{(0)}(\yy_1)\Big)\,,
    \end{split}
   \end{equation}
 where $\xi$ is the singular charge of $g$, $\xi^{(0)}$ is the spherically symmetric component of $\xi$, and $R\to +\infty$. With the high energy regularisation, each $g\in\mathcal{D}(\mathscr{H}_\sigma^{\mathrm{he}})$ and the corresponding regular part $\phi^\lambda$ of $g$ satisfy
 \begin{equation}\label{eq:allBPTMS-regularised-again2}
    \begin{split}
      \int_{\!\substack{ \\ \\ \pp_2\in\mathbb{R}^3 \\ |\pp_2|<R}}\widehat{g}(\pp_1,\pp_2)\,\ud\pp_2\;&=\;4\pi R\,\widehat{\xi}(\pp_1)+2\pi^2(\sigma_0+\sigma)\sqrt{\frac{3}{4}\pp_1^2+\lambda}\,\widehat{\xi^{(0)}}(\pp_1)+o(1)\,,\!\!\!\!\!\!\!\!\!\!\! \\
      \int_{\mathbb{R}^3}\widehat{\phi^\lambda}(\pp_1,\pp_2)\,\ud\pp_2\;&=\;(\widehat{T_{\lambda,\sigma}\xi})(\pp_1)+2\pi^2(\sigma_0+\sigma)\sqrt{\frac{3}{4}\pp_1^2+\lambda}\,\widehat{\xi^{(0)}}(\pp_1)\,.
    \end{split}
   \end{equation}

\def\cprime{$'$}

\end{document}